\documentclass[reqno,a4paper,11pt]{book}
\usepackage[english]{babel}
\usepackage{color}

\usepackage{tocloft}
\usepackage{enumitem}
\usepackage{graphicx}
\usepackage{fancyhdr}
\usepackage{titlesec}

\usepackage{pdfpages}

\usepackage{tikz-cd}

\usepackage[all]{xy} 
\input xy 
\xyoption{all}
\xyoption{2cell} 
\xyoption{v2}
\SelectTips{cm}{11} 

\usepackage[
	inner=4cm,
	outer=2cm,
	bottom=2cm,
	top=2cm,
	includefoot,
	heightrounded,          
]{geometry}

\usepackage{amsmath,amsthm,amssymb,accents}
\usepackage{mathrsfs, mathtools}
\usepackage{dsfont,bbold}

\allowdisplaybreaks

\usepackage[all]{xy}
\input xy 
\xyoption{all}
\xyoption{2cell} 
\xyoption{v2}
\SelectTips{cm}{11} 

\usepackage[hidelinks]{hyperref}

\usepackage[numbers,sort&compress]{natbib}

\usepackage[mathscr]{euscript}

\newcommand{\diag}{\mathrm{diag}}
\newcommand{\pr}{\mathrm{pr}}
\newcommand{\rmsi}{\mathrm{si}}

\newcommand{\ssub}{\mathrm{ssub}}
\newcommand{\open}{\mathrm{open}}
\newcommand{\coker}{\mathrm{coker}}
\newcommand{\ran}{\mathrm{ran}}
\newcommand{\rmflat}{\mathrm{flat}}
\newcommand{\rmuni}{\mathrm{uni}}
\newcommand{\obj}{\mathrm{obj}}
\newcommand{\URep}{\mathrm{URep}}
\newcommand{\trunc}{\mathrm{trunc}}
\newcommand{\Tor}{\mathrm{Tor}}
\newcommand{\sh}{\mathrm{sh}}
\newcommand{\pt}{\mathrm{pt}}

\newcommand{\sw}{\mathrm{sw}}
\newcommand{\rmpar}{\mathrm{par}}
\newcommand{\sep}{\mathrm{sep}}
\newcommand{\FP}{\mathrm{FP}}
\newcommand{\rmH}{\mathrm{H}}
\newcommand{\Tot}{\mathrm{Tot}}

\newcommand{\rmD}{\mathrm{D}}
\newcommand{\rmd}{\mathrm{d}}
\newcommand{\rms}{\mathrm{s}}
\newcommand{\rmN}{\mathrm{N}}

\newcommand{\rmr}{\mathrm{r}}
\newcommand{\rml}{\mathrm{l}}
\newcommand{\rmL}{\mathrm{L}}
\newcommand{\rmW}{\mathrm{W}}

\newcommand{\ev}{\mathrm{ev}}
\newcommand{\DD}{\mathrm{DD}}
\newcommand{\dR}{\mathrm{dR}}

\newcommand{\dd}{\mathrm{d}}
\newcommand{\End}{\mathrm{End}}
\newcommand{\Plot}{\mathrm{Plot}}
\newcommand{\Hom}{\mathrm{Hom}}
\newcommand{\Diff}{\mathrm{Diff}}
\newcommand{\opp}{\mathrm{op}}
\newcommand{\rank}{\mathrm{rk}}
\newcommand{\cl}{\mathrm{cl}}

\newcommand{\curv}{\mathrm{curv}}
\newcommand{\disc}{\mathrm{disc}}
\newcommand{\UF}{\mathrm{UF}}

\newcommand{\scH}{\mathscr{H}}

\newcommand{\scA}{\mathscr{A}}
\newcommand{\scC}{\mathscr{C}}
\newcommand{\scR}{\mathscr{R}}
\newcommand{\scB}{\mathscr{B}}
\newcommand{\scD}{\mathscr{D}}
\newcommand{\scE}{\mathscr{E}}

\newcommand{\CG}{\mathcal{G}}
\newcommand{\CI}{\mathcal{I}}
\newcommand{\CH}{\mathcal{H}}
\newcommand{\CT}{\mathcal{T}}
\newcommand{\CL}{\mathcal{L}}

\newcommand{\CP}{\mathcal{P}}
\newcommand{\CC}{\mathcal{C}}
\newcommand{\CB}{\mathcal{B}}
\newcommand{\CU}{\mathcal{U}}

\newcommand{\CF}{\mathcal{F}}

\newcommand{\CV}{\mathcal{V}}
\newcommand{\CW}{\mathcal{W}}

\newcommand{\CO}{\mathcal{O}}
\newcommand{\CA}{\mathcal{A}}

\newcommand{\sfG}{\mathsf{G}}

\newcommand{\sfor}{\mathsf{or}}

\newcommand{\sfR}{\mathsf{R}}
\newcommand{\sfU}{\mathsf{U}}
\newcommand{\sfS}{\mathsf{S}}
\newcommand{\sfF}{\mathsf{F}}

\newcommand{\sfJ}{\mathsf{J}}

\newcommand{\sfc}{\mathsf{c}}
\newcommand{\sfH}{\mathsf{H}}
\newcommand{\sft}{\mathsf{t}}
\newcommand{\sfd}{\mathsf{d}}
\newcommand{\hol}{\mathsf{hol}}

\newcommand{\inv}{\mathsf{inv}}
\renewcommand{\hom}{\mathsf{hom}}

\newcommand{\Mat}{\mathsf{Mat}}
\newcommand{\Asc}{\mathsf{Asc}}
\newcommand{\Desc}{\mathsf{Desc}}
\newcommand{\rev}{\mathsf{rev}}

\newcommand{\FC}{\mathds{C}}
\newcommand{\NN}{\mathds{N}}
\newcommand{\FR}{\mathds{R}}
\newcommand{\RZ}{\mathds{Z}}

\newcommand{\PP}{\mathds{P}}
\newcommand{\One}{\mathds{1}}
\newcommand{\Null}{\mathbb{0}}

\newcommand{\bbL}{\mathbb{L}}

\newcommand{\frg}{\mathfrak{g}}
\newcommand{\fru}{\mathfrak{u}}
\newcommand{\frUF}{\mathfrak{UF}}

\newcommand{\HLBdl}{\mathscr{HL}\hspace{-0.02cm}\mathscr{B}\mathrm{dl}}
\newcommand{\BGrb}{\mathscr{B}\hspace{-.025cm}\mathscr{G}\mathrm{rb}}
\newcommand{\Vect}{\mathscr{V}\mathrm{ect}}
\newcommand{\Set}{\mathscr{S}\mathrm{et}}
\newcommand{\scCat}{\mathscr{C}\mathrm{at}}

\newcommand{\Hilb}{\mathscr{H}\mathrm{ilb}}
\newcommand{\HVBdl}{\mathscr{HVB}\mathrm{dl}}

\newcommand{\Ab}{\mathscr{A}\mathrm{b}}

\newcommand{\Mfd}{\mathscr{M}\mathrm{fd}}
\newcommand{\DfgSp}{\mathscr{D}\mathrm{fg}\mathscr{S}\mathrm{p}}

\newcommand{\arisom}{\overset{\cong}{\longrightarrow}}

\newcommand{\iu}{\mathrm{i}}
\newcommand{\tr}{\mathrm{tr}}

\newcommand{\<}{\langle}
\renewcommand{\>}{\rangle}

\newcommand{\textint}{{\textstyle{\int} \hspace{-0.1cm}}}

\newcommand{\qen}{\hfill$\triangleleft$}

\newenvironment{myitemize}{\begin{itemize}[itemsep=-0.1cm, leftmargin=*, topsep=0cm]}{\end{itemize}}
\newenvironment{myenumerate}{\begin{enumerate}[itemsep=-0.1cm, leftmargin=*, topsep=0cm, label=(\arabic*)]}{\end{enumerate}}

\newcommand{\theeq}{\tag{\theequation}}

\newcommand{\ddownarrow}{\rotatebox[origin=c]{90}{$\twoheadleftarrow$}}

\newtheoremstyle{thm} 								
{0.25cm}   					
{0.2cm}   	 				
{\itshape} 	
{}         				
{\bfseries}				
{}        				
{0.2cm} 					
{\thmname{#1}~\thmnumber{#2}\thmnote{ (#3)}}%

\newtheoremstyle{rmk} 								
{0.25cm}   					
{0.2cm}   	 				
{} 	
{}         				
{\bfseries}				
{}        				
{0.2cm} 					
{}         				

\theoremstyle{thm}
\newtheorem{theorem}[equation]{Theorem}

\newtheorem{corollary}[equation]{Corollary}

\newtheorem{lemma}[equation]{Lemma}
\newtheorem{proposition}[equation]{Proposition}
\newtheorem{definition}[equation]{Definition}

\theoremstyle{rmk}
\newtheorem{example}[equation]{Example}
\newtheorem{remark}[equation]{Remark}

\numberwithin{equation}{chapter}

\titleformat{\chapter}[display]
    {\normalfont\LARGE\bfseries}{\chaptertitlename\ \thechapter}{0.25cm}{\huge}
\titlespacing*{\chapter}{0cm}{2cm}{1.5cm}

\titleformat{\section}[hang]
    {\normalfont\Large\bfseries}{\thesection}{0.5cm}{\Large}

\titleformat{\subsection}[hang]
    {\normalfont\large\bfseries}{\thesubsection}{5pt}{\large}

\begin{document}

\pagestyle{empty}
\hfill
\vspace{3cm}

\begin{center}
	\textbf{\huge{Categorical Structures on Bundle Gerbes\\[0.4cm]and Higher Geometric Prequantisation}}\\
	\vspace{1cm}
	\Large Severin Bunk\\
	\vspace{2cm}
	\large Submitted for the degree of\\
	Doctor of Philosophy\\
	\vfill
	\includegraphics[scale=0.2,origin=c]{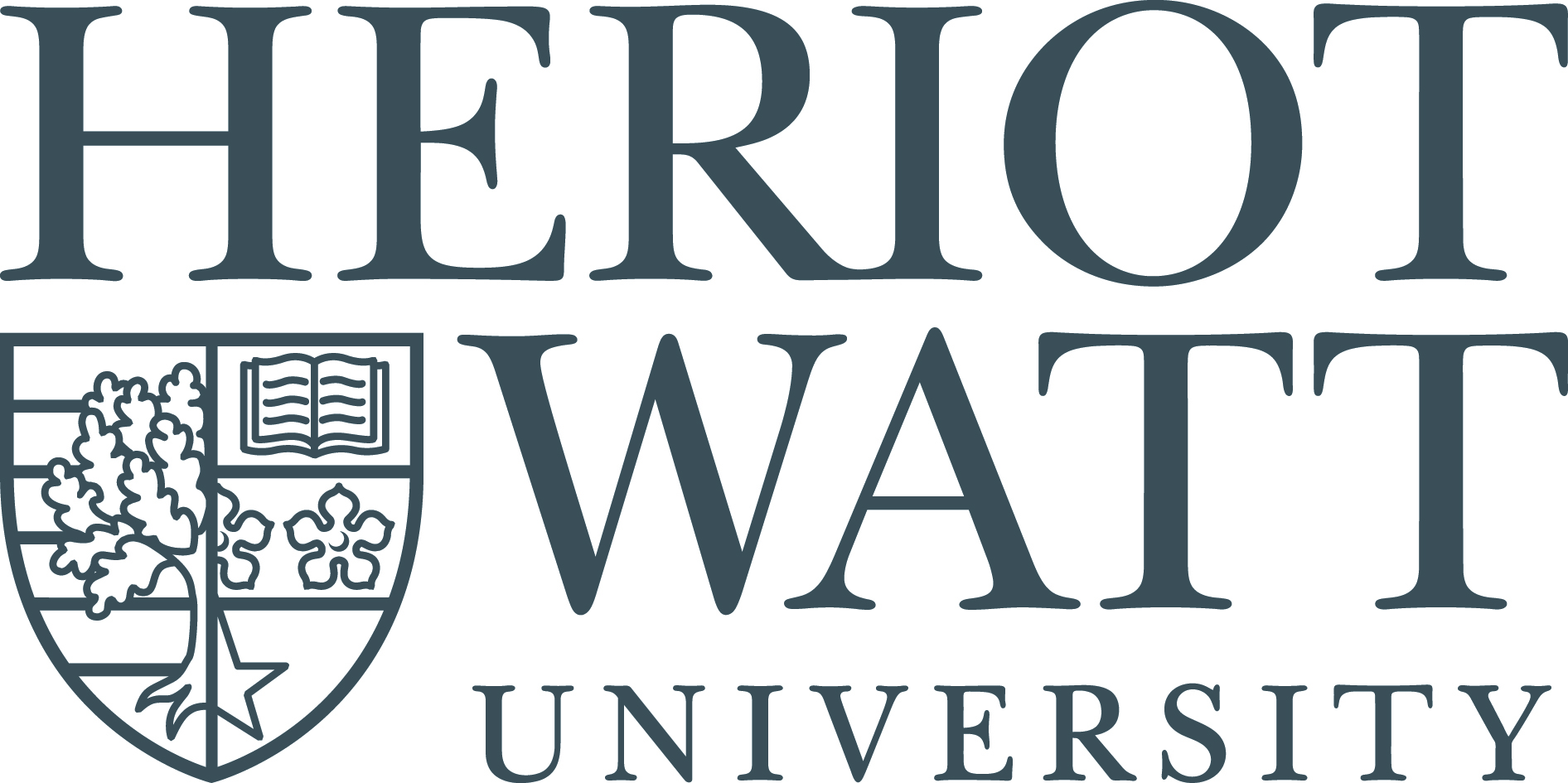}\\
	\vspace{1.cm}
	Department of Mathematics\\
	School of Mathematics and Computer Sciences\\
	\vspace{1cm}
	September 2017\\
	\vspace{1cm}
	\begin{footnotesize}
		\noindent
		The copyright in this thesis is owned by the author.
		Any quotation from the thesis or use of any of the information contained in it must acknowledge this thesis as the source of the quotation or information.
	\end{footnotesize}
\end{center}

\frontmatter

\cleardoublepage
\pagestyle{plain}
\section*{Abstract}

We present a construction of a 2-Hilbert space of sections of a bundle gerbe, a suitable candidate for a prequantum 2-Hilbert space in higher geometric quantisation.
We start by briefly recalling the construction of the 2-category of bundle gerbes, with minor alterations that allow us to endow morphisms with additive structures.
The morphisms in the resulting 2-categories are investigated in detail.
We introduce a direct sum on morphism categories of bundle gerbes and show that these categories are cartesian monoidal and abelian.
Endomorphisms of the trivial bundle gerbe, or higher functions, carry the structure of a rig-category, a categorified ring, and we show that generic morphism categories of bundle gerbes form module categories over this rig-category.

We continue by presenting a categorification of the hermitean bundle metric on a hermitean line bundle.
This is achieved by introducing a functorial dual that extends the dual of vector bundles to morphisms of bundle gerbes, and constructing a two-variable adjunction for the aforementioned rig-module category structure on morphism categories.
Its right internal hom is the module action, composed by taking the dual of the acting higher functions, while the left internal hom is interpreted as a bundle gerbe metric.

Sections of bundle gerbes are defined as morphisms from the trivial bundle gerbe to the bundle gerbe under consideration.
We show that the resulting categories of sections carry a rig-module structure over the category of finite-dimensional Hilbert spaces with its canonical direct sum and tensor product.
A suitable definition of 2-Hilbert spaces is given, modifying previous definitions by the use of two-variable adjunctions.
We prove that the category of sections of a bundle gerbe, with its additive and module structures, fits into this framework, thus obtaining a 2-Hilbert space of sections.
In particular, this can be constructed for prequantum bundle gerbes in problems of higher geometric quantisation.

We define a dimensional reduction functor and show that the categorical structures introduced on the 2-category of bundle gerbes naturally reduce to their counterparts on hermitean line bundles with connections.
In several places in this thesis, we provide examples, making 2-Hilbert spaces of sections and dimensional reduction very explicit.

\newpage
\hfill
\newpage

\thispagestyle{empty}
\hfil

\vspace{12cm}
\hspace{8cm}
\textit{To my parents.}

\newpage
\hfill
\newpage
\section*{Acknowledgements}

First and foremost, I would like to thank Richard Szabo for his supervision during the course of my PhD programme.
I am grateful for all the time, support and advice he has given me over the last three years, for always taking my ideas and concerns seriously, for our often long discussions, and for always being frank with me.

I would like to thank Christian S\"amann for being my second supervisor, but even more for his input into this project and being approachable, caring and supportive throughout my PhD.

I am grateful to Alexander Schenkel for countless discussions, his interest in my progress and career, and for giving me the opportunity to speak in Nottingham and at his MFO Mini-Workshop.

During the second half of my PhD I have had the pleasure of collaborating with Konrad Waldorf.
I would like to thank him in this place for inviting me to Greifswald, for the insightful discussions we have had, and for his support outside of our collaboration.

Moreover, I would like to thank Michael Murray and Danny Stevenson for several interesting discussions on bundle gerbes and their modules, Branislav Jur\v{c}o and Jos\'{e} Figueroa-O'Farrill for being my examiners, and Des Johnston for agreeing to watch over my viva.

I would like to thank David Jordan for running his TQFT seminar, Gwendolyn Barnes for running a category theory reading group, and Tim Weelinck, Matt Booth, Lukas M\"uller and Jenny August for their enthusiasm and their contributions to the homotopy theory and higher categories reading group.

Finally, I gratefully acknowledge support granted by Heriot-Watt University through a James-Watt Scholarship.

\newpage
\hfill

\includepdf[width=\textwidth, height=\textheight, keepaspectratio]{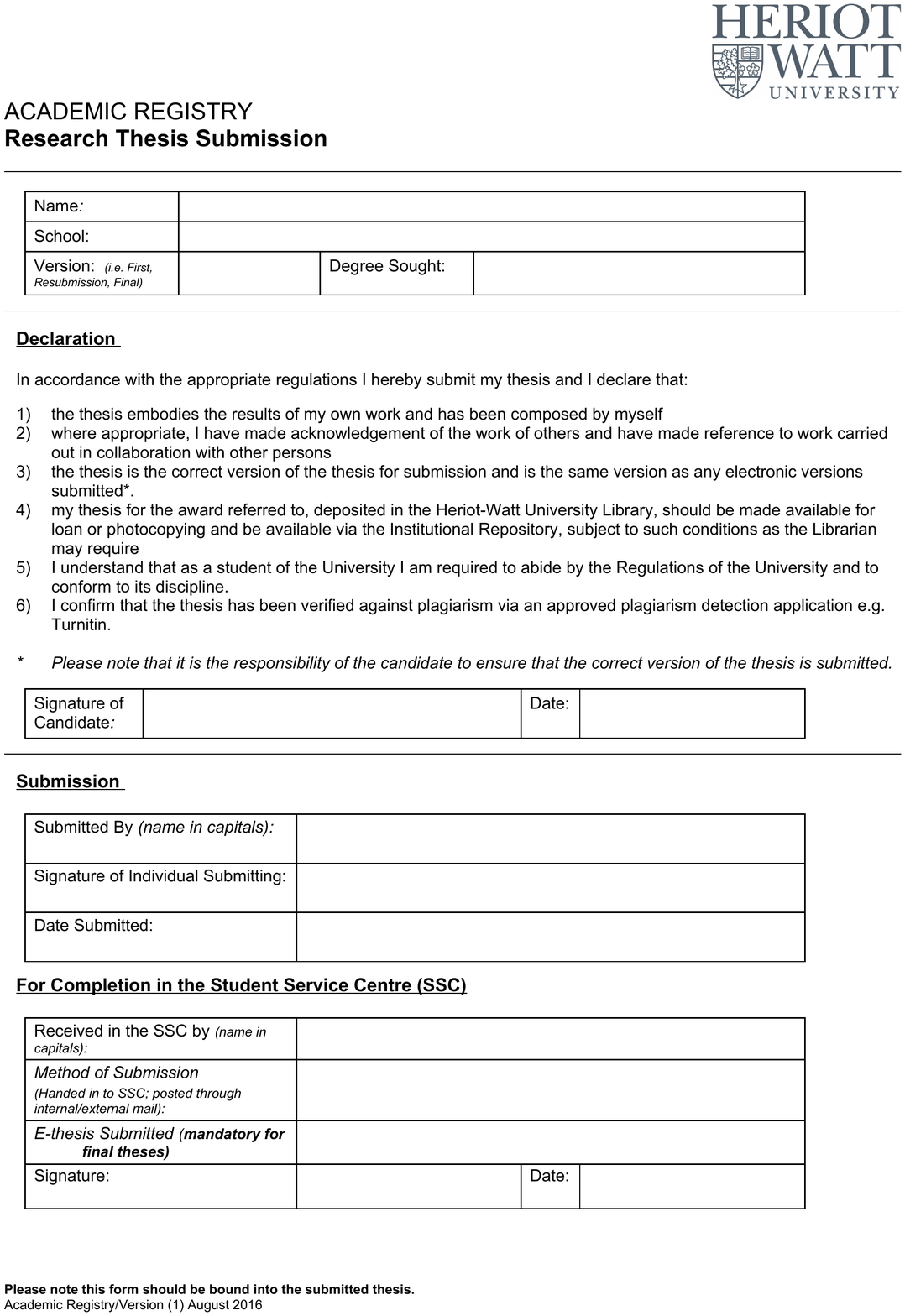}

\cleardoublepage

\pdfbookmark{\contentsname}{toc}
\tableofcontents

\bigskip


\section*{Notation}
\label{sect:notation}

\begin{myitemize}
	\item We use the term `2-category' to refer to what is sometimes called a weak 2-category, or bicategory, i.e. we allow for non-trivial associators and unitors.
	If we explicitly refer to 2-categories where these are trivial, we will use the term `strict 2-category'.
	
	\item For a category $\scC$ and objects $a,b \in \scC$, we denote the collection of morphisms from $a$ to $b$ in $\scC$ by $\scC(a,b)$.
	Similarly, if $\scC$ is a 2-category and $a,b \in \scC$ are objects, the category of morphisms from $a$ to $b$ in $\scC$ is denoted $\scC(a,b)$.
	For two 1-morphisms $\phi, \psi \colon a \to b$ in $\scC$, the collection of 2-morphisms from $\phi$ to $\psi$ is $\scC(\phi,\psi)$.
	
	\item In a 2-category $\scC$, we denote the horizontal composition of 2-morphisms by $\circ_1$, while the vertical composition is denoted $\circ_2$.
	
	\item For an $n$-category $\scC$, we denote its underlying $n$-groupoid by $\scC_\sim$, i.e. $\scC_\sim$ is the $n$-category obtained from $\scC$ by discarding all non-invertible $k$-morphisms for $k = 1, \ldots, n$, while $\pi_0 \scC$ denotes the collection of objects of $\scC$ modulo the equivalence relation $a \sim b$ if there exists a zig-zag of morphisms between $a$ and $b$.
	
	\item $\Mfd$ is the category of smooth manifolds and smooth maps.
	We we define manifolds to be second countable and Hausdorff.
	
	\item $\DfgSp$ is the category of diffeological spaces and diffeological maps.
	
	\item $\HVBdl$ is the relative sheaf of categories (see Section~\ref{sect:Coverings_and_sheaves_of_cats}) of smooth hermitean vector bundles, with smooth fibrewise linear maps as morphisms, with respect to the Grothendieck topology given by smooth surjective submersions on $\Mfd$ (cf.~\cite{NS--Equivariance_in_higher_geometry}).
	Here, we include non-invertible morphisms of vector bundles.
	$\HLBdl$ is the full sub-sheaf of categories of hermitean line bundles on $\Mfd$.
	
	\item For a morphism $\psi \in \HVBdl(M)(E,F)$, we write $\psi^* \colon F \to E$ for its adjoint, defined by $h_E(\psi^*(f),e) = h_F(f, \psi(e))$ via the hermitean metrics $h_E$ and $h_F$ on $E$ and $F$, respectively, with $e \in E$ and $f \in F$.
	The transpose of $\psi$ is denoted $\psi^\sft \colon F^* \to E^*$ and given by $\psi^\sft(\phi)(e) = \phi(\psi(e))$ for $e \in E$ and $\phi \in F^*$.
	
	\item $\HVBdl^\nabla$ is the sheaf of categories of smooth hermitean vector bundles on $\Mfd$ with hermitean connection and smooth fibrewise linear morphisms.
	Usually, by a `hermitean vector bundle with connection' on a manifold $M$ we shall mean an object of $\HVBdl^\nabla(M)$.
	The category $\HVBdl^\nabla_\rmpar(M)$ is the subcategory of $\HVBdl^\nabla(M)$ of smooth hermitean vector bundles on $M$ with hermitean connections and parallel morphisms, while $\HVBdl_\rmuni^\nabla(M)$ is the sub-groupoid of hermitean vector bundles and unitary, parallel isomorphisms.
	
	\item $\HLBdl^\nabla(M)$ is the full subcategory of $\HVBdl^\nabla(M)$ of hermitean line bundles with connection, and $\HLBdl^\nabla_\rmpar(M)$ and $\HLBdl^\nabla_\rmuni(M)$ are the corresponding full subcategories of $\HVBdl^\nabla_\rmpar(M)$ and $\HVBdl^\nabla_\rmuni(M)$, respectively.
	
	\item $\BGrb^\nabla(M)$ denotes the 2-category of hermitean line bundle gerbes with connection on $M \in \Mfd$, with morphisms constructed from arbitrary common refinements of coverings and $\HVBdl^\nabla$, i.e. without a restriction on the trace of the curvature of the vector bundles underlying 1-morphisms.
	$\BGrb^\nabla_\rmpar(M)$ shall refer to the sub-2-category of $\BGrb^\nabla(M)$ whose 2-morphisms are constructed from $\HVBdl^\nabla_\rmpar$.
	A yet smaller 2-category is $\BGrb^\nabla_\rmflat(M)$, which has 1-morphisms satisfying the trace condition~\eqref{eq:trace_condition} and 2-morphisms built from unitary parallel isomorphisms of hermitean vector bundles.
	
	\item For conventions regarding monoidal structures on vector bundles and surjective submersions see Appendix~\ref{app:monoidal_structures_and_strictness}.
\end{myitemize}

\mainmatter

\pagestyle{fancy}
 \fancyhf{}
 \fancyhead[CE]{\nouppercase{\textit{\leftmark}}}
 \fancyhead[CO]{\nouppercase{\textit{\rightmark}}}
 \fancyfoot[CE,CO]{\thepage}
 \renewcommand{\headrulewidth}{0cm}

\setcounter{page}{1}
\pagenumbering{arabic}

\chapter{Introduction}
\label{ch:Introduction}

\section{Quantisation and categorification}
\label{sect:quantisation and categorification}

The two guiding principles underlying this thesis are \emph{quantisation} and \emph{categorification}.
They have a remarkable conceptual similarity: both refer to problems where a certain simple form of a set-up is known and understood, but one desires to construct a different, richer set-up solely based on the simple framework and some expectation of what the outcome should look like.

In quantisation -- or, more specifically, first quantisation -- the starting point is a classical, mechanical system (though other situations can be considered as well~\cite{APW--Geometric_quatisation_of_CS_gauge_theory}).
This could consist, for instance, of a finite number of point particles moving in a fixed space, subject to constraints, interactions and external forces.
Observables are real-valued functions on the phase space, and time-evolution is governed by the Hamiltonian function and a Poisson structure on the algebra of observables.
(For a detailed mathematical treatment of classical mechanics, see, for instance, \cite{Abraham-Marsden--Foundations_of_mechanics,Scheck--Mechanics}.)
Quantisation refers to a hypothetical process that associated to any such classical system a quantum system, i.e. that turns classical states into elements of a Hilbert space $\CH$ of quantum states, represents the Poisson algebra of classical observables as self-adjoint operators on $\CH$ in first order in $\hbar$, and, thereby, turns the Hamiltonian function into a Hamiltonian operator.

It is, in general, very hard to find an appropriate Hilbert space and to represent classical observables on it in an appropriate manner.
The most fundamental consistency check is the so-called \emph{classical limit} $\hbar \to 0$.
In this limit, the quantum system degenerates, and should yield back the original, classical system.
Thus, there exists a systematic way of going from quantum to classical systems.
This is not surprising from the point of view that quantum mechanics is the more fundamental and complete model of the real world.
Quantisation, hence, tries to obtain a more complete description of a physical problem from the less complete classical one:
\begin{equation}
\begin{tikzcd}[column sep=2cm, row sep=1.25cm]
	\text{Quantum systems} \ar[d, shift left=0.1cm, "\text{Classical limit}"] & \text{Higher geometric structures} \ar[d, shift left=0.1cm, "\text{Transgression / Dimensional reduction}"]
	\\
	\text{Classical systems} \ar[u, shift left=0.1cm, "\text{``Quantisation''}"] & \text{Geometric structures} \ar[u, shift left=0.1cm, "\text{``Categorification''}"]
\end{tikzcd}
\end{equation}

Categorification is a term with a broad meaning that can apply to many concepts in mathematics.
In geometry, it means finding analogues of known geometric structures guided by the principle that one should replace sets by categories.
In fact, as there is a whole hierarchy of higher categories, one can iterate categorification and replace $n$-categories by $(n{+}1)$-categories.
Like in quantisation, there is no general scheme which tells us how to write down the next higher version, in this hierarchy, of a geometric object.
Similarly to the situation in quantisation, often all that is known is the object one starts with and some intuition about the structure one would like to obtain.
However, dimensional reduction along a circle in the base manifold should, at least in many cases, allow to obtain lower structures from higher ones.

While neither quantisation, nor categorification have been made precise in terms of a functor -- for quantisation this presumably cannot possibly be achieved -- in certain situations both admit canonical constructions that apply to a whole class of problems.
For quantisation, several such frameworks are known.
The most prominent ones are geometric quantisation~\cite{Woodhouse--Geometric_quanitsation,APW--Geometric_quatisation_of_CS_gauge_theory,Brylinski--Loop_spaces_and_geometric_quantisation} and deformation quantisation~\cite{Fedosov--Index_Thm_for_DefQuan,EW--DifGeo_of_DefQuan,Kontsevich--DefQuan_of_Poisson_Mfds}.
Regarding categorification, several procedures are known for different problems.
Higher principal bundles are usually tackled in a sheaf theoretic way and can be categorified using classifying space constructions~\cite{FSS--Higher_stacky_perspective,FRS--Higher_U1-gerbe_connections_in_geometric_prequant,FSS--Cech_cocycles_and_diff_char_classes}.
Higher versions of modules, or vector spaces are either treated via the framework of internalisation~\cite{Baez-Crans--HdimA-Lie_2-algebras} or (higher) module categories~\cite{KV--2-Cats_and_Zam_eqns}.

In this thesis, we will, in some sense, aim to \emph{categorify geometric quantisation}.
The motivation for doing so is that there are several physical as well as mathematical situations where a method of quantisation is desirable, but which lie outside of the scope of ordinary geometric quantisation.
In particular, we are interested in quantising classical systems which are described mathematically in terms of a 2-plectic geometry.
An $n$-plectic form~\cite{Rogers--Thesis} on a manifold $M$ is an $(n{+}1)$-form $\varpi$ on $M$ such that $\dd \varpi = 0$ and, for any tangent vector $X \in TM$, $\iota_X \varpi = \varpi(X,-) = 0$ if and only if $X = 0$.
Geometries which carry such a form are, for instance, compact, simple, simply-connected Lie groups $\sfG$ with their fundamental left-invariant 3-form, any orientable three-dimensional manifold with a choice of volume form, or products thereof.
It has been conjectured that, if found, a higher geometric quantisation of $\sfG$ should be related to the representations of the central extensions of the loop group of $\sfG$~\cite{Brylinski--Loop_spaces_and_geometric_quantisation,PS--Loop_groups}.

From the physical point of view, a higher analogue of a point particle in the presence of a symplectic form can arise from a closed string moving through a background with an NS-NS $H$-flux, which is the 3-form field strength of the locally defined Kalb-Ramond $B$-field~\cite{Kapustin--D-branes_in_top_nontriv_B-field}.
Several approaches have been made to quantising this system.
As the configuration space of a closed string in a target manifold $M$ is the loop space of $M$, one approach is to extend the formalism of geometric quantisation to loop spaces~\cite{Saemann-Szabo--Groupoid_quant_of_loop_spaces,Saemann-Szabo--Groupoids_loop_spaces_Geoquan,Saemann-Szabo--Quantisation_of_2-plectic_mfds}.
From a different point of view, one can argue that, as string theory is a theory of quantum gravity, quantisation should take place on the original target space $M$ rather than its loop space.
In order to apply geometric quantisation on $M$, however, one has to account for the higher geometric nature of the $B$-field: it is part of a connection on a higher version of a line bundle -- a \emph{bundle gerbe}.
Interestingly, the aforementioned central extensions of loop groups of Lie groups are instances of the same geometric objects (cf. Section~\ref{sect:Ex:Tautological_BGrbs}).
Higher geometric quantisation of the string in an $H$-flux background is expected to detect the non-associativity of spacetime which is induced by an $H$-flux in string theory~\cite{Lust--TD_and_closed_string_NCCDG,MSS--Membrane_sigma-models,MSS--NAG_and_twist_DefQuan,BSS--Working_with_NAG,Aschieri-Szabo--Triproducts,BDLPR--Non-geo_fluxes_and_NAG,Blumenhagen-Plauschinn--NAG_in_string_theory,CFL--Asymmetric_orbifolds}.
Bundle gerbes, moreover, can be used to describe $C$-fields and anomalies in M-theory, see, for instance~\cite{Bunk-Szabo--Fluxes_brbs_2Hspaces,Aschieri-Jurco--Gerbes_M5-brane_anomalies_and_E_8_gauge_theory} and references therein.

\section{A glance at bundle gerbes and higher quantisation}
\label{sect:A glance at BGrbs}

As pointed out in Section~\ref{sect:quantisation and categorification}, the geometric structure that models the $B$-field in string theory is a bundle gerbe.
There exist several different, but related, models for the geometry we wish to describe~\cite{Nikolaus--Thesis,Nikoalus-Waldorf--Four_equiv_versions_of_non-ab-gerbes,FSS--Cech_cocycles_and_diff_char_classes,FSS--Higher_stacky_perspective,Brylinski--Loop_spaces_and_geometric_quantisation,NSS--Principal_infty_bundles_I,NSS--Principal_infty_bundles_II,Karoubi--Twisted_bundles_and_twisted_K-theory,BCMMS,Atiyah-Segal--Twisted_K-theory}.
As we demonstrate in this thesis, bundle gerbes allow for a theory of non-invertible morphisms that conceptually parallels that of morphisms of line bundles.
The other approaches might be able to incorporate the necessary non-invertible morphisms, but this has, to our knowledge, not yet been worked out.
Those models are usually studied on the level of principal bundles only, rather than from the perspective of their associated higher vector bundles.

What is the structure that we need to geometrically quantise a 2-plectic form?
In ordinary geometric quantisation, the fundamental geometric object is a hermitean line bundle with connection $(L, \nabla^L)$ on the phase space $M$.
Its curvature 2-form $\curv(\nabla^L)$ is required to satisfy $\curv(\nabla^L) = 2\pi\, \iu\, \omega$, where $\omega$ is the symplectic form on $M$.
Given a good open covering $(U_a)_{a \in \Lambda}$ of $M$, i.e. an open covering such that all possible finite intersections of patches $U_a$ are diffeomorphic to $\FR^n$, the data $(L, \nabla^L)$ can, equivalently, be described by \v{C}ech data $(g_{ab}, A_a)_{a \in \Lambda}$, where $g_{ab} \colon U_{ab} \to \sfU(1)$ are smooth functions and $A_a \in \Omega^1(U_a, \iu\, \FR)$.
Note that we have set $U_{a_0 \ldots a_n} = \bigcap_{i = 0}^n\, U_{a_i}$ for $a_0, \ldots a_n \in \Lambda$.
For all $a,b,c \in \Lambda$, these data satisfy the relations
\begin{align}
\label{eq:cocycle for line bdls}
	g_{ac|U_{abc}} &= g_{ab|U_{abc}}\, g_{bc|U_{abc}}\,,
	\\
	\label{eq:connection and transition functions}
	g_{ab}^{-1}\, \dd g_{ab} &= A_{b|U_{ab}} - A_{a|U_{ab}}\,,
	\\
	\label{eq:curvature for line bundles}
	\dd A_a &= \curv(\nabla^L)_{|U_a} = 2\pi\, \iu\, \omega_{|U_a}\,.
\end{align}
We can, thus, describe a hermitean line bundle with connection in terms of $\sfU(1)$-valued functions and $1$-forms.
For a closed 1-form $\eta \in \Omega^1(M, \iu\, \FR)$ with integer periods, we may demand that it be represented by a smooth function $g \colon M \to \sfU(1)$ via
\begin{equation}
\label{eq:curv of a function}
	\curv(g) \coloneqq g^{-1}\, \dd g
	= 2\pi\, \iu\, \eta\,.
\end{equation}
In this sense, $\sfU(1)$-valued functions represent closed 1-forms with integer periods in the same way as hermitean line bundles with connection represent closed 2-forms with integer periods.
We, therefore, take the view that hermitean line bundles with connection provide a categorification of $\sfU(1)$-valued functions.

The hermitean line bundle $(L, \nabla^L)$ with connection on $M$ was obtained from \v{C}ech data that featured $\sfU(1)$-valued transition functions.
That is, we have encoded a higher structure using local cocycle data made of a lower structure.
If we iterate this step, we are thus led to consider ``transition functions'' $(L_{ab}, \nabla^{L_{ab}})_{a, b \in \Lambda}$, which consist of hermitean line bundles with connection on $U_{ab}$.
The cocycle relation~\eqref{eq:cocycle for line bdls} does not make sense any longer in a strict way; instead, we now have to impose the existence of an isomorphism that establishes it.
We require that there be a unitary, parallel isomorphism
\begin{equation}
	\mu_{abc} \colon L_{ab} \otimes L_{bc} \arisom L_{ac}\,,
\end{equation}
which is associative over quadruple overlaps.
In addition to the transition functions $g_{ab}$ on $U_{ab}$, we had to use 1-forms $A_a$ on $U_a$, which were related to the ``curvature'' of the transition functions via $A_b - A_a = \curv(g_{ab})$ on $U_{ab}$ (in the notation of~\eqref{eq:curv of a function}).
The curvature of the higher transition functions $(L_{ab}, \nabla^{L_{ab}})$ is a 2-form, whence we require the existence of 2-forms $B_a \in \Omega^2(U_a, \iu\, \FR)$ such that the analogue of~\eqref{eq:connection and transition functions} is satisfied:
\begin{equation}
	\curv(\nabla^{L_{ab}}) = B_{b|U_{ab}} - B_{a|U_{ab}} \quad \forall\, a,b \in \Lambda\,.
\end{equation}
Let us denote the data gathered so far by
\begin{equation}
\label{eq:BGrb -- first example}
	\big( \CG, \nabla^\CG \big) \coloneqq \big( L_{ab}, \nabla^{L_{ab}}, \mu_{abc}, B_a)_{a,b,c \in \Lambda}\,.
\end{equation}
By the Bianchi identity, it follows that we obtain a unique closed, globally defined 3-form by setting
\begin{equation}
	\curv(\nabla^\CG) \in \Omega^3(M, \iu\, \FR)\,, \quad
	\curv(\nabla^\CG)_{|U_a} \coloneqq \dd B_a\,.
\end{equation}
As it turns out, $\curv(\nabla^\CG)$ has periods in $2\pi\, \iu\, \RZ$ (see, for instance, \cite{Waldorf--Thesis,Brylinski--Loop_spaces_and_geometric_quantisation} and Section~\ref{sect:Deligne_coho_and_higher_cats}).
Consequently, the data $(\CG, \nabla^\CG)$ can be interpreted as a geometric structure that allows us to realise a 2-plectic form $\varpi$ on $M$ in the desired way: it now makes sense to ask for
\begin{equation}
\label{eq:prequantum BGrb condition}
	\curv(\nabla^\CG) = 2\pi\, \iu\, \varpi\,.
\end{equation}

Data $(\CG, \nabla^\CG)$ as in~\eqref{eq:BGrb -- first example} is a special instance of a \emph{hermitean (line) bundle gerbe with connection}.
The 3-form $\curv(\nabla^\CG)$ is called its \emph{curvature 3-form}.
If the identity~\eqref{eq:prequantum BGrb condition} is satisfied, we call $(\CG, \nabla^\CG)$ a \emph{prequantum bundle gerbe for $\varpi$}.
Bundle gerbes with connection are central to this work; we define them in full generality in Section~\ref{sect:The_2-category_of_BGrbs}.
In string theory, $(B_a)_{a \in \Lambda}$ is the Kalb-Ramond $B$-field, while $\curv(\nabla^\CG)$ is, accordingly, the 3-form $H$-flux.
The replacement
\begin{equation}
	g_{ab} \mapsto \big( L_{ab}, \nabla^{L_{ab}} \big)
\end{equation}
contains the most important heuristic principle in this thesis: the next higher version of a number is a vector space, or a Hilbert space.

Having at hand a notion of higher line bundle (i.e. bundle gerbes), we can now pursue the line of geometric quantisation.
The first step is to define the analogue of the prequantum (pre)Hilbert space, i.e. the hermitean inner product space of smooth sections of the prequantum line bundle.
For two sections $\phi, \psi \in \Gamma(M, L)$ and $\dim(M) = 2n$, their inner product is given by
\begin{equation}
\label{eq:inner_product_on_GammaL}
	\<\phi, \psi\>_{\Gamma(M, L)} = \int_M h_L(\phi, \psi)\, \frac{\omega^n}{n!}\,,
\end{equation}
where $h_L$ is the hermitean bundle metric on $L$.
Finding an analogue of this for a bundle gerbe poses several questions:
\begin{myitemize}
	\item What is the space of sections of $(\CG, \nabla^\CG)$?
	
	\item Does it come as a categorified module, or vector space?
	
	\item Does it carry an inner product, canonically determined by the geometry?
\end{myitemize}
We answer these questions in this thesis and relate the categorified structures we obtain in this way to known structures on line bundles via dimensional reduction in Sections~\ref{sect:transgression_of_additional_structures} and~\ref{sect:dimensional_reduction}.

\section{Outline and main results}
\label{sect:Outline and main results}

We begin this thesis with a brief review of sheaves of categories.
Throughout the thesis, we frequently make use of the fact that hermitean vector bundles with connection satisfy a descent property.
However, we include non-invertible morphisms into the framework and, therefore, have to give a slightly modified version of the definition of a sheaf of categories.
These \emph{relative sheaves of categories} are defined in Section~\ref{sect:Coverings_and_sheaves_of_cats}.

The main background is laid out in Section~\ref{sect:The_2-category_of_BGrbs}, where we recall the construction of the 2-category of bundle gerbes from~\cite{Waldorf--Thesis,Waldorf--More_morphisms}.
The definitions of morphisms that we give, i.e. Definition~\ref{def:1-morphisms_of_BGrbs} and Definition~\ref{def:2-morphisms_of_BGrbs}, are slightly altered versions of those in~\cite{Waldorf--Thesis,Waldorf--More_morphisms}: we drop the condition~\eqref{eq:trace_condition} on the curvature on 1-morphisms and allow for more general 2-morphisms.
This enables us to find structures on morphisms of bundle gerbes in Chapter~\ref{ch:structures_on_morphisms_of_bgrbs} that are key in relating bundle gerbes to higher geometric quantisation.
Before exploring those structures and presenting examples, we take a detour in Section~\ref{sect:Deligne_coho_and_higher_cats} and investigate the local theory of bundle gerbes.
We recall the \v{C}ech-Deligne double complex and use the Dold-Kan correspondence to explicitly work out the relation between our 2-categorical theory of bundle gerbes and the $\infty$-categorical language used, for instance, in~\cite{FSS--Cech_cocycles_and_diff_char_classes,FSS--Higher_stacky_perspective}.
Proposition~\ref{st:Deligne_2-skeleton_and_Bgrbs_over_CU} and Theorem~\ref{st:Deligne_2-skeleton_and_Bgrbs} refine the corresponding results obtained previously in~\cite{Waldorf--Thesis}.
As a by-product, we obtain a better understanding of the local theory of bundle gerbes and their classification in terms of Deligne cohomology.

In Chapter~\ref{ch:structures_on_morphisms_of_bgrbs} we develop the categorical structures that enable us to relate bundle gerbes to higher geometric quantisation.
The key idea is to exploit the identification
\begin{equation}
	\Gamma(M,L) \cong \HLBdl^\nabla(M) \big( I_0, (L, \nabla^L) \big)\,
\end{equation}
for $I_0$ the trivial hermitean line bundle with connection and $\HLBdl^\nabla(M)$ denoting the category of hermitean line bundles with connection on $M$ and smooth, fibrewise linear morphisms.
There exists a trivial bundle gerbe with connection, denoted $\CI_0$.
Hence, for a bundle gerbe with connection $(\CG, \nabla^\CG)$, we define
\begin{equation}
	\Gamma \big( M, (\CG, \nabla^\CG) \big) \coloneqq \BGrb^\nabla_\rmpar(M) \big( \CI_0, (\CG, \nabla^\CG) \big)\,.
\end{equation}
Here, $\BGrb^\nabla(M)$ is the 2-category of bundle gerbes with connection introduced in Definition~\ref{def:2-categories_of_BGrbs}, and the subscript ``$\rmpar$'' indicates the restriction to the sub-2-category which has only parallel 2-morphisms.
Consequently, sections of $(\CG, \nabla^\CG)$ define a category.
Similarly, the bijection
\begin{equation}
	C^\infty(M,\FC) \cong \HLBdl^\nabla(M) (I_0, I_0)
\end{equation}
leads us to consider $\BGrb_\rmpar^\nabla(M)(\CI_0, \CI_0)$ as the category of higher functions.

Motivated by these definitions, we investigate the morphism categories in $\BGrb^\nabla(M)$ in detail.
Section~\ref{sect:Additive_structures_on_morphisms_in_BGrb} is devoted to the construction of additive structures on 1-morphisms and 2-morphisms of bundle gerbes in Proposition~\ref{st:additive_structure_on_2-morphisms_of_BGrbs} and Theorem~\ref{st:direct_sum_structure_on_morphisms_of_BGrbs}.
More precisely, we show that for a pair $(\CG_0, \nabla^{\CG_0})$, $(\CG_1, \nabla^{\CG_1}) \in \BGrb^\nabla(M)$, the morphism categories $\BGrb^\nabla_\rmpar(M)((\CG_0, \nabla^{\CG_0}), (\CG_1, \nabla^{\CG_1}))$ carry a monoidal direct sum and are abelian categories.
In particular, we can add sections of bundle gerbes, analogously to how we can add sections of hermitean line bundles.
We prove, moreover, that the direct sum of morphisms is compatible with compositions and tensor products in $\BGrb^\nabla_\rmpar(M)$.
This allows us to view $\BGrb^\nabla_\rmpar(M)(\CI_0, \CI_0)$ as a rig-category, a categorified version of a ring, and $\BGrb^\nabla_\rmpar((\CG_0, \nabla^{\CG_0}), (\CG_1, \nabla^{\CG_1}))$ as a rig-module category, a categorified ring module, over $\BGrb^\nabla_\rmpar(M)(\CI_0, \CI_0)$.
This is summarised in Theorem~\ref{st:enrichment_in_BGrb}, which is the central result of Section~\ref{sect:Additive_structures_on_morphisms_in_BGrb}.
Consequently, $\Gamma ( M, (\CG, \nabla^\CG))$ is a rig-module over the higher functions, in analogy to how sections of a line bundle form a module over the ring of smooth functions.

Next, we address the existence of a bundle metric on a hermitean line bundle in Section~\ref{sect:Pairings_and_inner_hom_of_morphisms_in_BGrb}.
It can be seen as a $C^\infty(M, \FC)$-linear map
\begin{equation}
	h_L \colon \overline{\Gamma(M, L)} \times \Gamma(M,L) \to C^\infty(M,\FC)\,,
\end{equation}
where $\overline{\Gamma(M, L)}$ is the space of sections of $L$ with the complex conjugate $C^\infty(M,\FC)$-module structure.
The fundamental input for carrying this structure over to bundle gerbes is the adjunction in Theorem~\ref{st:internal_hom_of_morphisms_in_BGrb--existence_and_naturality}.
It provides a $\BGrb^\nabla_\rmpar(M)(\CI_0, \CI_0)$-bilinear bifunctor
\begin{equation}
	[-,-] \colon \overline{\BGrb^\nabla_\rmpar(M) (\CG_0, \CG_1)} \times \BGrb^\nabla_\rmpar(M) (\CG_0, \CG_1)
	\to \BGrb^\nabla_\rmpar(M)(\CI_0, \CI_0)\,,
\end{equation}
where we have not written out the connections, and where $\overline{\BGrb^\nabla_\rmpar(M) (\CG_0, \CG_1)}$ is the category $\BGrb^\nabla_\rmpar(M) (\CG_0, \CG_1)^\opp$, with the action of higher functions composed by the dual functor introduced in Theorem~\ref{def:Riesz_dual_functor}, which categorifies complex conjugation.
We recall the definition of a two-variable adjunction and a closed module category from~\cite{Hovey--Model_categories}, extend it to the setting of rig-categories and their modules, and use this powerful notion to summarise the structure of the morphism categories $\BGrb^\nabla_\rmpar(M) (\CG_0, \CG_1)$ in Theorem~\ref{st:2-var_adjunction_on_BGrb}, the central result of Chapter~\ref{ch:structures_on_morphisms_of_bgrbs}.

We conclude Chapter~\ref{ch:structures_on_morphisms_of_bgrbs} by providing two well-known examples of bundle gerbe constructions and commenting on the structures discovered in this chapter.

Having considered module structures over higher functions, Chapter~\ref{ch:2Hspaces_from_bundle_gerbes} deals with module structures over the rig-category $\Hilb$ of finite-dimensional Hilbert spaces with its canonical direct sum and tensor product.
We begin by very briefly reviewing geometric quantisation and commenting on its higher analogues in Section~\ref{sect:HGeo_Quan}.
This leads us to identifying carefully the relations between structures on the category $\HLBdl^\nabla(M)$ and their higher analogues on the 2-category $\BGrb^\nabla(M)$ in Section~\ref{sect:higher_geometric_structures}.
In particular, we derive evidence for what our model of a 2-Hilbert space, the higher analogue of a Hilbert space, should look like.

The formal study of 2-Hilbert spaces happens in Section~\ref{sect:2-Hspaces}.
In Definition~\ref{def:2Hspace} we introduce the notion of 2-Hilbert space that we use in this work.
Its new feature is the use of two-variable adjunctions and closed $\Hilb$-module structures.
This turns out be a strong definition, with many additional useful properties automatically implied as we show in Proposition~\ref{st:structure_on_2Hspaces}.

Section~\ref{sect:2-Hspace_of_a_BGrb} finally contains the definition of the 2-Hilbert space of a bundle gerbe in Theorem~\ref{st:2-Hspace_of_a_Bgrb}.
A closed $\Hilb$-module structure on $\Gamma(M, (\CG, \nabla^\CG))$ is obtained by constructing an inclusion functor $\Hilb \hookrightarrow \BGrb^\nabla_\rmpar(M)(\CI_0, \CI_0)$.
This is in analogy with the inclusion of $\FC$ into smooth functions as constants.
We define a $\Hilb$-valued pairing on the category $\scH_0(\CG, \nabla^\CG) = \Gamma(M, (\CG, \nabla^\CG))$ via
\begin{equation}
	\<-,-\>_{\scH_0(\CG, \nabla^\CG)} \cong \BGrb^\nabla_\rmpar(M)( -,-)\,,
\end{equation}
which we relate to the expression~\eqref{eq:inner_product_on_GammaL} conceptually via a higher integral.
The form of the inner product is fixed up to natural isomorphism by our definition of a 2-Hilbert space.

In the following two sections, we provide two examples of bundle gerbes and construct and investigate their 2-Hilbert spaces of sections very explicitly.
We finish Chapter~\ref{ch:2Hspaces_from_bundle_gerbes} by providing several comments on the no-go statement of Proposition~\ref{st:no-go_for_sections_of_BGrbs} in Section~\ref{sect:ways_around_the_torsion_constraint}.

The concluding Chapter~\ref{ch:Transgression_and_reduction} addresses dimensional reduction.
The key idea here is that every $S^1$-bundle $K \to M$ with a smoothly varying orientation on each of its fibres naturally defines a map from $M$ into the space of unparameterised smooth loops in $K$.
Therefore, we can define dimensional reduction after (what might first seem to be) a digression on transgression.
We begin by reviewing diffeological spaces in Section~\ref{sect:DfgSp_and_diffeological_bundles}, as these provide a useful framework for treating the geometry of spaces of smooth maps.
In Section~\ref{sect:transgression_functor}, we merge the approaches to transgression of~\cite{Waldorf--Transgression_II} and~\cite{CJM--Holonomy_on_D-branes} in order to obtain a transgression functor defined on the entire 2-category of bundle gerbes.
The crucial relations between the new categorical structures that we have introduced on bundle gerbes and their morphisms in Chapter~\ref{ch:structures_on_morphisms_of_bgrbs} and known structures on hermitean line bundles with connections are obtained in Section~\ref{sect:transgression_of_additional_structures}.
All structures that we have introduced transgress to their acclaimed counterparts on line bundles.

In Section~\ref{sect:unparameterised_loops_and_Reality} we divide out parameterisations of the loops, thus being left with a transgression functor to unparameterised, oriented loops in the base $M$.
Along the way we realise in Theorem~\ref{st:Real_structure_on_T_0CG} that the transgression line bundle naturally carries a Real structure with respect to the involution induced on loop space by the inversion on $S^1$.

After this detour, we can readily define a dimensional reduction, or pushforward, in Definition~\ref{def:reduction_functor} in Section~\ref{sect:dimensional_reduction}.
Its well-definedness is shown in Theorem~\ref{st:reduction_functor}.
Even though this version of dimensional reduction uses transgression to the infinite-dimensional loop space of $K$, we show that it is less cumbersome than it might appear on first sight, by unravelling its definition in a concrete example.
It follows that all categorical structures reduce correctly, in accordance with the analogies obtained earlier.

We defer some technical details regarding the construction of the 2-category of bundle gerbes to Appendix~\ref{ch:App:special_morphisms} and two longer proofs to Appendix~\ref{ch:App:proofs}.

Parts of the results of this thesis are contained in the preprint~\cite{BSS--HGeoQuan} and the paper~\cite{Bunk-Szabo--Fluxes_brbs_2Hspaces}.

\chapter{Bundle gerbes}
\label{ch:bundle_gerbes}

\section{Preliminaries on coverings and sheaves of categories}
\label{sect:Coverings_and_sheaves_of_cats}

In this section we briefly introduce some notation for surjective submersions of manifolds and recall the definitions of sheaves of categories, stacks and descent, which will be used frequently in the remainder of this thesis.
Since our choice of morphisms of hermitean vector bundles with connections is not the most commonly used one, we clarify how the resulting categories can still be viewed as a marginally weakened version of sheaves of categories.
The notions of Grothendieck topology and (pre)sheaves of categories that we use here can be found, for instance, in~\cite{MacLane-Moerdijk--Sheave_in_geometry_and_logic,Vistoli--Grothendieck_Tops_fibred_cats_and_descent,Waldorf--Thesis,Moerdijk--Intro_to_stacks_and_gerbes}.

If $\pi \colon Y \to M$ is a surjective submersion of manifolds, we write $Y^{[n]} \coloneqq Y {\times}_M \cdots {\times}_M Y$ for the $n$-fold fibre product of $Y$ with itself over $M$.
Elements of $Y^{[n]}$ are given by tuples $(y_0, \ldots, y_{n-1})$ of points in $Y$ such that $\pi(y_0) = \ldots = \pi(y_{n-1})$.
The fibre products $Y^{[n]}$ assemble into a simplicial manifold $\rmN_\bullet(Y,\pi,M)$ with $\rmN_n(Y,\pi,M) = Y^{[n+1]}$, called the \emph{\v{C}ech nerve of $\pi \colon Y \to M$}.
Its face and degeneracy morphisms are given by $d_i(y_0, \ldots, y_{n-1}) = (y_0, \ldots, \widehat{y_i}, \ldots, y_{n-1})$ and $s_i(y_0, \ldots, y_{n-1}) = (y_0, \ldots, y_i, y_i, \ldots, y_{n-1})$, where the hat over $y_i$ denotes omission of that entry of the tuple.
In the simplex category $\Delta$ of ordered finite ordinals $[n]$ and non-decreasing maps, we define
\begin{equation}
	p^{i_0, \ldots, i_n} \colon [n] \to [m]\,, \quad
	p^{i_0, \ldots, i_n}(k) = i_k\,,
\end{equation}
for $m \geq n > k \geq 0$ and $i_0 \leq i_1 \leq \ldots \leq i_n \in [m]$.
For a simplicial manifold $N_\bullet$ we write $N_\bullet p^{i_0, \ldots, i_n} = p_{i_0, \ldots, i_n}$.
Note that if $N_\bullet = \rmN_\bullet(Y,\pi,M)$ is a \v{C}ech nerve of a surjective submersion $\pi \colon Y \to M$, then $p_{i_0, \ldots, i_n} \colon Y^{[m+1]} \to Y^{[n+1]}$ acts as $(y_0, \ldots y_m) \mapsto (y_{i_0}, \ldots, y_{i_n})$.
If $f \colon Y_0 \to Y_1$ defines a morphism of surjective submersions $\pi_i \colon Y_i \to M$, for $i=0,1$, i.e. is a smooth map covering the identity on $M$, we write $f^{[n]} \colon Y_0^{[n]} \to Y_1^{[n]}$ for the $n$-th fibre product of $f$ with itself with respect to $\pi_0$ and $\pi_1$.

For a manifold $Y \in \Mfd$, we denote by $\HVBdl^\nabla(Y)$ the category with objects hermitean vector bundles with connection and morphisms smooth morphisms of vector bundles (i.e. not respecting metric or connection).
The sub-category $\HVBdl^\nabla_\rmpar(Y)$ has the same objects, but as its morphisms only those morphisms of vector bundles that are parallel with respect to the covariant derivatives induced on morphisms.
This category, in turn, has a subgroupoid, denoted $\HVBdl^\nabla_\rmuni(Y)$, with the same objects, but only unitary parallel isomorphisms as morphisms.
The categories $\HLBdl^\nabla(Y)$, $\HLBdl^\nabla_\rmpar(Y)$ and $\HLBdl^\nabla_\rmuni(Y)$ are defined to be the full subcategories of $\HVBdl^\nabla(Y)$, $\HVBdl^\nabla_\rmpar(Y)$ and $\HVBdl^\nabla_\rmuni(Y)$, respectively, whose objects are hermitean line bundles with connections.
We write $\curv(\nabla^E)$ for the curvature 2-form of a hermitean vector bundle with connection $(E, \nabla^E)$.
Forgetting connections, we obtain the categories $\HVBdl(Y)$, $\HVBdl_\rmuni(Y)$, $\HLBdl(Y)$ and $\HLBdl_\rmuni(Y)$, where $\HVBdl_\rmuni(Y)$ and $\HLBdl_\rmuni(Y)$ now have unitary isomorphisms as morphisms.

\begin{definition}[Presheaf of categories]
\label{def:presheaf_of_Cats}
Let $\scC$ be a category.
A \emph{presheaf of categories on $\scC$} is a (weak) 2-functor from the 2-category $\disc(\scC^\opp)$, obtained by adding identity 2-morphisms to $\scC^\opp$, to the 2-category of categories $\scCat$.
Explicitly, this assigns
\begin{myenumerate}
	\item a category $\CF X \in \scCat$ to every $X \in \scC$,
	
	\item a functor $\CF \phi \colon \CF Y \to \CF X$ to every $\phi \in \scC(X,Y)$,
	
	\item a natural isomorphism $\CF(\phi, \psi) \colon \CF(\psi \circ \phi) \to \CF \phi \circ \CF \psi$ to every composable pair of morphisms $\psi \in \scC(Y,Z)$ and $\phi \in \scC(X,Y)$ in $\scC$,
\end{myenumerate}
such that, for $\rho \in \scC(Z,U)$ we have
\begin{equation}
	\CF \rho \big( \CF(\phi, \psi) \big) \circ \CF(\psi \circ \phi, \rho)
	=
	\big( \CF(\phi, \psi) \big)_{\CF \rho} \circ \CF(\phi, \rho \circ \psi)\,.
\end{equation}
\end{definition}

\begin{definition}[Morphisms of presheaves of categories]
Let $\CF$, $\CH$ be presheaves of categories on $\scC$.
A \emph{morphism of presheaves of categories} $\CF \to \CH$ is an assignment of a functor $\Phi_X \colon \CF X \to \CH X$ to every object $X \in \scC$, and of a natural isomorphism $\eta^\phi \colon \Phi_X \circ \CF \phi \to \CH \phi \circ \Phi_Y$ to every morphism $\phi \in \scC(X,Y)$, such that for every $\psi \in \scC(Y,Z)$ the following diagram of functors and natural transformations commutes:
\begin{equation}
\begin{tikzcd}[column sep=3cm, row sep=1.25cm]
	\Phi_X \circ \CF (\psi \circ \phi) \ar[r, "{\Phi_X\, \CF(\phi, \psi)}"] \ar[dd, "\eta^{\psi \circ \phi}"'] & \Phi_X \circ \CF \phi \circ \CF \psi \ar[d, "\eta^\phi_{\CF \psi}"]
	\\
	&  \CH \phi \circ \Phi_Y \circ \CF \psi \ar[d, "(\CH \phi)\, \eta^\psi"]
	\\
	\CF (\psi \circ \phi) \circ \Phi_Z & \CH \phi \circ \CH \psi \circ \Phi_Z \ar[l, "{(\CH(\phi, \psi)_{\Phi_Z})^{-1}}"]
\end{tikzcd}
\end{equation}
Let $(\Psi, \epsilon) \colon \CF \to \CH$ be another morphism of presheaves of categories.
A \emph{2-morphism of presheaves of categories} $(\Phi, \eta) \to (\Psi, \epsilon)$ is an assignment of a natural transformation $\sigma^X \colon \Phi_X \to \Psi_X$ to every $X \in \scC$ such that, for every morphism $\phi \in \scC(X,Y)$, we have the following commutative diagram of functors and natural transformations:
\begin{equation}
\begin{tikzcd}[column sep=2cm, row sep=1.25cm]
	\Phi_X \circ \CF f \ar[r, "\eta^\phi"] \ar[d, "\sigma^X_{\CF f}"'] & \CH f \circ \Phi_Y \ar[d, "\CH f\, \sigma^Y"]
	\\
	\Psi_X \circ \CF f \ar[r, "\epsilon^\phi"'] & \CH f \circ \Psi_Y
\end{tikzcd}
\end{equation}
\end{definition}

\begin{definition}[Grothendieck topology, Grothendieck site]
\label{def:Grothendieck_top_and_site}
A \emph{Grothendieck topology} on a category $\scC$ is an assignment of a set $\tau(X) \in \Set$ to every object $X \in \scC$ such that
\begin{myenumerate}
	\item the set $\tau(X)$ has as elements families of morphisms $\{ \pi_a \in \scC(U_a, X)\}_{a \in \Lambda}$,
	
	\item every isomorphism $\pi \in \scC_\sim(Y,X)$ with target $X$ is in $\tau(X)$,
	
	\item if $\{\pi_a \in \scC(U_a,X)\}_{a \in \Lambda} \in \tau(X)$ and for every $a \in \Lambda$ we have $\{\zeta_{ab} \in \scC(V_{ab}, U_a) \}_{b \in \Lambda_a} \in \tau(U_a)$, then the collection of composites $\{\pi_a \circ \zeta_{ab} \in \scC(V_{ab}, X) \}_{a \in \Lambda,\, b \in \Lambda_a}$ is in $\tau(X)$,
	
	\item if $\phi \in \scC(Y,X)$ and $\{ \pi_a \in \scC(U_a, X)\}_{a \in \Lambda} \in \tau(X)$, then the pullbacks $Y {\times}_X U_a \to $ exist and $\{ Y {\times}_X U_a \to Y \}_{a \in \Lambda} \in \tau(Y)$.
\end{myenumerate}
The elements of $\tau(X)$ are called \emph{coverings of $X$}.
A category $\scC$ endowed with a fixed Grothendieck topology $\tau$ is called a \emph{Grothendieck site}.
\end{definition}

\begin{example}
The two crucial examples for us are the following well-known Grothendieck topologies on $\scC = \Mfd$ (see, for instance, \cite{NS--Equivariance_in_higher_geometry,Nikolaus--Thesis}):
\begin{myenumerate}
	\item
	For $M \in \Mfd$, let $\tau_\open(M)$ be the set of all maps $\pi \colon \CU \to M$, where $\CU = \bigsqcup_{a \in \Lambda}\, U_a$ is the total space of an open covering $(U_a)_{a \in \Lambda}$ of $M$, with the canonical map to $M$.
	
	\item
	For $M \in \Mfd$ let elements of $\tau_\ssub(M)$ be surjective submersions of manifolds with target $M$.
	\qen
\end{myenumerate}
\end{example}

In these cases, a covering $\{\pi_a \in \Mfd(Y_a, M) \}_{a \in \Lambda}$ consists of a single morphism, i.e. is of the form $\{\pi_a \in \Mfd(Y_a, M) \}_{a \in \Lambda} = \{\pi \in \Mfd(Y, M)\}$. 
Having a notion of covering at hand, we can now define gluing, or descent data and impose sheaf conditions on presheaves of categories.
Let $(\scC, \tau)$ be a Grothendieck site.
Given two coverings $\{\pi_a \in \scC(U_a, X)\}$ and $\{\nu_i \in \scC(V_i, X)\}_{i \in \Xi}$ of an object $X \in \scC$, it follows from axiom (4) of Definition~\ref{def:Grothendieck_top_and_site} that $\{ U_a {\times}_X V_i \to V_i \}_{a \in \Lambda}$ is a covering of $V_i \in \scC$ for every $i \in \Xi$.
Thus, axiom (3) implies that $\{ U_a {\times}_X V_i \to X \}_{a \in \Lambda,\, i \in \Xi}$ is a covering of $X$.
Observe that there are canonical morphisms $p_{U_a} \in \scC(U_a {\times}_X V_i, U_a)$ and $p_{V_i} \in \scC(U_a {\times}_X V_i, V_i)$.
This construction can be iterated for multiple coverings of the same object $X \in \scC$.
We write $p_{a_{i_0} \ldots a_{i_m}} \colon U_{a_0} {\times}_X \ldots {\times}_X U_{a_n} \to U_{a_{i_0}} {\times}_X \ldots {\times}_X U_{a_{i_m}}$ for the projection maps in the $n$-fold fibre product of the covering with itself.%
\footnote{We chose to neglect possible re-bracketing isomorphisms in the fibre products, as these will not be of relevance in this work.}

\begin{definition}[Descent category]
Let $(\scC, \tau)$ be a Grothendieck site, and let $\CF$ be a presheaf of categories on $\scC$.
For a fixed covering $\{ \pi_a \in \scC(U_a, X)\}_{a \in \Lambda}$, we define the \emph{descent category of $\CF$ with respect to $\{ \pi_a \in \scC(U_a, X)\}_{a \in \Lambda}$}, denoted $\Desc(\CF, \{\pi_a\})$, as follows:
Its objects are pairs $(E_a, \alpha_{ab})_{a \in \Lambda}$ consisting of an object $E_a \in \CF U_a$ for every $a \in \Lambda$, and isomorphisms
\begin{equation}
	\alpha_{ab} \in \CF_\sim(U_a {\times}_X U_b) \big( \CF p_b (E_b),\, \CF p_a (E_a) \big)
\end{equation}
such that the following diagram commutes:
\begin{equation}
\begin{tikzcd}[column sep=1.5cm, row sep=1.25cm]
	(\CF p_a)(E_a) \ar[r, "{\CF(p_{ab}, p_a)}"] \ar[d, "{\CF(p_{ac}, p_a)}"'] & (\CF p_{ab} \circ \CF p_a) (E_a) \ar[r, "{(\CF p_{ab})(\alpha_{ab})^{-1}}"] & (\CF p_{ab} \circ \CF p_b) (E_b) \ar[d, "{\CF(p_{ab}, p_b)^{-1}}"]
	\\
	(\CF p_{ac} \circ \CF p_a) (E_a) \ar[d, "{(\CF p_{ac})(\alpha_{ac})^{-1}}"'] & & (\CF p_b) (E_b) \ar[d, "{\CF(p_{bc}, p_c)}"]
	\\
	(\CF p_{ac} \circ \CF p_c) (E_c) \ar[d, "{\CF(p_{ac}, p_c)}"'] & & (\CF p_{bc} \circ \CF p_b) (E_b) \ar[d, "{(\CF p_{bc})(\alpha_{bc})^{-1}}"]
	\\
	(\CF p_c)(E_c) & & (\CF p_{bc} \circ \CF p_c) (E_c) \ar[ll, "{\CF(p_{bc}, p_a)^{-1}}"]
\end{tikzcd}
\end{equation}
A morphism $(E_a, \alpha_{ab})_{a,b \in \Lambda} \to (F_a, \beta_{ab})_{a,b \in \Lambda}$ in the category $\Desc(\CF, \{\pi_a\})$ consists of a collection $(\phi_a \in \CF(U_a)(E_a, F_a))_{a \in \Lambda}$ such that we have the following commutative diagram in $\CF (U_a {\times}_X U_b)$:
\begin{equation}
\begin{tikzcd}[column sep=2cm, row sep=1.25cm]
	(\CF p_a) (E_a) \ar[r, "\alpha_{ab}^{-1}"] \ar[d, "{(\CF p_a) (\phi_a)}"'] & (\CF p_b) (E_b) \ar[d, "{(\CF p_b) (\phi_b)}"]
	\\
	(\CF p_a) (F_a) \ar[r, "\beta_{ab}^{-1}"'] & (\CF p_b) (F_b)
\end{tikzcd}
\end{equation}
\end{definition}

For any Grothendieck site $(\scC, \tau)$ and presheaf of categories $\CF$ on $\scC$, and given some covering $\{\pi_a \in \scC(U_a, X)\}_{a \in \Lambda}$ of an object $X \in \scC$, there exists a canonical functor
\begin{equation}
\begin{aligned}
	&\Asc_{\{\pi_a\}} \colon \CF X \to \Desc(\CF, \{\pi_a\})\,,
	\\
	&E \mapsto \big( (\CF \pi_a) E,\, (\CF(p_b, \pi_b) \circ \CF(p_a, \pi_a)^{-1})_E \big)_{a \in \Lambda}\,,
	\\
	&\CF X (E,F) \ni \phi \mapsto \big( (\CF p_a) \phi \big)_{a \in \Lambda}\,.
\end{aligned}
\end{equation}

\begin{definition}[Sheaf of categories, stack]
Let $(\scC, \tau)$ be a Grothendieck site.
A presheaf of categories $\CF$ on $(\scC, \tau)$ is called a \emph{sheaf of categories} if, for every possible covering $\{\pi_a \in \scC(U_a, X)\}_{a \in \Lambda}$ of an object in $\scC$, the functor $\Asc_{\{\pi_a\}}$ is an equivalence of categories.
A \emph{stack on $(\scC, \tau)$} is a sheaf of categories $\CF$ which takes values in groupoids.
\end{definition}

For example, $\HVBdl_\rmuni$ and $\HVBdl^\nabla_\rmuni$ are stacks:
for $\pi \colon Y \to M$ a surjective submersion in $\Mfd$, an object in $\Desc( \HVBdl^\nabla_\rmuni, \pi)$ is a pair $(E,\alpha)$ of a hermitean vector bundle with connection $E \to Y$ and a unitary, parallel isomorphism $d_0^*E \to d_1^*E$ over $Y {\times}_M Y = Y^{[2]}$.
However, $\HVBdl^\nabla$ is not a sheaf of categories.
The reason is that the isomorphisms $\alpha$ in a descent object is not necessarily unitary or parallel for our definition of $\HVBdl^\nabla$.
Thus, we have to slightly weaken the definition of the descent category and a sheaf of categories.
We call a subcategory of some category \emph{wide} if it contains all objects of the ambient category.

\begin{definition}[Relative sheaf of categories]
\label{def:relative_sheaf_of_Cats}
Let $(\scC, \tau)$ be a Grothendieck site, and let $\CF$ and $\CF_\Desc$ be presheaves of categories on $\scC$ together with a morphism $(\Phi, \eta) \colon \CF_\Desc \to \CF$, such that $\Phi_X \colon \CF_\Desc X \hookrightarrow \CF_\sim X$ is an inclusion of a wide subgroupoid.
We call such a tuple $(\CF, \CF_\Desc, \Phi, \eta)$ a \emph{relative presheaf of categories}.%
\footnote{This is loosely modelled on the notion of a \emph{relative category}.}
Given a relative presheaf of categories on $\scC$ and a covering $\{ \pi_a \in \scC(U_a, X)\}_{a \in \Lambda}$, we define the \emph{relative descent category} $\Desc( \CF, \CF_\Desc, \{\pi_a\})$ to be the full subcategory of $\Desc(\CF, \{\pi_a\})$ over the objects $(E_a, \alpha_{ab})_{a,b \in \Lambda}$ where
\begin{equation}
	\alpha_{ab} \in \CF_\Desc(U_a {\times}_X U_b) \big( \CF p_b (E_b),\, \CF p_a (E_a) \big) \quad \forall\, a,b \in \Lambda\,.
\end{equation}
We say a relative presheaf of categories is a \emph{relative sheaf of categories} if for each covering $\{ \pi_a \in \scC(U_a, X)\}_{a \in \Lambda}$ the ascent  functor $\Asc_{\{\pi_a\}}$ is valued in the relative descent category $\Desc( \CF, \CF_\Desc, \{\pi_a\})$ and is an equivalence of categories.
\end{definition}

\begin{example}
All of the categories of vector bundles introduced above define relative sheaves of categories on $(\Mfd, \tau_\open)$ as well as $(\Mfd, \tau_\ssub)$ if we take the sub-presheaf to be
\begin{equation}
	(\HVBdl^\nabla)_\Desc = \HVBdl^\nabla_\rmuni
\end{equation}
for $\HVBdl^\nabla$ and $\HVBdl^\nabla_\rmpar$, with the canonical inclusions into $\HVBdl^\nabla$ and $\HVBdl^\nabla_\rmpar$, respectively, and accordingly for line bundles.
Here we treat the morphisms $\CF(\phi, \psi)$ as the identities, as they virtually have no effect in our formulae expect for cluttering the notation (compare also the remarks in Appendix~\ref{app:monoidal_structures_and_strictness}).
\qen
\end{example}

This is well-known for $\HVBdl$, $\HVBdl^\nabla$, $\HVBdl_\rmuni$ and their analogues in line bundles.
It is less well-known for $\HVBdl^\nabla_\rmpar$: here the descent of morphisms can be seen either directly, or using~\cite[Proposition 2.17]{NS--Equivariance_in_higher_geometry} to first transfer the problem to $(\Mfd, \tau_\open)$, where it is straightforward to see that parallel morphisms glue over overlaps of patches.
However, in the following we will usually write
\begin{equation}
	\Desc \big( \HVBdl^\nabla,\, \HVBdl^\nabla_\rmuni, \pi \big) \eqqcolon \Desc(\HVBdl^\nabla, \pi)
\end{equation}
in order to straighten up notation.
When we refer to $\Mfd$ as a Grothendieck site without specifying a Grothendieck topology, we shall mean the Grothendieck topology of surjective submersions, i.e. the site $(\Mfd, \tau_\ssub)$, unless stated otherwise.

\section{The 2-category of bundle gerbes}
\label{sect:The_2-category_of_BGrbs}

Bundle gerbes are the central objects in this thesis.
They have been defined in~\cite{Murray--Bundle_gerbes}, the 2-categorical theory of their morphisms has been initiated in~\cite{Murray-Stevenson:Bgrbs--stable_isomps_and_local_theory} and has been significantly extended in~\cite{Waldorf--More_morphisms}.
In this section, we will recall definitions and results from these references, generally in the language of~\cite{Waldorf--More_morphisms,Waldorf--Thesis}.
The only originality in this section is contained in a slight generalisation of the definition of 1-morphisms (Definition~\ref{def:1-morphisms_of_BGrbs}) and 2-morphisms (Definition~\ref{def:2-morphisms_of_BGrbs}) as compared to~\cite{Waldorf--More_morphisms,Waldorf--Thesis}.
Proposition~\ref{st:determinants_of_morphisms_of_BGrbs} has been known~\cite{BCMMS} before, but we make its relation to the tensor product of bundle gerbes over generic surjective submersions to $M$ explicit.

\begin{definition}[Bundle gerbe, connection]
We make the following definitions.
\begin{myenumerate}
	\item A \emph{hermitean line bundle gerbe}, or \emph{bundle gerbe} on a manifold $M \in \Mfd$ is a tuple $\CG = (L, \mu, Y, \pi)$, where $\pi \colon Y \to M$ is a surjective submersion, $L \in \HLBdl(Y^{[2]})$, and
	\begin{equation}
		\mu \in \HLBdl_\rmuni(Y^{[3]}) \big( d_2^*L \otimes d_0^*L,\, d_1^*L \big)
	\end{equation}
	is a unitary isomorphism.
	Over $Y^{[4]}$ it is required to satisfy the associativity condition expressed by the commutativity of the diagram
	\begin{equation}
	\begin{tikzcd}[column sep=2.5cm, row sep=1.5cm]
		p_{01}^*L \otimes p_{12}^*L \otimes p_{23}^*L \ar[r, "1 \otimes d_0^*\mu"] \ar[d, "d_3^*\mu \otimes 1"'] & p_{01}^*L \otimes p_{13}^*L \ar[d, "d_2^*\mu"]
		\\
		p_{02}^*L \otimes p_{23}^*L \ar[r, "d_1^*\mu"'] & p_{03}^*L
	\end{tikzcd}
	\end{equation}
	
	\item Let $\CG = (L, \mu, Y, \pi)$ be a bundle gerbe.
	A \emph{connection on $\CG$} is a pair $\nabla^\CG = (\nabla^L, B)$ of a hermitean connection $\nabla^L$ on $L$ and a 2-form $B \in \Omega^1(Y, \iu\, \FR)$ such that $\mu$ is parallel with respect to the the connections induced by $\nabla^L$ on the respective pullbacks and tensor products, and such that $\curv(\nabla^L) = d_0^*B - d_1^*B$, where $\curv(\nabla^L)$ is the field strength of $\nabla^L$.
	The 2-form $B$ is called the \emph{curving of $\nabla^\CG$}.
	It can be shown~\cite{Murray--Bundle_gerbes} that $\dd B$ descends to $M$ to define a 3-form $\curv(\nabla^\CG) \in \Omega^3(M, \iu\, \FR)$ which is called the \emph{curvature of the bundle gerbe with connection $(\CG,\nabla^\CG)$}.
\end{myenumerate}
\end{definition}

The morphism $\mu$ of a bundle gerbe is often referred to as a \emph{multiplication} or \emph{composition}.
Let us denote the trivial hermitean line bundle on $M \in \Mfd$ by $I \coloneqq M {\times} \FC$.
It can be endowed with a hermitean connection by specifying a global connection 1-form $A \in \Omega^1(M, \iu\, \FR)$, and we denote the resulting hermitean line bundle with connection by $I_A \coloneqq (I, \dd + A)$.

\begin{example}
A particularly easy, but very important bundle gerbe is the \emph{trivial bundle gerbe} $\CI \coloneqq ( I,\, m,\, M,\, 1_M)$.
Its multiplication is given by that on $\FC$ via $m((x,z),(x,z')) = (x, z z')$, for $x \in M$ and $z, z' \in \FC$.
Similarly to how $I$ can be endowed with a hermitean connection, we can put a connection on $\CI$ by specifying a 2-form $\rho \in \Omega^2(M,\iu\,\FR)$ and setting $\nabla^{\CI_\rho} = (\dd, \rho)$.
We write
\begin{equation}
\CI_\rho \coloneqq ( I_0,\, m,\, \rho,\, M,\, 1_M)
\end{equation}
for the trivial bundle gerbe with connection $(\dd, \rho)$.
\qen
\end{example}

\begin{example}
\begin{myenumerate}
	\item Given a bundle gerbe with connection $(\CG,\nabla^\CG) = ( L, \nabla^L, \mu, B, Y, \pi )$, we can define a new bundle gerbe by setting
	\begin{equation}
		(\CG,\nabla^\CG)^*
		\coloneqq  \big( (L, \nabla^L)^*, \mu^{-\sft}, -B, Y, \pi \big)\,.
	\end{equation}
	Here, $(L, \nabla^L)^*$ is the dual hermitean line bundle with connection of $(L,\nabla^L)$, and $\mu^{-\sft} = (\mu^\sft)^{-1} = (\mu^{-1})^\sft$ is the inverse transpose of $\mu$.
	This bundle gerbe with connection is called the \emph{dual bundle gerbe with connection of $(\CG,\nabla^\CG)$}.
	
	\item Let $(\CG_i, \nabla^{\CG_i}) = ( L_i, \nabla^{L_i}, \mu_i, B_i, Y_i, \pi_i )$, for $i= 0,1$, be two bundle gerbes with connection on $M$.
	Their \emph{tensor product} is the bundle gerbe with connection given by
	\begin{equation}
	\label{eq:tensor_product_of_BGrbs}
	\begin{aligned}
		&(\CG_0, \nabla^{\CG_0}) \otimes (\CG_1, \nabla^{\CG_1})
		\\
		&\coloneqq \Big( \pr_{Y_0}^{[2]*}(L_0, \nabla^{L_0}) \otimes \pr_{Y_1}^{[2]*} (L_1, \nabla^{L_1}),\, \pr_{Y_0}^{[3]*}\mu_0 \otimes \pr_{Y_1}^{[3]*}\mu_1,\, Y_0 {\times}_M Y_1,\, \pi_i \circ \pr_{Y_i} \Big)\,,
	\end{aligned}
	\end{equation}
	where $\pr_{Y_i} \colon Y_0 {\times}_M Y_1 \to Y_i$ is the projection map.
	Note that $\pi_0 \circ \pr_{Y_0} = \pi_1 \circ \pr_{Y_1}$.
	
	\item If $(\CG,\nabla^\CG) = ( L, \nabla^L, \mu, B, Y, \pi )$ is a bundle gerbe on $M$ and $f \in \Mfd(N,M)$, there is a \emph{pullback bundle gerbe} on $N$ induced by these data.
	It reads as
	\begin{equation}
		f^*(\CG,\nabla^\CG) = \big( f_\pi^{[2]*} (L,\nabla^L),\, f_\pi^{[3]*} \mu,\, f_\pi^*B,\, f^*Y,\, \pi_f \big)\,,
	\end{equation}
	where $\pi_f \colon f^*Y \to N$ is the induced surjective submersion $f^*Y \to N$, and $f_\pi \colon f^*Y \to Y$ is the canonically induced map covering $f$. \qen
\end{myenumerate}
\end{example}

\begin{definition}[Morphisms of bundle gerbes]
\label{def:1-morphisms_of_BGrbs}
Given two bundle gerbes with connection $(\CG_i, \nabla^{\CG_i}) = ( L_i, \nabla^{L_i}, \mu_i, B_i, Y_i, \pi_i )$, for $i= 0,1$, on $M$, a \emph{1-morphism of bundle gerbes with connection} from $(\CG_0, \nabla^{\CG_0})$ to $(\CG_1, \nabla^{\CG_1})$ is a tuple $(E, \nabla^E, \alpha, Z, \zeta)$ consisting of the following data.
First, $\zeta \colon Z \to Y_{01} \coloneqq Y_0 {\times}_M Y_1$ is a surjective submersion, and $(E, \nabla^E) \in \HVBdl^\nabla(Z)$ is a hermitean vector bundle with connection on $Z$.
We write $\zeta_{Y_i} \coloneqq \pr_{Y_i} \circ \zeta$ and set $\zeta_M \coloneqq \pi_0 \circ \zeta_{Y_0} = \pi_1 \circ \zeta_{Y_1}$.
Then,
\begin{equation}
	\alpha \in \HVBdl^\nabla_\rmuni(Z^{[2]}) \big( \zeta_{Y_0}^{[2]*} L_0 \otimes d_0^*E,\, d_1^*E \otimes \zeta_{Y_1}^{[2]*} L_1 \big)
\end{equation}
is a unitary parallel isomorphism of bundles over $Z^{[2]} = Z {\times}_M Z$ satisfying
\begin{equation}
\label{eq:1-morphisms_compatibility_with_BGrb_multiplications}
	\big( 1 \otimes \zeta_{Y_1}^{[3]*} \mu_1 \big) \circ (d_2^*\alpha \otimes 1) \circ (1 \otimes d_0^* \alpha)
	= d_1^*\alpha \circ \big( \zeta_{Y_0}^{[3]*} \mu_0 \otimes 1 \big)\,,
\end{equation}
or, equivalently, making the following diagram commute:
\begin{equation}
\begin{tikzcd}[column sep=3cm, row sep=1.5cm]
	p_{01}^* \zeta_{Y_0}^{[2]*} L_0 \otimes p_{12}^* \zeta_{Y_0}^{[2]*} L_0 \otimes p_2^* E \ar[dd, "\zeta_{Y_0}^{[3]*}\mu_0 \otimes 1"'] \ar[r, "1 \otimes d_0^*\alpha"] & p_{01}^* \zeta_{Y_0}^{[2]*} L_0 \otimes p_1^* E \otimes p_{12}^* \zeta_{Y_1}^{[2]*} L_1 \ar[d, "d_2^*\alpha \otimes 1"]
	\\
	& p_0^* E \otimes p_{01}^* \zeta_{Y_1}^{[2]*} L_1 \otimes p_{12}^* \zeta_{Y_1}^{[2]*} L_1 \ar[d, "1 \otimes \zeta_{Y_1}^{[3]*}\mu_1"]
	\\
	p_{02}^* \zeta_{Y_0}^{[2]*} L_0 \otimes p_2^* E \ar[r, "d_1^*\alpha"'] & p_0^* E \otimes p_{02}^* \zeta_{Y_1}^{[2]*} L_1
\end{tikzcd}
\end{equation}
\end{definition}

Often we will abbreviate 1-morphisms by writing
\begin{equation}
	(E,\alpha) \coloneqq (E, \nabla^E, \alpha, Z, \zeta)\,.
\end{equation}
Given a triple of bundle gerbes $(\CG_i, \nabla^{\CG_i}) = ( L_i, \nabla^{L_i}, \mu_i, B_i, Y_i, \pi_i )$, for $i= 0,1,2$, and two 1-morphisms $(E, \nabla^E, \alpha, Z, \zeta) \colon (\CG_0, \nabla^{\CG_0}) \to (\CG_1, \nabla^{\CG_1})$ and $(E', \nabla^{E'}, \alpha', Z', \zeta') \colon (\CG_1, \nabla^{\CG_1}) \to (\CG_2, \nabla^{\CG_2})$, their composition is given by the 1-morphism
\begin{align}
	&(E', \nabla^{E'}, \alpha', Z', \zeta') \circ (E, \nabla^E, \alpha, Z, \zeta) \notag
	\\
	&= \big( \pr_Z^*(E,\nabla^E) \otimes \pr_{Z'}^*(E', \nabla^{E'}),\, (1_{\pr_Z^*d_1^*E} \otimes \pr_{Z'}^{[2]*} \alpha') \circ (\pr_Z^{[2]*} \alpha \otimes 1_{\pr_{Z'}^*d_0^*E'}), \notag
	\\
	&\hspace{1cm} Z {\times}_{Y_1} Z',\, (\zeta_{Y_0} \circ \pr_Z) {\times}_{M} (\zeta'_{Y_2} \circ \pr_{Z'}) \big)\,.
\end{align}

Morphisms of bundle gerbes without connection are given just as in Definition~\ref{def:1-morphisms_of_BGrbs}, after discarding the connections on $\CG_i$ and the hermitean vector bundle $E$.

\begin{remark}
\label{rmk:trace_condition_on_1-morphisms}
Morphisms of bundle gerbes with connections are sometimes (see for instance~\cite{Waldorf--More_morphisms,Waldorf--Thesis}) required to satisfy the additional property that
\begin{equation}
\label{eq:trace_condition}
	\tr \big( \curv(\nabla^E) - (\zeta_{Y_1}^*B_1 - \zeta_{Y_0}^*B_0) \otimes 1_E \big) = 0\,.
\end{equation}
However, this is a rather strong condition, and so we do not impose it in the general definition of 1-morphisms.
Instead, in our case these morphisms define a subcategory of the 2-category of bundle gerbes (see Definition~\ref{def:2-categories_of_BGrbs}).
The trace condition~\eqref{eq:trace_condition} poses a condition on the first Chern class of the bundle $E$.
Moreover, since $\dd\, \tr(\curv(\nabla^E)) = 0$ for any hermitean vector bundle with connection, this condition implies $\dd\, (\zeta_{Y_1}^*B_1 - \zeta_{Y_0}^*B_0) = 0$, which in turn yields
\begin{equation}
	\curv(\nabla^{\CG_1}) = \curv(\nabla^{\CG_0})\,.
\end{equation}
This poses a very strong restriction on pairs of bundle gerbes which admit morphisms between them.

An even stronger constraint which is often encountered in higher gauge theory is the so-called fake-curvature condition~\cite{BS--Higher_gauge_theory}
\begin{equation}
\label{eq:fake_curvature_condition}
	\curv(\nabla^E) - (\zeta_{Y_1}^*B_1 - \zeta_{Y_0}^*B_0) \otimes 1_E = 0\,.
\end{equation}
From a physical point of view this seems too strong a condition to impose generally, as, for instance, it trivialises the DBI action of D-branes~\cite{Szabo--String-theory_and_D-brane_dynamics}.
We will see in Section~\ref{sect:transgression_functor} that one of the consequences of the fake-curvature condition is that sections with this property transgress to parallel sections over the loop space of $M$.
\qen
\end{remark}

The \emph{identity 1-morphism} on a bundle gerbe with connection $(\CG,\nabla^\CG)$ is given by
\begin{equation}
\label{eq:identity_1-morphism_of_BGrb}
	1_{(\CG,\nabla^\CG)} = \big( (L, \nabla^L), p_{013}^*\mu^{-1} \circ p_{023}^*\mu, Y^{[2]}, 1_{Y^{[2]}} \big) \,,
\end{equation}
where we are using the identification $Y^{[2]} {\times}_M Y^{[2]} \cong Y^{[4]}$ and where $p_{ij}$ denotes the projection on the $i$-th and $j$-th factors in $Y^{[4]}$.

For any hermitean line bundle with connection $(L,\nabla^L)$ on a manifold $M$ there exists a unitary parallel isomorphism $\delta_L \in \HLBdl^\nabla_\rmuni(M)( L \otimes L^*, I_0)$, which is sends an endomorphism $\psi = f \cdot 1_L$ of $L$ to $f$.
The morphism of line bundles $\delta_L$ has an analogue for bundle gerbes.
For a bundle gerbe $(\CG,\nabla^\CG)$ we define
\begin{equation}
\begin{aligned}
	\delta_{(\CG,\nabla^\CG)} &{}\coloneqq \big( (L, \nabla^L)^*, (1 \otimes \delta_{p_{12}^*L} \otimes \delta_{p_{23}^*L^*}) \circ (p_{012}^*\mu^\sft \otimes p_{123}^*\mu^{-1} \otimes 1), Y^{[2]}, 1_{Y^{[2]}} \big)
	\\
	&\quad \colon (\CG_, \nabla^{\CG})^* \otimes (\CG_, \nabla^{\CG}) \to \CI_0\,.
\end{aligned}
\end{equation}

\begin{example}
\begin{myenumerate}
	\item Let $\rho_0, \rho_1 \in \Omega^2(M,\iu\, \FR)$.
	Morphisms $\CI_{\rho_0} \to \CI_{\rho_1}$ consist of a surjective submersion $Z \to M$, a hermitean vector bundle $(E,\nabla^E)$ over $Z$ with connection, and a unitary parallel isomorphism $\alpha \colon d_1^*(E,\nabla^E) \to d_0^*(E,\nabla^E)$ over $Z^{[2]}$ which satisfies a cocycle relation over $Z^{[3]}$ (as~\eqref{eq:1-morphisms_compatibility_with_BGrb_multiplications} becomes trivial in this case).
	That is, morphisms $\CI_{\rho_0} \to \CI_{\rho_1}$ are equivalently descent data for hermitean vector bundles with connections on $M$.
	If $(E,\alpha)$ is such a morphism, we will denote the corresponding descended hermitean vector bundle with connection by $\sfR(E,\alpha)$ (see Appendix~\ref{app:special_morphisms_and_descent} for more on morphisms and descent).
	The observation that 1-morphisms $\CI_0 \to \CI_0$ form descent data for hermitean vector bundles with connection on $M$, and that, in particular, every $(E, \nabla^E) \in \HVBdl^\nabla(M)$ defines such a morphism will be crucial in order to understand bundle gerbes as higher line bundles in Chapter~\ref{ch:2Hspaces_from_bundle_gerbes}.
	
	\item Given two pairs of bundle gerbes $(\CG_i, \nabla^{\CG_i}) = ( L_i, \nabla^{L_i}, \mu_i, B_i, Y_i, \pi_i )$ and $(\CG'_i, \nabla^{\CG'_i}) = ( L'_i, \nabla^{L'_i}, \mu'_i, B'_i, Y'_i, \pi'_i )$, for $i= 0,1$, and 1-morphisms $(E, \nabla^E, \alpha, Z, \zeta) \colon (\CG_0, \nabla^{\CG_0}) \to (\CG_1, \nabla^{\CG_1})$ and $(E', \nabla^{E'}, \alpha', Z', \zeta') \colon (\CG'_0, \nabla^{\CG'_0}) \to (\CG'_1, \nabla^{\CG'_1})$, the morphism
	\begin{equation}
	\begin{aligned}
		&(E, \nabla^E, \alpha, Z, \zeta) \otimes (E', \nabla^{E'}, \alpha', Z', \zeta')
		\\
		&= \big(  \pr_Z^*(E,\nabla^E) \otimes \pr_{Z'}^*(E',\nabla^{E'}),\, \pr_Z^{[2]*} \alpha \otimes \pr_{Z'}^{[2]*} \alpha' ,\, Z {\times}_M Z',\, \zeta {\times}_M \zeta' \big)
	\end{aligned}
	\end{equation}
	is called the \emph{tensor product of $(E,\alpha)$ and $(E',\alpha')$}.
	It provides a 1-morphism
	\begin{equation}
		(E,\alpha) \otimes (E',\alpha') \colon (\CG_0,\nabla^{\CG_0}) \otimes (\CG'_0,\nabla^{\CG'_0}) \to (\CG_1,\nabla^{\CG_1}) \otimes (\CG'_1,\nabla^{\CG'_1})\,.
	\end{equation}
	
	\item If $(E, \nabla^E, \alpha, Z, \zeta) \colon (\CG_0, \nabla^{\CG_0}) \to (\CG_1, \nabla^{\CG_1})$, then we obtain the \emph{transposed morphism} $(E,\alpha)^\sft \colon (\CG_1,\nabla^{\CG_1})^* \to (\CG_0, \nabla^{\CG_0})^*$ by setting
	\begin{equation}
	\label{eq:tranpose_of_1-morphism_of_BGrbs}
		(E, \nabla^E, \alpha, Z, \zeta)^\sft
		= \big( (E, \nabla^E),\, (1 \otimes \delta_{L_1}) \circ (1 \otimes \alpha \otimes 1) \circ (\delta_{L_0}^{-1} \otimes 1),\, Z,\, \sw \circ \zeta \big)\,,
	\end{equation}
	where we have omitted pullbacks and where $\sw \colon Y_0 {\times_M} Y_1 \to Y_1 {\times_M} Y_0$ swaps the factors in the fibre product.%
	\footnote{In~\cite{Waldorf--More_morphisms} this is called the \emph{dual morphism} instead of the transpose.
	However, we feel that this nomenclature would be misleading in view of Section~\ref{sect:Pairings_and_inner_hom_of_morphisms_in_BGrb}.}
	It is worth to note that for a morphism $(E,\alpha) \colon \CI_{\rho_0} \to \CI_{\rho_1}$, the data of its transpose is exactly the same as that of $(E,\alpha)$, but now $(E,\alpha)^\sft \colon \CI_{-\rho_1} \to \CI_{-\rho_0}$.
	\qen
\end{myenumerate}
\end{example}

The set-up for bundle gerbes so far fails to assemble bundle gerbes and their morphisms into a category.
Strictly speaking, composition of morphisms is not associative, as the tensor product of vector spaces and vector bundles is not associative, but under the conventions spelled out in Appendix~\ref{app:monoidal_structures_and_strictness} it becomes virtually associative.
However, one can check that pre- and postcomposition by the identity morphism do not act as the identity on morphisms of bundle gerbes.
The way to resolve this issue is to go to the less strict and, therefore, richer world of 2-categories.
As before, we write $Y_{01} \coloneqq Y_0 {\times}_M Y_1$.

\begin{definition}[2-morphisms of bundle gerbes]
\label{def:2-morphisms_of_BGrbs}
Let $(\CG_i, \nabla^{\CG_i}) \in \BGrb^\nabla(M)$ be bundle gerbes with connection on $M$ for $i = 0,1$.
Let $(E, \nabla^E, \alpha, Z, \zeta)$ and $(E', \nabla^{E'}, \alpha', Z', \zeta')$ be 1-morphisms $(\CG_0, \nabla^{\CG_0}) \to (\CG_1, \nabla^{\CG_1})$ of bundle gerbes with connections on $M$.
A \emph{2-morphism $(E,\alpha) \to (E',\alpha')$} is an equivalence class of triples $(W,\omega, \psi)$, where $\omega \colon W \to Z {\times}_{Y_{01}} Z'$ is a surjective submersion and $\psi \in \HVBdl^\nabla(W)(\omega_Z^*E, \omega_{Z'}^*E')$, where we have set $\omega_Z = \pr_Z \circ \omega$ and $\omega_{Z'} = \pr_{Z'} \circ \omega$.
Note that we can define $\omega_{Y_i} \coloneqq \zeta_{Y_i} \circ \omega_Z = \zeta'_{Y_i} \circ \zeta'_{Y_i}$.
This morphism is required to make the following diagram over $W^{[2]} = W {\times}_M W$ commute:
\begin{equation}
\label{eq:2-morphism_compatibility_with_alphas}
\begin{tikzcd}[column sep=3cm, row sep=1.25cm]
	\omega_{Y_0}^{[2]*} L_0 \otimes d_0^*\omega_Z^*E \ar[r, "\omega_Z^{[2]*} \alpha"] \ar[d, "1 \otimes d_0^*\psi"'] & d_1^*\omega_Z^*E \otimes \omega_{Y_1}^{[2]*} L_1 \ar[d, "d_1^*\psi \otimes 1"]
	\\
	\omega_{Y_0}^{[2]*} L_0 \otimes d_0^*\omega_{Z'}^*E' \ar[r, "\omega_{Z'}^{[2]*} \alpha'"'] & d_1^*\omega_{Z'}^*E' \otimes \omega_{Y_1}^{[2]*} L_1
\end{tikzcd}
\end{equation}

Two triples $(W,\omega,\psi)$ and $(W',\omega', \psi')$ are equivalent if there exists a triple $(X, \chi_W, \chi_{W'})$ of a manifold $X$ and surjective submersions $\chi_W \colon X \to W$ and $\chi_{W'} \colon X \to W'$ such that we have a commuting diagram
\begin{equation}
\begin{tikzcd}
	& X \ar[dl, "\chi_W"'] \ar[dr, "\chi_{W'}"] &
	\\
	W \ar[dr, "\omega"'] & & W' \ar[dl, "\omega"]
	\\
	& Z {\times}_{Y_{01}} Z' &
\end{tikzcd}
\end{equation}
and such that $\chi_{W'}^*\psi' = \chi_W^* \psi$.
\end{definition}

In Proposition~\ref{st:2-morphisms_have_simple_representatives} we show that every 2-morphism $[W, \omega, \psi]$ as above has a representative of the form $[Z {\times}_{Y_{01}} Z',\, 1_{(Z {\times}_{Y_{01}} Z')},\, \phi]$.
For representatives of this special form we will often use the shorthand notation
\begin{equation}
	[Z {\times}_{Y_{01}} Z',\, 1,\, \phi] \coloneqq [Z {\times}_{Y_{01}} Z',\, 1_{(Z {\times}_{Y_{01}} Z')},\, \phi]\,.
\end{equation}
Consider three bundle gerbes with connection $(\CG_i, \nabla^{\CG_i}) \in \BGrb^\nabla(M)$, for $i = 0,1,2$ and 1-morphisms $(E,\nabla^E,\alpha,Z,\zeta)$, $(E',\nabla^{E'},\alpha',Z',\zeta')$, and $(E'',\nabla^{E''},\alpha'',Z'',\zeta'')$ from $(\CG_0, \nabla^{\CG_0})$ to $(\CG_1, \nabla^{\CG_1})$, as well as $(F, \nabla^F, \beta,X, \xi)$, $(F', \nabla^{F'}, \beta', X', \xi)$ from $(\CG_1,\nabla^{\CG_1})$ to $(\CG_2,\nabla^{\CG_2})$.
The \emph{vertical composition} of $[W, \omega, \phi] \colon (E,\alpha) \to (E',\alpha',)$ and $[W',\omega',\phi'] \colon (E',\alpha') \to (E'',\alpha'')$ is the equivalence class of
\begin{equation}
\label{eq:vertical_composition_of_2-morphisms}
	(W',\omega',\phi') \circ_2 (W, \omega, \phi)
	= \big( W {\times}_{Z'} W',\, (\omega_Z \circ \pr_W) {\times}_{Y_{10}} (\omega'_{Z''} \circ \pr_{W'}),\, \pr_{W'}^*\phi' \circ \pr_W^*\phi \big)\,.
\end{equation}
The \emph{horizontal composition} of a pair of 2-morphisms $[W, \omega, \phi] \colon (E,\alpha) \to (E',\alpha')$ and $[W'', \omega'', \phi''] \colon (F, \beta) \to (F', \beta')$ is represented by
\begin{equation}
\label{eq:horizontal_composition_of_2-morphisms}
\begin{aligned}
	(W'', \omega'', \phi'') \circ_1 (W, \omega, \phi) &= (U, \omega^U, \psi) \,, \quad \text{where}
	\\
	U &= (X {\times}_{Y_1} Z) {\times}_{Y_{02}} (W \times_{Y_1} W'') {\times}_{Y_{02}} (X' {\times}_{Y_1} Z')
	\\
	\omega^U &= (\pr_{(X {\times}_{Y_1} Z)}) {\times}_{M} (\pr_{(X' {\times}_{Y_1} Z')})\,,
	\\
	\psi &= \dd_{(F',\beta') \circ (E',\alpha')} \circ (\phi \otimes\phi'') \circ \dd_{(F,\beta) \circ (E,\alpha)} \,,
\end{aligned}
\end{equation}
with pullbacks omitted, and $\dd_{(-)}$ as defined in Lemma~\ref{st:dd-def_and_properties}.

\begin{example}
\label{eg:2-morphisms--new_from_old_identity_unitors}
Important 2-morphisms of bundle gerbes are given as follows:
\begin{myenumerate}
	\item Consider bundle gerbes $(\CG_i, \nabla^{\CG_i}) \in \BGrb^\nabla(M)$ for $i = 0, 1, 2, 3$ and 1-morphisms $(E, \nabla^E, \alpha, Z, \zeta)$, $(E', \nabla^{E'}, \alpha', Z', \zeta') \colon (\CG_0, \nabla^{\CG_0}) \to (\CG_1, \nabla^{\CG_1})$ as well as $(F, \nabla^F, \beta, X, \xi)$, $(F', \nabla^{F'}, \beta', X', \xi') \colon (\CG_2, \nabla^{\CG_2}) \to (\CG_3, \nabla^{\CG_3})$.
	Given, moreover, a pair of 2-morphisms $[W, \omega, \psi] \colon (E,\alpha) \to (E', \alpha')$ and $[U, \nu, \phi] \colon (F, \beta) \to (F', \beta')$, their \emph{tensor product} reads as
	\begin{equation}
	\label{eq:tensor_product_of_2-morphisms}
	\begin{aligned}
		&[U, \nu, \phi] \otimes [W, \omega, \psi]
		= \big[ U {\times}_M W,\, \nu {\times}_M \omega,\, \pr_U^*\phi \otimes \pr_W^*\psi \big]
		\\
		&\quad \colon (F, \beta) \otimes (E,\alpha) \to (F', \beta') \otimes (E',\alpha')\,.
	\end{aligned}
	\end{equation}
	
	\item The \emph{transpose of a 2-morphism} $[W, \omega, \psi] \colon (E,\alpha) \to (F, \beta)$ is given by~\cite[p.~60]{Waldorf--Thesis}
	\begin{equation}
		[W, \omega, \psi]^\sft = [W, \omega, \psi] \colon (E,\alpha)^\sft \to (F, \beta)^\sft\,.
	\end{equation}
	
	\item We define the \emph{adjoint of a 2-morphism} $[W, \omega, \psi] \colon (E,\alpha) \to (F, \beta)$ to be
	\begin{equation}
	\label{eq:adjoint_of_2-morphism_in_BGrb}
		[W, \omega, \psi]^* \coloneqq [W, \sw \circ \omega, \psi^*] \colon (F, \beta) \to (E,\alpha)\,,
	\end{equation}
	where $\sw \colon X {\times}_{Y_{01}} Z \to Z {\times}_{Y_{01}} X$, $(x,z) \mapsto (z,x)$.
	The compatibility condition and change of direction is seen by applying $(-)^*$ to~\eqref{eq:2-morphism_compatibility_with_alphas}.
	
	\item The \emph{identity 2-morphism} on $(E, \nabla^E, \alpha, Z, \zeta)$ is~\cite[p.~41]{Waldorf--Thesis}
	\begin{equation}
		1_{(E,\alpha)} = \big[ Z {\times}_{Y_{01}} Z, 1_{Z {\times}_{Y_{01}} Z}, \dd_{(E,\alpha)} \big]\,,
	\end{equation}
	where $\dd_{(E,\alpha)}$ is defined in Lemma~\ref{st:dd-def_and_properties}.
	
	\item There exist \emph{left and right unitors} for morphisms of bundle gerbes~\cite[Section 2.3]{Waldorf--Thesis}, i.e. natural 2-isomorphisms
	\begin{equation}
		\lambda_{(E,\alpha)} \colon (E,\alpha) \circ 1_{(\CG_0, \nabla^{\CG_0})} \to (E,\alpha)\,,
		\quad
		\rho_{(E,\alpha)} \colon 1_{(\CG_1, \nabla^{\CG_1})} \circ (E,\alpha) \to (E,\alpha)\,,
	\end{equation}
	which establish the identity 1-morphisms~\eqref{eq:identity_1-morphism_of_BGrb} as weak identities.
	\qen
\end{myenumerate}
\end{example}

We have now gathered structure sufficient to organise the collection of bundle gerbes and their 1-morphisms and 2-morphisms into a 2-category.
Theorem~\ref{st:2-categories_of_bundle_gerbes} and Definition~\ref{def:2-categories_of_BGrbs} slightly modify those in~\cite{Waldorf--More_morphisms,Waldorf--Thesis} to have different morphism categories.
However, the important structural statements derived there and generalised in Appendix~\ref{app:special_morphisms_and_descent} hold true in any of the below settings.
This is a consequence of the fact that $\HVBdl^\nabla$, $\HVBdl^\nabla_\rmpar$, and $\HVBdl$, and $\HVBdl^\nabla_\rho$ for any $\rho \in \Omega^2(M, \iu\, \FR)$ form sheaves of categories.
The latter sheaf of categories has objects $(E,\nabla^E) \in \HVBdl^\nabla(M)$ satisfying $\tr(\curv(\nabla^E) - \rho \cdot 1_E) = 0$, while morphisms are given by parallel morphisms of hermitean vector bundles with connections.
All structural morphisms which we used on hermitean vector bundles are in fact unitary parallel isomorphisms, and hence are inhabitants to any of the above sheaves of categories (for $\HVBdl$ after forgetting connections).

\begin{theorem}[{\cite[Chapter 2]{Waldorf--Thesis}}]
\label{st:2-categories_of_bundle_gerbes}
Bundle gerbes with connection on $M$, together with their 1-morphisms and 2-morphisms as defined above, form a symmetric monoidal 2-category.
Composition of 1-morphisms is strictly associative, and the tensor product is strictly associative as well as strictly unital.
\end{theorem}

\begin{remark}
Strictly speaking, the second part of Theorem~\ref{st:2-categories_of_bundle_gerbes} is a consequence of the conventions we have adapted regarding the monoidal structures on $\HVBdl^\nabla(M)$ and the category of surjective submersions over $M$, as well as on the composability of pullbacks of vector bundles.
For more on this, see Appendix~\ref{app:monoidal_structures_and_strictness}.
\qen
\end{remark}

\begin{definition}[2-categories of bundle gerbes]
\label{def:2-categories_of_BGrbs}
We refer to the symmetric monoidal 2-category thus defined as the \emph{2-category of bundle gerbes with connections on $M$} and denote it by $(\BGrb^\nabla(M), \otimes\,)$.
\\
Using the same objects and 1-morphisms, but restricting 2-morphisms to be made of parallel morphisms of hermitean vector bundles with connection, we obtain the symmetric monoidal sub-2-category $(\BGrb^\nabla_\rmpar(M), \otimes\,)$ of \emph{bundle gerbes with connection on $M$ and parallel 2-morphisms}.
\\
This, in turn, has a symmetric monoidal sub-2-category $(\BGrb^\nabla_\rmflat(M), \otimes\,)$ which has 1-morphisms those 1-morphisms in $\BGrb^\nabla(M)$ that satisfy~\eqref{eq:trace_condition}, and 2-morphisms made from unitary, parallel isomorphisms of hermitean vector bundles with connection.
We call this 2-category the \emph{2-category of bundle gerbes with connections on $M$ and flat morphisms}.
\\
Forgetting about connections on all levels of $\BGrb^\nabla(M)$ yields the symmetric monoidal \emph{2-category $(\BGrb(M),\otimes\, )$ of bundle gerbes on $M$}.
\end{definition}

We can now state the first crucial results about these 2-categories.

\begin{proposition}[{\cite[Proposition 2.3.4]{Waldorf--Thesis}}]
\label{st:classification_of_isomps_of_BGrbs}
A morphism $(E, \alpha)$ in any of the above 2-categories is (weakly) invertible if and only if its underlying vector bundle $E$ is of rank one.
\end{proposition}

\begin{proposition}
\label{st:determinants_of_morphisms_of_BGrbs}
Let $(\CG_i, \nabla^{\CG_i}) = (L_i, \nabla^{L_i}, \mu_i, B_i, Y_i, \pi_i)$ be bundle gerbes with connection on $M$ for $i=0,1$.
The following statements hold true.
\begin{myenumerate}
	\item If $(Y_0, \pi_0) = (Y_1, \pi_1) = (Y, \pi)$, i.e. the bundle gerbes are defined over the same surjective submersion onto $M$, there exist isomorphisms
	\begin{equation}
	\label{eq:simplified_tensor_product_over_same_sursub}
		(\CG_0, \nabla^{\CG_0}) \otimes (\CG_1, \nabla^{\CG_1}) \cong \big( (L_0, \nabla^{L_0}) \otimes (L_1, \nabla^{L_1}), \mu_0 \otimes \mu_1, B_0 + B_1, Y, \pi \big)\,.
	\end{equation}
	
	\item For generic bundle gerbes $(\CG_i, \nabla^{\CG_i})$, i.e. without the assumption made under (1), if there exists a morphism $(E,\alpha) \colon (\CG_0, \nabla^{\CG_0}) \to (\CG_1, \nabla^{\CG_1})$, then there exists an isomorphism $(\CG_0, \nabla^{\CG_0})^{\otimes \rank(E)} \cong (\CG_1, \nabla^{\CG_1})^{\otimes \rank(E)}$.
\end{myenumerate}
\end{proposition}

\begin{proof}
Writing out the tensor product (see~\eqref{eq:tensor_product_of_BGrbs}) explicitly for this case yields
\begin{equation}
\begin{aligned}
	&(\CG_0, \nabla^{\CG_0}) \otimes (\CG_1, \nabla^{\CG_1})
	\\*
	&= \big( p_{02}^* (L_0, \nabla^{L_0}) \otimes p_{13}^* (L_1, \nabla^{L_1}),\, p_{024}^*\mu_0 \otimes p_{135}^*\mu_1,\, Y^{[2]},\, \pi^{[2]} \big)\,.
\end{aligned}
\end{equation}
An isomorphism as required is given by the tuple
\begin{equation}
\begin{aligned}
	(J, \beta) \coloneqq &\Big( p_{02}^*(L_0, \nabla^{L_0}) \otimes p_{12}^*(L_1, \nabla^{L_1}),
	\\
	&\qquad \big( p_{025}^* \mu_0^{-1} \circ p_{035}^* \mu_0 \big) \otimes \big( p_{125}^*\mu_1^{-1} \circ p_{145}^* \mu_1 \big),\,	Y^{[3]},\, \pi^{[3]} \Big)\,.
\end{aligned}
\end{equation}
This proves (1).

In order to see (2), we observe that we obtain a morphism
\begin{align}
	&\big( \det(E, \nabla^E), \det(\alpha), Z, \zeta \big) \colon
	\\
	&\big( (L_0, \nabla^{L_0})^{\otimes n}, \mu_0 \otimes \ldots \otimes \mu_0, n B_0, Y_0, \pi_0 \big) \to \big( (L_1, \nabla^{L_1})^{\otimes n}, \mu_1 \otimes \ldots \otimes \mu_1, n B_1, Y_1, \pi_1 \big)\,, \notag
\end{align}
where $n = \rank(E)$.
This is an isomorphism by Proposition~\ref{st:classification_of_isomps_of_BGrbs}.
From (1) we know, moreover, that source and target of this isomorphism are canonically isomorphic to $(\CG_0, \nabla^{\CG_0})^{\otimes n}$ and $(\CG_1, \nabla^{\CG_1})^{\otimes n}$, respectively.
\end{proof}

The observation that the determinant of a bundle gerbe morphism induces an isomorphism has been made in~\cite{BCMMS}.
Here we extended that statement to bundle gerbes over different surjective submersions by finding an explicit isomorphism~\eqref{eq:simplified_tensor_product_over_same_sursub} to the simplified version of the tensor product of bundle gerbes defined over the same surjective submersions.
The existence of such an isomorphism could have been deduced implicitly from classification results in terms of Deligne cohomology (cf. Definition~\ref{def:Deligne_complex_and_cohomology} and Theorem~\ref{st:Classification_of_BGrbs_by_Deligne_coho}).

\begin{definition}[Trivialisations of bundle gerbes]
Given a bundle gerbe with connection $(\CG, \nabla^{\CG})$ on $M$, we call a flat isomorphism $(S,\beta) \in \BGrb^\nabla_{\rmflat\sim}((\CG, \nabla^{\CG}), \CI_\rho)$ for some $\rho \in \Omega^2(M, \iu\, \FR)$ a \emph{trivialisation of $(\CG, \nabla^{\CG})$}.
\end{definition}

\begin{corollary}
\label{st:determinants_of_sections_and_dual_sections}
If there exists a morphism $(E,\alpha) \colon (\CG, \nabla^{\CG}) \to \CI_\rho$, or alternatively a morphism $(E', \alpha') \colon \CI_\rho \to (\CG, \nabla^{\CG})$, then there exists a trivialisation of $(\CG, \nabla^{\CG})^{\otimes \rank(E)}$.
\end{corollary}

\begin{proof}
Consider the case where $(E, \nabla^E, \alpha, Z, \zeta) \colon (\CG, \nabla^\CG) \to \CI_\rho$ and let $n = \rank(E)$.
As in the proof of Proposition~\ref{st:determinants_of_morphisms_of_BGrbs}, we obtain a 1-isomorphism $\det(E,\alpha) \colon ( (L, \nabla^L)^{\otimes n}, \mu \otimes \ldots \otimes \mu, n B, Y, \pi \big) \to \CI_{\rho}$.
Precomposing with the isomorphism from part (1) of Proposition~\ref{st:determinants_of_morphisms_of_BGrbs} thus yields an isomorphism $(F, \nabla^F, \beta, X, \chi) \colon (\CG,\nabla^\CG)^{\otimes n} \to \CI_\rho$, whose explicit form we do not need here.
While this is an isomorphism in $\BGrb^\nabla(M)$, it is not necessarily flat, for there is no relation between $\curv(\nabla^F)$ and $\rho$.
Let
\begin{equation}
	B_{\CG^{\otimes n}} = \sum_{i = 0}^{n-1}\, d_i^* B \quad \in \Omega^2(Y^{[n]}, \iu\, \FR)
\end{equation}
be the curving of $(\CG, \nabla^\CG)^{\otimes n}$.
Then, defining $\eta \in \Omega^2(M,\iu\, \FR)$ to be the unique 2-form such that $\chi_M^* \eta = \tr( \curv(\nabla^F) + \chi^*B_{\CG^{\otimes n}} \cdot 1_F)$, we see that $(F, \beta) \colon (\CG, \nabla^\CG)^{\otimes n} \to \CI_\eta$ is a flat isomorphism of bundle gerbes, i.e. a trivialisation of $(\CG,\nabla^\CG)^{\otimes n}$.
\end{proof}

Proposition~\ref{st:determinants_of_morphisms_of_BGrbs} and Corollary~\ref{st:determinants_of_sections_and_dual_sections} show that the theory of bundle gerbes with connections is much stricter than for example that of line bundle with connections, where the existence of generic morphisms between line bundles does not have implications of this kind.
We will comment on possible ways to weaken or circumvent this restriction in Section~\ref{sect:ways_around_the_torsion_constraint}.

\begin{remark}
Proposition~\ref{st:determinants_of_morphisms_of_BGrbs} and Corollary~\ref{st:determinants_of_sections_and_dual_sections} have analogues in the case without connection:
If there exists a morphism $(E,\alpha)$ between two bundle gerbes without connection, then the determinant induces an isomorphism in $\BGrb(M)$ between their $\rank(E)$-th tensor powers.
\end{remark}

\section{Deligne cohomology and higher categories}
\label{sect:Deligne_coho_and_higher_cats}

The main goal of this section is to describe the relation between the geometric theory of bundle gerbes and the sheaf and homotopy theoretic description of gerbes using simplicial abelian groups of local data for gerbes.
The Deligne complex is seen to interpolate between the two points of view, and we make this interpolation very explicit.
The main result is Theorem~\ref{st:Deligne_2-skeleton_and_Bgrbs}, which is a refinement of~\cite[Proposition 2.6.1]{Waldorf--Thesis}.
In order to prove the result, we have to make use of our results from Appendix~\ref{app:BGrbs_with_mutual_surjective_submersions} in several places.

The results in this section are interesting mostly from the conceptual point of view.
Apart from Definition~\ref{def:Deligne_complex_and_cohomology} and Theorem~\ref{st:Classification_of_BGrbs_by_Deligne_coho}, findings in this section do not make appearances in other parts of this thesis.
Hence, readers unfamiliar with the language of simplicial sets may well take note of Definition~\ref{def:Deligne_complex_and_cohomology} and Theorem~\ref{st:Classification_of_BGrbs_by_Deligne_coho} and proceed to Section~\ref{sect:Additive_structures_on_morphisms_in_BGrb}.

Denote by $\HLBdl_\rmuni(M)$ the symmetric monoidal groupoid of hermitean line bundles on $M$ and unitary isomorphisms.
There exists an isomorphism of abelian groups
\begin{equation}
	\pi_0 \big( \HLBdl_\rmuni(M), \otimes\, \big) \cong \rmH^1 \big( M,\sfU(1) \big) \cong \rmH^2(M,\RZ)\,.
\end{equation}
This classification can be refined to include connections:
let $(U_a)_{a \in \Lambda}$ be a good open covering of $M$, i.e. such that all possible finite intersections of the $U_a$ are diffeomorphic to $\FR^{\dim(M)}$.
Good open coverings exist on any manifold, as they can be constructed from geodesic balls with respect to a Riemannian metric, and every manifold admits a Riemannian metric due to a partition of unity argument.

Recall the definition of the descent category and the ascent functor with respect to a covering and a (relative) sheaf of categories from Section~\ref{sect:Coverings_and_sheaves_of_cats}.
We write $\CU = \bigsqcup_{a \in \Lambda} U_a$ and note that the canonical map $\pi \colon \CU \to M$, $(x,a) \mapsto x$, defines a surjective submersion.
Any $(L, \nabla^L) \in \HLBdl^\nabla(M)$ induces descent data $(\pi^*L, \alpha) \eqqcolon \Asc_\pi(L,\nabla^L) \in \Desc(\HLBdl^\nabla, \pi)$ with bundle $\pi^*L_{|(x,a)} = L_{|x}$ and $\alpha_{|((x,a), (x,b))} = 1_{L_{|x}}$.
As $U_a \cong \FR^{\dim(M)}$ for every $a \in \Lambda$, we can find, for every $a \in \Lambda$, parallel unitary isomorphisms $\psi_a \colon L_{|U_a} \to I_{A_a}$ for some $A_a \in \Omega^1(U_a, \iu\, \FR)$.
Any such family of isomorphisms gives rise to a descent isomorphism
\begin{equation}
	\Asc_\pi(L, \nabla^L) \cong \big( \{I_{A_a}\},\, \psi_b \circ \psi_a^{-1} \big) \eqqcolon (A_a, g_{ab})\,,
\end{equation}
where $\{I_{A_a}\} \to \CU$ is the hermitean line bundle defined by $\{I_{A_a}\}_{|U_a} = I_{A_a}$.
By construction, the $g_{ab}$ are isomorphisms of line bundles with connections, so that we have $A_b = A_a + g_{ab}^*\,\mu_{\sfU(1)}$ with $\mu_{\sfG}$ denoting the Maurer-Cartan form of a Lie group $\sfG$.
The ascent functor $\Asc_\pi$ is an equivalence of categories.
Therefore, adjoining the isomorphisms $\psi_a$ to the image $\Asc_\pi\phi$ of a morphism of line bundles, we obtain a bijection
\begin{equation}
\label{eq:local_morphs_of_HLBdls}
\begin{aligned}
	&\HLBdl_{\rmuni}^\nabla(M) \big( (L, \nabla^L), (L', \nabla^{L'}) \big)
	\\
	&\quad \cong \big\{ (\phi_a \colon U_a \to \sfU(1))_{a \in \Lambda}\, \big| \, \phi_b\, g_{ab} = g'_{ab}\, \phi_a,\, A'_a = A_a + \phi_a^*\,\mu_{\sfU(1)} \ \forall\, a,b \in \Lambda \big\}\,,
\end{aligned}
\end{equation}
with $(\{I_{A'_a}\}, g'_{ab})$ a chosen representative of $\Asc_\pi(L', \nabla^{L'})$.
We thus define a groupoid $\scD^1(\CU)$ with objects being collections $(A_a, g_{ab})_{a,b \in \Lambda}$ as above and morphisms $(A_a, g_{ab}) \to (A'_a, g'_{ab})$ given by collections $(\phi_a \colon U_a \to \sfU(1))_{a \in \Lambda}$ such that $\phi_b\, g_{ab} = g'_{ab}\, \phi_a$ for all $a,b \in \Lambda$.

It follows that there exists an equivalence of symmetric monoidal groupoids
\begin{equation}
\label{eq:groupoid_of_Deligne_descent_data_for_HLBdl^nabla}
	\big( \HLBdl^\nabla_\rmuni(M), \otimes\, \big) \cong \big( \scD^1(\CU), \otimes\, \big)\,,
\end{equation}
where $(A_a, g_{ab}) \otimes (A'_a, g'_{ab}) = (A_a + A'_a, g_{ab}\, g'_{ab})$ on objects and $(\phi_a) \otimes (\phi'_a) = (\phi_a\, \phi'_a)$ on morphisms.
This induces an isomorphism of abelian groups
\begin{equation}
	\pi_0 \big( \HLBdl^\nabla_\rmuni(M), \otimes\, \big) \cong \pi_0 \big( \scD^1(\CU), \otimes\, \big) \eqqcolon \rmH^1(M,\scD^\diamond_1)\,.
\end{equation}

This is part of a much more general scheme.
To a good open covering $\pi \colon \CU \to M$ we can associate a double chain complex $\check{\scD}^n_{\diamond,\diamond}(\CU)$ which reads as
\begin{equation}
\label{eq:Cech_Deligne_double_complex}
\begin{tikzcd}[row sep=1cm]
	& 0 \ar[d] & 0 \ar[d] & & 0 \ar[d] & 0 \ar[d] &
	\\
	0 \ar[r] & \widetilde{\Omega}^0(\CU) \ar[r, "{\widetilde{\dd}}"] \ar[d, "{\check{\delta}}"] & \widetilde{\Omega}^1(\CU) \ar[r, "{\widetilde{\dd}}"] \ar[d, "{\check{\delta}}"] & \ \ldots\ \ar[r, "{\widetilde{\dd}}"] & \widetilde{\Omega}^{n-1}(\CU) \ar[r, "{\widetilde{\dd}}"] \ar[d, "{\check{\delta}}"] & \widetilde{\Omega}^n(\CU) \ar[r] \ar[d, "{\check{\delta}}"] & 0
	\\
	0 \ar[r] & \widetilde{\Omega}^0(\CU^{[2]}) \ar[r, "{\widetilde{\dd}}"] \ar[d, "{\check{\delta}}"] & \widetilde{\Omega}^1(\CU^{[2]}) \ar[r, "{\widetilde{\dd}}"] \ar[d, "{\check{\delta}}"] & \ \ldots\ \ar[r, "{\widetilde{\dd}}"] & \widetilde{\Omega}^{n-1}(\CU^{[2]}) \ar[r, "{\widetilde{\dd}}"] \ar[d, "{\check{\delta}}"] & \widetilde{\Omega}^n(\CU^{[2]}) \ar[r] \ar[d, "{\check{\delta}}"] & 0
	\\
	0 \ar[r] & \widetilde{\Omega}^0(\CU^{[3]}) \ar[r, "{\widetilde{\dd}}"] \ar[d, "{\check{\delta}}"] & \widetilde{\Omega}^1(\CU^{[3]}) \ar[r, "{\widetilde{\dd}}"] \ar[d, "{\check{\delta}}"] & \ \ldots\ \ar[r, "{\widetilde{\dd}}"] & \widetilde{\Omega}^{n-1}(\CU^{[3]}) \ar[r, "{\widetilde{\dd}}"] \ar[d, "{\check{\delta}}"] & \widetilde{\Omega}^n(\CU^{[3]}) \ar[r] \ar[d, "{\check{\delta}}"] & 0
	\\
	& \vdots & \vdots & &\vdots & \vdots &
\end{tikzcd}
\end{equation}
Here we have set $\widetilde{\Omega}^k \coloneqq \Omega^k$ for $k > 0$ and $\widetilde{\Omega}^0 \coloneqq \Mfd(-,\sfU(1))$, with $\check{\delta}$ denoting the \v{C}ech differential and $\widetilde{\dd}^k \coloneqq \dd^k$ denoting the de Rham differential for $k > 0$, while $\widetilde{\dd}^0 \coloneqq \dd \log \colon \widetilde{\Omega}^0 \to \Omega^1$.
We assign to $\widetilde{\Omega}^n(\CU)$ the degree $(0,0)$.
To this double chain complex there is an associated total chain complex $(\Tot_\diamond( \check{\scD}^n_{\diamond, \diamond}), \rmD)$.
In the case of $n = 1$, this has $\widetilde{\Omega}^0(\CU)$ in degree $i=1$, i.e. at degree one it consists of families $(\phi_a)_{a \in \Lambda}$ of smooth $\sfU(1)$-valued maps.
In degree $i = 0$ it has pairs $(A_a, g_{ab})_{a,b \in \Lambda}$ where $A_a \in \Omega^1(U_a, \iu\, \FR)$ and $g_{ab} \in \Mfd(U_a, \sfU(1))$.
However, the above discussion shows that only those pairs which satisfy $\rmD_{-1}(A_a, g_{ab}) = (0, A_b - A_a - g_{ab}^*\,\mu_{\sfU(1)}, \check{\delta}(g_{ab})) = 0$ provide descent data for a hermitean line bundle with connection on $M$.
Thus, we observe that
\begin{equation}
	\ker(\rmD_{-1}) = \obj \big( \scD^1(\CU) \big)\,.
\end{equation}
Consequently, we consider the truncation $\trunc_{\geq 0} \Tot_\diamond( \check{\scD}^1_{\diamond, \diamond})$ at degree $i = 0$, i.e. we replace $\Tot_0( \check{\scD}^1_{\diamond, \diamond})$ by $\ker(\rmD_{-1})$.
In this way, we enforce that we obtain a chain complex of abelian groups whose degree-zero group represents descent data for hermitean line bundles with connection on $M$:
\begin{equation}
	\big( \trunc_{\geq 0} \Tot_\diamond( \check{\scD}^1_{\diamond, \diamond}) \big)_0
	= \obj \big( \scD^1(\CU) \big)
	\subset \Desc (\HLBdl^\nabla, \pi)\,.
\end{equation}

For a category $\scC$ we write $s\scC \coloneqq \scCat(\Delta^\opp, \scC)$ for the category of simplicial objects in $\scC$.
We obtain a simplicial abelian group by applying the Dold-Kan correspondence to the complex $\trunc_{\geq 0} \Tot_\diamond( \check{\scD}^n_{\diamond, \diamond})$~\cite{GJ--Simplicial_homotopy_theory}.
This correspondence provides an equivalence of categories
\begin{equation}
	NC_\diamond :\, \rms\scA \longleftrightarrow \mathrm{Ch}_{\geq 0}(\scA)\, : \Gamma_\bullet\,,
\end{equation}
where $\scA$ is an abelian category, $\rms\scA = \scA^{\Delta^{\opp}}$ is the category of simplicial objects in $\scA$, and $\mathrm{Ch}_{\geq 0}(\scA)$ is the category of non-negatively graded chain complexes in $\scA$.
For $\scA = \Ab$, the category of abelian groups, even more is true:
There exists a forgetful functor $U \colon \rms\Ab \to \rms\Set$, and saying that $f \in \rms\Ab(A_\bullet,B_\bullet)$ is a weak equivalence or fibration if $Uf$ is so in $\rms\Set$ as well as defining cofibrations via the left-lifting property with respect to the trivial fibrations thus obtained induces a simplicial model structure on $\rms\Ab$.
The equivalence $(NC_\diamond, \Gamma_\bullet)$ is even a Quillen equivalence~\cite{SS--Equivalences_of_mon_MoCats} and induces isomorphisms $\pi_n(UA_\bullet) = \pi_n(A) \cong H_n(NC_\diamond (A_\bullet))$ and $H_n(C_\diamond) \cong \pi_n(\Gamma_\bullet (C_\diamond))$~\cite{GJ--Simplicial_homotopy_theory}.

We can, thus, define a simplicial abelian group
\begin{equation}
	(\check{\scB}^n_\nabla \sfU(1))_\bullet (\CU) \coloneqq
	\Gamma_\bullet \big( \trunc_{\geq 0} \Tot (\check{\scD}^n_{\diamond, \diamond}(\CU)) \big) \quad \in \rms\Ab\,.
\end{equation}
Note that because the chain complex $\Tot_\diamond (\check{\scD}^n_{\diamond, \diamond}(\CU))$ is concentrated in degrees $n \geq i \geq 0$, all $i$-simplicies in $(\check{\scB}^n_\nabla \sfU(1))_\bullet (\CU)$ with $i \geq n$ are degenerate.
In technical terms, the simplicial set $U \big( (\check{\scB}^n_\nabla \sfU(1))_\bullet (\CU) \big) \in \rms\Set$ is an $n$-skeleton.

This has an interpretation in terms of higher categories:
First, the underlying simplicial set $UA_\bullet \in \rms\Set$ of any simplicial abelian group $A_\bullet \in \rms\Ab$ is a \emph{Kan complex} (a fibrant simplicial set)~\cite{GSW--BGrbs_for_orientifolds}, meaning it has fillers for all horns $\Lambda^k_l \to UA_\bullet$.
In particular, $UA_\bullet$ has fillers for all inner horns, and thus forms a \emph{weak Kan complex}, or \emph{quasicategory}.
Quasicategories are considered models for $(\infty,1)$-categories (see e.g.~\cite{Lurie--HTT,Riehl--Categorical_homotopy_theory,Boardman-Vorgt--Homotopy_invariant_structures_on_Top}), while Kan complexes consequently model $(\infty, 0)$-categories, or $\infty$-groupoids.%
\footnote{This is motivated, for example, by the Quillen equivalence induced by the geometric realisation and singular complex functors, together with the homotopy hypothesis.}
A simplicial abelian group can, therefore, be interpreted as a \emph{strict symmetric monoidal $\infty$-groupoid}.
Motivated by the case of line bundles, $(\check{\scB}^n_\nabla \sfU(1))_\bullet (\CU)$ is often used as a model for the $\infty$-groupoid of \emph{$\sfU(1)$ $n$-bundles on $M$} (see, for instance, \cite{FSS--Higher_stacky_perspective}).
As we had observed from the Dold-Kan construction of $(\check{\scB}^n_\nabla \sfU(1))_\bullet (\CU)$, all $k$-simplices for $k > n$ in $(\check{\scB}^n_\nabla \sfU(1))_\bullet (\CU)$ are degenerate.
This means that all $k$-morphisms in the $\infty$-groupoid $(\check{\scB}^n_\nabla \sfU(1))_\bullet (\CU)$ for $k > n$ are trivial, reflecting the expected fact that these $n$-bundles should form a symmetric monoidal $n$-groupoid.

The reader might have objected to assigning the degree $(n,0)$ to $\widetilde{\Omega}^0(\CU)$ in~\eqref{eq:Cech_Deligne_double_complex} and viewing the double complex as a double chain complex.
It might seem more natural to view~\eqref{eq:Cech_Deligne_double_complex} as a double cochain complex $\check{\scD}_n^{\diamond, \diamond}$ with $\widetilde{\Omega}^0(\CU)$ in degree $(0,0)$.
This is the \emph{\v{C}ech-Deligne double cochain complex} (as used for example in~\cite{Waldorf--Thesis,Waldorf--More_morphisms}).
Write $\NN_0$ for the set of non-negative integers.

\begin{definition}[Deligne complex, Deligne cohomology]
\label{def:Deligne_complex_and_cohomology}
For $n \in \NN_0$, the cochain complex of sheaves of abelian groups
\begin{equation}
	\scD^\diamond_n(M) \coloneqq
	\begin{tikzcd}
		0 \ar[r] & \widetilde{\Omega}^0_M \ar[r, "\widetilde{\dd}^0"] & \widetilde{\Omega}^1_M \ar[r, "\widetilde{\dd}^1"] & \ldots \ar[r, "\widetilde{\dd}^{n-1}"] & \widetilde{\Omega}^n_M \ar[r] & 0
	\end{tikzcd}
\end{equation}
with sheaves given by $\widetilde{\Omega}^k_M(U) = \widetilde{\Omega}^k(U)$ for $U \subset M$ open, and $\widetilde{\Omega}^0_M$ in degree $0$, is called the \emph{Deligne complex of degree $n$} of $M$.
Its $k$-th hypercohomology group is denoted $\rmH^k(M,\scD^\diamond_n) \coloneqq \rmH^k(\scD^\diamond_n(M))$ and is called the \emph{$k$-th Deligne cohomology group of $M$ in degree $n$}.
The group
\begin{equation}
	\hat{\rmH}^n(M,\RZ) \coloneqq \rmH^n(M,\scD^\diamond_n)
\end{equation}
is also called the \emph{$n$-th differential cohomology group of $M$}.
\end{definition}

For $\CU \to M$ a good open covering, the hypercohomology groups of the \v{C}ech-Deligne double complex $\check{\rmH}^k(\CU, \check{\scD}_n^{\diamond, \diamond}) = \rmH^k( \Tot^\diamond (\check{\scD}_n(\CU)))$ are isomorphic to the Deligne cohomology groups~\cite{Waldorf--Transgression_II}.
Under the relabelling, which relates $\check{\scD}^n_{\diamond, \diamond}(\CU)$ and $\check{\scD}^{\diamond, \diamond}_n(\CU)$, together with the isomorphism from homology groups to homotopy groups from the Dold-Kan correspondence, we have a chain of isomorphisms of abelian groups for $k = 0, \ldots, n$,
\begin{equation}
\begin{aligned}
	\rmH^k(M, \scD_n^\diamond) &\cong \rmH^k \big( \Tot^\diamond (\check{\scD}_n(\CU)) \big)
	\\
	&\cong \rmH_{n-k} \big( \trunc_{\geq 0} \Tot_\diamond(\check{\scD}^n(\CU)) \big)
	\\
	&\cong \pi_{n-k} \big( (\check{\scB}^n_\nabla \sfU(1))_\bullet (\CU) \big)\,.
\end{aligned}
\end{equation}
Thus, we have derived the following statement:

\begin{proposition}
\label{st:homotopy_groups_of_Deligne_infty_groups}
The homotopy groups of $(\check{\scB}^n_\nabla \sfU(1))_\bullet (\CU)$ are given by the Deligne cohomology groups of $M$ as
\begin{equation}
	\pi_k \big( (\check{\scB}^n_\nabla \sfU(1))_\bullet (\CU) \big) \cong
	\begin{cases}
		\rmH^{n-k} \big( M,\scD^\diamond_n \big)\,, & 0 \leq k \leq n\,,
		\\
		0\,, & k > n\,.
	\end{cases}
\end{equation}
\end{proposition}

\begin{proof}
For $0 \leq k \leq n$ the arguments have just been given, and the statement for $k > n$ follows immediately from the fact that $(\check{\scB}^n_\nabla \sfU(1))_\bullet (\CU)$ is an $n$-skeleton.
\end{proof}

\begin{remark}
The chain complex in positive degrees $\trunc_{\geq 0} \Tot_\diamond( \check{\scD}^n_{\diamond, \diamond})$ is, in fact, the homotopy limit of the cosimplicial chain complex whose $k$-th level is given by the $k$-th row in~\eqref{eq:Cech_Deligne_double_complex}~\cite{Dugger--Primer_on_HoColims,BSS--Hocolims_and_global_observables}.
This, together with the fact that the \v{C}ech-Deligne double complex computes the same homology groups, shows that $\scD^\diamond_n$, shifted in degree by $n$, provides a homotopy sheaf of positively graded chain complexes on the Grothendieck site $(\Mfd, \tau_\open)$ (cf.~\cite{Dugger--Sheaves_and_homotopy_theory}).
\qen
\end{remark}

A vertex of $(\check{\scB}^1_\nabla \sfU(1))_\bullet (\CU)$ is a family $(A_a, g_{ab})_{a,b \in \Lambda}$ which satisfies $0 = \rmD_{-1}(A_a, g_{ab}) = A_b - A_a - g_{ab}^*\, \mu_{\sfU(1)}$.
A 1-simplex in $(\check{\scB}^1_\nabla \sfU(1))_\bullet (\CU)$ is a family $((A_a, g_{ab}), (A'_a, g'_{ab}), h_a)_{a,b \in \Lambda}$ where $(A'_a, g'_{ab}) - (A_a, g_{ab}) = \rmD_0(h_a)$, i.e. $g'_{ab}\, g_{ab}^{-1} = h_b\, h_a^{-1}$ and $A'_a = A_a + h_a^*\, \mu_{\sfU(1)}$ for all $a,b \in \Lambda$.
In short, the symmetric monoidal groupoid extracted from the simplicial abelian group $(\check{\scB}^1_\nabla \sfU(1))_\bullet (\CU)$ is precisely $\scD^1(\CU)$ (as described after~\eqref{eq:local_morphs_of_HLBdls}).

Let us investigate $(\check{\scB}^2_\nabla \sfU(1))_\bullet (\CU)$ in more detail, for this should be related to the symmetric monoidal 2-groupoid $\BGrb^\nabla_{\rmflat \sim}(M)$.
A vertex in $(\check{\scB}^2_\nabla \sfU(1))_\bullet (\CU)$ is a triple $(B_a, A_{ab}, g_{abc})_{a,b,c \in \Lambda} \in \Tot_0(\check{\scD}^2_{\diamond, \diamond}(\CU))$, which is closed, i.e. satisfies
\begin{equation}
\begin{aligned}
	g_{bcd}\, g_{abd} &= g_{acd}\, g_{abc}\,,
	\\
	A_{bc} - A_{ac} + A_{ab} &= g_{abc}^*\,\mu_{\sfU(1)}\,,
	\\
	B_b - B_a &= \dd A_{ab}
\end{aligned}
\end{equation}
for all $a,b,c,d \in \Lambda$.
Here, $g_{abc} \colon U_{abc} \to \sfU(1)$, $A_{ab} \in \Omega^1(U_{ab}, \iu\, \FR)$, and $B_a \in \Omega^2(U_a, \iu, \FR)$.
That is, every $(B_a, A_{ab}, g_{abc})_{a,b,c \in \Lambda} \in (\trunc_{\geq 0} \Tot(\check{\scD}^2_{\diamond, \diamond}(\CU)))_0 = \ker(\rmD_{-1})$ defines a bundle gerbe with connection on $M$ given by $(\CG, \nabla^\CG) = \big( \{I_{A_{ab}}\},\, m,\, \{B_a\},\, \CU,\, \pi \big)$.
The bundle gerbe multiplication is given by the multiplication $m$ on $\FC$.
A 1-simplex $S \in (\check{\scB}^2_\nabla \sfU(1))_1 (\CU)$ with $d_1 S = (B_a, A_{ab}, g_{abc})_{a,b,c \in \Lambda}$ and $d_0 S = (B'_a, A'_{ab}, g'_{abc})_{a,b,c \in \Lambda}$ is a tuple
\begin{equation}
	S = \big( (B_a, A_{ab}, g_{abc}),\, (B'_a, A'_{ab}, g'_{abc}),\, (\eta_a, h_{ab}) \big)_{a,b,c \in \Lambda}
\end{equation}
with $\eta_a \in \Omega^1(U_a, \iu\, \FR)$ and $h_{ab} \colon U_{ab} \to \sfU(1)$ such that
\begin{equation}
\begin{aligned}
	(B'_a, A'_{ab}, g'_{abc}) &= (B_a, A_{ab}, g_{abc}) + \rmD (\eta_a, h_{ab})\,, \quad \text{i.e. such that}
	\\
	h_{bc}\, h_{ac}^{-1}\, h_{ab} &= g'_{abc}\, g_{abc}^{-1}\,,
	\\
	\eta_b - \eta_a + h_{ab}^*\,\mu_{\sfU(1)} &= A'_{ab} - A_{ab}\,,
	\\
	\dd \eta_a &= B'_a - B_a
\end{aligned}
\end{equation}
for all $a,b,c \in \Lambda$.
We can assemble these data into the pair
\begin{equation}
	(\{I_{\eta_a}\}, h_{ab}) \in \HVBdl^\nabla \big( (\CG, \nabla^{\CG}), (\CG', \nabla^{\CG'}) \big)
\end{equation}
of a rank-one $(\CG{-}\CG')$-twisted hermitean vector bundle with connection.
These are defined in Appendix~\ref{app:BGrbs_with_mutual_surjective_submersions}, where it is also shown that they form a category equivalent to $\BGrb^\nabla(M)((\CG, \nabla^{\CG}), (\CG', \nabla^{\CG'}))$, provided that $(\CG, \nabla^{\CG})$ and $(\CG', \nabla^{\CG'})$ are defined over the same surjective submersion onto $M$.
This, however, is the case here, as all bundle gerbes obtained in the above manner from vertices of $(\check{\scB}^2_\nabla \sfU(1))_\bullet (\CU)$ are defined with respect to $\pi \colon \CU \to M$.

\begin{remark}
It is important to note that objects of $\BGrb^\nabla(M)((\CG, \nabla^{\CG}), (\CG', \nabla^{\CG'}))$ are, in the sense of Definition~\ref{def:1-morphisms_of_BGrbs}, defined over $\CU^{[2]}$, i.e. feature hermitean vector bundles with connection $E_{ab} \to U_{ab}$ rather than $E_a \to U_a$, thus making the step from morphisms to twisted vector bundles necessary.
This appears to have been overlooked in the literature so far.
The equivalence of twisted vector bundles and morphisms of bundle gerbes over the same surjective submersion has, to our knowledge, not been spelled out before.
\qen
\end{remark}

Finally, a 2-simplex in $(\check{\scB}^2_\nabla \sfU(1))_\bullet (\CU)$ is explicitly given by
\begin{equation}
	\psi = \big( (\CG, \nabla^{\CG}), (\CG', \nabla^{\CG'}), (\CG'', \nabla^{\CG''}), S, S', S'', \psi_a \big)\,,
\end{equation}
where $\psi_a \colon U_a \to \sfU(1)$.
Here we have abbreviated vertices as $(\CG, \nabla^{\CG}) = (B_a, A_{ab}, g_{abc})$, $(\CG', \nabla^{\CG'}) = (B'_a, A'_{ab}, g'_{abc})$ and $(\CG'', \nabla^{\CG''}) = (B''_a, A''_{ab}, g''_{abc})$, and  have written 1-simplices as $S = ((\CG, \nabla^\CG), (\CG', \nabla^{\CG'}), (\eta_a, h_{ab}))$, $S' = ((\CG', \nabla^{\CG'}), (\CG'', \nabla^{\CG''}), (\eta'_a, h'_{ab}))$, as well as $S'' = ((\CG, \nabla^\CG), (\CG'', \nabla^{\CG''}), (\eta''_a, h''_{ab}))$.
We then have
\begin{equation}
\begin{aligned}
	S' - S'' + S &= \check{\delta} (\psi_a)\,, \quad \text{or, explicitly,}
	\\
	h'_{ab}\, h_{ab}\, \psi_a &= \psi_b\, h''_{ab}\,,
	\\
	\eta'_a + \eta_a &= \eta''_a + \psi_a^*\mu_{\sfU(1)}
\end{aligned}
\end{equation}
for all indices $a,b \in \Lambda$.
In other words,
\begin{equation}
	\psi \in \HVBdl^\nabla \big( (\CG, \nabla^\CG), (\CG'', \nabla^{\CG''}) \big) \big( S' \circ S,\, S'' \big)\,,
\end{equation}
if we define (suppressing connections in the notation)
\begin{equation}
\begin{aligned}
	&(-) \circ (-)\, \colon \HVBdl^\nabla(\CG', \CG'') \times \HVBdl^\nabla(\CG, \CG') \to \HVBdl^\nabla(\CG, \CG'')\,,
	\\
	&\big( (E,\alpha), (F, \beta) \big) \mapsto \big( E \otimes F,\, (\alpha \otimes 1) \circ (1 \otimes \beta) \big)\,,
	\\
	&(\phi, \phi') \mapsto \phi \otimes \phi'\,,
\end{aligned}
\end{equation}
for $(\CG, \nabla^\CG)$, $(\CG', \nabla^{\CG'})$ and $(\CG'', \nabla^{\CG''})$ all defined over the same surjective submersion to $M$.
Note that this differs slightly from the (modified) tensor product of bundle gerbes over the same surjective submersion given by (cf.~\eqref{eq:simplified_tensor_product_over_same_sursub})
\begin{equation}
\begin{aligned}
	&\widetilde{\otimes} \colon \HVBdl^\nabla(\CG_0, \CG_0') \times \HVBdl^\nabla(\CG_1, \CG_1') \to \HVBdl^\nabla(\CG_0\, \widetilde{\otimes}\, \CG_1, \CG'_0\, \widetilde{\otimes}\, \CG_1')\,,
	\\
	&\big( (E,\alpha), (F, \beta) \big) \mapsto \big(E \otimes F,\, \alpha \otimes \beta \big)\,,
	\\
	&(\phi, \phi') \mapsto \phi \otimes \phi'\,.
\end{aligned}
\end{equation}
Here $(\CG_0, \nabla^{\CG_0})\, \widetilde{\otimes}\, (\CG_1, \nabla^{\CG_1})$ is the bundle gerbe given by the tensor product of the line bundles of $\CG_0$ and $\CG_1$ over $\CU^{[2]}$ as well as the tensor product of their bundle gerbe multiplications over $\CU^{[3]}$ (compare also Proposition~\ref{st:determinants_of_morphisms_of_BGrbs}).
Thus we obtain the following statement:

\begin{proposition}
\label{st:Deligne_2-skeleton_and_Bgrbs_over_CU}
For $\CU \to M$ a good open covering, the strict symmetric monoidal 2-category $\scD^2(\CU)$ defined by the 2-skeleton $(\check{\scB}^2_\nabla \sfU(1))_\bullet (\CU)$ is equivalent to the symmetric monoidal 2-groupoid $(\BGrb^\nabla_{\rmflat \sim}(\pi \colon \CU \to M), \widetilde{\otimes}\,)$ of bundle gerbes defined over the fixed surjective submersion $\pi \colon \CU \to M$ whose 1-morphisms are twisted hermitean line bundles with connections over $\CU$ and whose 2-morphisms are morphisms thereof (cf. Appendix~\ref{app:BGrbs_with_mutual_surjective_submersions}).
\end{proposition}

\begin{proof}
This is essentially~\cite[Proposition 2.6.1]{Waldorf--Thesis}, but the proof has to be refined by the observation that for bundle gerbes over a fixed surjective submersion the category of their twisted vector bundles is equivalent to the full category of morphisms.
This is Proposition~\ref{st:morphism_categories_and_twisted_HVBdls_for_same_sur_sub}.
\end{proof}

We now quote the well-known classification result for bundle gerbes with and without connection on $M$:

\begin{theorem}[{\cite[Theorem 4.1]{Murray-Stevenson:Bgrbs--stable_isomps_and_local_theory}}]
\label{st:Classification_of_BGrbs_by_Deligne_coho}
There are isomorphisms of abelian groups
\begin{equation}
\begin{aligned}
	\pi_0 \big( \BGrb^\nabla_{\rmflat \sim}(M), \otimes\, \big) &\cong \rmH^2(M,\scD^\bullet_2(M)) \cong \hat{\rmH}^3(M, \RZ)\,, \quad \text{and}
	\\
	\pi_0 \big( \BGrb_\sim(M), \otimes\, \big) &\cong \rmH^3(M, \RZ)\,.
\end{aligned}
\end{equation}
\end{theorem}

\begin{definition}[Deligne class, Dixmier-Douady class]
Let $(\CG, \nabla^\CG) \in \BGrb^\nabla(M)$.
The class of $(\CG, \nabla^\CG)$ in $\hat{\rmH}^3(M, \RZ)$ is called the \emph{Deligne class} $\rmD(\CG, \nabla^\CG)$ of $(\CG,\nabla^\CG)$, whereas the class $\DD(\CG)$ of $\CG$ in $\rmH^3(M,\RZ)$ is called the \emph{Dixmier-Douady class} of $\CG$.
\end{definition}

Note that Theorem~\ref{st:Classification_of_BGrbs_by_Deligne_coho} is proven in~\cite{Murray-Stevenson:Bgrbs--stable_isomps_and_local_theory} for stable isomorphisms of bundle gerbes only, i.e. isomorphisms in $\BGrb^\nabla_\FP(M)((\CG_0, \nabla^{\CG_0}), (\CG_1, \nabla^{\CG_1}))$ (see Appendix~\ref{app:special_morphisms_and_descent}), but because of Corollary~\ref{st:BGrb_FP_hookrightarrow_BGrb_is_equivalence} the same result holds for the full 2-groupoid $\BGrb^\nabla_{\rmflat \sim}(M)$, as has been shown already in~\cite{Waldorf--Thesis}.

We can now deduce a stronger version of Proposition~\ref{st:Deligne_2-skeleton_and_Bgrbs_over_CU}.

\begin{theorem}
\label{st:Deligne_2-skeleton_and_Bgrbs}
For $\pi \colon \CU \to M$ a good open covering, the strict symmetric monoidal 2-category $\scD^2(\CU)$ defined by 2-skeletal simplicial abelian group $(\check{\scB}^2_\nabla \sfU(1))_\bullet (\CU)$ is equivalent to the symmetric monoidal 2-groupoid $(\BGrb^\nabla_{\rmflat \sim}(M), \otimes\, )$.
\end{theorem}

\begin{proof}
For $\CU \to M$ a good open covering, the total cohomology of the \v{C}ech-Deligne complex $\check{\rmH}^2(\CU, \check{\scD}^{\diamond, \diamond}_2(\CU))$ is isomorphic to Deligne cohomology.
Choosing a representative \v{C}ech cocycle for a Deligne class, there exists a bundle gerbe with connection $(\CG, \nabla^\CG)$ on $M$ defined over $\CU \to M$ which represents this Deligne class under the isomorphism in Theorem~\ref{st:Classification_of_BGrbs_by_Deligne_coho}.
Explicitly, this is the bundle gerbe constructed above from a closed triple $(B_a, A_{ab}, g_{abc})_{a,b,c \in \Lambda}$ and the isomorphism in Proposition~\ref{st:homotopy_groups_of_Deligne_infty_groups}.
By Theorem~\ref{st:Classification_of_BGrbs_by_Deligne_coho}, every bundle gerbe with connection on $M$ is isomorphic in $\BGrb^\nabla_\rmflat$ to one represented by such a vertex in $(\check{\scB}^2_\nabla \sfU(1))_\bullet (\CU)$.
This shows that the inclusion $\BGrb^\nabla_{\rmflat \sim}(\pi \colon \CU \to M) \hookrightarrow \BGrb^\nabla_{\rmflat \sim}(M)$ is essentially surjective in the 2-categorical sense (see~\cite{Leinster--Basic_bicategories}).
It is, furthermore, fully faithful as follows from Proposition~\ref{st:Deligne_2-skeleton_and_Bgrbs_over_CU} (together with Proposition~\ref{st:morphism_categories_and_twisted_HVBdls_for_same_sur_sub}), and, thus, is an equivalence of 2-categories.%
\footnote{Here we used that every twisted vector bundle $(E_a, \alpha_{ab})_{a,b \in \Lambda}$ is isomorphic to one of the form $(U_a {\times} \FC^{\rank(E)}, \beta_{ab})_{a,b \in \Lambda}$, since all $U_a$ are contractible.}
As composition of morphisms is compatible with tensor products in both source and target, this equivalence, moreover, respects the symmetric monoidal structures.
\end{proof}

Consequently, the simplicial group $(\check{\scB}^2_\nabla \sfU(1))_\bullet (\CU)$ contains essentially all information about bundle gerbes and their flat isomorphisms on $M$, analogously to how $(\check{\scB}^1_\nabla \sfU(1))_\bullet (\CU)$ describes hermitean line bundles with connection.
In particular, it captures all isomorphism classes not only of bundle gerbes, but also of the flat isomorphisms between these, and thus allows to study the homotopy properties of the 2-category of bundle gerbes in a local setting using small categories.
Such studies have been pursued, and are still being pushed further, in even greater generality in, for instance, \cite{FSS--Higher_stacky_perspective,FRS--Higher_U1-gerbe_connections_in_geometric_prequant}.
Nevertheless, while the scheme outlined in this section provides enough technology to investigate these features of bundle gerbes, local descriptions are often rather cumbersome in geometric, explicit constructions, as it is usually hard to find good open coverings explicitly.
We will see examples of bundle gerbes in Section~\ref{sect:examples_of_BGrbs} which are difficult to describe explicitly in terms of local constructions based on good open coverings, or which admit a much more elegant and accessible description in terms of more general surjective submersions $Y \to M$.

We close this section with a known~\cite[Corollary 1.3.10]{Waldorf--Thesis}, but important proposition.

\begin{proposition}
\label{st:DD_triviality_implies_trivialisability}
Let $(\CG, \nabla^\CG) \in \BGrb^\nabla(M)$.
If there exists a 1-isomorphism $\CI \to \CG$ of bundle gerbes without connection, then there exists a flat 1-isomorphism $\CI_\rho \to (\CG, \nabla^\CG)$ for some $\rho \in \Omega^2(M, \iu\, \FR)$.
\end{proposition}

\begin{proof}
Choose a good open covering $(U_a)_{a \in \Lambda}$ of $M$.
Let the Deligne class of $(\CG, \nabla^\CG)$ be represented by the \v{C}ech-Deligne cocycle $[g_{abc}, A_{ab}, B_a] \in \rmH^2(\trunc_{\geq 0} \Tot(\check{\scD}^2_{\diamond, \diamond}(\CU)))$.
The Dixmier-Douady class of $\CG$ has \v{C}ech representative $[g_{abc}] \in \check{\rmH}^2(\CU, \sfU(1))$.
This being trivial means that there exists a family $(h_{ab})_{a,b \in \Lambda}$, with $h_{ab} \colon U_{ab} \to \sfU(1)$, such that $g_{abc} = h_{ab}\, h_{bc}^{-1}\, h_{ac}$ for all $a,b,c \in \Lambda$.
We have
\begin{equation}
	[g_{abc}, A_{ab}, B_a] - \rmD(h_{ab},0)
	= [1,\, A_{ab} - h_{ab}^*\, \mu_{\sfU(1)},\, B_a]\,.
\end{equation}
The fact that $\rmD$ applied to this still vanishes, together with the exactness of the \v{C}ech resolutions of sheaves of differential forms, implies that there exists a family $(C_a)_{a \in \Lambda}$, with $C_a \in \Omega^1(U_a, \iu\, \FR)$, such that
\begin{equation}
	A_{ab} - h_{ab}^*\, \mu_{\sfU(1)} = C_b - C_a\,.
\end{equation}
Thus,
\begin{equation}
	[g_{abc}, A_{ab}, B_a] - \rmD(h_{ab}, C_a)
	= [1,\, 0 ,\, B_a - \dd C_a]\,.
\end{equation}
Finally, by the same reasons as in the preceding step, we observe that there exists $\rho \in \Omega^2(M, \iu\, \FR)$ such that $B_a - \dd C_a = \rho_{|U_a}$.
Combining this with Theorem~\ref{st:Classification_of_BGrbs_by_Deligne_coho} yields the assertion.
\end{proof}

\chapter[Categorical structures on morphisms of bundle gerbes]{Categorical structures\\on morphisms of bundle gerbes}
\label{ch:structures_on_morphisms_of_bgrbs}

\section{Additive structures on morphisms}
\label{sect:Additive_structures_on_morphisms_in_BGrb}

In this section we show that the 2-categories of bundle gerbes introduced in Section~\ref{sect:The_2-category_of_BGrbs} carry a much richer structure than outlined there.
What follows is mostly motivated from the theory of vector bundles with connection together with the crucial role they play in the definition of morphisms of bundle gerbes (Definitions~\ref{def:1-morphisms_of_BGrbs} and~\ref{def:2-morphisms_of_BGrbs}).
A yet more direct relation in restricted cases is made precise in Proposition~\ref{st:morphism_categories_and_twisted_HVBdls_for_same_sur_sub}.
We start by investigating the 2-morphisms more closely.

\begin{proposition}
\label{st:additive_structure_on_2-morphisms_of_BGrbs}
For any two bundle gerbes with connection $(\CG_i, \nabla^{\CG_i})$ on $M$ and any pair of 1-morphisms $(E,\nabla^E, \alpha, Z, \zeta)$ and $(E' ,\nabla^{E'}, \alpha', Z', \zeta')$ from $(\CG_0, \nabla^{\CG_0})$ to $(\CG_1, \nabla^{\CG_1})$, the set of 2-morphisms $\BGrb^\nabla(M)((E,\alpha), (E', \alpha'))$ is an abelian group.
\end{proposition}

\begin{proof}
Let $[W, \omega, \psi]$ and $[W', \omega', \psi']$ be 2-morphisms from $(E,\alpha)$ to $(E',\alpha')$ in $\BGrb^\nabla(M)$.
We write $\hat{Z} = Z' {\times}_{Y_{01}} Z$ and set
\begin{equation}
\label{eq:sum_of_2-morphisms}
	[W', \omega', \psi'] + [W, \omega, \psi]
	\coloneqq \big[ W' {\times}_{\hat{Z}} W,\, \omega' {\times}_{\hat{Z}} \omega,\, \pr_{W'}^* \psi' + \pr_W^* \psi \big]\,.
\end{equation}
If $(\widetilde{W}, \widetilde{\omega}, \widetilde{\psi})$ and $(\widetilde{W}', \widetilde{\omega}', \widetilde{\psi}')$ are different representatives for $(W, \omega, \psi)$ and $(W', \omega', \psi')$, respectively, with the equivalences established by surjective submersions $\chi'_{\widetilde{W}'} \colon X' \to \widetilde{W}'$ and $\chi'_{W'} \colon X' \to W'$, as well as $\chi_{\widetilde{W}} \colon X \to \widetilde{W}$ and $\chi_W \colon X \to W$, respectively, then $(X' {\times}_{\hat{Z}} X,\, \chi'_{\widetilde{W}'} {\times}_{\hat{Z}} \chi_{\widetilde{W}}, \chi_{W'} {\times}_{\hat{Z}} \chi_W)$ establishes the equality of $[W', \omega', \psi'] + [W, \omega, \psi]$ and $[\widetilde{W}', \widetilde{\omega}', \widetilde{\psi}'] + [\widetilde{W}, \widetilde{\omega}, \widetilde{\psi}]$.
\\
The unit element with respect to this operation is $[\hat{Z}, 1_{\hat{Z}}, 0]$, using that $W {\times}_{\hat{Z}} \hat{Z} \cong W$.
\\
The symmetry of the operation~\eqref{eq:sum_of_2-morphisms} follows from considering the equivalence induced by $X = (W' {\times}_{\hat{Z}} W) {\times}_{\hat{Z}} (W' {\times}_{\hat{Z}} W)$, and the surjective submersion
given by the fibre product (or pullback) of the swap diffeomorphism $\sw \colon (W' {\times}_{\hat{Z}} W) \to (W {\times}_{\hat{Z}} W')$ and the identity on $(W' {\times}_{\hat{Z}} W)$.
\\
Associativity follows from the associativity of addition of morphisms of vector bundles together with the (weak) associativity of fibre products of surjective submersions.
\end{proof}

Alternatively, we could have used the unique representatives of the 2-morphisms (see Proposition~\ref{st:2-morphisms_have_simple_representatives}) over the minimal surjective submersion $1_{\hat{Z}}$, and just set
\begin{equation}
\label{eq:sum_of_2-morphisms_simple_rep}
	[\hat{Z}, 1_{\hat{Z}}, \psi'] + [\hat{Z}, 1_{\hat{Z}}, \psi]
	\coloneqq \big[ \hat{Z}, 1_{\hat{Z}},\, \psi' + \psi \big]\,.
\end{equation}
One can check that this is compatible with the equivalence relation and, hence, sufficient to define an abelian group structure.
Unitality and associativity are straightforward in this definition of the sum of 2-morphisms.
From the expressions~\eqref{eq:sum_of_2-morphisms_simple_rep} and~\eqref{eq:vertical_composition_of_2-morphisms} it is, furthermore, apparent that the sum of 2-morphisms is compatible with vertical composition of 2-morphisms.

Consider 1-morphisms $(E,\nabla^E,\alpha,Z,\zeta)$, $(E',\nabla^{E'},\alpha',Z',\zeta')$ from $(\CG_0, \nabla^{\CG_0})$ to $(\CG_1, \nabla^{\CG_1})$, and $(F, \nabla^F, \beta,X, \xi)$, $(F', \nabla^{F'}, \beta', X', \xi')$ from $(\CG_1,\nabla^{\CG_1})$ to $(\CG_2,\nabla^{\CG_2})$.
We write $\hat{Z} = Z {\times}_{Y_{01}} Z'$ and $\hat{X} = X {\times}_{Y_{12}} X$.
Let, moreover, $[\hat{Z}, 1_{\hat{Z}}, \phi], [\hat{Z}, 1_{\hat{Z}}, \phi'] \colon (E,\alpha) \to (E',\alpha')$ and $[\hat{X}, 1_{\hat{X}}, \psi] \colon (F, \beta) \to (F', \beta')$.
We have
\addtocounter{equation}{1}
\begin{align*}
	&[\hat{X}, 1_{\hat{X}}, \psi] \circ_1 \big( [\hat{Z}, 1_{\hat{Z}}, \phi] + [\hat{Z}, 1_{\hat{Z}}, \phi'] \big)
	\\*[0.2cm]
	&= \big[ (Z {\times}_{Y_1} X) {\times}_{Y_{01}} (\hat{Z} \times_{Y_1} \hat{X}) {\times}_{Y_{12}} (Z' {\times}_{Y_1} X'),\, (\pr_{Z {\times}_{Y_1} X}) {\times}_{M} (\pr_{Z' {\times}_{Y_1} X'}),\,
	\\
	& \qquad \dd_{(F',\beta') \circ (E',\alpha')} \circ (\psi \otimes (\phi + \phi')) \circ \dd_{(F,\beta) \circ (E,\alpha)} \big]  \theeq
	\\[0.2cm]
	&= \big[ (Z {\times}_{Y_1} X) {\times}_{Y_{01}} (\hat{Z} \times_{Y_1} \hat{X}) {\times}_{Y_{12}} (Z' {\times}_{Y_1} X'),\, (\pr_{Z {\times}_{Y_1} X}) {\times}_{M} (\pr_{Z' {\times}_{Y_1} X'}),\,
	\\
	& \qquad \big( \dd_{(F',\beta') \circ (E',\alpha')} \circ (\psi \otimes \phi ) \circ \dd_{(F,\beta) \circ (E,\alpha)} \big)
	+ \big( \dd_{(F',\beta') \circ (E',\alpha')} \circ (\psi \otimes \phi') \circ \dd_{(F,\beta) \circ (E,\alpha)} \big) \big]
	\\*[0.2cm]
	&= \big( [\hat{X}, 1_{\hat{X}}, \psi] \circ_1 [\hat{Z}, 1_{\hat{Z}}, \phi] \big)
	+ \big( [\hat{X}, 1_{\hat{X}}, \psi] \circ_1 [\hat{Z}, 1_{\hat{Z}}, \phi'] \big)\,.
\end{align*}
Here we have only used the distributivity of composition and tensor product of morphisms of vector bundles over their sum.
From~\eqref{eq:tensor_product_of_2-morphisms} it follows that the tensor product of bundle gerbes is distributive with respect to this additive structure on the level of 2-morphisms as well.
Finally, part (2) and (3) of Example~\ref{eg:2-morphisms--new_from_old_identity_unitors} let us see straightforwardly that taking the transpose and adjoint of 2-morphisms is compatible with sums.
Thus, we have proven the next proposition.

\begin{proposition}
\label{st:Ab-enrichment_of_BGrb^nabla}
The 2-category $\BGrb^\nabla(M)$ is enriched in preadditive categories as a symmetric monoidal 2-category, meaning that all its morphism categories are preadditive, and composition of 1-morphisms as well as the tensor products are additive functors.
Moreover, addition of 2-morphisms is compatible with the transpose functor and the adjoint.
\end{proposition}

Independently of this additive structure, we can show the following property of 2-morphisms in the sub-2-category $\BGrb^\nabla_\rmpar(M) \subset \BGrb^\nabla(M)$.
Recall that this subcategory has the same 1-morphisms as $\BGrb^\nabla(M)$, but only parallel 2-morphisms.

\begin{proposition}
\label{st:kernels_of_2-morphisms_in_BGrb^nabla_par}
Consider 1-morphisms $(E, \alpha)$ and $(E', \alpha')$ as above, and consider a 2-morphism $[W, \omega, \psi] \in \BGrb^\nabla_\rmpar(M)((E,\alpha), (E',\alpha'))$.
Let $(\hat{Z}, 1_{\hat{Z}}, \psi_0)$ be the unique representative of this 2-morphism over $\hat{Z} = Z {\times}_{Y_{01}} Z'$.
We write $\hat{\zeta} \colon \hat{Z} \to Y_{01} = Y_0 {\times}_M Y_1$ for the induced surjective submersion.
The following statements hold true:
\begin{myenumerate}
	\item The representative $(W, \omega, \psi)$ defines a 1-morphism
	\begin{equation}
	\label{eq:kernels_of_2-morphism_in_BGrb^nabla_par}
	\begin{aligned}
		&\big( \ker(\psi),\, (\omega_Z^*\nabla^E)_{|\ker(\psi)},\, \omega_Z^{[2]*}\alpha_{|(\omega_{Y_0}^{[2]*}L_0 \otimes d_0^*\ker(\psi))},\, W,\, \hat{\zeta} \circ \omega \big)
		\\
		&\qquad\in \BGrb^\nabla_\rmpar(M)\big( (\CG_0, \nabla^{\CG_0}), (\CG_1, \nabla^{\CG_1}) \big)\,,
	\end{aligned}
	\end{equation}
	where $\omega_{Y_i} \coloneqq \zeta_{Y_i} \circ \omega_Z = \zeta'_{Y_i} \circ \omega_{Z'} \colon W \to Y_i$.
	
	\item The 1-morphism~\eqref{eq:kernels_of_2-morphism_in_BGrb^nabla_par} is 2-isomorphic in $\BGrb^\nabla_\rmpar(M)$ by a unitary 2-isomorphism to $\big( \ker(\psi_0),\, (\pr_Z^*\nabla^E)_{|\ker(\psi_0)},\, \pr_Z^{[2]*}\alpha_{|(L_0 \otimes d_0^*\ker(\psi_0))},\, \hat{Z},\, \hat{\zeta} \big)$.
	
	\item Any of the above 1-morphisms represents the categorical kernel of $[W, \omega, \psi]$:
	\begin{equation}
		\ker \big( [W, \omega, \psi] \big) \cong \big( \ker(\psi_0),\, (\pr_Z^*\nabla^E)_{|\ker(\psi_0)},\, \pr_Z^{[2]*}\alpha_{|(L_0 \otimes d_0^*\ker(\psi_0))},\, \hat{Z},\, \hat{\zeta} \big)\,.
	\end{equation}
\end{myenumerate}
\end{proposition}

\begin{proof}
Ad (1):
First note that since $\psi \in \HVBdl^\nabla_\rmpar(W)$ is parallel it commutes with the parallel transport of $\omega_Z^*(E,\nabla^E)$.
Thus, its (fibrewise) kernel is preserved by $\omega_Z^*\nabla^E$.
Hence, we can use the parallel transport to locally trivialise the disjoint union of the fibrewise kernels of $\psi$, making $\ker(\psi)$ into a hermitean subbundle of $\omega_Z^*(E,\nabla^E)$  with connection.
The connection on $\ker(\psi)$ is given by the restriction of $\omega_Z^*\nabla^E$ to the subbundle $\ker(\psi) \subset \omega_Z^*E$.
We have to show that $\omega_Z^{[2]*}\alpha$ restricts to an isomorphism $\omega_{Y_0}^{[2]*}L_0 \otimes d_0^*\ker(\psi) \to d_1^*\ker(\psi) \otimes \omega_{Y_1}^{[2]*}L_1$.
Here the defining property~\eqref{eq:2-morphism_compatibility_with_alphas} of 2-morphisms comes to help:
It readily implies that $(d_1^*\psi \otimes 1) \circ \omega_Z^*\, \zeta^{[2]*} \alpha$ vanishes on pairs $(\ell \otimes e)$ with $\ell \in \omega_{Y_0}^{[2]*}L_0$ and  $e \in \ker(d_0^*\psi) \cong d_0^*\ker(\psi)$.
The compatibility of the restriction of $\omega_Z^{[2]*} \alpha$ with the bundle gerbe multiplications follows from the original compatibility of the bundle gerbe multiplications and $\alpha$.

Ad (2):
A representative for a 2-isomorphism as desired is given by $(W, \omega {\times}_{\hat{Z}} 1_W, \nu)$, where $\nu \colon \ker(\omega^*\psi) \to \omega_Z^*\ker(\psi)$ is the restriction of the isomorphism $\omega^*\pr_Z^*E \cong \omega_Z^*E$ to the subbundle $\ker(\omega^*\psi) \subset \omega^*\pr_Z^*E$.

Ad (3):
First, there is a parallel 2-morphism
\begin{equation}
\label{eq:kernel inclusion}
	\widehat{\iota}_{\ker(\psi_0)} \coloneqq \big[ \hat{Z} {\times}_{Y_{01}} Z, 1_{Z {\times}_{Y_{10}} \hat{Z}}, \pr_{(Z {\times}_{Y_{01}} Z)}^* \dd_{(E,\alpha)} \circ \iota_{\ker(\psi_0)} \big]
\end{equation}
from the 1-morphism in (2) to $(E,\alpha)$, where $\iota_{\ker(\psi_0)} \colon \ker(\psi_0) \hookrightarrow E$ is the inclusion of the subbundle $\ker(\psi_0) \subset E$ and $\dd_{(E,\alpha)}$ is as in Lemma~\ref{st:dd-def_and_properties}.
Let $(F, \nabla^F, \beta, X, \xi)$ be another 1-morphism $(\CG_0, \nabla^{\CG_0}) \to (\CG_1, \nabla^{\CG_1})$, and let $[X {\times}_{Y_{01}} Z, 1, \phi] \colon (F, \beta) \to (E,\alpha)$ be another parallel 2-morphism.
Using representatives of 2-morphisms over the minimal surjective submersion $1_{\hat{Z}}$, vertical composition reads as
\begin{equation}
	[\hat{Z}, 1_{\hat{Z}}, \psi_0] \circ_2 [X {\times}_{Y_{01}} Z, 1, \phi]
	= \big[ X {\times}_{Y_{01}} \hat{Z}, 1, \pr_{\hat{Z}}^*\psi_0 \circ \pr_{(X {\times}_{Y_{01}} Z)}^*\phi \big]\,.
\end{equation}
Thus, if this composition is the zero 2-morphism, the fibrewise image of $\pr_{(X {\times}_{Y_{01}} Z)}^*\phi$ must be contained in $\ker(\pr_{\hat{Z}}^*\psi_0) \cong \pr_{\hat{Z}}^*\ker(\psi_0)$.
The fact that $\phi$ is compatible with $\alpha$ and $\beta$, together with the fact that $\alpha$ is compatible with the restriction to $\ker(\psi_0)$ readily implies that $[X {\times}_{Y_{01}} Z {\times}_{Y_{01}} Z', 1, \pr_{(X {\times}_{Y_{01}} Z)}^* \phi]$ is a 2-morphism from $(F, \beta)$ to the 1-morphism under (2).
As $\iota_{\ker(\psi_0)} \circ \phi = \phi$, and using that $\ker(\psi_0)$ lives over $Z {\times}_{Y_{01}} Z'$ we see that
\begin{equation}
\begin{aligned}
	&\big[ \hat{Z} {\times}_{Y_{01}} Z,\, 1,\, \pr_{(Z {\times}_{Y_{01}} Z)}^* \dd_{(E,\alpha)}\circ \iota_{\ker(\psi_0)}\big]
	\circ_2 \big[ X {\times}_{Y_{01}} \hat{Z}, 1, \pr_{(X {\times}_{Y_{01}} Z)}^* \phi \big]
	\\
	&= \big[ X {\times}_{Y_{01}} Z {\times}_{Y_{01}} Z' {\times}_{Y_{01}} Z,\, 1,\, \pr_{(Z {\times}_{Y_{01}} Z)}^* \dd_{(E,\alpha)} \circ \pr_{(Z {\times}_{Y_{01}} X)}^*\phi \big]
	\\
	&= \big[ X {\times}_{Y_{01}} Z {\times}_{Y_{01}} Z,\, 1,\, \pr_{(Z {\times}_{Y_{01}} Z)}^* \dd_{(E,\alpha)} \circ \pr_{(X {\times}_{Y_{01}} Z)}^*\phi \big]
	\\
	&= 1_{(E,\alpha)} \circ_2 [X {\times}_{Y_{01}} Z, 1, \phi]
	\\
	&= [X {\times}_{Y_{01}} Z, 1, \phi]\,,
\end{aligned}
\end{equation}
where in the second step we changed the representative of the 2-morphism in order to get rid of the unused factor of $Z'$.
Thus, $[X {\times}_{Y_{01}} Z {\times}_{Y_{01}} Z', 1, \pr_{(X {\times}_{Y_{01}} Z)}^* \phi]$ provides the unique 2-morphism factorising $[X {\times}_{Y_{01}} Z, 1, \phi]$ through the inclusion $[\hat{Z} {\times}_{Y_{01}} Z, 1, \pr_{(Z {\times}_{Y_{01}} Z)}^* \dd_{(E,\alpha)} \circ \iota_{\ker(\psi_0)}]$, where uniqueness is a consequence of the use of the kernel in $\HVBdl^\nabla$.
\end{proof}

\begin{remark}
Analogously, the categorical cokernel of a parallel 2-morphism is represented by
\begin{equation}
	\coker \big( [W, \omega, \psi] \big) \cong \big( \coker(\psi_0),\, (\pr_{Z'}^*\nabla^E)_{|\coker(\psi_0)},\, \pr_{Z'}^{[2]*}\alpha'_{|(L_0 \otimes d_0^*\coker(\psi_0))},\, \hat{Z},\, \hat{\zeta} \big)\,.
\end{equation}
Here we can explicitly write $\coker(\psi_0) \cong (\ran(\psi_0) )^\perp$.
Combining fact that $\psi_0$ commutes with the parallel transport with the parallelity of the hermitean metric on $E$ shows that $\nabla^E$ preserves $(\ran(\psi_0) )^\perp$.
This gives rise to a vector bundle again by using the parallel transport to obtain local trivialisations.
\end{remark}

Now we descend one level and consider 1-morphisms of bundle gerbes.
For twisted vector bundles between two bundle gerbes over a mutual surjective submersion $Y \to M$ one can straightforwardly introduce a direct sum, simply by using the direct sum in $\HVBdl^\nabla(Y)$.
In the case of general morphisms of bundle gerbes we have to account for different surjective submersions appearing in the data of the bundle gerbes as well as the morphisms.
This makes abstract computations rather tedious, but, nevertheless, the desired structure exists.

\begin{theorem}
\label{st:direct_sum_structure_on_morphisms_of_BGrbs}
The following statements hold true:
\begin{myenumerate}
	\item On every morphism category $\BGrb^\nabla(M)((\CG_0, \nabla^{\CG_0}), (\CG_1, \nabla^{\CG_1}))$ there exists a symmetric monoidal structure, given on objects by
	\addtocounter{equation}{1}
	\begin{align*}
	\label{eq:direct_sum_on_1-morphisms}
		&(E', \nabla^{E'}, \alpha', Z', \zeta') \oplus (E, \nabla^E, \alpha, Z, \zeta) \theeq
		\\
		&= \Big( \pr_{Z'}^* (E', \nabla^{E'}) \oplus \pr_Z^* (E, \nabla^E),\, \sfd_r^{-1} \circ \big( \pr_{Z'}^{[2]*} \alpha' \oplus \pr_Z^{[2]*} \alpha \big) \circ \sfd_l,\, \hat{Z},\, \hat{\zeta} \Big)\,,
	\end{align*}
	where $\sfd_l$ and $\sfd_r$ denote the natural isomorphisms witnessing the distributivity of the tensor product over the direct sum in $\HVBdl^\nabla$ from the left and from the right, respectively, and where we write $\hat{Z} = Z {\times}_{Y_{01}} Z'$ as well as $\hat{\zeta} = \zeta {\times}_{Y_{01}} \zeta'$.
	On morphisms, the monoidal structure acts as
	\begin{equation}
	\label{eq:direct_sum_on_2-morphisms}
		[W', \omega', \psi'] \oplus [W, \omega, \psi]
		= \big[ W' {\times}_{Y_{01}} W,\, \omega' {\times}_{Y_{01}} \omega,\, \pr_{W'}^* \psi' \oplus \pr_W^* \psi \big]\,.
	\end{equation}
	
	\item The direct sum introduced in (1) is additive with respect to the sum on morphisms introduced in Proposition~\ref{st:additive_structure_on_2-morphisms_of_BGrbs}, and the tensor product of morphisms of bundle gerbes distributes over it:
	There are natural isomorphisms
	\begin{equation}
	\begin{aligned}
		\delta_{l, (F,E',E)} \colon (F, \beta) \otimes \big( (E',\alpha') \oplus (E, \alpha) \big)
		&\arisom \big( (F,\beta) \otimes (E', \alpha') \big) \oplus \big( (F, \beta) \otimes (E,\alpha) \big)
		\\
		\delta_{r, (E',E,F)} \colon \big( (E', \alpha') \oplus (E,\alpha) \big) \otimes (F,\beta)
		&\arisom \big( (E', \alpha') \otimes (F,\beta) \big) \oplus \big( (E,\alpha) \otimes (F, \beta) \big)\,,
	\end{aligned}
	\end{equation}
	making $(-) \otimes (-)$ monoidal with respect to $(-) \oplus (-)$ in each argument.
	
	\item Composition of 1-morphisms distributes over the direct sum introduced in (1) in the sense that there are natural isomorphisms
	\begin{equation}
	\begin{aligned}
		\sfc_{l,(F,E',E)} \colon (F, \beta) \circ \big( (E',\alpha') \oplus (E, \alpha) \big)
		&\arisom \big( (F,\beta) \circ (E', \alpha') \big) \oplus \big( (F, \beta) \circ (E,\alpha) \big)
		\\
		\sfc_{r, (E',E,F)} \colon \big( (E', \alpha') \oplus (E,\alpha) \big) \circ (F,\beta)
		&\arisom \big( (E', \alpha') \circ (F,\beta) \big) \oplus \big( (E,\alpha) \circ (F, \beta) \big)\,,
	\end{aligned}
	\end{equation}
	making $(-) \circ (-)$ monoidal with respect to $(-) \oplus (-)$ in each argument.
\end{myenumerate}
\end{theorem}

As indicated above, the proof of Theorem~\ref{st:direct_sum_structure_on_morphisms_of_BGrbs} is somewhat cumbersome due to the occurrence of multiple common refinements of surjective submersions, though no deep argument has to be used.
The situation is improved by Proposition~\ref{st:2-morphisms_have_simple_representatives}, but still the proof deserves to be deferred to Appendix~\ref{app:proof_of_direct_sum_theorem}.

There exist inclusion 2-morphisms
\begin{equation}
\label{eq:direct_sum_inclusion}
	\iota_{(E,\alpha)} \coloneqq \big[ Z {\times}_{Y_{01}} \hat{Z}, 1, \iota_E \circ \pr_{(Z {\times}_{Y_{01}} Z)}^* \dd_{(E,\alpha)} \big]
	\colon (E,\alpha) \to (E',\alpha') \oplus (E,\alpha)\,,
\end{equation}
where $\iota_E$ denotes the inclusion of the pullback of $E$ into the direct sum, and analogously for $\iota_{(E',\alpha')} \colon (E',\alpha') \to (E',\alpha') \oplus (E,\alpha)$.
Moreover, there also exist projection 2-morphisms
\begin{equation}
\label{eq:direct_sum_projection}
	\pr_{(E,\alpha)} \coloneqq \big[ \hat{Z} {\times}_{Y_{01}} Z, 1, \pr_{(Z {\times}_{Y_{01}} Z)}^* \dd_{(E,\alpha)} \circ \pr_E \big]
	\colon (E',\alpha') \oplus (E,\alpha) \to (E,\alpha)\,,
\end{equation}
where $\pr_E$ denotes the projection from the direct sum onto the pullback of $E$.
An analogous projection 2-morphism $\pr_{(E',\alpha')}$ exists for $(E',\alpha')$.
We can check that
\begin{equation}
	\pr_{(E,\alpha)} \circ_2 \iota_{(E,\alpha)} = 1_{(E,\alpha)}\,, \quad \text{and} \quad
	\pr_{(E',\alpha')} \circ_2 \iota_{(E',\alpha')} = 1_{(E',\alpha')}\,.
\end{equation}
Observe that, omitting pullbacks, we have
\begin{equation}
	(\iota_{E'} \circ \dd_{(E',\alpha')} \circ \pr_{E'}) \oplus (\iota_E \circ \dd_{(E,\alpha)} \circ \pr_E)
	= \dd_{(E',\alpha') \oplus (E,\alpha)}\,,
\end{equation}
which implies
\begin{equation}
	\big( \iota_{(E',\alpha')} \circ_2 \pr_{(E',\alpha')} \big) \oplus \big( \iota_{(E,\alpha)} \circ_2 \pr_{(E,\alpha)} \big)
	= 1_{(E',\alpha') \oplus (E,\alpha)}\,.
\end{equation}
Denote by $\Ab$ the category of abelian groups.

\begin{proposition}
The direct sum in the morphism categories of $\BGrb^\nabla(M)$ as defined in Theorem~\ref{st:direct_sum_structure_on_morphisms_of_BGrbs} is the categorical product and coproduct in these categories.
\end{proposition}

\begin{proof}
This follows from the $\Ab$-enrichment of the morphism categories in $\BGrb^\nabla(M)$ together with the properties of the morphisms~\eqref{eq:direct_sum_inclusion} and~\eqref{eq:direct_sum_projection}, as well as~\cite[Section VII.2, Theorem 2]{ML--Categories_for_the_working_mathematician}.
\end{proof}

We now introduce further categorical structures which will be essential in establishing morphisms of bundle gerbes as a higher analogue of morphisms of line bundles, and in applying this theory to higher geometric quantisation in Section~\ref{sect:2-Hspace_of_a_BGrb}.

\begin{definition}[(Commutative) rig category{~\cite{BDRR--Ring_completion_of_rig_categories,Laplaza--Coherence_for_distributivity}}]
\label{def:rig_category}
A \emph{(commutative) rig category} consists of a tuple $(\scR, \otimes, \One_\scR, \oplus, \Null_\scR, \delta^\scR_l, \delta^\scR_r, a^\scR_l, a^\scR_r)$ of a category $\scR$ together with two (symmetric) monoidal structures $(\otimes, \One_\scR)$ and $(\oplus, \Null_\scR)$ together with left and right distributivity natural isomorphisms
\begin{equation}
\begin{aligned}
	&\delta^\scR_{l,(x,y,z)} \colon x \otimes (y \oplus z) \to (x \otimes y) \oplus (x \otimes z)
	\\
	&\delta^\scR_{r,(x,y,z)} \colon (x \oplus y) \otimes z \to (x \otimes z) \oplus (y \otimes z)\,
\end{aligned}
\end{equation}
and left and right absorption natural isomorphisms
\begin{equation}
\begin{aligned}
	&a^\scR_{l,x} \colon \Null_\scR \otimes x \to \Null_\scR\,,
	\\
	&a^\scR_{r,x} \colon x \otimes \Null_\scR \to \Null_\scR\,,
\end{aligned}
\end{equation}
satisfying the axioms of a (commutative) rig (a ring without negatives) up to coherent natural isomorphisms.
\end{definition}

We will usually abbreviate a rig category $(\scR, \otimes, \One_\scR, \oplus, \Null_\scR, \delta^\scR_l, \delta^\scR_r, a^\scR_l, a^\scR_r)$ by writing $(\scR, \otimes, \One_\scR, \oplus, \Null_\scR)$, or $\scR$, if the remaining data have been clearly specified.

\begin{example}
The prime examples of commutative rig categories in this work are $\Hilb$, the category of finite-dimensional Hilbert spaces with its tensor product and direct sum~\cite{BDRR--Ring_completion_of_rig_categories}, $\HVBdl(M)$, $\HVBdl^\nabla(M)$, and $\HVBdl^\nabla_\rmpar(M)$, each with tensor product and direct sum of vector bundles providing the two symmetric monoidal structures.
It is noteworthy that in $\Hilb$ and $\HVBdl^\nabla_\rmpar(M)$ the monoidal structure $\oplus$ is cartesian, i.e. it coincides with the categorical product.
In fact, those two categories are abelian; we shall see more on this in due course.\qen
\end{example}

\begin{definition}[Rig module categories]
\label{def:rig_module_category}
Let $(\scR, \otimes, \One_\scR, \oplus, \Null_\scR)$ be a (commutative) rig category.
A \emph{right $\scR$-module category} is a tuple $(\scC, \oplus_\scC, \Null_\scC, \otimes
)$ of a symmetric monoidal category $(\scC, \oplus_\scC, \Null_\scC)$ together with a functor $\otimes \colon \scC {\times} \scR \to \scC$, which satisfies the axioms of a left module over a (commutative) ring up to coherent isomorphism.
A \emph{left $\scR$-module category} is defined accordingly.
\end{definition}

\begin{example}
Every commutative rig category is a bimodule category over itself.\qen
\end{example}

Thus, $\scR$-module categories $\scC$ are in particular module categories over the monoidal category $(\scR, \otimes, \One_\scR)$ (cf.~\cite{Hovey--Model_categories}), but here $\scR$ and $\scC$ are endowed with an additional symmetric monoidal structure, and the module action has to be compatible with both these structures.
Theorem~\ref{st:direct_sum_structure_on_morphisms_of_BGrbs} has the following corollary, which significantly extends the results in~\cite{Waldorf--More_morphisms,Waldorf--Thesis} regarding module actions of vector bundles on morphisms of bundle gerbes.

\begin{corollary}
\label{st:module_structures_on_morphisms_in_BGrb}
Let $(\CG_0, \nabla^{\CG_0})$, $(\CG_1, \nabla^{\CG_1}) \in \BGrb^\nabla(M)$.
The following statements hold true:
\begin{myenumerate}
	\item The category $\BGrb^\nabla(M)(\CI_0, \CI_0)$ is a commutative rig category with respect to the direct sum in $\BGrb^\nabla(\CI_0, \CI_0)$ and the tensor product inherited from that on $\BGrb^\nabla(M)$.
	
	\item The morphism category $\BGrb^\nabla(M)((\CG_0, \nabla^{\CG_0}), (\CG_1, \nabla^{\CG_1}))$ is a right module category over $\BGrb^\nabla(M)(\CI_0, \CI_0)$
	via the direct sum in $\BGrb^\nabla(M)((\CG_0, \nabla^{\CG_0}), (\CG_1, \nabla^{\CG_1}))$ and the module action induced by the tensor product on $\BGrb^\nabla(M)$.
	
	\item Any 2-morphism set in $\BGrb^\nabla(M)$ canonically has the structure of a bimodule over the algebra $\BGrb^\nabla(M)(1_{\CI_0}, 1_{\CI_0}) \cong C^\infty(M,\FC)$ (as follows from Proposition~\ref{st:2-morphisms_have_simple_representatives}) via
	\begin{equation}
		f \cdot [Z, 1, \psi] = [M, 1_M, f] \otimes [Z, 1, \psi] = [Z, 1, \zeta_M^*f \cdot \psi]\,.
	\end{equation}
	
	\item The corresponding statements hold true with $\BGrb^\nabla_\rmpar(M)$ in place of $\BGrb^\nabla(M)$, where now $\BGrb^\nabla_\rmpar(M)(1_{\CI_0}, 1_{\CI_0}) \cong \FC$.
\end{myenumerate}
\end{corollary}

\begin{proof}
This follows from the fact that $\BGrb^\nabla(M)$ is a symmetric monoidal 2-category (see Theorem~\ref{st:2-categories_of_bundle_gerbes}) together with Theorem~\ref{st:direct_sum_structure_on_morphisms_of_BGrbs}.
\end{proof}

\begin{remark}
All the above structures are strictly associative and unital as a consequence of the conventions adapted in Appendix~\ref{app:monoidal_structures_and_strictness}.
A treatment which explicitly displays the unitors and associators would also have unitors and associators for these actions.
However, they never enter non-trivially in any expression or computation.
\qen
\end{remark}

We further investigate the categorical structures on the morphism categories:

\begin{lemma}
Let $(\CG_0, \nabla^{\CG_0})$, $(\CG_1, \nabla^{\CG_1}) \in \BGrb^\nabla(M)$.
The following statements hold true:
\begin{myenumerate}
	\item The zero 1-morphism is a zero object in $\BGrb^\nabla(M)((\CG_0, \nabla^{\CG_0}), (\CG_1, \nabla^{\CG_1}))$.
	
	\item Monomorphisms in $\BGrb^\nabla(M)((\CG_0, \nabla^{\CG_0}), (\CG_1, \nabla^{\CG_1}))$ are the 2-morphisms whose underlying morphism of vector bundles is injective (i.e. monic).
	
	\item Epimorphisms in $\BGrb^\nabla(M)((\CG_0, \nabla^{\CG_0}), (\CG_1, \nabla^{\CG_1}))$ are the 2-morphisms whose underlying morphism of vector bundles is surjective (i.e. epic).
\end{myenumerate}
\end{lemma}

\begin{proof}
These statements follow from the corresponding well-known assertions in the category $\HVBdl^\nabla$ evaluated on the respective surjective submersions.
\end{proof}

\begin{lemma}
In $\BGrb^\nabla_\rmpar(M)((\CG_0, \nabla^{\CG_0}), (\CG_1, \nabla^{\CG_1}))$, every monomorphism is a kernel, and every epimorphism is a cokernel.
\end{lemma}

\begin{proof}
We prove the epic half; the monic half is similar.
Consider $(F, \nabla^F, \beta, X, \xi)$ and $(E, \nabla^E, \alpha, Z, \zeta) \in \BGrb^\nabla_\rmpar(M)((\CG_0, \nabla^{\CG_0}), (\CG_1, \nabla^{\CG_1}))$, and let $[X {\times}_{Y_{01}} Z,\, 1,\, \psi] \colon (F, \beta) \to (E,\alpha)$ be an epic in $\BGrb^\nabla_\rmpar(M)((\CG_0, \nabla^{\CG_0}), (\CG_1, \nabla^{\CG_1}))$.
Recall from Proposition~\ref{st:kernels_of_2-morphisms_in_BGrb^nabla_par} that every 2-morphism in this sub-2-category of bundle gerbes with connection has a kernel.
We show that if $[X {\times}_{Y_{01}} Z,\, 1,\, \psi]$ is epic, then it is the cokernel of $\widehat{\iota}_{\ker(\psi)}$ (cf.~\eqref{eq:kernel inclusion}).

Consider a third 1-morphism $(E', \nabla^{E'}, \alpha', Z', \zeta') \colon (\CG_0, \nabla^{\CG_0}) \to (\CG_1, \nabla^{\CG_1})$ and a 2-morphism $[X {\times}_{Y_{01}} Z', 1, \phi] \colon (F, \beta) \to (E', \alpha')$ such that $[X {\times}_{Y_{10}} Z', 1, \phi] \circ_2 \widehat{\iota}_{\ker(\psi)} = 0$.
Then, over $X {\times}_{Y_{01}} Z {\times}_{Y_{01}} X {\times}_{Y_{01}} Z'$, the composition of the pullbacks of $\phi$ and $\dd_{(F, \beta)} \circ \iota_{\ker(\psi)}$ is zero.
Because of Lemma~\ref{st:dd_and_2-morphisms} together with $\dd_{(E',\alpha')|Z' {\times}_{Z'} Z'} = 1$ (see proof of Proposition~\ref{st:2-morphisms_have_simple_representatives}), we can infer that over $Z {\times}_{Y_{01}} X {\times}_{Y_{01}} Z'$, the morphism $\pr_{(X {\times}_{Y_{01}} Z')}^*\phi$ vanishes on $\ker(\psi)$.
This is where, again, the well-known fact that every epic in $\HVBdl^\nabla_\rmpar(Z {\times}_{Y_{01}} X {\times}_{Y_{01}} Z')$ is a cokernel can be used.
We define
\begin{equation}
	[Z {\times}_{Y_{01}} X {\times}_{Y_{01}} Z', 1, \phi'] \colon (E,\alpha) \to (E', \alpha')\,,
\end{equation}
with $\phi'$ given in the usual way, by setting $\phi'(e) = \phi(f)$ for $e$ in the pullback of $E$ and $f$ any preimage of $e$ under the pullback of $\psi$.
The morphism $\phi'$ is the unique candidate to yield such a 2-morphism as $\HVBdl^\nabla_\rmpar(Z {\times}_{Y_{01}} X {\times}_{Y_{01}} Z')$ is abelian.
Finally, compatibility of $\phi'$ with $\alpha$ and $\alpha'$ is a consequence of $\ker(\psi)$ being preserved by $\beta$ and the compatibility of $\phi$ and $\psi$ with $\alpha$, $\beta$, and $\alpha'$.
\end{proof}

We can now summarise the results of this section as the following two assertions.

\begin{theorem}
\label{st:enrichment_in_BGrb}
The 2-categories of bundle gerbes with connection on $M$ have the following structures.
\begin{myenumerate}
	\item The 2-category $\BGrb^\nabla(M)$ is canonically enriched in cartesian monoidal, preadditive, right $\BGrb^\nabla(M)(\CI_0, \CI_0)$-rig-module categories.
	
	\item The 2-category $\BGrb^\nabla_\rmpar(M)$ is canonically enriched in cartesian monoidal, semisimple abelian, right $\BGrb^\nabla_\rmpar(M)(\CI_0, \CI_0)$-rig-module categories.
\end{myenumerate}
\end{theorem}

\begin{proof}
We are only left to show the acclaimed semisimplicity in (2).
Any 2-endomorphism $[W, \omega, \psi]$ of a 1morphism $(E,\alpha)$ in $\BGrb^\nabla_\rmpar(M)$ that is not a multiple of the identity splits the morphism into morphisms built from the fibrewise eigenspaces of the underlying morphism of vector bundles.
These are the kernels of $\lambda\, 1_{(E,\alpha)} - [W, \omega, \psi]$ for $\lambda$ an eigenvalue of $\psi$.%
\footnote{Observe that $\lambda$ does not depend on the fibre since both $1_E$ and $\psi$ are parallel.}
There is an upper bound on the number of simple 1-morphisms that a given 1-morphism decomposes into in this way, given by $\rank(E)$.
Thus, any 1-morphism in $\BGrb^\nabla_\rmpar(M)$ is a direct sum of finitely many simple 1-morphisms.
\end{proof}

\begin{remark}
In general, there will exist infinitely many 2-isomorphism classes of simple objects in $\BGrb^\nabla_\rmpar(M)((\CG_0, \nabla^{\CG_0}), (\CG_1, \nabla^{\CG_1}))$.
\qen
\end{remark}

The above module structures can be reduced to smaller, but equivalent, rig categories by precomposing the module actions with the equivalences of categories
\begin{equation}
\begin{alignedat}{2}
	\HVBdl^\nabla(M) &\hookrightarrow \BGrb^\nabla(M)(\CI_0, \CI_0) \qquad &&\text{in case (1), and}
	\\
	\HVBdl^\nabla_\rmpar(M) &\hookrightarrow \BGrb^\nabla_\rmpar(M)(\CI_0, \CI_0) \qquad &&\text{in case (2).}
\end{alignedat}
\end{equation}

We can even extend this structure somewhat further, although this is not going to be important to us in the remainder of this thesis.
\begin{definition}[$\scR$-algebra]
Let $(\scR, \otimes, \One_\scR, \oplus, \Null_\scR)$ be a rig category (see Definition~\ref{def:rig_category}) and let $(\scC, \oplus_\scC, \Null_\scC, \otimes)$ be a left/right $\scR$-module category (Definition~\ref{def:rig_module_category}) endowed with an additional (symmetric) monoidal structure $(\otimes_\scC, 1_\scC)$ such that $(\scC, \oplus_\scC, \Null_\scC, \otimes_\scC, 1_\scC)$ is itself a rig category which satisfies the categorical versions of the left/right algebra axioms with respect to the $\scR$-action $\otimes$.
Then we call $(\scC, \otimes_\scC, \One_\scC, \oplus_\scC, \Null_\scC, \otimes)$ a \emph{left/right $\scR$-algebra}.
\end{definition}

\begin{example}
\begin{myenumerate}
	\item Every rig category is both a left and a right algebra over itself.
	
	\item $\Hilb$ is a commutative rig category, and hence an algebra over itself.
	
	\item For any bundle gerbe with connection $(\CG,\nabla^\CG)$ on $M$, its category of endomorphisms $\BGrb^\nabla(M)((\CG, \nabla^\CG), (\CG, \nabla^\CG))$ is an algebra category over $\BGrb^\nabla(M)(\CI_0,\CI_0)$ via the tensor product of bundle gerbes (compare Theorem~\ref{st:enrichment_in_BGrb}).
	The algebra product of endomorphisms of $(\CG, \nabla^\CG)$ is given by composition.
	Note that these algebra actions are naturally isomorphic when taken from the left or from the right by the symmetry of the tensor product in $\HVBdl^\nabla(M)$.
	\qen
\end{myenumerate}
\end{example}

\section{Pairings of morphisms -- closed structures}
\label{sect:Pairings_and_inner_hom_of_morphisms_in_BGrb}

In this section we will develop pairings of morphisms in the 2-categories of bundle gerbes introduced in Section~\ref{sect:The_2-category_of_BGrbs}.
Consider two bundle gerbes $(\CG_0,\nabla^{\CG_0})$ and $(\CG_1,\nabla^{\CG_1})$ with connection on $M$, and let $(E, \nabla^E, \alpha, Z, \zeta) \in \BGrb^\nabla(M)((\CG_0,\nabla^{\CG_0}),(\CG_1,\nabla^{\CG_1}))$ be a 1-morphism.
Set
\begin{equation}
\label{eq:dual_of_1-morphism_in_BGrb}
	\Theta(E, \nabla^E, \alpha, Z, \zeta) \coloneqq \big( (E,\nabla^E)^*,\, \alpha^{-\sft},\, Z,\, \zeta \big)\,.
\end{equation}
Applying $(-)^{-\sft}$ to~\eqref{eq:1-morphisms_compatibility_with_BGrb_multiplications}, noting that both the operations of inverse and transpose reverse the order of composition, we see that this defines a 1-morphism
\begin{equation}
	\Theta (E,\alpha) \colon (\CG_0,\nabla^{\CG_0})^* \to (\CG_1,\nabla^{\CG_1})^*\,.
\end{equation}
Observe that
\begin{equation}
	\Theta (1_{(\CG_0,\nabla^{\CG_0})}) = 1_{(\CG_0,\nabla^{\CG_0})^*}.
\end{equation}

For $[W,\omega,\psi]$ a 2-morphism from $(E, \nabla^E, \alpha, Z, \zeta)$ to $(F, \nabla^F, \beta, X, \xi)$, we set
\begin{equation}
\label{eq:dual_of_2-morphism_in_BGrb}
	\Theta \big( [W, \omega, \psi] \big) \coloneqq [W, \sw \circ \omega, \psi^\sft] \colon \Theta(F, \beta) \to \Theta(E,\alpha)\,,
\end{equation}
where, as before, $\sw \colon Z {\times}_{Y_{01}} X \to X {\times}_{Y_{01}} Z$, $(x,z) \mapsto (z,x)$.
It follows readily that this defines a 2-morphism (apply $(-)^\sft$ to~\eqref{eq:2-morphism_compatibility_with_alphas}).
Lemma~\ref{st:dd-def_and_properties} implies that $\Theta(1_{(E,\alpha)}) = 1_{\Theta(E,\alpha)}$.
Since dual, inverse and transpose in $\HVBdl^\nabla$ are compatible with direct sums and tensor products, it follows, moreover, that the operation thus defined is compatible with composition and direct sums of 1-morphisms and 2-morphisms, as well as the tensor product of bundle gerbes.
Thus, we obtain

\begin{theorem}
\label{st:Riesz_dual_is_functorial}
The assignments~\eqref{eq:dual_of_1-morphism_in_BGrb} and~\eqref{eq:dual_of_2-morphism_in_BGrb}, together with the dual of bundle gerbes, define a 2-functor
\begin{equation}
	\Theta \colon \BGrb^\nabla(M)^{\opp_2} \to \BGrb^\nabla(M)
\end{equation}
which is covariant on the level of 1-morphisms and contravariant on 2-morphisms.
It is compatible with tensor products, direct sums, and the enrichments from Section~\ref{sect:Additive_structures_on_morphisms_in_BGrb}.
\end{theorem}

\begin{definition}[Riesz dual functor]
\label{def:Riesz_dual_functor}
The functor $\Theta \colon \BGrb^\nabla(M) \to \BGrb^\nabla(M)$ from Theorem~\ref{st:Riesz_dual_is_functorial} is called the \emph{Riesz dual}.
\end{definition}
It will become apparent in Section~\ref{sect:2-Hspace_of_a_BGrb} why this nomenclature makes sense.

\begin{theorem}
\label{st:internal_hom_of_morphisms_in_BGrb--existence_and_naturality}
Let $(\CG_i, \nabla^{\CG_i})$, for $i \in \{0, 1, 2, 3\}$, be bundle gerbes with connection on $M$, and let $(F, \nabla^F, \beta, X, \xi) \colon (\CG_1, \nabla^{\CG_1}) \to (\CG_3, \nabla^{\CG_3})$.
There is an adjoint pair of functors
\begin{equation}
\begin{tikzcd}[column sep=2cm]
	\BGrb^\nabla(M)\big( \CG_0, \CG_2 \big) \ar[r, shift left=0.175cm, "\perp"', "{(-) \otimes (F, \beta)}"] & \BGrb^\nabla(M)\big( \CG_0 \otimes \CG_1, \CG_2 \otimes \CG_3 \big) \ar[l, shift left=0.175cm, "{[(F, \beta), -]}"] \,.
\end{tikzcd}
\end{equation}
\end{theorem}

The proof of Theorem~\ref{st:internal_hom_of_morphisms_in_BGrb--existence_and_naturality} is lengthy and has therefore been deferred to Appendix~\ref{app:Proof_of_adjunction_theorem}.
The main idea is to set, for a 1-morphism $(G, \nabla^G, \gamma, U, \chi) \colon (\CG_0, \nabla^{\CG_0}) \otimes (\CG_2, \nabla^{\CG_2}) \to (\CG_1, \nabla^{\CG_1}) \otimes (\CG_3, \nabla^{\CG_3})$,
\begin{align}
\label{eq:internal_hom_on_1-morphisms}
	\big[ (F, \beta), (G,\gamma) \big]
	&= \Big( \pr_U^* (G, \nabla^G) \otimes \pr_X^* (F, \nabla^F)^*,
	\\*
	&\qquad (1 \otimes \pr_{Y_3}^{[2]*}\delta_{L_3}) \circ ( \pr_U^{[2]*} \gamma \otimes \pr_X^{[2]*} \beta^{-\sft}) \circ (1 \otimes \pr_{Y_1}^{[2]*}\delta_{L_1}^{-1}),\,
	\notag\\*
	&\qquad X {\times}_{Y_{13}} U,\, \chi_{Y_{02}} \circ \pr_U \Big)\,, \notag
\end{align}
with $\chi_{Y_{02}} = \pr_{Y_{02}} \circ \chi$.
In this way, we cancel the twists introduced by the line bundles of $(\CG_1, \nabla^{\CG_1})$ and $(\CG_3, \nabla^{\CG_3})$, by combining $L_1$ with $L_1^*$ and $L_3$ with $L_3^*$.
In order to do so, we have to work over the fibre product $U {\times}_{Y_{13}} X$.
Over $(U {\times}_{Y_{13}} X)^{[2]}$ we can easily produce pairs $L_1 \otimes L_1^*$ and $L_3 \otimes L_3^*$, which we can then use to define a strucutral isomorphism as spelled out in~\eqref{eq:internal_hom_on_1-morphisms}.
In the proof we use descent for $\HVBdl^\nabla$ together with results from Appendix~\ref{app:special_morphisms_and_descent} to get around the problems caused by the occurrence of several different surjective submersions in the morphisms involved.

The most important case for us is that of $(\CG_0, \nabla^{\CG_0}) = (\CG_2, \nabla^{\CG_2}) = \CI_0$.
In this situation, we have an adjunction
\begin{equation}
\begin{tikzcd}[column sep=2cm]
	\BGrb^\nabla(M)\big( \CI_0, \CI_0 \big) \ar[r, shift left=0.175cm, "\perp"', "{(-) \otimes (F, \beta)}"] & \BGrb^\nabla(M)\big( \CG_1, \CG_3 \big)\,. \ar[l, shift left=0.175cm, "{[(F, \beta), -]}"]
\end{tikzcd}
\end{equation}
Explicitly, for $(E, \nabla^E, \alpha, Z, \zeta) \in \BGrb^\nabla(M) ((\CG_1, \nabla^{\CG_1}), (\CG_3, \nabla^{\CG_3}))$,
\begin{align}
\label{eq:[-,-]_explicit}
	\big[ (F, \beta), (E,\alpha) \big] &= \Big( \pr_Z^* (E, \nabla^E) \otimes \pr_X^* (F, \nabla^F)^*,
	\notag\\*
	&\qquad \pr_{Y_3}^{[2]*}\delta_{L_3} \circ ( \pr_Z^{[2]*} \alpha \otimes \pr_X^{[2]*} \beta^{-\sft}) \circ \pr_{Y_1}^{[2]*}\delta_{L_1}^{-1},\,
	\notag\\*
	&\qquad X {\times}_{Y_{13}} Z,\, \zeta \circ \pr_Z = \xi \circ \pr_X \Big)
	\notag\\
	&= \Big( \Hom \big( \pr_X^* (F, \nabla^F), \pr_Z^* (E, \nabla^E) \big),\, \Hom(\beta, \alpha),
	\\*
	&\qquad X {\times}_{Y_{13}} Z,\, \zeta \circ \pr_Z = \xi \circ \pr_X \Big)\,. \notag
\end{align}
Furthermore, $\Hom \big( \pr_X^* (F, \nabla^F), \pr_Z^* (E, \nabla^E) \big)$ is the hermitean vector bundle with connection whose fibre consists of linear maps between the respective fibres of $\pr_X^*F$ and $\pr_Z^*E$ (i.e. $\Hom$ is the internal hom in $(\HVBdl^\nabla, \otimes\,)$), while $\Hom(\beta, \alpha)$ is the composition
\begin{equation}
\begin{tikzcd}[column sep=2cm, row sep=1cm]
	d_0^*\, \Hom(\pr_X^*F, \pr_Z^*E) \ar[r, "\cong"] \ar[ddd, dashed, "{\Hom(\beta, \alpha)}"] & \Hom \big( d_0^*\, \pr_X^*F,\, d_1^*\, \pr_Z^*E \big) \ar[d, "\cong"]
	\\
	& \Hom \big( L_1 \otimes d_0^*\, \pr_X^*F, \,L_1 \otimes d_0^*\, \pr_Z^*E \big) \ar[d, "{\alpha \circ (-) \circ \beta^{-1}}"]
	\\
	& \Hom \big( d_0^*\, \pr_X^*F \otimes L_3,\, d_0^*\, \pr_Z^*E \otimes L_3 \big) \ar[d, "\cong"]
	\\
	d_1^*\, \Hom(\pr_X^*F, \pr_Z^*E) & \Hom \big( d_0^*\, \pr_X^*F,\, d_0^*\, \pr_Z^*E \big) \ar[l, "\cong"]
\end{tikzcd}
\end{equation}

Observing that there are canonical 2-isomorphisms
\begin{equation}
\label{eq:Riesz_dual_and_internal_hom}
	\big[ (F, \beta), (E,\alpha) \big] \cong \Theta \big[ (E,\alpha),(F, \beta) \big]
	\cong \big[ \Theta(E,\alpha), \Theta(F, \beta) \big]
\end{equation}
we thus obtain a bifunctor
\begin{equation}
	[-,-] \colon \BGrb^\nabla(M)(\CG_1, \CG_3)^{\opp} \times \BGrb^\nabla(M)(\CG_1, \CG_3) \to \BGrb^\nabla(M)(\CI_0, \CI_0)\,.
\end{equation}
By construction, for $(K, \kappa), (K', \kappa') \colon \CI_0 \to \CI_0$, there are further natural isomorphisms
\begin{equation}
\label{eq:internal_hom_and_tensors}
\begin{aligned}
	&[(F, \beta), (E,\alpha) \otimes (K, \kappa)] \cong [(F, \beta), (E,\alpha)]  \otimes (K, \kappa)\,,
	\\
	&[(F, \beta) \otimes (K, \kappa), (E,\alpha)] \cong \Theta (K, \kappa) \otimes [(F, \beta), (E,\alpha)]\,,
	\\
	&[(K, \kappa), (K', \kappa')] \cong \Theta(K, \kappa) \otimes (K', \kappa')\,.
\end{aligned}
\end{equation}

The structure of the tensor product of morphisms in $\BGrb^\nabla(M)$, the Riesz dual, and the bifunctor $[-,-]$ are related to each other in a well studied manner:

\begin{definition}[Two-variable adjunction, closed module category~\cite{Hovey--Model_categories}]
\label{def:2-Var_adjunction_and_closed_module_cat}
Let $\scC, \scD, \scE$ be categories.
A \emph{two-variable adjunction} $\scC {\times} \scD \to \scE$ is defined to be a tuple $( \otimes, \hom_l, \hom_r, \varphi_l, \varphi_r)$, consisting of functors
\begin{equation}
	\otimes \colon \scC {\times} \scD \to \scE\,, \quad
	\hom_l \colon \scC^\opp {\times} \scE \to \scD\,, \quad
	\hom_r \colon \scD^\opp {\times} \scE \to \scC\,,
\end{equation}
together with isomorphisms, natural in each variable,
\begin{equation}
\begin{tikzcd}
	\scD \big( d, \hom_l(c,e) \big) & \scE \big( c \otimes d, e \big) \ar[r, "(\varphi_r)_{c,e}^d"', "\cong"] \ar[l, "(\varphi_l)_{d,e}^c", "\cong"'] & \scD \big( c, \hom_r(d,e) \big)\,.
\end{tikzcd}
\end{equation}
A (symmetric) monoidal category $(\scD, \otimes_\scD, \One_\scC)$ such that $\otimes_\scD$ is part of a two-variable adjunction is called a \emph{closed monoidal category}.
If $(\scD, \otimes_\scD, \One_\scD)$ is closed monoidal, a right $\scD$-module category $\scC$ is called a \emph{closed $\scD$-module category} if the module action functor $\otimes$ is part of a two-variable adjunction
\begin{equation}
	\otimes \colon \scC {\times} \scD \to \scC\,, \quad
	\hom_l \colon \scC^\opp {\times} \scC \to \scD\,, \quad
	\hom_r \colon \scD^\opp {\times} \scC \to \scC\,.
\end{equation}
\end{definition}

\begin{definition}[Closed rig-category, closed rig-module category]
\label{def:closed_rig-module}
We call a rig-category $(\scR, \otimes, \One_\scR, \oplus, \Null_\scR)$ a \emph{closed rig-category} if $\otimes_\scR$ is part of a two-variable adjunction whose $\hom_l$ and $\hom_r$ are monoidal with respect to $\oplus_\scR$ in each argument.
A rig-module category over a closed rig category $(\scR, \otimes, \One_\scR)$ (cf. Definition~\ref{def:rig_module_category}) is called a \emph{closed rig-module category over $\scR$} if the action bifunctor $\otimes$ is part of a two-variable adjunction whose $\hom_l$ and $\hom_r$ are monoidal with respect to both $\oplus_\scC$ and $\oplus_\scR$ in their respective arguments.
\end{definition}

\begin{theorem}
\label{st:2-var_adjunction_on_BGrb}
The rig-category $(\BGrb^\nabla(M)(\CI_0, \CI_0), \otimes, 1_{\CI_0}, \oplus, 0_{\CI_0})$ is closed in the sense of Definition~\ref{def:closed_rig-module}, and
$\BGrb^\nabla((\CG_0, \nabla^{\CG_0}), (\CG_1, \nabla^{\CG_1}))$ is a closed symmetric monoidal rig-module category over $(\BGrb^\nabla(M)(\CI_0, \CI_0), \otimes, 1_{\CI_0}, \oplus, 0_{\CI_0})$.
Both categories have their two-variable adjunction given by the tensor product inherited from $\BGrb^\nabla(M)$, together with
\begin{equation}
	\hom_l = [-,-]\,, \quad \text{and} \quad \hom_r = (-) \otimes \Theta(-)\,.
\end{equation}
The natural isomorphisms $\varphi_l$, $\varphi_r$ are those worked out in the proof of Theorem~\ref{st:internal_hom_of_morphisms_in_BGrb--existence_and_naturality} as well as the dualities in~\eqref{eq:Riesz_dual_and_internal_hom}.
\end{theorem}

In particular, for any $(E,\alpha), (F, \beta) \colon (\CG_0, \nabla^{\CG_0}) \to (\CG_1, \nabla^{\CG_1})$ and $(K, \kappa) \colon \CI_0 \to \CI_0$ we have natural isomorphisms
\begin{equation}
\label{eq:2-var_adjunction_of_[-,-]}
\begin{aligned}
	&\BGrb^\nabla(M) \big( (K,\kappa) \otimes (F, \beta),\, (E,\alpha) \big)
	\\*
	&\cong \BGrb^\nabla(M) \big( (F, \beta),\, (E,\alpha) \otimes \Theta(K, \kappa) \big)
	\\*
	&\cong \BGrb^\nabla(M) \big( (K, \kappa),\, [(F, \beta), (E,\alpha)] \big)\,,
\end{aligned}
\end{equation}
coherent with the respective distributivities.

\begin{proof}
A natural isomorphism from the first to the third line in~\eqref{eq:2-var_adjunction_of_[-,-]} already follows from Theorem~\ref{st:internal_hom_of_morphisms_in_BGrb--existence_and_naturality} (together with~\eqref{eq:Riesz_dual_and_internal_hom} to get naturality in $(F, \beta)$).
The isomorphism from the first to the second line is another instance of Theorem~\ref{st:internal_hom_of_morphisms_in_BGrb--existence_and_naturality}, applied to $- \otimes (K, \kappa)$, using the symmetry of the tensor product in $\BGrb^\nabla(M)$ and the observation that, for $(K, \kappa)$ an endomorphism of $\CI_0$, we have $[(K, \kappa), (E, \alpha)] = \Theta (K, \kappa) \otimes (E, \alpha)$.
The proof of the compatibility with the direct sum in both arguments of $[-,-]$ is analogous to the proof of the distributivity of the tensor product over the direct sum (see Theorem~\ref{st:direct_sum_structure_on_morphisms_of_BGrbs}), since $[(E, \alpha), (F, \beta)]$ is a reduced version of the tensor product $\Theta(E, \alpha) \otimes (F, \beta)$.
\end{proof}

\begin{remark}
\label{rmk:2-var_adjunction_restricts_to_parallel_subcat}
It is important to note that Theorem~\ref{st:internal_hom_of_morphisms_in_BGrb--existence_and_naturality} and Theorem~\ref{st:2-var_adjunction_on_BGrb} restrict to the 2-category $\BGrb^\nabla_\rmpar(M)$.
The arguments in the respective proofs restrict verbatim to parallel 2-morphisms of bundle gerbes.
\end{remark}

Theorem~\ref{st:2-var_adjunction_on_BGrb} and equation~\eqref{eq:2-var_adjunction_of_[-,-]} have a useful corollary:

\begin{corollary}
\label{st:sections_of_reduced_pairing_agree_with_2-homs}
For two 1-morphisms $(E,\alpha), (F, \beta) \colon (\CG_0, \nabla^{\CG_0}) \to (\CG_1, \nabla^{\CG_1})$ there are natural bijections
\begin{equation}
\begin{aligned}
	\BGrb^\nabla(M) \big( (E,\alpha), (F, \beta) \big) &\cong \BGrb^\nabla(M) \big( 1_{\CI_0}, [(E,\alpha), (F, \beta)] \big)
	\\*
	&\cong \Gamma \big( M, \sfR[(E,\alpha), (F, \beta)] \big)\,, \quad \text{and}
	\\[0.2cm]
	\BGrb^\nabla_\rmpar(M) \big( (E,\alpha), (F, \beta) \big) &\cong \BGrb^\nabla_\rmpar(M) \big( 1_{\CI_0}, [(E,\alpha), (F, \beta)] \big)
	\\*
	&\cong \Gamma_\rmpar \big( M, \sfR[(E,\alpha), (F, \beta)] \big)\,.
\end{aligned}
\end{equation}
\end{corollary}

Here we have made use of the equivalence of categories
\begin{equation}
	\sfR : \BGrb^\nabla(M) \big( (\CG_0, \nabla^{\CG_0}), (\CG_1, \nabla^{\CG_1}) \big) \longleftrightarrow
	\BGrb^\nabla_\FP(M) \big( (\CG_0, \nabla^{\CG_0}), (\CG_1, \nabla^{\CG_1}) \big) : \sfS
\end{equation}
found in~\cite{Waldorf--Thesis,Waldorf--More_morphisms} (see also Theorem~\ref{st:BGrb_FP_hookrightarrow_BGrb_is_equivalence}).

\begin{proof}
From Theorem~\ref{st:internal_hom_of_morphisms_in_BGrb--existence_and_naturality} and Theorem~\ref{st:2-var_adjunction_on_BGrb} we have natural bijections
\begin{equation}
\begin{aligned}
	\BGrb^\nabla(M) \big( (E,\alpha), (F, \beta) \big)
	&\cong \BGrb^\nabla(M) \big( 1_{\CI_0} \otimes (E,\alpha), (F, \beta) \big)
	\\*
	&\cong \BGrb^\nabla(M) \big( 1_{\CI_0}, [(E,\alpha), (F, \beta)] \big)
	\\
	&\cong \HVBdl^\nabla(M) \big( I_0, \sfR([(E,\alpha), (F, \beta)]) \big)
	\\
	&\cong \Gamma \big( M, \sfR[(E,\alpha), (F, \beta)] \big)\,.
\end{aligned}
\end{equation}
The statement can also be checked directly by observing from~\eqref{eq:[-,-]_explicit} that descent data for a section of $\sfR[(E,\alpha), (F, \beta)]$ corresponds under the internal hom $\tau$ of $\HVBdl^\nabla$ to 2-morphisms $(E,\alpha) \to (F, \beta)$ defined over the minimal surjective submersion.
\end{proof}

As a by-product, Corollary~\ref{st:sections_of_reduced_pairing_agree_with_2-homs} specifies the fibre of $\sfR [(E,\alpha), (F, \beta)]$:
For $x \in M$ and $\iota_x \colon \pt \hookrightarrow M$ its inclusion into $M$, there is a canonical isomorphism in $\Hilb$
\begin{equation}
\label{eq:fibre_of_R[-,-]}
	\big( \sfR [(E,\alpha), (F, \beta)] \big)_{|x}
	\cong
	\BGrb(\pt) \big( \iota_x^*(E,\alpha)\, \iota_x^*(F, \beta) \big)\,.
\end{equation}
Furthermore, we infer the following structure on 2-morphisms in $\BGrb^\nabla_\rmpar(M)$:

\begin{corollary}
\label{st:Hilb-enrichment_of_MorCats_in_BGrb}
The categories $\BGrb^\nabla_\rmpar(M) \big( (E,\alpha), (F, \beta) \big)$ are enriched in $\Hilb$.
\end{corollary}

\begin{proof}
An inner product of 2-morphisms is obtained from the isomorphism~\eqref{eq:fibre_of_R[-,-]}, or from
\begin{equation}
	\BGrb^\nabla(M) \big( 1_{\CI_0}, [(E,\alpha), (F, \beta)] \big)
	\cong \Gamma \big( M, \sfR[(E,\alpha), (F, \beta)] \big)
\end{equation}
since by construction $\sfR[(E,\alpha), (F, \beta)] \in \HVBdl^\nabla(M)$ is hermitean, so that the evaluation of the bundle metric on any pair of parallel sections is constant and can, therefore, be used as the inner product.
Alternatively, one can form the composition $[W,\omega, \psi]^* \circ_2 [W', \omega', \psi']$ for two 2-morphisms $[W,\omega, \psi]$, $[W', \omega', \psi'] \colon (E,\alpha) \to (F,\beta)$ (with the adjoint of a 2-morphism defined in Example~\ref{eg:2-morphisms--new_from_old_identity_unitors}) and then consider the trace of the underlying morphism of vector bundles, which is a constant in this situation.%
\footnote{For non-parallel 2-morphisms it is a function which descends to $M$.}
\end{proof}

\begin{remark}
If instead of using $\BGrb^\nabla_\rmpar(M)$ we worked with the larger 2-category $\BGrb^\nabla(M)$, we would obtain $\BGrb^\nabla(M)((E,\alpha), (F, \beta)) \cong \Gamma(M, \sfR [(E,\alpha), (F, \beta)])$ which are infinite-dimensional vector spaces.
These are pre-Hilbert spaces under the $\rmL^2$-product on $M$, and one could try to to work with their Hilbert space completion.
\qen
\end{remark}

\begin{remark}
The isomorphism~\eqref{eq:fibre_of_R[-,-]} shows that there is a canonical algebra structure on the fibres of $\sfR[(E,\alpha), (E, \alpha)]$ as these are composed of 2-endomorphisms of $(E,\alpha)$ in this case.
In other words, $\sfR[(E,\alpha), (E, \alpha)]$ canonically has the additional structure of a bundle of algebras.
Similarly, the bundle $\sfR[(F, \beta), (E, \alpha)]$ carries a right module structure over $\sfR[(F, \beta), (F, \beta)]$ and a compatible left module structure over $\sfR[(E, \alpha), (E, \alpha)]$, making it into a bundle of bimodules.
This relates directly to the considerations in~\cite{STV--Serre-Swan_for_BGrbMods,Karoubi--Twisted_bundles_and_twisted_K-theory}.
\qen
\end{remark}

Note that for $(F, \beta) \colon \CG_1 \to \CG_3$ and $(G, \gamma) \colon \CG_0 \otimes \CG_1 \to \CG_2 \otimes \CG_3$ there is another way of obtaining a morphism $\CG_0 \to \CG_2$, more emphasising the 2-categorical structures at hand.
Consider the expression (compare also~\cite{BSS--HGeoQuan})
\begin{equation}
	\sfH \big( (F, \beta), (G, \gamma) \big) \coloneqq (1_{\CG_2} \otimes \delta_{\CG_1}) \circ \big( (G, \gamma) \otimes \Theta(F, \beta) \big) \circ (1_{\CG_0} \otimes \delta_{\CG_3}^{-1})\,.
\end{equation}
The most important difference from $[(F, \beta) , (G, \gamma)]$ is that the tensor product of the hermitean vector bundles involved lives over $X {\times}_M U$ in this case, rather than over $X {\times}_{Y_{13}} U$.
This means that we cannot use the internal hom $\tau$ on $\HVBdl^\nabla$ as we did in the construction of $[-,-]$ in order to obtain a 1-morphism, but we need the additional factors of $L_3$ and $L_1^*$ coming from $\delta_{\CG_3}$ and $\delta_{\CG_1}^{-1}$, respectively, making it harder to handle 2-morphisms in this framework.
The additional factors of $L_i$ in $\sfH(-,-)$ allow us to use the structural isomorphisms $\gamma$ and $\beta$ to pass to a subspace $X {\times}_{Y_{13}} U \hookrightarrow X {\times}_M U$ where we can proceed as above.
Thus, the descents of $\sfH((F, \beta), (G, \gamma))$ and $[(F, \beta), (G, \gamma)]$ to their minimal surjective submersion $Y_{02} \to M$ agree, i.e. represent the same 1-morphism, whence the original 1-morphisms must be 2-isomorphic according to Theorem~\ref{st:BGrb_FP_hookrightarrow_BGrb_is_equivalence}.

\begin{remark}
The difference between $\sfH$ and $[-,-]$ becomes most clear when we consider $(Y, \pi) = (\CU, \pi)$ to be the total space $\CU = \bigsqcup_{a \in \Lambda} U_a$ of an open covering $(U_a)_{a \in \Lambda}$ of $M$.
Let $(E_a, \alpha_a)$ and $(F_a, \beta_a)$ form two morphisms of bundle gerbes which are both defined with respect to this surjective submersion.
The subtle point is that $(E,\alpha) \otimes \Theta(F, \beta)$ consists of the hermitean vector bundles $E_{a|U_{ab}} \otimes F^*_{b|U_{ab}} \to U_{ab}$ rather than the seemingly more natural $E_a \otimes F^*_a \to U_a$ of $[(E,\alpha), (F, \beta)]$.
However, the definition of the tensor product in $\BGrb^\nabla(M)$ requires us to take the common refinement of the involved surjective submersions over $M$.
Here, this amounts to taking the refinement of $\CU$ with itself, i.e. $\CU^{[2]}$, which consists of all two-fold intersections of patches in $\CU$.
The bifunctor $[-,-]$ directly employs the reduced version of the tensor product over $\CU {\times}_\CU \CU \cong \CU$, and, hence, appears to be the more natural choice for a pairing of two 1-morphisms which is supposed to cancel the twists given by the source and target bundle gerbes.
Nevertheless, recall from Proposition~\ref{st:determinants_of_morphisms_of_BGrbs} (and see also Appendix~\ref{app:BGrbs_with_mutual_surjective_submersions}) that we can pass between $\CU$ and $\CU^{[2]}$ when describing morphisms without loosing information about the bundle gerbes over $\CU \to M$ or their tensor products.
\end{remark}

\section{Examples}
\label{sect:examples_of_BGrbs}

In this section we provide several examples for bundle gerbes and some of the structures which we have encountered in the preceding sections.
Most of them are well-known already and have appeared in the literature before.
Nevertheless, they provide important applications as well as testing ground for the abstract theory.

\subsection{Local bundle gerbes}
\label{sect:Ex:Local_BGrbs}

The most basic form of bundle gerbe are so-called \emph{local bundle gerbes}.
We have already encountered them in Section~\ref{sect:Deligne_coho_and_higher_cats}, where we have found that local bundle gerbes are sufficient to capture all 2-categorical properties of $\BGrb^\nabla(M)$ in the sense that these bundle gerbes already provide an equivalent subcategory of $\BGrb^\nabla(M)$.
We call a bundle gerbe on $M$ with connection $(\CG, \nabla^\CG)$ a local bundle gerbe if its surjective submersion $(Y, \pi) = (\CU, \pi)$ is the total space $\CU = \bigsqcup_{a \in \Lambda} U_a$ of an open covering $(U_a)_{a \in \Lambda}$ of $M$ with its canonical map to $M$.
Thus, the data of a local bundle gerbe with respect to the open covering $(\CU,\pi)$ is a tuple consisting of a family of hermitean line bundles with connection $(L_{ab} \to U_{ab})_{a,b \in \Lambda}$ over the double intersections of patches of the covering, and isomorphisms $(\mu_{abc} \colon L_{ab} \otimes L_{bc} \to L_{ac})_{a,b,c \in \Lambda}$ over the triple intersections of patches from $\CU$, which are associative over quadruple intersections.
Moreover, the curving is given by a collection $B_a \in \Omega^2(U_a, \iu, \FR)$, for $a \in \Lambda$, such that $\curv(L_{ab}) = B_{b|U_{ab}} - B_{a|U_{ab}}$ for all $a,b \in \Lambda$.

An even more specific, but categorically equally general situation is given by considering good open coverings, i.e. open coverings of $M$ such that all possible finite intersections of patches from the covering are diffeomorphic to $\FR^n$.
We have already encountered good coverings in Section~\ref{sect:Deligne_coho_and_higher_cats}.
Their point there was that any hermitean line bundle with connection on $M$ could be seen to be isomorphic to a hermitean line bundle with connection which arises as the descent of a family $(I_{A_a})_{a \in \Lambda}$, where $A_a \in \Omega^1(U_a, \iu\, \FR)$, with respect to some $\sfU(1)$-valued 1-cocycle $(g_{ab})_{a,b \in \Lambda}$.
In this way, we could relate isomorphism classes of hermitean line bundles to Deligne 1-cocycles in degree 1.

If we consider a local bundle gerbe that is defined over a good open covering of $M$, there are parallel, unitary isomorphisms $\phi_{ab} \colon L_{ab} \to I_{A_{ab}}$ for every $a,b \in \Lambda$ by the contractibility of $U_{ab}$, for some collection $(A_{ab})_{a,b \in \Lambda}$ of 1-forms $A_{ab} \in \Omega^1(U_{ab}, \iu\, \FR)$.
Setting $g_{abc} \coloneqq \phi_{ac} \circ \mu_{abc} \circ (\phi_{ab} \otimes \phi_{bc})$, we obtain a bundle gerbe $(\{ I_{A_{ab}} \},\, \{ g_{abc} \},\, \{B_a\},\, \CU,\, \pi)$ with connection on $M$.
Proposition~\ref{st:From_naive_isomorphisms_to_1-isomorphisms--same_sursub} shows that this new, particularly simple bundle gerbe is isomorphic, even in $\BGrb^\nabla_\rmflat(M)$, to the original local bundle gerbe.
Note that the objects $g_{abc}$, $A_{ab}$, and $B_a$ form a Deligne 2-cocycle in degree 2, establishing the relation between bundle gerbes and $\rmH^2(M,\scD^\bullet_2(M)) \cong \hat{\rmH}^3(M, \RZ)$.
Up to isomorphism, any bundle gerbe can be described in the above simple form with respect to a good open covering of $M$ (see Section~\ref{sect:Deligne_coho_and_higher_cats}, especially Theorem~\ref{st:Deligne_2-skeleton_and_Bgrbs}).

A 1-morphism between two local bundle gerbes
\begin{equation}
\label{eq:local_bundle_gerbes--morphism}
	\big( \{ I_{A_{ab}} \},\, \{ g_{abc} \},\, \{B_a\},\, \CU,\, \pi \big) \to \big( \{ I_{A'_{ab}} \},\, \{ g'_{abc} \},\, \{B'_a\},\, \CU,\, \pi \big)
\end{equation}
employs a hermitean vector bundle with connection over an additional surjective submersion, $E \to Z \to \CU {\times}_M \CU$ (cf. Definition~\ref{def:1-morphisms_of_BGrbs}), and, therefore, is built from a family of bundles $E_{ab} \to Z_{|U_{ab}} \to U_{ab}$.
It is thus, in general, different from the twisted vector bundles considered for instance in~\cite{Rogers--Thesis,Karoubi--Twisted_bundles_and_twisted_K-theory}.
Even if we take the additional surjective submersion $Z \to \CU {\times}_M \CU$ to be the identity, we still have bundles defined over overlaps $E_{ab} \to U_{ab}$ rather than over patches $E_a \to U_a$.
This may seem like a small technical detail at first, but recall that the additional surjective submersions are vital to the many useful properties which the 2-category $\BGrb^\nabla(M)$ and its variations enjoy.
The category of twisted vector bundles is, however, equivalent to the full category of 1-morphisms between these bundle gerbes, as they are defined over the same surjective submersions onto $M$, as we have shown in Proposition~\ref{st:morphism_categories_and_twisted_HVBdls_for_same_sur_sub}.

A twisted vector bundle between local bundle gerbes as in~\eqref{eq:local_bundle_gerbes--morphism} consists of the data $(\{(E_a,\nabla^{E_a})\}, \{ \alpha_{ab} \})_{a,b \in \Lambda}$, where $(E_a,\nabla^{E_a}) \in \HVBdl^\nabla(U_a)$ and $\alpha_{ab} \colon I_{A_b} \otimes E_b \to E_a \otimes I_{A'_a}$ is a unitary, parallel isomorphism satisfying $\alpha_{ac}\, g_{abc} = g'_{abc}\, \alpha_{ab} \circ \alpha_{bc}$ for all $a,b,c \in \Lambda$.
On the level of twisted vector bundles, the structures we have introduced in Section~\ref{sect:Additive_structures_on_morphisms_in_BGrb} and Section~\ref{sect:Pairings_and_inner_hom_of_morphisms_in_BGrb} are visible more directly.
Given two twisted vector bundles $(\{(E_a,\nabla^{E_a})\}, \{ \alpha_{ab} \})_{a,b \in \Lambda}$ and $(\{(F_a,\nabla^{F_a})\}, \{ \beta_{ab} \})_{a,b \in \Lambda}$, one can readily write down their direct sum as $(\{(E_a \oplus F_a,\nabla^{E_a} \oplus \nabla^{F_a})\}, \{ \alpha_{ab} \oplus \beta_{ab} \})_{a,b \in \Lambda}$.
Their image under $[-,-]$ can be seen schematically as the family of hermitean vector bundle with connection $[E,F]_a = F_a \otimes E_a^* \cong \Hom(E_a, F_a)$ on $U_a$, where $\Hom$ denotes the functor that assigns to a pair of hermitean vector bundles the bundle of fibrewise homomorphisms between the bundles with the induced connection and hermitean metric.

\subsection{Tautological bundle gerbes}
\label{sect:Ex:Tautological_BGrbs}

Tautological bundle gerbes have been defined already in~\cite{Murray--Bundle_gerbes}; another account can be found in~\cite{Bunk-Szabo--Fluxes_brbs_2Hspaces}.
They provide a rather general class of bundle gerbes with connection in the case where the base manifold $M$ is 2-connected, i.e. $\pi_i(M) = 0$ for $i = 0,1,2$.
Tautological bundle gerbes are constructed using the based path space covering $P_0 M \to M$, where, for a fixed $x_0 \in M$,
\begin{equation}
	P_0 M \coloneqq \big\{ \gamma \in \Mfd([0,1], M)\, | \, \gamma(0) = x_0,\, \exists\, U \subset [0,1] \text{ open}: \{0,1\} \subset U,\, \gamma_{|U} = const. \big\}\,.
\end{equation}
This is actually the space of paths in $M$ with so-called \emph{sitting instants}, i.e. there exists an open neighbourhood of the boundary of the interval on which $\gamma$ is constant.
The projection $\pi \colon P_0M \to M$ is evaluation at parameter value $1$.
The reason for considering based paths is technical: the space of generic smooth based paths in $M$ is not closed under concatenation of paths.
In particular, it is desirable that a pair $\gamma, \gamma' \in P_0 M$ with the same endpoint gives rise to a closed path $\overline{\gamma'} * \gamma \colon S^1 \to M$, where $\overline{\gamma'}$ is the path $t \mapsto \gamma'(1-t)$, and $*$ denotes the concatenation of paths.
Based paths with sitting instants are stable, in this sense, under concatenation, but the sitting instants prevent us from endowing $P_0M$ with a differentiable structure, even in the Fr\'echet sense (see e.g.~\cite{Hamilton--Inverse_function_theorem}).
Instead, one has to pass to a more general notion of differentiable space, called \emph{diffeological spaces}.
Their technical treatment shall not concern us at this point in order not to distract the reader from the main purpose of this section.
We shall rather focus on the structural aspects of the tautological bundle gerbe construction.
For details on diffeological spaces we instead refer the reader to~\cite{Iglesias-Zemmour--Diffeology,Waldorf--Transgression_I} and Section~\ref{sect:DfgSp_and_diffeological_bundles}.

Reverting back to tautological bundle gerbes, let $M$ be a 2-connected manifold and $H \in \Omega^3_\cl(M,\iu\,\FR)$ a closed 3-form on $M$ with periods in $2\pi\, \iu\, \RZ$, i.e. for any closed 3-manifold $N$ and $f \in \Mfd(N,M)$, we have $\int_N f^*H \in 2 \pi\, \iu\, \RZ$.
Consider the diagram of diffeological spaces
\begin{equation}
\label{eq:tautological_BGrb_diagramm}
	\begin{tikzcd}
		\Omega^3 M \ar[r, shift left=0.1cm, "r_0"] \ar[r, shift right=0.1cm, "r_1"'] & P_0 \Omega^2 M \ar[d, "\partial"] & 
		\\
		 & \Omega^2 M \ar[r, shift left=0.1cm, "r_0"] \ar[r, shift right=0.1cm, "r_1"'] & P_0 \Omega M \ar[d, "\partial"] & 
		\\
		 & & \Omega M \ar[r, shift left=0.1cm, "r_0"] \ar[r, shift right=0.1cm, "r_1"'] & P_0 M \ar[d, "\partial"]
		\\
		 & & & M
	\end{tikzcd}
\end{equation}
Here, $\Omega X$ denotes the space of based, smooth loops with sitting instants a diffeological space, i.e. the space of smooth maps $\gamma \colon S^1 \to X$ such that there exists an open neighbourhood of $\{-1,1\} \subset S^1 \subset \FC$ on which $f$ is constant.%
\footnote{The additional sitting instant at $-1 \in S^1$ is for the purpose of having an isomorphism $\Omega M \cong (P_0 M)^{[2]}$.}
The horizontal arrows act as $r_1\gamma(t) = \gamma(w(\frac{1}{2} t))$ and $r_0 \gamma(t) = \gamma(w(1 - \frac{1}{2} t))$, where $w \colon [0,1] \to S^1$, $t \mapsto \exp(2\pi\, \iu\, t)$.
The vertical maps are evaluation of a path at $t = 1$.
Note that there are canonical inclusions $P_0 \Omega^n M \hookrightarrow \Mfd(D^n, M)$ and $\Omega^n M \hookrightarrow \Mfd(S^n, M)$.
Under these, the vertical maps act as restriction of a smooth map $D^n \to M$ to the boundary $\partial D^n \cong S^{n-1}$.
All vertical maps are, moreover, surjective (and diffeologically submersive~\cite{Waldorf--Transgression_I,Iglesias-Zemmour--Diffeology}) because of the vanishing of the homotopy groups $\pi_i(M)$ for $i=0,1,2$, which allows us to write every smooth map $S^n \to M$ as the restriction of a smooth map $D^n \to M$ to the boundary of the sphere for $0 < n < 3$.

We obtain an element of $\sfU(1)$ from an element $f \in P_0 \Omega^2 M$ by setting $\hat{\lambda}(f) = \exp( \int_{D^3} f^*H) \in \sfU(1)$.
For a smooth map $f \in \Mfd(N,N')$, we denote its restriction to the boundary of the source manifold by $\partial f \coloneqq f_{|\partial N} \in \Mfd(\partial N, N')$.
Since $H$ is closed and has periods in $2\pi\, \iu\, \RZ$, any two maps $f$ and $f'$ with $\partial f = \partial f'$ define the same complex number, $\hat{\lambda}(f) = \hat{\lambda}(f')$.
Thus, we obtain a smooth function $\lambda \colon \Omega^2 M \to \sfU(1)$.
One can now check that, given maps $g_i \in P_0 \Omega M$ for $i = 0,1,2$ such that $\partial g_i = g_j$ for $i,j = 0,1,2$, the map $\lambda$ has a cocycle property
\begin{equation}
	\lambda(\overline{g_0} * g_2) = \lambda(\overline{g_1} * g_2)\, \lambda(\overline{g_0} * g_1)
\end{equation}
and, thus, defines descent data for a hermitean line bundle $L \to \Omega M$.
An element in the fibre over $\gamma \in \Omega M$ is a an equivalence class of pairs $[\hat{\gamma}, z]$ of $\hat{\gamma} \in P_0 \Omega M$ with $\partial \hat{\gamma} = \gamma$ and $z \in \FC$.
For $\hat{\gamma}' \in \Omega M$ a different choice of disc filling the loop $\gamma$, the equivalence relation is given by setting
\begin{equation}
	[\hat{\gamma}, z] = \big[ \hat{\gamma}',\, \exp \big( \textint_{D^3} f^*H \big)\, z \big]
\end{equation}
for any $f \in P_0 \Omega^2 M$ which fills in the 2-sphere obtained by gluing $\hat{\gamma}'$ and $\hat{\gamma}$ along their common boundary.
Note that the exponential factor is always in $\sfU(1)$ so that $L$ is canonically endowed with a hermitean metric.
Given a triple of paths $\gamma_i \in P_0 M$, $i= 0,1,2$, with common endpoint, we obtain three loops $\overline{\gamma_i} * \gamma_j$ for $j > i \in \{0,1,2\}$.
We can find 2-discs $\hat{\gamma}_{ij}$ that fill in the respective loops, i.e. such that $\partial \hat{\gamma}_{ij}  = \overline{\gamma_i} * \gamma_j$.
These discs, in turn, glue together to give a smooth map $S^2 \to M$, which can once again be filled in by some map $f \colon D^3 \to M$.
We obtain a bundle gerbe product on $L \to \Omega M$ by setting
\begin{equation}
	\mu_{\gamma_2, \gamma_1, \gamma_0} \colon [\hat{\gamma}_{12}, z_1] \otimes [\hat{\gamma}_{01}, z_0] \mapsto
	\big[ \hat{\gamma}_{02},\, \exp \big( \textint_{D^3} f^* H \big)\, z_1 z_0 \big]\,.
\end{equation}
The tuple
\begin{equation}
	\CG_H \coloneqq (L, \mu, P_0 M, \partial)
\end{equation}
is called the \emph{tautological bundle gerbe} of $(M, H)$.

From the 3-form $H$ on $M$ we can, furthermore, construct a connection on $\CG_H$ as follows.
First, we obtain a connection on the line bundle $L$ via descent of $I_A \to P_0 \Omega M$ with the transition isomorphism given by $\hat{\lambda}$, where $A \in \Omega^1(P_0 M, \iu\, \FR)$ is defined as
\begin{equation}
	A_{|\hat{\gamma}}(X) = - \int_{D^2} \hat{\gamma}^*(\iota_X H)\,,
\end{equation}
where $X \in \Gamma(D^2, \hat{\gamma}^*TM)$ is a vector field along $\hat{\gamma}$ in $M$ and $\iota_{(-)}$ denotes insertion of a vector field into the first slot of a differential form.
One can check that $(I_A, \lambda)$ is descent data for a hermitean connection $\nabla^L$ on $L$ with field strength
\begin{equation}
	F^L_{|\gamma}(X_0,X_1) = \int_{S^1} \gamma^*(\iota_{X_0 \wedge X_1} H)
\end{equation}
for tangent vectors $X_0, X_1 \in \Gamma(S^1, \gamma^*TM)$.
A rigorous proof of the relevant transgression formulae can be found in~\cite[Appendix C]{BSS--HGeoQuan}.
This 2-form splits naturally into a difference $F^L = r_1^*B - r_0^*B$ for $B \in \Omega^2(P_0 M, \iu\, \FR)$ given by
\begin{equation}
	B_\sigma(X'_0, X'_1) = \int_{[0,1]} \sigma^*(\iota_{X_0 \wedge X_1} H)\,,
\end{equation}
for $\sigma \in P_0 M$ and $X_0', X_1' \in \Gamma([0,1], \sigma^*TM)$.
This makes
\begin{equation}
	(\CG_H, \nabla^{\CG_H}) = \big( L, \nabla^L, \mu, P_0 M, \partial \big)
\end{equation}
into a hermitean bundle gerbe with connection on $M$.
Finally, observe that
\begin{equation}
	\curv(\nabla^{\CG_H}) = H\,.
\end{equation}

We have chosen to write the evaluation of forms on tangent vectors explicitly, but one could as well use the general formula for the transgression of forms on $M$ to mapping spaces.
For manifolds $M, N \in \Mfd$ and the evaluation map $\Mfd(N,M) {\times} N \to M$, there is a linear map%
\footnote{Note that it is the evaluation map which induces the diffeological structure on $\Mfd(N,M)$ so that the forms we obtain on the mapping spaces are automatically smooth in the diffeological sense.}
\begin{equation}
	\CT_P \colon \Omega^n(M) \to \Omega^{n - \dim(N)}(\Mfd(N,M))\,, \quad
	\omega \mapsto \int_N \ev^*\omega\,.
\end{equation}
One may use subsets of $\Mfd(N,M)$ here instead, such as, for instance, pointed smooth maps with certain sitting instants.

A particularly important instance of the tautological bundle gerbe is the case where $M = \sfG$ is a simply connected, compact, simple Lie group.
By a Theorem of Cartan's, $\sfG$ is 2-connected and has $\rmH^3(\sfG,\RZ) \cong \RZ$~\cite{Cartan--Topologie_des_groupes_de_Lie}.
Examples which are of particular relevance to physics are the special unitary groups $\sfS\sfU(n)$, and especially $\sfS\sfU(2) \cong S^3$.
This set-up gives rise to the so-called Wess-Zumino-Witten theory, which allows to investigate D-branes in string theory with target space $\sfS\sfU(2)$~\cite{Wess-Zumino--Consequences_of_anomalous_Ward_identities,Witten--NonAb_bosonization,Gepner-Witten--Strings_on_group_manifolds}.
The Wess-Zumino-Witten model has provided several applications of bundle gerbes to string theory and D-branes, see, for instance, \cite{Gawedzki-Reis:WZW-branes_and_gerbes,Gawedzki:branes_in_WZW-models_and_gerbes,Gawedzki--Topological_actions,CJMSW--BGrbs_for_CS_and_WZW_theories}.
From the above construction we readily obtain a bundle gerbe with connection on $\sfG$ whose field strength realises any given closed 3-form $H \in \Omega^3_\cl(\sfG, \iu\, \FR)$ on $\sfG$ with periods in $2\pi\,\iu\, \RZ$.
A bundle gerbe with connection whose curvature is the fundamental 3-form $H_0 = \frac{\iu}{6}\, \< -, [-,-]_\frg \>_\frg$, where $\<-,-\>_\frg$ is the Killing form on the Lie algebra $\frg$ of $\sfG$, and $[-,-]_\frg$ denotes the Lie bracket on $\frg$, is called the \emph{basic bundle gerbe on $\sfG$}.
Several models for the basic bundle gerbe are known; it is not necessary to go to infinite-dimensional spaces, or diffeological spaces, in order to construct them, see e.g.~\cite{Waldorf--Thesis,Gawedzki-Reis--Basic_BGrb_over_non-simply_conn_grps,Murray-Stevenson--Basic_BGrb_on_unitary_groups}.
However, note that since $\rmH^3(\sfG,\RZ)$ is torsion-free, there are no 1-morphisms from the trivial bundle gerbe on $\sfG$ to the basic bundle gerbe, and no 1-morphisms in the opposite direction either, according to Proposition~\ref{st:determinants_of_sections_and_dual_sections} and Theorem~\ref{st:Classification_of_BGrbs_by_Deligne_coho}.

Basic bundle gerbes carry more structure than generic bundle gerbes.
For example, they can be described as \emph{multiplicative bundle gerbes}~\cite{Waldorf--Multiplicative_BGrbs}.
These are bundle gerbes on Lie groups where the group structure lifts, in a 2-categorically weakened manner, to the bundle gerbe.
From the construction of the basic gerbe as a tautological gerbe one can already expect this to be true, for all spaces in the diagram~\eqref{eq:tautological_BGrb_diagramm} naturally inherit group structures from the base $M = \sfG$.
Finally, let us remark that also some of the canonical involutions on $\sfG$, such as inversion for general $\sfG$, or $g \mapsto -g^\sft$ for $\sfG = \sfS\sfU(n)$ with $n$ even, have 2-categorical lifts to the basic gerbe.
These currently find applications in physics for instance in the mathematical description of topological phases of matter~\cite{Gawedzki--FKM_and_sewing_matrix,Gawedzki--BGrbs_for_TIs}.

\chapter{2-Hilbert spaces from bundle gerbes}
\label{ch:2Hspaces_from_bundle_gerbes}

\section{Higher geometric quantisation}
\label{sect:HGeo_Quan}

In this section, we approach the main application in this thesis of the bundle gerbe technology developed in Chapter~\ref{ch:bundle_gerbes} and Chapter~\ref{ch:structures_on_morphisms_of_bgrbs}, which is to higher geometric quantisation.
In ordinary geometric quantisation, one starts from a symplectic manifold $(M,\omega)$, where $M \in \Mfd$ is a $2n$-dimensional manifold and $\omega \in \Omega^2(M, \FR)$ is a closed, non-degenerate 2-form on $M$.
Non-degeneracy of a $p$-form $\eta \in \Omega^p(M)$ means that $\iota_X \eta_{|x} = 0$ for a tangent vector $X \in T_x M$ if and only if $X = 0$, where $\iota_X$ denotes the insertion of the tangent vector into the first argument of a differential form.
Symplectic geometry is a natural framework for classical mechanics.
For textbook references, see, for instance, \cite{Scheck--Mechanics,Abraham-Marsden--Foundations_of_mechanics}.
Hamiltonian functions, Hamiltonian vector fields, conserved quantities, symmetries, and many other concepts find appropriate formalisations in this set-up.
The most common example of a symplectic manifold arising in classical mechanics is the cotangent bundle of a manifold with its canonical symplectic form.

In order to quantise Hamiltonian systems, one further requires that $\omega$ have integer cycles, so that there exists a hermitean line bundle with connection $(L, \nabla^L)$ on $M$ whose field strength satisfies $\curv(\nabla^L) = 2\pi\, \iu\, \omega$.
Such an $(L, \nabla^L) \in \HLBdl^\nabla(M)$ is called a \emph{prequantum line bundle} for $(M,\omega)$.
A candidate Hilbert space to model the quantum theory on is given by the Hilbert space completion of the space of smooth global sections of $L$.
The inner product space
\begin{equation}
	\big( \CH_0(L), \<-,-\>_{\CH_0(L)} \big) \coloneqq \Big( \Gamma(M, L),\, (\psi, \phi) \mapsto \int_M h_L(\psi, \phi)\, \frac{\omega^n}{n!} \Big)
\end{equation}
is not complete, i.e. it only forms a pre-Hilbert space, the \emph{prequantum pre-Hilbert space} of $(M,L,\omega)$.%
\footnote{In the case of non-compact base manifold $M$ one has to consider compactly supported sections in order to ensure that the inner product of any pair of sections in the pre-Hilbert space exists.}
Pre-Hilbert spaces can, however, always be completed uniquely into a Hilbert space by defining elements to be equivalence classes of Cauchy sequences.
The resulting Hilbert space is $\rmL^2(M,L; \mu_\omega)$, the space of square integrable sections of $L$ with respect to the measure $\mu_\omega$ induced by the symplectic volume form $\frac{1}{n!}\, \omega$.
It is called the \emph{prequantum Hilbert space} of $(M, \omega, L)$.
We include $L$ into the data here, as one has to chose an appropriate line bundle by hand in generic situations.
However, this choice is based solely on the requirement that $\curv(\nabla^L) = 2\pi\, \iu\, \omega$, and is, therefore, ambiguous.
One can tensor $(L, \nabla^L)$ by flat line bundles without changing this requirement on the curvature.

Observables in Hamiltonian mechanics are real-valued functions $f \in C^\infty(M, \FR)$.
The symplectic structure endows these with a Poisson bracket as follows.
First, given a real-valued function $f$ on $M$, we can define its associated \emph{Hamiltonian vector field} $X_f \in \Gamma(M, TM)$ by the equation $\dd f = - \iota_{X_f} \omega$.
The vector field $X_f$ is defined uniquely because of the non-degeneracy of $\omega$.
Given two functions $f, g \in C^\infty(M, \FR)$, we set
\begin{equation}
	\{f,g\} \coloneqq \iota_{X_f \wedge X_g} \omega\,.
\end{equation}
This defines a Poisson structure on the algebra $C^\infty(M,\FR)$ which is represented on $\CH_0(L)$ via the so-called \emph{Kostant-Souriau prequantisation map}
\begin{equation}
	f \mapsto \CO_f = - \iu\, \hbar\, \nabla^L_{X_f} + 2\pi\,\hbar\, f\, \psi\,.
\end{equation}
It provides a quantisation of the Poisson algebra $(C^\infty(M,\FR), \{-,-\})$ in the sense that $\CO_f$ is an operator on $\CH(L)$ and satisfies
\begin{equation}
	[\CO_f, \CO_g] = -\iu\, \hbar\, \CO_{\{f,g\}}
	\quad
	\forall\, f,g \in C^\infty(M, \FR)\,.
\end{equation}

It turns out, however, that the prequantum Hilbert space obtained as the Hilbert space completion of $\CH_0(L)$ is too large a Hilbert space to correctly describe the physical properties of the system at hand.
An easy examples where this is obvious is the case of a point particle in $\FR^3$.
The phase space of this system is the cotangent bundle $T^* \FR^3 \cong \FR^3 {\times} \FR^3$ with symplectic form $\omega = \dd x^0 \wedge \dd x^3 + \dd x^1 \wedge \dd x^4 + \dd x^2 \wedge \dd x^5 = \dd q^k \wedge \dd p_k$, where $k = 0,1,2$ and $q^k$ are the position coordinates while $p_k$ are the momentum coordinates.
A hermitean prequantum line bundle $L \to T^* \FR^3$ is given by $I_\theta$, where $\theta = 2\pi\, \iu\ p_k\, \dd q^k$.
Square integrable sections of $I_\theta$ depend on both position and momentum coordinates, in general, and, consequently, do not represent physically viable states.
In order to cure this problem, and, hence, to obtain the physical Hilbert space, one has to restrict to a subspace of the prequantum Hilbert space.
There are several ways of defining such a physical subspace, all of which require a certain amount of choice.
In the above example, valid choices are, for instance, the subspaces of functions which depend solely on either the $q^k$ or the $p_k$.
The resulting subspaces are different, but encode the same physics: they correspond to the position and momentum space representations of physical states, and the Fourier transform on $\FR^3$ provides a transformation between them.
That is, while the prequantum Hilbert space is obtained canonically from the set-up once a prequantum line bundle has been chosen, there is no canonical way of constructing the physical Hilbert space in geometric quantisation.
Methods of reducing the prequantum Hilbert space to a physically viable subspace are usually referred to as \emph{choices of polarisations}.
Finding appropriate polarisations for a given specific system is considered the hardest step in the geometric quantisation program, as the choice of a polarisation is generically non-canonical.
One way of obtaining polarisations is by restricting to sections of the prequantum line bundle which are parallel with respect to some Lagrangian distribution on $M$, as given by either of the two factors of $\FR^3$ in the above example.
In several important cases, such choices can be implemented and interesting quantum systems have been obtained.
See for instance~\cite{APW--Geometric_quatisation_of_CS_gauge_theory,Woodhouse--Geometric_quanitsation} and references therein.

From the mathematical as well as physical perspective, there exist important systems which do not lie in the scope of geometric quantisation.
Most notably, geometric quantisation only applies to symplectic manifolds.
It is often interesting, however, to consider higher-degree analogues of symplectic forms.
String theory naturally contains a 3-form field, the so-called $H$-flux (see e.g.~\cite{Szabo--String-theory_and_D-brane_dynamics,Kiritsis--String_theory_in_a_nutshell}), while M-theory contains a 3-form field called the $C$-flux, and both of these provide interesting candidates for a modified version of geometric quantisation from the physical point of view.

There exist different approaches to this so-called problem of higher geometric quantisation.
As in string theory the configuration space of a closed bosonic string in a spacetime manifold $M$ is the loop space of $M$, it seems natural to try to transfer the programme of geometric quantisation to the loop space.
This line has been pursued, for example, in~\cite{Saemann-Szabo--Groupoid_quant_of_loop_spaces,Saemann-Szabo--Groupoids_loop_spaces_Geoquan,Saemann-Szabo--Quantisation_of_2-plectic_mfds}.
A different point of view from a conceptual angle is that since string theory is a theory of quantum gravity, it should not quantise the configuration space of strings, but it should quantise spacetime itself.
The quantisation of spacetime should then, in turn, induce a quantisation of the loop space of $M$.
In this thesis, we take the latter point of view and investigate higher geometric quantisation on $M$ rather than on its loop space.
The input we have is a closed non-degenerate 3-form $\varpi \in \Omega_\cl^3(M, \FR)$, and we have seen that if such a 3-form has integer cycles, $2 \pi\, \iu\, \varpi$ can be realised as the curvature 3-form of a bundle gerbe with connection on $M$.
Hence, bundle gerbes provide a candidate geometric object to replace the prequantum line bundle of geometric quantisation.

From a mathematical perspective, perhaps the most interesting spaces to consider are simple, compact, and simply connected Lie groups with their canonical 3-forms.
An appropriately modified procedure of geometric quantisation is expected to bare strong relations to the representation theory of the loop group $\Omega \sfG$.
We have already seen in Section~\ref{sect:Ex:Tautological_BGrbs} that the 3-form on a compact, simple, and simply connected Lie group $\sfG$ can be obtained as the field strength of a bundle gerbe.
Central extensions $\widehat{\Omega}_k \sfG \to\Omega \sfG$ of level $k$ of the loop group $\Omega \sfG$ can be described as tautological bundle gerbes on $\sfG$, denoted $\CG_k$.
This is evident from diagram~\eqref{eq:tautological_BGrb_diagramm} and the multiplicativity properties of the canonical 3-form~\cite{Waldorf--Multiplicative_BGrbs}.
A continuous unitary representation $\rho$ of $\widehat{\Omega}_k\sfG$ on a (not necessarily finite-dimensional) Hilbert space $\CH$ gives rise to a morphism of bundle gerbes without connection $\CG_k \to \CI$ which uses the trivial $\CH$-bundle over $P_0\Omega\sfG$.
Thus, such representations fit into the framework of bundle gerbes.
However, note that Proposition~\ref{st:determinants_of_sections_and_dual_sections} together with Theorem~\ref{st:Classification_of_BGrbs_by_Deligne_coho} immediately tell us that $\CH$ must have infinite dimension since the Dixmier-Douady class of $\CG_k$ agrees with $k\, [H]_\dR \in \rmH^3_\dR(\sfG, \FR) \cong \RZ$ and is, thus, non-torsion.
We will revisit this argument in more detail in Section~\ref{sect:higher_geometric_structures}.

Likewise it has been shown in string theory that the $H$-flux is the curvature 3-form of a gerbe, or bundle gerbe, with curving locally given by the Kalb-Ramond $B$-field~\cite{Kapustin--D-branes_in_top_nontriv_B-field}.
Above we commented that higher geometric quantisation in string theory should be related to quantisations of the background spacetime.
Combining this with the point of view that the object central to higher geometric quantisation is a bundle gerbe that geometrically realises the Kalb-Ramond $B$-field and the $H$-flux via its connection, the presence of these fields should be related to quantum structures on spacetime.
In fact, this has been proven to be true:
using techniques from string theory and conformal field theory, it has been shown that a constant $B$-field gives rise to noncommutative structures on spacetime~\cite{Seiberg-Witten--Strings_and_NCG}.
In geometric terms a constant $B$-field is a flat bundle gerbe with connection, i.e. a pair $(\CG, \nabla^\CG)$ with $\curv(\nabla^\CG) = 0$.
More recently, these computations have been extended to non-constant $B$-fields, i.e. where $H = \curv(\nabla^\CG) \neq 0$; see, for instance, \cite{BDLPR--Non-geo_fluxes_and_NAG,Blumenhagen-Plauschinn--NAG_in_string_theory}.

We thus follow the principle that the object providing a higher analogue of the prequantum line bundle should be a bundle gerbe with connection.
In order to investigate and exploit this analogy, we first have to better understand the structural analogies between line bundles and bundle gerbes.

Let us remark that it is not entirely clear what a polarisation in higher geometric quantisation should be.
This problem has been addressed in the language of local bundle gerbes in~\cite{Rogers--Thesis} on the level of higher observable algebras as well as sections, but we will not consider that issue here.

\section{Higher geometric structures}
\label{sect:higher_geometric_structures}

We relate the structures on the 2-category of bundle gerbes unveiled thus far to structures on the the category of line bundles.
This will provide the necessary intuition and understanding for finding analogues of constructions familiar from geometry and geometric quantisation.
Our central application will be higher geometric quantisation in Section~\ref{sect:2-Hspace_of_a_BGrb}.

The discussion in Section~\ref{sect:Deligne_coho_and_higher_cats} suggested already that bundle gerbes provide the next ``higher'' version of hermitean line bundles with connection.
This intuition is central to the investigation of bundle gerbes and related geometric objects~\cite{Brylinski--Loop_spaces_and_geometric_quantisation,Gajer--Geometry_of_Deligne_cohomology,BS--Higher_gauge_theory,FSS--Higher_stacky_perspective}.

Let us start with the simplest possible line bundle with connection, the trivial bundle $I_0 \to M$.
Its endomorphisms (in $\HLBdl^\nabla(M)$) canonically identify with the algebra of smooth functions on $M$, i.e. there is a canonical isomorphism
\begin{equation}
	\HLBdl^\nabla(M)(I_0,I_0) \cong C^\infty(M, \FC)\,,\quad
	f \cdot 1_{I_0} \mapsto f\,.
\end{equation}
In fact, there exists a bijection like this for any complex line bundle $L$ on $M$ (with or without connection),
\begin{equation}
\label{eq:endomps_of_lbdls_to_functions_via_multiples_of_1}
	\HLBdl^\nabla(M) \big( (L, \nabla^L), (L, \nabla^L) \big) \cong
	\HLBdl(M)(L,L) \cong C^\infty(M, \FC)\,,\quad
		f \cdot 1_L \mapsto f\,,
\end{equation}
stemming from the fact that the fibres of $L$ are complex lines.
Inspecting this closely, there exists an algebra structure on $\HLBdl^\nabla(M)((L, \nabla^L), (L, \nabla^L))$ for any $(L,\nabla^L) \in \HLBdl^\nabla(M)$, whose multiplication is given by the composition of endomorphisms.
For $I_0$, however, technically there exist two such structures given by composition and the tensor product of line bundles.
As $I_0$ is the unit object in $\HLBdl^\nabla(M)$, the resulting multiplications agree by an Eckmann-Hilton type argument.
The algebra $\HLBdl^\nabla(M)(I_0,I_0)$ carries a natural involution, induced by complex conjugation
\begin{equation}
	f \cdot 1_{I_0} \mapsto
	\bar{f} \cdot 1_{I_0}\,.
\end{equation}

Any pair $(L,\nabla^L)$, $(J,\nabla^J) \in \HLBdl^\nabla(M)$ of hermitean line bundles with connection on $M$ gives rise to a bimodule $\HLBdl^\nabla(M)((L, \nabla^L), (J, \nabla^J))$ over $(\HLBdl^\nabla(M)(I_0,I_0), \otimes\,)$.
The action is via the tensor product of line bundles.
This is an instance of how in any monoidal category any morphism set carries actions of the endomorphisms of the unit object from both sides.

There is a different way of obtaining an algebra isomorphism $\HLBdl(M)(L,L) \to C^\infty(M, \FC)$, employing the hermitean structure which we so far have not use at all.
The hermitean metric $h_L$ on $L$ allows us, among other things, to define a map
\begin{equation}
	\HLBdl(M)(L,L) \cong C^\infty(M, \FC)\,, \quad
	\psi \mapsto h_{\End(L)}(1_L, \psi)\,,
\end{equation}
which agrees with the map~\eqref{eq:endomps_of_lbdls_to_functions_via_multiples_of_1}.
Here, $h_{\End(L)}$ is the hermitean metric canonically induced by $h_L$ on the endomorphism bundle of $L$.
The algebra structure, however, is much clearer from the original point of view.
For a pair $(L,\nabla^L)$, $(J,\nabla^J) \in \HLBdl^\nabla(M)$, the metrics on $L$ and $J$ induce a morphism of $\HLBdl^\nabla(M)(I_0,I_0)$-modules
\begin{equation}
\label{eq:pairing_of_LBdl_morphs}
\begin{aligned}
	&\overline{\HLBdl^\nabla(M)(L,J)} \times \HLBdl^\nabla(M)(L,J) \to \HLBdl^\nabla(M)(I_0,I_0)\,,
	\\
	&(\psi, \phi) \mapsto h_{J \otimes L^*}(\psi, \phi) \cdot 1_{\CI_0}\,,
\end{aligned}
\end{equation}
where the bar over the first argument indicates antilinearity in this argument.
Here, $h_{J \otimes L^*}$ is the metric induced on $J \otimes L^*$ by the metrics on $h_L$ and $h_J$.
We have $h_{J \otimes L^*}(\psi, \phi) \cdot 1_L = \psi^* \circ \phi$.
It is worth noting how much~\eqref{eq:pairing_of_LBdl_morphs} and the identities
\begin{equation}
	h_{J \otimes L^*}(\psi \cdot f, \phi) = h_{J \otimes L^*}(\psi, \bar{f} \cdot \phi) = \bar{f} \otimes h_{J \otimes L^*}(\psi, \phi)
\end{equation}
resemble the relations~\eqref{eq:internal_hom_and_tensors} and~\eqref{eq:2-var_adjunction_of_[-,-]} of a two-variable adjunction.

Higher analogues of several of these structures on and within the category $\HLBdl^\nabla(M)$ have already been found on the 2-category $\BGrb^\nabla(M)$.
The tensor product of bundle gerbes, tranposes of morphisms, as well as a module action of endomorphisms of $\CI_0$ have been introduced and examined in~\cite{Waldorf--More_morphisms,Waldorf--Thesis}.
We have enlarged the 2-categories of bundle gerbes slightly by dropping the requirements that 2-morphisms should be constructed from parallel, unitary isomorphisms of hermitean vector bundles with connection (cf. Definition~\ref{def:2-categories_of_BGrbs}) and dropping the condition on the trace of the curvature of a 1-morphism (cf. Definition~\ref{def:1-morphisms_of_BGrbs} and Remark~\ref{rmk:trace_condition_on_1-morphisms}).
The additive structure on morphisms of bundle gerbes which parallels the addition of morphisms of hermitean line bundles is given by the direct sum from Theorem~\ref{st:direct_sum_structure_on_morphisms_of_BGrbs}.
Together with the symmetric monoidal structure induced by the tensor product of bundle gerbes, this turns $\BGrb^\nabla(M)(\CI_0,\CI_0)$ into a closed rig category (cf. Theorem~\ref{st:2-var_adjunction_on_BGrb}) which we view as the higher analogue of the ring $\HLBdl^\nabla(M)(I_0,I_0) \cong C^\infty(M, \FC)$.
Note that here we have not yet completely matched the structure on $\HLBdl^\nabla(M)(I_0,I_0)$, which is even a $\FC$-algebra.
We will address this point in due course.
A higher pendant of the hermitean bundle metric acting on $\HLBdl^\nabla(M)((L, \nabla^L), (J, \nabla^J))$ is now provided by the bifunctor (omitting connections)
\begin{equation}
\label{eq:Pairing to higher functions}
	[-,-] \colon \overline{\BGrb^\nabla(M)(\CG_0, \CG_1)} \times \BGrb^\nabla(M)(\CG_0, \CG_1) \to \BGrb^\nabla(M)(\CI_0,\CI_0)
\end{equation}
from Theorem~\ref{st:internal_hom_of_morphisms_in_BGrb--existence_and_naturality} and Theorem~\ref{st:2-var_adjunction_on_BGrb}.
Here, the category $\overline{\BGrb^\nabla(M)((\CG_0, \nabla^{\CG_0}), (\CG_1, \nabla^{\CG_1}))}$ is the category $\big( \BGrb^\nabla(M)((\CG_0, \nabla^{\CG_0}), (\CG_1, \nabla^{\CG_1})) \big)^\opp$ with the conjugate $\BGrb^\nabla(M)(\CI_0,\CI_0)$-module structure, where a higher function $(K,\kappa) \in \BGrb^\nabla(\CI_0,\CI_0)$ acts on a 1-morphism $(E,\alpha) \in \BGrb^\nabla(M)((\CG_0, \nabla^{\CG_0}), (\CG_1, \nabla^{\CG_1}))$ as
\begin{equation}
	\big( (E,\alpha), (K,\kappa) \big) \mapsto \Theta(K, \kappa) \otimes (E,\alpha)\,.
\end{equation}

We now consider the algebra structure of $\HLBdl^\nabla(M)(I_0,I_0)$.
For any pair of line bundles $L, J \in \HLBdl(M)$, the space $\HLBdl(M)(L,J)$ carries a natural $\FC$-module structure.
In geometric terms, this is given by multiplying a homomorphism of line bundles fibrewise by a complex number.
The same action of a complex number is obtained by first using the inclusion $\sfc \colon \FC \hookrightarrow \HLBdl(M)(I,I)$ of algebras, which maps $z \in \FC$ to the constant function with value $z$, and then using the module action of $\HLBdl(M)(I,I)$ on $\HLBdl(M)(L,J)$.
For bundle gerbes we first have to find a suitable replacement for the field of complex numbers $\FC$.
Note that we obtain $\FC$ back from the category of hermitean line bundles via the canonical isomorphism $\HLBdl(\pt)(I,I) \cong \FC$, where $\pt$ is the one-point set, viewed as a 0-dimensional manifold.
For bundle gerbes, Theorems~\ref{st:BGrb_FP_hookrightarrow_BGrb_is_equivalence} and~\ref{st:morphism_categories_and_twisted_HVBdls_for_same_sur_sub} yield the existence of an equivalence of rig categories
\begin{equation}
	\big( \BGrb^\nabla(\pt)(\CI_0,\CI_0), \otimes, 1_{\CI_0}, \oplus, 0_{\CI_0} \big) \cong \big( \Hilb, \otimes, \FC, \oplus, 0 \big)\,.
\end{equation}
Here we have set $0_{\CI_0} = (0,1_0, M,1_M) \colon \CI_0 \to \CI_0$ for the unit object of the direct sum in $\BGrb^\nabla(\pt)(\CI_0,\CI_0)$, i.e. the zero 1-morphism $\CI_0 \to \CI_0$.
Consequently, we view the rig category $(\Hilb, \otimes, \FC, \oplus, 0)$ as the higher analogue of the ground field $\FC$.
This is, in fact, a closed rig category under $\hom_r(V,W) = \hom_l(V,W) = W \otimes V^*$.
As above, there exists a canonical inclusion of higher numbers into higher functions as constant higher functions:
it is given by the functor
\begin{equation}
\label{eq:inclusion_Hilb_into_higher functions}
	\sfc \colon \Hilb \to \BGrb^\nabla_\rmpar(\pt)(\CI_0,\CI_0)\,,\quad
	 \big( V \overset{\psi}{\to} V' \big) \mapsto \big( (M {\times} V,\, \dd) \overset{x \mapsto \psi}{\longrightarrow} (M {\times} V',\, \dd\,) \big)\,,
\end{equation}
with $(M {\times} V, \dd\,)$ denoting the trivial hermitean vector bundle over $M$ with fibre $V$ and the trivial connection given by the de Rham differential.
This is an inclusion of rig categories.

\begin{remark}
Composing the inclusion $\Hilb \hookrightarrow \BGrb^\nabla(M)(\CI_0,\CI_0)$ by the rig category action of $\BGrb^\nabla(M)(\CI_0,\CI_0)$ on any morphism category $\BGrb^\nabla(M)((\CG_0,\nabla^{\CG_0}),(\CG_1,\nabla^{\CG_1}))$, every morphism category in $\BGrb^\nabla(M)$ canonically becomes a $\Hilb$-module as well (cf. Theorem~\ref{st:enrichment_in_BGrb}).
\qen
\end{remark}

The object central to geometric prequantisation is the space of sections of the prequantum line bundle.
Given a line bundle with connection $(L,\nabla^L) \in \HLBdl^\nabla(M)$, its space of smooth sections can be defined geometrically as smooth maps from the base $M$ to the total space $L$ which are right inverses to the projection onto the base.
However, there is a purely categorical way of defining sections as well: we have a canonical bijection
\begin{equation}
	\Gamma(M,L) \cong \HLBdl^\nabla(M)(I_0,L)\,, \quad
	s \mapsto \big( (x,z) \mapsto z\, s(x) \big)
\end{equation}
for $s \in \Gamma(M,L)$, $x \in M$ and $z \in \FC$.
This perspective on sections of $L$ can readily be translated to bundle gerbes.

\begin{definition}[Sections of a bundle gerbe]
\label{def:Sections_of_a_BGrb}
For a bundle gerbe with connection $(\CG, \nabla^\CG) \in \BGrb^\nabla(M)$, we define its \emph{category of sections} as
\begin{equation}
\label{eq:sections_of_a_BGrb}
	\Gamma \big( M, (\CG, \nabla^\CG) \big) \coloneqq \BGrb^\nabla_\rmpar(M) \big( \CI_0, (\CG, \nabla^\CG) \big)\,.
\end{equation}
Equivalently, the \emph{global section functor of bundle gerbes} is the representable
\begin{equation}
	\Gamma \big( M, - \big) \coloneqq \BGrb^\nabla_\rmpar(M) \big( \CI_0, - \big)\,.
\end{equation}
\end{definition}

The idea to define sections of a bundle gerbe in this way has been around as folklore and is due to Konrad Waldorf, but, to our knowledge, its first appearance in the literature has been in~\cite{BSS--HGeoQuan}.

\begin{remark}
Instead of $\BGrb^\nabla_\rmpar(M)( \CI_0, (\CG, \nabla^\CG))$, we could have used the bigger category $\BGrb^\nabla(M)( \CI_0, (\CG, \nabla^\CG))$ to define sections of $(\CG,\nabla^\CG)$.
In this sense, our definition of sections is not the most general one.
However, in Section~\ref{sect:2-Hspace_of_a_BGrb} we will make use of the more refined structure of $\BGrb^\nabla_\rmpar(M)( \CI_0, (\CG, \nabla^\CG))$ (see Theorem~\ref{st:enrichment_in_BGrb}).
\qen
\end{remark}

Denote by $\rmH_\dR^p(M,\iu\, \FR) = \iu\, \rmH_\dR^p(M,\FR)$ the $p$-th de Rham cohomology group of $M$ with coefficients in $\iu\, \FR$.
Proposition~\ref{st:determinants_of_sections_and_dual_sections} together with Theorem~\ref{st:Classification_of_BGrbs_by_Deligne_coho} combine to an important no-go statement:

\begin{proposition}
\label{st:no-go_for_sections_of_BGrbs}
A bundle gerbe with connection $(\CG, \nabla^\CG)$ with $[H]_\dR \neq 0 \in \rmH_\dR^3(M,\iu\,\FR)$ admits no non-zero sections.
\end{proposition}

\begin{proof}
If there is a non-trivial section of $(\CG,\nabla^\CG)$, then the Dixmier-Douady class of $\CG$ is torsion by Proposition~\ref{st:determinants_of_sections_and_dual_sections} and Theorem~\ref{st:Classification_of_BGrbs_by_Deligne_coho}.
The statement then follows from the observation that the image of $\DD(\CG)$ in de Rham cohomology agrees with $[H]_\dR$~\cite{Waldorf--Thesis} and the fact that the kernel of the map $\rmH^k(M,\RZ) \to \rmH^k_\dR(M, \iu\, \FR)$ is precisely the torsion subgroup of $\rmH^k(M,\RZ)$.
\end{proof}

This puts a strong restriction on the class of bundle gerbes which admits non-trivial sections.
A bundle gerbe can admit non-trivial sections only if its Dixmier-Douady class is torsion.
There are several approaches to how to circumvent this constraint which we will survey in Section~\ref{sect:ways_around_the_torsion_constraint}.
We will see a non-trivial example of a bundle gerbe with torsion Dixmier-Douady class and investigate its category of sections in Section~\ref{sect:Ex:Cup_product_lens_space_BGrbs}.

\begin{remark}
Note that Proposition~\ref{st:no-go_for_sections_of_BGrbs} also implies that sections of bundle gerbes over $M$ can not generally be extended from closed subsets of $M$.
Most strikingly, since every bundle gerbe over $\pt$ is trivial, we always have $\Gamma(\pt, x^*(\CG, \nabla^\CG)) \cong \Hilb$ for any map $x \colon \pt \to M$.
However, even if $(\CG, \nabla^\CG)$ has a torsion class, only those sections over $x \in M$ with appropriate ranks may be restrictions of global sections of $(\CG, \nabla^\CG)$.
\qen
\end{remark}

We close this section by giving a definition for the bundle gerbe metric which is strongly motivated by~\eqref{eq:Pairing to higher functions}.

\begin{definition}[Bundle gerbe metric]
We call the bifunctor
\begin{equation}
	[-,-] \colon \overline{\Gamma \big( M, (\CG, \nabla^{\CG}) \big)} \times \Gamma \big( M, (\CG, \nabla^{\CG}) \big) \to \BGrb^\nabla(M)(\CI_0, \CI_0)
\end{equation}
the \emph{bundle gerbe metric of $(\CG, \nabla^{\CG})$}.
\end{definition}

\section{2-Hilbert spaces}
\label{sect:2-Hspaces}

In Section~\ref{sect:higher_geometric_structures} we have established several correspondences between notions derived from hermitean line bundles with connection and hermitean bundle gerbes with connection.
In order to proceed further along the lines of geometric quantisation in the higher setting, in the sense that we replace line bundles by bundle gerbes, we have to answer the question what the higher analogue of the pre-Hilbert space of smooth sections of the prequantum line bundle should be.
To that end, we first have to settle for a notion of a \emph{2-Hilbert space}.
There currently exist different notions of 2-vector spaces.
The most well-known are \emph{Kapranov-Voevodsky 2-vector spaces}~\cite{KV--2-Cats_and_Zam_eqns}, which are module categories over $\Vect$, the category of finite-dimensional complex vector spaces, \emph{Baez-Crans 2-vector spaces}~\cite{Baez-Crans--HdimA-Lie_2-algebras}, which are categories internal to $\Vect$, and finitely semi-simple, linear, abelian categories (see, for instance, \cite{FSS--Trace_for_bimodule_categories}).

We have seen that we should think of the rig category $\Hilb$ as the correct higher replacement of $\FC$ as the ground field, and that morphism categories of bundle gerbes naturally carry the structure of $\Hilb$-module categories.
Therefore, we follow the line of Kapranov and Voevodsky and consider $\Hilb$-module categories as the foundational structure for 2-Hilbert spaces.
Our definition of 2-Hilbert spaces will be closely related to that initiated in~\cite{Freed--Extended_structures,Crane--Clock_and_Cat} and worked out in detail in~\cite{Baez--2-Hilbert_spaces}, but it will differ in that we put more emphasis on the $\Hilb$-module structure.

\begin{definition}[2-Hilbert space]
\label{def:2Hspace}
A \emph{2-Hilbert space} a cartesian monoidal, abelian category $(\scH, \oplus, \Null_\scH)$ with the following additional data:
\begin{myenumerate}
	\item a closed rig-module structure over $(\Hilb, \otimes, \FC, \oplus, 0)$ (cf. Definition~\ref{def:closed_rig-module}), where we set $\<-,-\>_\scH \coloneqq \hom_l$ and call this bifunctor the \emph{inner product, or pairing of $\scH$}, and
	
	\item a natural isomorphism $\<\CA,\CB\>_\scH \cong \<\CB,\CA\>_\scH^*$ for $\CA, \CB \in \scH$, compatible with the rig-module structure.
\end{myenumerate}
\end{definition}

Despite the close similarity to the definition used in~\cite{Baez--2-Hilbert_spaces}, the application of the structure of 2-variable adjunctions in the definition of a 2-Hilbert space is new.
Note that the symmetry of the inner product may be achieved by a star structure as used in~\cite{Baez--2-Hilbert_spaces}.
The structure in the definition of a 2-Hilbert space is strong and, thus, has several implications:

\begin{proposition}
\label{st:structure_on_2Hspaces}
Let $\scH$ be a 2-Hilbert space with inner product $\<-,-\>_\scH$, and let $\CA, \CB \in \scH$.
The following statements hold true.
\begin{myenumerate}
	\item $\scH$ is enriched, tensored and cotensored over $\Hilb$.
	
	\item There is a natural isomorphism $\<-,-\>_\scH \cong \scH(-,-)$.
	
	\item $\<-,-\>_\scH$ is non-degenerate.
	
	\item The right internal hom is fixed up to natural isomorphism as $\hom_r \cong (-) \otimes (-)^*$.
\end{myenumerate}
\end{proposition}

\begin{proof}
Ad (1):
Any closed module category over a closed symmetric monoidal category is automatically enriched, tensored and cotensored over the symmetric monoidal category by the properties of the two-variable adjunction, see~\cite[Chapter 3]{Riehl--Categorical_homotopy_theory}.
In fact the properties enriched, tensored and cotensored are equivalent to being a closed module category in that set-up.

Ad (2):
The two-variable adjunction yields natural isomorphisms
\begin{equation}
	\scH( \CV, \CW) \cong \scH( \FC \otimes \CV, \CW)
	\cong \Hilb \big( \FC, \< \CV, \CW \>_\scH \big)
	\cong \< \CV, \CW \>_\scH\,.
\end{equation}

Ad (3):
If $\CA = 0$, then $\<\CA, \CB\>_\scH \cong \scH(\CA, \CB) = 0$.
In the opposite way, assume that $\<\CA, \CB\>_\scH \cong \scH(\CA,\CB) = 0$ for all $\CB \in \scH$.
That is, $\CA$ is initial in $\scH$.
By the symmetry of the pairing $\<-,-\>_\scH$ the same argument also yields that $\CA$ is terminal, and thus a zero object in $\scH$.

Ad (4):
We have natural isomorphisms
\begin{equation}
\begin{aligned}
	\scH(\CA, \CB \otimes V^*)
	&\cong \scH(\CB \otimes V^*, \CA)^*
	\cong \big( \Hilb(V^*, \scH(\CB, \CA)) \big)^*
	\cong V^* \otimes \scH(\CB, \CA)^*
	\\*
	&\cong V^* \otimes \scH(\CA, \CB)
	\cong \scH(\CA \otimes V, \CB)
	\\*
	&\cong \scH(\CA, \hom_r(V, \CB))\,.
\end{aligned}
\end{equation}
Thus, by the enriched Yoneda Lemma (cf., for instance, \cite{Riehl--Categorical_homotopy_theory}), there exists a natural isomorphism as claimed.
\end{proof}

\begin{remark}
In certain applications it might be desirable to weaken the definition to an abelian category rather than a cartesian monoidal abelian category.
The symmetry of the inner product together with the natural isomorphism $\<-,-\>_\scH \cong \scH(-,-)$ then still yield a possibly non-functorial version of sesquilinearity of the inner product.
Note also that because of the properties of a two-variable adjunction, all the involved functors are compatible with either products or coproducts in each of their arguments.
Here we assume that the coproduct in $\scH$ can be made functorial, i.e. we work in the cartesian monoidal setting.
\qen
\end{remark}

\begin{remark}
Extending the remark on the implications of the adjointness properties of the closed module structure, it might actually be enough to demand that $\scH$ be abelian and cartesian monoidal with a closed $(\Hilb, \otimes, \FC)$-module structure, i.e. without referring to any direct sums, in order to obtain a rig-module structure.
Because every $- \otimes V$ is a left adjoint, the cartesian monoidal structure, which is automatically also cocartesian in an abelian category, is preserved by any $- \otimes V$.
That is, we automatically obtain isomorphisms $(\CA \oplus \CB) \otimes V \cong \CA \otimes V \oplus \CB \otimes V$.
The same argument applies to the other argument of $- \otimes -$ since $\Hilb$ is cartesian monoidal as well.
It is unclear, however, whether these isomorphisms are automatically coherent.
\qen
\end{remark}

As a final remark, note that if we require merely a monoidal structure $\oplus$ on $\scH$ instead of a cartesian monoidal structure, we automatically obtain from the sequilinearity and Proposition~\ref{st:structure_on_2Hspaces} that
\begin{equation}
	\scH(\CC, \CA \oplus \CB) \cong \<\CC, \CA \oplus \CB\>_\scH
	\cong \<\CC, \CA \>_\scH \oplus \<\CC, \CB\>_\scH
	\cong \scH(\CC, \CA) \oplus \scH(\CC, \CB)\,.
\end{equation}
Together with the morphism $\CA \oplus \CB \to \CA$ obtained as the image of $1_\CA \oplus 0$ under the isomorphism $\scH(\CA, \CA) \oplus \scH(\CB, \CA) \cong \scH(\CA \oplus \CB, \CA)$ and analogously for $\CA \oplus \CB \to \CB$, this establishes $\CA \oplus \CB$ as the categorical product $\CA {\times} \CB$.
A similar construction for the coproduct shows that we automatically obtain that $(\scH, \oplus, \Null_\scH)$ is cartesian monoidal.

\begin{example}
\begin{myenumerate}
	\item The primordial 2-Hilbert space~\cite{Baez--2-Hilbert_spaces} is $\Hilb$ itself:
	the action functor $\otimes$ is the usual tensor product of finite-dimensional Hilbert spaces,
	and $\hom_l (V,W) = \<V,W\>_{\Hilb} = \Hilb(V,W) \cong V^* \otimes W$.
	
	\item Very similarly, the $n$-fold product $\Hilb^n$ is a 2-Hilbert space in the sense of Definition~\ref{def:2Hspace}:
	the module action reads as $(V_0, \ldots V_{n-1}) \otimes W = (V_0 \otimes W, \ldots, V_{n-1} \otimes W)$ and the inner product is
	\begin{equation}
		\big[ (V_0, \ldots V_{n-1}),\, (W_0, \ldots W_{n-1}) \big]
		= \bigoplus_{k=0}^{n-1}\ [V_k, W_k] \quad \in \Hilb\,.
	\end{equation}
	
	\item Let $\sfG$ be a compact Lie group, and denote its category of finite-dimensional unitary representations by $\URep(\sfG)$.
	We write $\CV = (W, \rho)$ for an object of $\URep(\sfG)$, where $W \in \Hilb$ and $\rho$ is a unitary representation of $\sfG$ on $W$.
	We take morphisms $(W, \rho) \to (W', \rho')$ to be linear maps $\varphi \in \Hilb(W,W')$ such that $\rho'(g) \circ \varphi \circ \rho(g^{-1}) = \varphi$ for all $g \in \sfG$.
	An action of $\Hilb$ on $\URep(\sfG)$ is given by $(W, \rho) \otimes V = (W, \rho) \otimes (V, 1) = (W \otimes V, \rho \otimes 1)$, where $1$ denotes the trivial representation of $\sfG$.
	The inner product has to be defined as
	\begin{equation}
		\big[ (W,\rho), (W', \rho') \big] = \URep(\sfG) \big( (W,\rho), (W', \rho') \big)\,.
	\end{equation}
	Its sesquilinearity can be seen as follows:
	consider the diagram of finite-dimensional Hilbert spaces
	\begin{equation}
	\begin{tikzcd}
			\URep(\sfG) \big( (W,\rho), (W', \rho') \big) \ar[r, "(-)^\sft"] \ar[dr, dashed] & \URep(\sfG) \big( (W',\rho')^*, (W, \rho)^* \big) \ar[d, "(-)^*"]
			\\
			 & \URep(\sfG) \big( (W',\rho')^*, (W, \rho)^* \big)\,.
	\end{tikzcd}
	\end{equation}
	Both solid arrows are isomorphisms, but the vertical arrow is antilinear.
	Therefore, the dashed arrow is an antilinear isomorphism, and we obtain an isomorphism
	\begin{equation}
		\URep(\sfG) \big( (W,\rho), (W', \rho') \big) \cong \URep(\sfG) \big( (W',\rho'), (W, \rho) \big)^*
	\end{equation}
	as desired.
	\qen
\end{myenumerate}
\end{example}

\begin{example}
\label{eg:Hilb_sep_is_not_a_2Hspace}
The category $\Hilb_\sep$ of possibly infinite-dimensional separable Hilbert spaces is not a 2-Hilbert space in this sense:
while an action of $\Hilb$ or $\Hilb_\sep$ on $\Hilb_\sep$ can be defined using the tensor product, an isomorphism $\<-,-\>_{\Hilb_\sep} \cong \Hilb_\sep(-,-)$ yields a contradiction.
This is because $\Hilb_\sep$ is not enriched over itself.
Instead, the morphism set $\Hilb_\sep(\CV, \CW)$ of bounded linear operators $\CV \to \CW$ carries the structure of a Banach space, while $\<\CV, \CW\>_{\Hilb_\sep}$ is supposed to be a Hilbert space isomorphic to $\Hilb_\sep(\CV, \CW)$.
Moreover, $\Hilb_\sep$ is not abelian: there exist injective continuous operators between infinite-dimensional separable Hilbert spaces whose image is not a closed subspace so that they are not the kernel of another bounded operator.
For example, the operator $T \colon \CV \to \CV$ defined on an orthonormal basis $(e_k)_{k \in \NN_0}$ by $T e_k = \frac{1}{k}\, e_k$ is injective, bounded ($T \leq 1_\CV$) and has dense image:
given $\psi = \sum_k \psi^k\, e_k \in \CV$, define $\psi_n \in \CV$ to have coefficients $(\psi_n)^k = k\, \psi^k$ for $k \leq n$ and $(\psi_n)^k = \psi^k$ for $k > n$.
Then we have $\Vert \psi - T(\psi_n) \Vert_\CV \to 0$, where $\Vert - \Vert_\CV$ is the Hilbert norm on $\CV \in \Hilb_\sep$.
Thus, there exists a sequence $(T \psi_n)_{n \in \NN_0}$ in $T \CV$ that converges to $\psi$.
However, the vector $\phi = \sum_k \frac{1}{k}\, e_k$ is not in the image of $T$.
Its preimage under $T$ would be $\phi' = \sum_k e_k$, which would have infinite norm.
Therefore, $T\CV \subset \CV$ is dense, but not closed, and we have thus found a monomorphism in $\Hilb_\sep$ which is not the kernel of any other morphism, since kernels of bounded operators are closed subspaces.
\qen
\end{example}

Note that for a complete analogy between 2-Hilbert spaces and Hilbert spaces there is one element missing entirely.
We have not accounted for the norm-completeness of a Hilbert space in Definition~\ref{def:2Hspace}.
There are several reasons for this, but probably the most striking one is the absence of a good notion of a difference of two objects.
While Baez suggests in~\cite{Baez--2-Hilbert_spaces} to understand cokernels of morphisms as differences, it is not quite clear whether this produces a useful concept.
A better, yet far more abstract way to incorporate differences into our framework might be to pass to ring completions of the rig categories $\Hilb$ and $\BGrb^\nabla(M)(\CI_0,\CI_0)$ and to construct corresponding structures on the categories $\BGrb^\nabla((\CG_0, \nabla^{\CG_0}), (\CG_1, \nabla^{\CG_1}))$.
A general formalism for such ring completions of rig categories has been worked out in~\cite{BDRR--Ring_completion_of_rig_categories}.
However, the resulting category has not been made explicit for the cases we need, and, while providing a very interesting question, doing so would go beyond the scope of this thesis.
It might also be interesting to see whether the framework of ``categories with norms'' recently proposed in~\cite{Kubis--Cats_with_norms} yields new insights.

\section{The 2-Hilbert space of a bundle gerbe}
\label{sect:2-Hspace_of_a_BGrb}

In geometric quantisation, the prequantum pre-Hilbert space is obtained as the space of smooth sections of the prequantum line bundle.
This comes endowed with a canonical inner product, obtained by inserting two sections into the bundle metric and integrating the resulting function over the base manifold $M$.
So far, we have seen that in higher geometric quantisation the prequantum line bundle should be replaced by a bundle gerbe.

\begin{definition}[2-plectic manifold, prequantum bundle gerbe~\cite{Rogers--Thesis}]
Let $M$ be a manifold.
A closed, non-degenerate 3-form $\varpi \in \Omega^3(M,\FR)$ is called a \emph{2-plectic form} on $M$.
The pair $(M, \varpi)$ is called a \emph{2-plectic manifold}.
If $\varpi$ has integer cycles, $(M,\varpi)$ is called \emph{prequantisable}.
A hermitean bundle gerbe with connection $(\CG,\nabla^\CG)$ on $M$ such that $\curv(\nabla^\CG) = 2 \pi\, \iu\, \varpi$ is called a \emph{prequantum bundle gerbe for $(M,\varpi)$}.
\end{definition}

In Section~\ref{sect:higher_geometric_structures} we have defined the category of sections of a bundle gerbe.
Thus, given a prequantisable 2-plectic manifold $(M,\varpi)$ with a choice of prequantum bundle gerbe $(\CG,\nabla^\CG)$, the natural candidate for the underlying $\Hilb$-module category of the prequantum 2-Hilbert space of this system is
\begin{equation}
	\scH_0(\CG, \nabla^\CG) \coloneqq \Gamma \big( M, (\CG, \nabla^\CG) \big)
	= \BGrb^\nabla_\rmpar(M) \big( \CI_0, (\CG, \nabla^\CG)  \big)\,.
\end{equation}
Theorem~\ref{st:enrichment_in_BGrb} implies that this is a cartesian monoidal, abelian $\Hilb$-module category.
Furthermore, Theorem~\ref{st:2-var_adjunction_on_BGrb} states that the bifunctors $[-,-]$ and $(-) \otimes \Theta(-)$ induce a two-variable adjunction and, hence, a closed $\BGrb^\nabla_\rmpar(M)(\CI_0, \CI_0)$-module structure on $\Gamma(M,(\CG,\nabla^\CG))$.
However, we need a two-variable adjunction with respect to the $\Hilb$-module action induced by the inclusion functor~\eqref{eq:inclusion_Hilb_into_higher functions}.
Part (2) of Proposition~\ref{st:structure_on_2Hspaces} fixes the inner product bifunctor up to natural isomorphism, while part (3) of the same statement fixes the right internal hom $\hom_r$ of the closed module structure.
We can use Corollary~\ref{st:Hilb-enrichment_of_MorCats_in_BGrb} to obtain an inner product bifunctor with values in $\Hilb$ in a slightly modified way:
we set
\begin{equation}
\label{eq:inner product on 2-Hspace_of BGrb}
\begin{aligned}
	\<-,-\>_{\scH_0(\CG, \nabla^\CG)} &\coloneqq \BGrb^\nabla_\rmpar(M) \big( 1_{\CI_0}, [-,-] \big)
	\\*
	&\cong \Gamma_\rmpar \big( M, \sfR[-,-] \big)
	\\*
	&\cong \BGrb^\nabla_\rmpar(M)( -,- )\,.
\end{aligned}
\end{equation}

\begin{remark}
Note that, for $\{ M_j \}_{j \in \pi_0(M)}$ denoting the connected components of $M$, we have
\begin{equation}
\begin{aligned}
	\scH_0(\CG, \nabla^\CG)
	&= \Gamma \big( M, (\CG, \nabla^\CG) \big)
	\\*
	&\cong \bigoplus_{j \in \pi_0(M)}\, \Gamma \big( M_j, (\CG, \nabla^\CG) \big)
	\\*
	&= \bigoplus_{j \in \pi_0(M)}\, \scH_0 \big( (\CG, \nabla^\CG)_{|M_j}
\end{aligned}
\end{equation}
and
\begin{equation}
\begin{aligned}
	\<-,-\>_{\scH_0(\CG, \nabla^\CG)}
	&\cong \Gamma_\rmpar \big( M, \sfR[-,-] \big)
	\\*
	&\cong \bigoplus_{j \in \pi_0(M)}\, \Gamma_\rmpar \big( M_j, \sfR[-,-] \big)
	\\*
	&\cong \bigoplus_{j \in \pi_0(M)}\, \<-,-\>_{\scH_0((\CG, \nabla^\CG)_{|M_j})}\,,
\end{aligned}
\end{equation}
in analogy with the structure of spaces of sections of hermitean line bundles.
The direct sum $\bigoplus_{j = 0}^n\, \scH_j$ of 2-Hilbert spaces $\scH_j$, $j = 0, \ldots, n$, is the product of the underlying categories endowed with the $\Hilb$-action
\begin{equation}
	(\CA_0, \ldots, \CA_n) \otimes V
	= \big( \CA_0 \otimes V, \ldots, \CA_n \otimes V \big)
\end{equation}
for $\CA_j \in \scH_j$.
The action of $\hom_r$ is given analogously, while, for $\CA'_j \in \scH_j$, the action of $\hom_l$ reads as
\begin{equation}
\begin{aligned}
	 \big[ (\CA_0, \ldots, \CA_n),\, (\CA'_0, \ldots, \CA'_n) \big]
	 &\cong \bigoplus_{j,k = 0}^n\, \Big( \bigoplus_{l = 0}^n\, \scH_l \Big) (\CA_j, \CA'_k)
	 \\*
	 &\cong \bigoplus_{j = 0}^n\, \scH_j (\CA_j, \CA'_j)
	 \\*
	 &= \bigoplus_{j = 0}^n\, [\CA_i, \CA'_i]\,,
\end{aligned}
\end{equation}
like we obtained above from geometric considerations.
\qen
\end{remark}

Let $V \in \Hilb$ and $(E, \alpha)$, $(F, \beta) \in \Gamma(M, (\CG, \nabla^\CG)$.
Recall the definition of the inclusion functor $\sfc \colon \Hilb \hookrightarrow \BGrb^\nabla(M)(\CI_0, \CI_0)$ from ~\eqref{eq:inclusion_Hilb_into_higher functions}.
Since parallel sections of the bundle $\sfc V = (M {\times} V, \dd)$ are locally constant, and, therefore, locally given by elements of $V$, the properties of $[-,-]$ yield natural isomorphisms
\addtocounter{equation}{1}
\begin{align*}
	\< (E, \alpha) \otimes \sfc V,\, (F, \beta) \>_{\scH_0(\CG, \nabla^\CG)}
	&\cong \BGrb^\nabla_\rmpar(M) \big( (E, \alpha) \otimes \sfc V,\, (F, \beta) \big)
	\\*[0.2cm]
	&\cong \BGrb^\nabla_\rmpar(M) \big( 1_{\CI_0},\, \big[ (E, \alpha)\otimes \sfc V, (F, \beta) \big] \big)
	\\[0.2cm]
	&\cong \BGrb^\nabla_\rmpar(M) \big( 1_{\CI_0},\, \big[ (E, \alpha), (F, \beta) \big] \otimes \sfc V^* \big)
	\\[0.2cm]
	&\cong \Gamma_\rmpar \big( M,\, \sfR \big[ (E, \alpha), (F, \beta) \big] \otimes \sfc V^* \big) \theeq
	\\[0.2cm]
	&\cong \bigoplus_{j \in \pi_0(M)}\, \Big( \Gamma_\rmpar \big( M_j,\, \sfR \big[ (E, \alpha), (F, \beta) \big] \big) \otimes \sfc V^* \Big)
	\\[0.2cm]
	&\cong \Big( \bigoplus_{j \in \pi_0(M)}\, \Gamma_\rmpar \big( M_j,\, \sfR \big[ (E, \alpha), (F, \beta) \big] \big) \Big) \otimes \sfc V^*
	\\*[0.2cm]
	&\cong \< (E, \alpha),\, (F, \beta) \>_{\scH_0(\CG, \nabla^\CG)} \otimes \sfc V^*\,.
\end{align*}
The linearity of $\<-,-\>_{\scH_0(\CG, \nabla^\CG)}$ in the other argument can be checked analogously.
Hence, the choice~\eqref{eq:inner product on 2-Hspace_of BGrb} provides a two-variable adjunction suitable for a 2-Hilbert space structure on $\Gamma(M,(\CG, \nabla^\CG))$.

One can check the symmetry of $\<-,-\>_{\scH_0(\CG, \nabla^\CG)}$ in the sense of axiom (2) of Definition~\ref{def:2Hspace}.
This is due to the natural isomorphisms found in Section~\ref{sect:Pairings_and_inner_hom_of_morphisms_in_BGrb}:
\addtocounter{equation}{1}
\begin{align*}
	\<(E,\alpha), (F, \beta) \>_{\scH_0(\CG, \nabla^\CG)}
	&= \BGrb^\nabla_\rmpar(M) \big( 1_{\CI_0},\, [ (E,\alpha), (F, \beta)] \big)
	\\*
	&\cong \Gamma_\rmpar \big( M, \sfR[(E,\alpha), (F, \beta)] \big)
	\\
	&\cong \Gamma_\rmpar \big( M, \sfR \circ \Theta [(F, \beta), (E,\alpha)] \big)
	\\
	&\cong \Gamma_\rmpar \big( M, \big( \sfR [(F, \beta), (E,\alpha)] \big)^* \big) \theeq
	\\
	&\cong \Big( \Gamma_\rmpar \big( M, \sfR [(F, \beta), (E,\alpha)] \big) \Big)^*
	\\
	&\cong \BGrb^\nabla_\rmpar(M) \big( 1_{\CI_0},\, [ (F, \beta), (E,\alpha)] \big) \Big)^*
	\\*
	&= \big( \<(F, \beta), (E,\alpha) \>_{\scH_0(\CG, \nabla^\CG)} \big)^*\,.
\end{align*}
Finally, consider the self-pairing
\begin{equation}
\begin{aligned}
	\<(E,\alpha), (E,\alpha) \>_{\scH(\CG, \nabla^\CG)}
	&= \BGrb^\nabla_\rmpar(M) \big( 1_{\CI_0},\, [ (E,\alpha), (E,\alpha)] \big)
	\\*
	&\cong \BGrb^\nabla_\rmpar(M) \big( (E,\alpha), (E,\alpha) \big)\,.
\end{aligned}
\end{equation}
This contains at least all multiples of $1_{(E,\alpha)}$, i.e. it is a finite-dimensional Hilbert space of dimension at least one, unless $(E,\alpha) = 0$.
That is, the pairing $\<-,-\>_{\scH_0(\CG, \nabla^\CG)}$ is non-degenerate.
We summarise the above results in

\begin{theorem}
\label{st:2-Hspace_of_a_Bgrb}
Given a bundle gerbe $(\CG,\nabla^\CG) \in \BGrb^\nabla(M)$, the pair
\begin{equation}
	\big( \scH_0(\CG, \nabla^\CG), \<-,-\>_{\scH_0(\CG, \nabla^\CG)} \big)
	= \Big( \Gamma \big( M, (\CG, \nabla^\CG) \big) ,\, \BGrb^\nabla_\rmpar(M) \big( 1_{\CI_0}, [-,-] \big) \Big)
\end{equation}
together with the $\Hilb$-module action on $\Gamma \big( M, (\CG, \nabla^\CG) \big)$ induced by the inclusion~\eqref{eq:inclusion_Hilb_into_higher functions} and $\hom_r = (-) \otimes \Theta(-)$ defines a 2-Hilbert space in the sense of Definition~\ref{def:2Hspace}.
\end{theorem}

\begin{definition}[Prequantum 2-Hilbert space]
For a prequantisable 2-plectic manifold $(M, \varpi)$ and a choice of a prequantum bundle gerbe $(\CG, \nabla^\CG) \in \BGrb^\nabla(M)$ for $(M,\varpi)$, we call the 2-Hilbert space $\big( \scH_0(\CG, \nabla^\CG), \<-,-\>_{\scH_0(\CG, \nabla^\CG)} \big)$ the \emph{prequantum 2-Hilbert space of $(M,\varpi)$}.
\end{definition}

Let us comment on the motivation for using the space of sections of $\sfR[(E,\alpha), (F, \beta)]$ as the 2-Hilbert space inner product of $(E,\alpha)$ and $(F, \beta)$ (at least up to isomorphism).
In Section~\ref{sect:higher_geometric_structures} we have argued that we can view $[-,-]$ as a higher version of the hermitean bundle metric.
We had also seen that the reduction functor $\sfR$ yields an equivalence between higher functions and hermitean vector bundles with connection.
Thus, in order to mimic the pairing of sections of a hermitean line bundle, we would first insert $(E,\alpha)$ and $(F, \beta)$ into the higher bundle metric.
Applying the equivalence $\sfR$, this produces a hermitean vector bundle with connection $\sfR[(E,\alpha), (F, \beta)] \in \HVBdl^\nabla(M)$.
For the pairing of sections $\psi, \phi$ of a line bundle $(L, \nabla^L)$ with bundle metric $h_L$, we would integrate the function $h_L(\psi, \phi)$ over $M$.
Viewing $\sfR[(E,\alpha), (F, \beta)]$ as a function valued in $\Hilb$, this would, intuitively, amount to taking a weighted direct sum of all fibres of $\sfR[(E,\alpha), (F, \beta)]$.
This procedure has been made precise in harmonic analysis, and is referred to as the \emph{direct integral} of a family of Hilbert spaces on $M$, see for instance~\cite{Folland--Harmonic_analysis}.
Heuristically, the direct integral of a family of Hilbert spaces amounts to forming the (infinite-dimensional) Hilbert space of square integrable sections of that family.
We, in contrast, would like to stay in the finite-dimensional setting here, which is why we take a reduced version of this general direct integral by forming the space of parallel sections only.

\begin{remark}
The hermitean structure of the bundle gerbe $(\CG, \nabla^\CG)$ only enters in the pairing of 2-morphisms.
That is, perhaps somewhat surprisingly, we can still define the functor $[-,-]$ for bundle gerbes without hermitean bundle metric on their defining line bundle.
This is similar to how we can still define a functorial pairing $(V, W) \mapsto W \otimes V^*$ on $\Vect$, the category of finite-dimensional $\FC$-vector spaces, which acts on morphisms as $(\psi \colon W \to W', \phi \colon V' \to V) \mapsto \phi \otimes \psi^\sft \colon W \otimes V^* \to W' \otimes V^{\prime\, *}$; the inner product is never used here.
Only when we desire to define an inner product of morphisms of vector spaces do we need inner products on the vector spaces here.
A possible interpretation in terms of higher vector spaces is the following.
Note, first, that we use the dual vector space as a higher replacement for the complex conjugate of a complex number.
For any 1-dimensional complex vector space $V \in \Vect$, there exist isomorphisms of vector spaces $V \otimes V^* \cong \Vect(V,V) \cong \FC$.
Thus, at the same time as being the conjugate of $V$, the vector space $V^*$ provides an inverse of $V$ with respect to the monoidal structure $\otimes$ on $\Vect$.
In this sense, the conjugate and the inverse of $V$ in $\Vect$ agree, which means that we should view any 1-dimensional vector space as a \emph{higher unitary number}, or a higher unitary object in the rig category $\Vect$.
This indicates that many of our results which do not use the inner product on any vector space should generalise to the weaker setting of bundle gerbes without a hermitean structure.
Their higher transition functions, i.e. their defining line bundles, are still valued in unitary objects of $\Vect$.
From this perspective, bundle gerbes without hermitean structure would still be hermitean 2-line bundles, while hermitean bundle gerbes inherit extra structure on the level of 2-morphisms.
\qen
\end{remark}

\section{Example: Bundle gerbes on $\FR^3$}
\label{sect:Ex:BGrbs_on_R3}

As a first example of the abstract theory developed in this and the preceding chapters, we consider the bundle gerbe $\CI_\rho$ on $M \in \Mfd$ for some 2-form $\rho \in \Omega^2(M, \iu\, \FR)$.
In this case, there exists an equivalence of symmetric monoidal categories
\begin{equation}
	\big( \Gamma( M, \CI_\rho), \oplus\, \big) \cong \big( \HVBdl_\rmpar^\nabla(M), \oplus\, \big)\,.
\end{equation}
The target category naturally carries an action of the rig-category
\begin{equation}
	\big( \Gamma( M, \CI_0), \otimes, \oplus\, \big) \cong \big( \HVBdl_\rmpar^\nabla(M), \otimes, \oplus\, \big)\,,
\end{equation}
corresponding to the rig-module structure of $\Gamma(M, \CI_\rho)$ over $\BGrb^\nabla(M)(\CI_0, \CI_0) \cong \Gamma(M,\CI_0)$.
This is analogous to the identity
\begin{equation}
	\Gamma(M, I_A) = \Gamma(M,I_0)
\end{equation}
of sections of the trivial line bundle; the choice of connection has no effect on the space of sections.
The rig-module structure on $\HVBdl_\rmpar^\nabla(M) \cong \Gamma(M,\CI_\rho)$ is given by the tensor product of vector bundles with connections, $\Theta$ acts by taking dual bundles and transpose morphisms, and the bundle gerbe metric acts via
\begin{equation}
	\big( (E, \nabla^E), (F, \nabla^F) \big) \mapsto (F, \nabla^F) \otimes (E, \nabla^E)^*\,, \quad \text{i.e.} \quad
	[-,-] = (-) \otimes (-)^* \circ \sw^\otimes\,,
\end{equation}
with $\sw^\otimes$ denoting the symmetry of the tensor product on $\HVBdl^\nabla_\rmpar(M)$.

Thus, we have an equivalence
\begin{equation}
	\big( \scH_0(\CI_\rho),\, \<-,-\>_{\scH_0(\CI_\rho)} \big)
	\cong \big( \HVBdl^\nabla_\rmpar(M),\, \Gamma_\rmpar(M,-) \circ \big( (-) \otimes (-)^* \big) \circ \sw^\otimes \big)
\end{equation}
as 2-Hilbert spaces.

The situation simplifies further if we assume a contractible base manifold, such as $M = \FR^3$.
As a curving on $\CI_\rho$ we can, for example, choose $\rho = \frac{\iu\, \pi}{3}\, \varepsilon_{jkl}\, x^j\, \dd x^k \wedge \dd x^l$, with $j,k,l \in \{0,1,2\}$.
Its differential reads $\curv(\CI_\rho) = \dd \rho = 2 \pi\, \iu\, \dd x^0 \wedge \dd x^1 \wedge \dd x^2 = 2\pi\, \iu\, \varpi_{\FR^3}$, for the canonical 2-plectic form $\varpi_{\FR^3} = \dd x^0 \wedge \dd x^1 \wedge \dd x^2$ on $\FR^3$.

On a contractible base manifold $M$, every hermitean vector bundle is trivialisable as a vector bundle without connection.
If $M$ is connected, there exist isomorphisms
\begin{equation}
	(E, \nabla^E) \cong \big( M {\times} \FC^{\rank(E)},\, \dd + A \big)
\end{equation}
for some $A \in \Omega^1(M, \fru(\rank(E)))$, where $\fru(n)$ denotes the Lie algebra of the unitary group $\sfU(n)$ for $n \in \NN_0$.
Thus, there exists a further equivalence
\begin{equation}
	\HVBdl^\nabla_\rmpar(M) \cong \Omega^1(M,\fru)\,,
\end{equation}
where $\Omega^1(M, \fru)$ is the following category:
its objects are Lie algebra-valued 1-forms $A \in \Omega^1(M, \fru(n))$ for some $n \in \NN_0$ and its morphism sets read as
\begin{equation}
\begin{aligned}
	\Omega^1(M,\fru) (A, A')
	&= \HVBdl^\nabla_\rmpar(M) \big( (M {\times} \FC^n, \dd + A),\, (M {\times} \FC^{n'}, \dd + A') \big)
	\\
	&= \big\{ f \in \Mfd \big(M, \Mat(n {\times} n', \FC) \big)\, \big|\, A'\,f = f\, A + \dd f \big\}\,,
\end{aligned}
\end{equation}
with $\Mat(n {\times} n', \FC)$ denoting the space of $(n {\times} n')$-matrices with complex coefficients.
Composition is given by multiplication of matrix-valued functions.
The category $\Omega^1(M, \fru)$ is endowed with the structure of a rig-category under the direct sum given on objects by $(A, A') \mapsto A \oplus A' \in \Omega^1(M, \fru(n + n'))$ and the tensor product $(A, A') \mapsto A \otimes \One + \One \otimes A' \in \Omega^1(M, \fru(n \cdot n'))$.
The dual acts via
\begin{equation}
	\big( A \overset{f}{\longrightarrow} A' \big) \mapsto \big( (-A^{\prime\, \sft}) \overset{f^\sft}{\longrightarrow} (-A^\sft) \big)\,.
\end{equation}
This is, indeed, a morphism in $\Omega^1(M, \fru)$: we have
\begin{equation}
	(-A^\sft)\, f^\sft
	= - (f\, A)^\sft
	= - ( A'\,f - \dd f )^\sft
	= f^\sft\, (-A^{\prime\, \sft}) + \dd f^\sft\,.
\end{equation}
For the inner product, we obtain
\begin{equation}
	\<A, A'\>_{\Omega^1(M, \fru)} = \big\{ f \in \Mfd \big(M, \Mat(n {\times} n', \FC) \big)\, \big|\, A'\,f = f\, A + \dd f \big\}\,.
\end{equation}
On morphisms it is given by pre- and postcomposition in the first and second slot, respectively.
We are left to check that $\<A, A'\>_{\Omega^1(M, \fru)} \in \Hilb$.
Consider $f,g \in \<A, A'\>_{\Omega^1(M, \fru)}$.
We have a non-degenerate, sesquilinear pairing on $\<A, A'\>_{\Omega^1(M, \fru)}$ given by
\begin{equation}
	(f,g) \mapsto \tr( f^*\, g)\,.
\end{equation}
To see that this is constant, first observe that the adjoint $f^*$ is a morphism $-A = A^* \to A^{\prime\, *} = -A'$.
We compute
\begin{equation}
	\dd\, \tr( f^*\, g)
	= \tr \big( \dd (f^*\, g) + (f^*\, g)\, A - A\, (f^*\, g) \big)
	= 0\,,
\end{equation}
where we have used the cyclicity of the trace and that $f^*\, g$ is an endomorphism of $A$.
Thus, we can describe the 2-Hilbert space of sections of $\CI_\rho$ on a contractible manifold and, in particular, the 2-Hilbert space of sections of the prequantum bundle gerbe $\CI_\rho$, as
\begin{equation}
	\big( \Gamma(M, \CI_\rho),\, \Gamma_\rmpar(M, \sfR [-,-]) \big)
	\cong \big( \Omega^1(M, \fru),\, \<-,-\>_{\Omega^1(M, \fru)} \big)\,.
\end{equation}

Denote by $0_n \in \Omega^1(M, \fru(n))$ the zero element of the vector space $\Omega^1(M, \fru(n))$.
The rig-category $(\Omega^1(M, \fru), \otimes, 0_1, \oplus, 0_0)$ resembles the Lie 2-algebra of observables constructed in~\cite{Rogers--Thesis,Rogers--L_infty_algebras}.
It would be interesting to see whether the rig-category $\Omega^1(M, \fru)$ does, in fact, contain this Lie 2-algebra and, thus, extends the observables considered in the above references.

\section{Example: A decomposable bundle gerbe}
\label{sect:Ex:Cup_product_lens_space_BGrbs}

In this section we provide an explicit example of a bundle gerbe whose Dixmier-Douady class is decomposable in the sense that it is a realisation of a cup product cohomology class~\cite{MMS--Index_of_projective_families-Decomposable_case}.
It is a special case of a decomposable bundle gerbe as defined in~\cite[Section 3.5]{Johnson--Constructions_with_bundle_gerbes}.
This example can also be found in~\cite{Bunk-Szabo--Fluxes_brbs_2Hspaces}.
We describe the category of sections and the resulting 2-Hilbert space of sections explicitly in terms of more familiar differential geometry.

\subsection{Construction of the bundle gerbe}

The three-sphere $S^3$ carries a free and transitive $\sfU(1)$-action; the principal $\sfU(1)$-fibration induced by it is the Hopf fibration $\sfU(1) \to S^3 \to S^2$.
Via the embedding $\RZ_p \hookrightarrow \sfU(1)$, which sends elements of $\RZ_p$ to $p$-th roots of unity, this yields a free action of $\RZ_p$ on $S^3$.
The quotient lens space $\bbL_p \coloneqq S^3/\RZ_p$ inherits a canonical projection $q \colon \bbL_p \to S^2$.
It is a connected, compact, orientable, 3-dimensional manifold, and one can show that~\cite[Example 2.43]{Hatcher--AT}
\begin{equation}
	\rmH_0(\bbL_p, \RZ) \cong \RZ \cong \rmH_3(\bbL_p, \RZ)\,, \quad
	\rmH_1(\bbL_p, \RZ) \cong \RZ_p\,, \quad
	\rmH_2(\bbL_p, \RZ) \cong 0\,.
\end{equation}
Hence, from the Universal Coefficient Theorem we deduce
\begin{equation}
	\rmH^0(\bbL_p, \RZ) \cong \RZ \cong \rmH^3(\bbL_p, \RZ)\,, \quad
	\rmH^1(\bbL_p, \RZ) \cong 0\,, \quad
	\rmH^2(\bbL_p, \RZ) \cong \RZ_p\,.
\end{equation}
If $K \to S^2$ is the hermitean line bundle associated to the Hopf fibration, we set $J \coloneqq q^*K \to \bbL_p$.
One can then show~\cite{Karoubi--K-theory-an_introduction} that its first Chern class $c_1(J)$ generates $\rmH^2(\bbL_p,\RZ) \cong \RZ_p$.

So far, we have found a manifold $\bbL_p$ which has torsion in its second cohomology.
In order to be able to construct a non-trivial bundle gerbe that admits non-trivial sections (cf. Proposition~\ref{st:no-go_for_sections_of_BGrbs}), we need torsion in the third cohomology group.
Therefore, we consider the space $M_p \coloneqq \bbL_p {\times} S^1$.
From the K\"unneth Theorem we derive
\begin{equation}
	\rmH^3(M_p, \RZ) \cong \RZ \oplus \RZ_p\,.
\end{equation}
By construction, the cup product $c_1(J) \smile [1_{S^1}]$ is a generator for the torsion subgroup $\Tor(\rmH^3(M_p, \RZ)) \cong \RZ_p$.%
\footnote{We will omit pullbacks to $M_p$ where they are clear from context.}
If we write $\omega \in \Omega^2(S^2, \FR)$ for the volume form on $S^2$, we can chose a connection $\nabla^J$ on $J \to M_p$ such that $\curv(\nabla^J) = 2\pi\, \iu\, q^*\omega$.
Such a connection can in fact be pulled back from $J \to S^2$.
The fact that the hermitean line bundle $J \to \bbL_p$ is torsion implies $[\curv(\nabla^J)]_\dR = 0 \in \rmH^2_\dR(M_p, \iu\, \FR)$.

In order to construct a bundle gerbe over the product space $M_p$, we consider the covering projection $\pi_{S^1} \colon \FR \to S^1$ and set $(Y,\pi) = (\bbL_p {\times} \FR, 1_{\bbL_p} {\times} \pi_{S^1})$.
This defines a surjective submersion onto $M_p$.
We have the following diagram:
\begin{equation}
\label{eq:cup_product_BGrb_diagram}
\begin{tikzcd}
	(J, \nabla^J)^{\otimes \RZ} \ar[d, "\FC"] & 
	\\
	\bbL_p {\times} \FR {\times} \RZ \ar[r, shift left=0.1cm] \ar[r, shift left=-0.1cm] & \bbL_p {\times} \FR \ar[d, "\pi", "\RZ"']
	\\
	 & M_p
\end{tikzcd}
\end{equation}
Here we have made use of the isomorphism $\bbL_p {\times} \FR {\times} \RZ \cong (\bbL_p {\times} \FR)^{[2]}$ which is given by mapping $((x, r),(x,s)) \in (\bbL_p {\times} \FR)^{[2]}$ to $(x,r, s-r) \in \bbL_p {\times} \FR {\times} \RZ$.
Accordingly, the degeneracies act as $d_0(x, r, n) = (x, r + n)$ and $d_1(x, r, n) = (x,r)$.
More generally, we have $\bbL_p {\times} \FR {\times} \RZ^k \cong (\bbL_p {\times} \FR)^{[k+1]}$ with
\begin{equation}
	d_i(x, r, n_0, \ldots n_{k-1}) =
	\begin{cases}
		(x, r+n_0, n_1, \ldots, n_{k-1})\,, & i = 0\,,
		\\
		(x, r, n_0, \ldots, n_{i-2}, n_{i-1} + n_i, n_{i+1}, \ldots, n_{k-1})\,, & 0<i<k\,,
		\\
		(x, r, n_0, \ldots, n_{k-2})\,, & i = k\,.
	\end{cases}
\end{equation}
The hermitean line bundle with connection $(J, \nabla^J)^{\otimes \RZ}$ on $\bbL_p {\times} \FR {\times} \RZ$ is defined via
\begin{equation}
	(J, \nabla^J)^{\otimes \RZ}_{|\bbL_p {\times} \FR {\times} \{n\}} \coloneqq \pi^*(J, \nabla^J)^{\otimes n}\,,
\end{equation}
where we understand $(J, \nabla^J)^{\otimes n} = ((J, \nabla^J)^*)^{\otimes (-n)}$ for $n < 0$.

A bundle gerbe product is given via the canonical isomorphism
\begin{equation}
	\mu_{|(x, r, n, m)} \colon (J,\nabla^J)^{\otimes \RZ}_{|(x,r,n)} \otimes (J,\nabla^J)^{\otimes \RZ}_{|(x,r+n,m)} \to (J,\nabla^J)^{\otimes \RZ}_{|(x,r, n+m)}
\end{equation}
which maps
\begin{equation}
	(J,\nabla^J)^{\otimes n}_{|x} \otimes (J,\nabla^J)^{\otimes m}_{|x} \to (J,\nabla^J)^{\otimes n+m}_{|x}\,.
\end{equation}
Note that this is already compatible with any connection on $J \to \bbL_p$.
This yields a bundle gerbe with Dixmier-Douady class the cup product $c_1(J) \smile [1_{S^1}]$~\cite{Johnson--Constructions_with_bundle_gerbes,Brylinski--Loop_spaces_and_geometric_quantisation}.
One can see that this bundle gerbe is torsion of order $p$ by observing that $J^p$ is trivial and using part (1) of Proposition~\ref{st:determinants_of_morphisms_of_BGrbs}.
The next step is to find a connection on $\CG = (J^{\otimes \RZ}, \mu, \bbL_p {\times} \FR, \pi)$.
We have already found a connection $(\nabla^J)^{\otimes \RZ}$ on $J^{\otimes \RZ}$ which is compatible with the bundle gerbe multiplication.
Its curvature reads as
\begin{equation}
\label{eq:cup_product_BGrb--curving}
\begin{aligned}
	\curv \big( (\nabla^J)^{\otimes \RZ} \big)_{|(x,r,n)}
	&= n\, \curv(\nabla^J)_{|x}
	\\*
	&= 2 \pi n\, \iu\ \omega_{|x}\,.
\end{aligned}
\end{equation}
Thus, a choice of curving is given by $B \in \Omega^2(\bbL_p {\times} \FR,\, \iu\, \FR)$ with
\begin{equation}
\begin{aligned}
	B_{|(x,r)} &= r\, \curv(\nabla^J)_{|x}
	\\*
	&= 2 \pi r\, \iu\ \omega_{|x}\,.
\end{aligned}
\end{equation}
To summarise, we have found a hermitean bundle gerbe with connection
\begin{equation}
	(\CG, \nabla^\CG) = \big( (J, \nabla^J)^{\otimes \RZ},\, \mu,\, B,\, \bbL_p {\times} \FR,\, \pi \big) \quad \in \BGrb^\nabla(M_p)
\end{equation}
that has torsion Dixmier-Douady class of order $p$.
Its curvature 3-form can be computed as
\begin{equation}
	\curv(\nabla^\CG) = \curv(\nabla^J) \wedge \dd r = 2\pi\, \iu\, \omega \wedge \dd r\,.
\end{equation}
Consequently, $(M_p, \omega \wedge \dd r)$ is a prequantisable 2-plectic manfiold, and $(\CG, \nabla^\CG)$ provides a prequantum bundle gerbe.

\subsection{The 2-Hilbert space of sections}

We proceed to investigate the category of sections of $(\CG, \nabla^\CG)$.
Recall from Definition~\ref{def:Sections_of_a_BGrb} that this category is the category of morphisms $\CI_0 \to (\CG, \nabla^\CG)$.
We will only consider an equivalent category here, namely the category of such morphisms over the minimal surjective submersion $\BGrb^\nabla_\FP(M_p)(\CI_0, (\CG, \nabla^\CG))$ (cf. Theorem~\ref{st:BGrb_FP_hookrightarrow_BGrb_is_equivalence}).
An object in this category consists of of a hermitean vector bundle with connection $(E,\nabla^E) \in \HVBdl^\nabla(\bbL_p {\times} \FR)$ defined over the fibre product $M_p {\times}_{M_p} (\bbL_p {\times} \FR) \cong \bbL_p {\times} \FR$, together with a unitary, parallel isomorphism $\alpha \colon d_0^*E \to d_1^*E \otimes J^{\otimes \RZ}$ which is compatible with $\mu$ in the sense of~\eqref{eq:1-morphisms_compatibility_with_BGrb_multiplications}.
For two such 1-morphisms, 2-morphisms $(E, \alpha) \to (F, \beta)$ are given by parallel morphisms of hermitean vector bundles $\psi \colon E \to F$ which are compatible with $\alpha$ and $\beta$.
Over a point $(x,r,n) \in \bbL_p {\times} \FR {\times} \RZ$ this reads as
\begin{equation}
	\alpha_{|(x,r,n)} \colon E_{|(x,r+n)} \to E_{|(x,r)} \otimes J^{\otimes n}_{|x}\,,
\end{equation}
which can be understood as a twisted version of $\RZ$-equivariance of $(E, \nabla^E)$, with the twist introduced by the line bundle $(J, \nabla^J)$.

We can undo this twist by passing to the covering $S^3 {\times} \FR \to \bbL_p {\times} \FR$ with fibre $\RZ_p$.
First, the line bundle $K \to S^2$ associated to the Hopf fibration $S^3 \to S^2$ is the descent of the trivial line bundle $I \to S^3$ with the $\sfU(1)$-action $(y, z) \cdot \lambda = (y \cdot \lambda, \lambda^{-1}\, z)$ for $y \in S^3$, $z \in \FC$ and $\lambda \in \sfU(1)$.
Restricting this to the action of $\RZ_p \subset \sfU(1)$ defines a $\RZ_p$-equivariant line bundle $I^{(p)} \to S^3$, which, in turn, provides descent data for the hermitean line bundle $J \to \bbL_p$.
The descent data for $\nabla^J$ is given by some $\RZ_p$-invariant $\kappa \in \Omega^1(S^3 {\times} S^1, \iu\, \FR)$, where $\RZ_p$ acts trivially on $S^1$.
An object $(E, \alpha)$ as above lifts to descent data $(\hat{E}, \phi^{\hat{E}})$ of a hermitean vector bundle for the projection $S^3 {\times} \FR \to \bbL_p {\times} \FR$ together with a descent isomorphism $\hat{\alpha}$.
Here,
\begin{equation}
\begin{aligned}
	\phi^{\hat{E}}_{(y, \zeta ,r)} &\colon \hat{E}_{|(y\, \zeta,r)} \to \hat{E}_{|(y,r)}\,,
	\\
	\hat{\alpha}_{|(y,r,n)} &\colon \hat{E}_{|(y,r+n)} \to \hat{E}_{|(y,r)} \otimes I^{(p)\, \otimes n}_{|y}
\end{aligned}
\end{equation}
for $y \in S^3$, $r \in \FR$, $n \in \RZ$ and $\zeta \in \RZ_p$.
The fact that $\hat{\alpha}$ is a descent isomorphism is equivalent to saying that $\hat{\alpha}$ intertwines the $\RZ_p$-actions in its source and target.
Consider the morphism
\begin{equation}
	m^{\hat{E}}_{|(y,r)} \colon \hat{E}_{|(y,r)} \otimes I^{(p)\, \otimes n}_{|y} \to \hat{E}_{|(y,r)}\,, \quad
	\big( e, (y,z) \big) \mapsto z \cdot e\,.
\end{equation}
This is an isomorphism of hermitean vector bundles over $S^3 {\times} \FR$, but it fails to be equivariant:
over $ (y, \zeta, r, n) \in S^3 {\times} \RZ_p \times \FR {\times} \RZ$ we have (omitting pullbacks)
\begin{equation}
\begin{aligned}
	m^{\hat{E}}_{|(y,r)} \big( \phi^{\hat{E}}_{(y, \zeta ,r)}(e) \otimes (y\, \zeta,\, \zeta^{-n}\, z) \big)
	&= \zeta^{-n}\, z\, \phi^{\hat{E}}_{(y, \zeta ,r)}(e)
	\\*
	&= \zeta^{-n}\, \phi^{\hat{E}}_{|(y, \zeta, r)} \big( m^{\hat{E}}_{|(y\, \zeta,r)} (e \otimes (y,z)) \big)\,.
\end{aligned}
\end{equation}
The geometric twist given by the line bundle $I^{(p)}$ can now be absorbed into the descent isomorphism by setting $\gamma^{\hat{E}} \coloneqq m^{\hat{E}} \circ \hat{\alpha} \colon d_0^*\hat{E} \to d_1^*\hat{E}$, or, more explicitly,
\begin{equation}
	\gamma^{\hat{E}}_{|(y, r, n)} \coloneqq m_{|(y,r)}^{\hat{E}} \circ \hat{\alpha}_{|(y,r,n)} \colon \hat{E}_{|(y, r+n)} \to \hat{E}_{|(y,r)}\,.
\end{equation}
Here we have to be careful because $m^{\hat{E}}$ is not parallel.
Instead, observing that the connection on $\hat{E} \otimes I^{(p)\, \otimes n}$ is given by $\nabla^E \otimes \One + \One \otimes (\dd + n \kappa)$, for which we will write $\nabla^{\hat{E}} + n\, \kappa$, we see that
\begin{equation}
	\big( \nabla^{\hat{E}} \circ m^{\hat{E}} - m^{\hat{E}} \circ \big( \nabla^{\hat{E}} + n\, \kappa \big) \big)_{|(y,r)}
	= - n\, \kappa_{|y} \otimes m^{\hat{E}}_{|(y,r)}\,.
\end{equation}
However, if we endow $\hat{E}$ with the connection $\nabla^{\hat{E}} + a$ with $a_{|(y,r)} = r\, \kappa_{|y}$, the composition $\gamma^{\hat{E}}$ becomes parallel:
\begin{equation}
\begin{aligned}
	\big( (\nabla^{\hat{E}} + a) \circ  \gamma^{\hat{E}} - \gamma^{\hat{E}} \circ (\nabla^{\hat{E}} + a) \big)_{|(y,r,n)}
	&= (a_{|(y,r)} - a_{|(y, r+n)} + n\, \kappa_{|y})\, \gamma^{\hat{E}}_{|(y,r,n)}
	\\*
	&= (n\, \kappa_{|y} - n\, \kappa_{|y})\, \gamma^{\hat{E}}_{|(y,r,n)}
	\\*
	&= 0\,.
\end{aligned}
\end{equation}
Note that the isomorphism $\phi^{\hat{E}}$ is still parallel with respect to the modified connection $\nabla^{\hat{E}} + a$ since $\phi^{\hat{E}}$ leaves the argument $r \in \FR$ unchanged.
Thus, the bundle $(\hat{E}, \nabla^{\hat{E}} + a) \in \HVBdl^\nabla(S^3 {\times} \FR)$ is endowed with two equivariant structures.
The isomorphisms $\phi^{\hat{E}}_{(y,\zeta, r)} \colon \hat{E}_{|(y\, \zeta, r)} \to \hat{E}_{|(y, r)}$ make $(\hat{E}, \nabla^{\hat{E}} + a)$ equivariant with respect to the $\RZ_p$-action on $S^3$, while $\gamma^{\hat{E}}_{(y, r, n)} \colon \hat{E}_{|(y, r+n)} \to \hat{E}_{|(y, r)}$ make $(\hat{E}, \nabla^{\hat{E}} + a)$ equivariant with respect to the $\RZ$-action on $\FR$.
Consequently, $(\hat{E}, \nabla^{\hat{E}} + a)$ could descend along both projections $S^3 \to \bbL_p$ and $\FR \to S^1$.
This is true for either of the projections separately, but not simultaneously;
the two equivariant structures do not commute and, hence, do not form a $\RZ_p {\times} \RZ$-equivariant structure:
instead, we have
\begin{equation}
\label{eq:plus-twisted_equivariance}
	\phi^{\hat{E}}_{(y,\zeta, r)} \circ \gamma^{\hat{E}}_{(y\, \zeta, r, n)}
	= \zeta^n \cdot \gamma^{\hat{E}}_{(y, r, n)} \circ \phi^{\hat{E}}_{(y,\zeta, r+n)}\,.
\end{equation}
Let us call a hermitean vector bundle with connection on a space with both a $\RZ_p {\times} \RZ$-action which has an equivariant structure for each of the separate group actions satisfying a twisted commutation relation as in~\eqref{eq:plus-twisted_equivariance} a \emph{$(\RZ_p{-}\RZ)_+$-twisted equivariant hermitean vector bundle with connection}.
We write $\HVBdl^\nabla_{\rmpar, (\RZ_p{-}\RZ)_+}(S^3 {\times} \FR)$ for their category.
Here we take as morphism of such bundles parallel morphisms of hermitean vector bundles with connection which commute with both equivariant structures.
Note that altering the connection by adding $a$ has no effect on the parallel morphisms $(E, \nabla^E) \to (F, \nabla^F)$.
We have, thus, derived an equivalence of categories
\begin{equation}
	\Gamma \big( M_p,\, (\CG, \nabla^\CG) \big)
	\cong \HVBdl^\nabla_{\rmpar, (\RZ_p{-}\RZ)_+}(S^3 {\times} \FR)\,.
\end{equation}
The category on the right-hand side naturally comes endowed with a direct sum, which the equivalence (which is merely descent for the relative sheaf of categories $\HVBdl^\nabla$, cf. Section~\ref{sect:Coverings_and_sheaves_of_cats}) maps to the direct sum in $\BGrb^\nabla_\FP(M)(\CI_0, (\CG, \nabla^\CG))$.
Further, note that there exists a natural action of $\HVBdl^\nabla(M_p) \cong \BGrb^\nabla(M_p)(\CI_0, \CI_0)$ on $\HVBdl^\nabla_{\rmpar, \RZ_p{-}\RZ}(S^3 {\times} \FR)$.
It is given by pulling back a hermitean vector bundle with connection along the fibration $S^3 {\times} \FR \to M_p$ and then taking the tensor product.
The functor $\Theta$ acts on $(\hat{E}, \nabla^{\hat{E}} + a, \gamma^{\hat{E}}, \phi^{\hat{E}})$ via taking the dual bundle and the inverse tranpose of the action morphisms.
Note that this affects the relation~\eqref{eq:plus-twisted_equivariance} by adding an inverse to $\zeta$, i.e. we have
\begin{equation}
\label{eq:minus-twisted_equivariance}
	\phi^{\hat{E}\, -\sft}_{(y,\zeta, r)} \circ \gamma^{\hat{E}\, -\sft}_{(y\, \zeta, r, n)}
	= \zeta^{-n} \cdot \gamma^{\hat{E}\, -\sft}_{(y, r, n)} \circ \phi^{\hat{E}\, -\sft}_{(y,\zeta, r+n)}\,.
\end{equation}
We call a hermitean vector bundle with a $\RZ_p$-equivariant structure and a $\RZ$-equivariant structure satisfying~\eqref{eq:minus-twisted_equivariance} a \emph{$(\RZ_p{-}\RZ)_-$-twisted equivariant hermitean vector bundle with connection} and write $\HVBdl^\nabla_{\rmpar, (\RZ_p{-}\RZ)_-}(S^3 {\times} \FR)$ for their category.

The tensor product of vector bundles induces a pairing
\begin{equation}
\begin{aligned}
	&\HVBdl^\nabla_{\rmpar, (\RZ_p{-}\RZ)_-}(S^3 {\times} \FR) \times \HVBdl^\nabla_{\rmpar, (\RZ_p{-}\RZ)_+}(S^3 {\times} \FR)
	\\*
	&\quad \to \HVBdl^\nabla_{\rmpar, \RZ_p{\times}\RZ}(S^3 {\times} \FR) \cong \HVBdl^\nabla_{\rmpar}(M_p)\,,
\end{aligned}
\end{equation}
where $\HVBdl^\nabla_{\rmpar, \RZ_p{\times}\RZ}(S^3 {\times} \FR)$ denotes the category of $(\RZ_p {\times} \RZ)$-equivariant vector bundles on $S^3 {\times} \FR$, and the equivalence on the target side is descent.
We thus obtain the inner product of sections as
\begin{equation}
\begin{aligned}
	&\<-,-\> = \Desc \circ \big( \Theta(-) \otimes (-) \big) \colon
	\\*
	&\quad \overline{ \HVBdl^\nabla_{\rmpar, (\RZ_p{-}\RZ)_+}(S^3 {\times} \FR) } \times \HVBdl^\nabla_{\rmpar, (\RZ_p{-}\RZ)_+}(S^3 {\times} \FR) \to \HVBdl^\nabla_{\rmpar}(M_p)\,.
\end{aligned}
\end{equation}

\section{Remarks on the torsion constraint}
\label{sect:ways_around_the_torsion_constraint}

In this section we revisit the torsion constraint (cf. Proposition~\ref{st:no-go_for_sections_of_BGrbs}) on the existence of non-zero sections of a bundle gerbe.
As pointed out in Section~\ref{sect:higher_geometric_structures}, this puts a strong restriction on the theory of bundle gerbes and their morphisms as opposed to line bundles, where every pair of line bundles admits a non-zero morphism between them by a partition of unity argument.
The absence of a partition of unity for higher line bundles, i.e. bundle gerbes, prohibits carrying over many constructive existence proofs from ordinary differential geometry.

However, at the heart of the no-go statement given by Proposition~\ref{st:no-go_for_sections_of_BGrbs} lies the determinant trick used in the proof of Proposition~\ref{st:determinants_of_morphisms_of_BGrbs}.
Consequently, in order to circumvent the no-go statement, this argument has to be bypassed.
To that end, we would have to weaken the notion of a 1-morphism of bundle gerbes in a way that would make it impossible to form the determinant bundle of the underlying bundle or similar object.
There are two straightforward candidate solutions:
allow for bundles of infinite rank in the definition of 1-morphisms, or weaken the notion of bundle to a family of Hilbert spaces whose rank may vary from point to point.

\subsection{Hilbert-bundle morphisms}

Let us address the first attempt to circumvent Proposition~\ref{st:no-go_for_sections_of_BGrbs}, i.e. let us, for a moment, alter Definition~\ref{def:1-morphisms_of_BGrbs} by allowing the bundle $E \to Z$ in a 1-morphism of bundle gerbes to be a bundle of Hilbert spaces of possibly infinite rank.
Such bundles still satisfy descent and form a relative sheaf of categories on $\Mfd$ in the sense of Definition~\ref{def:relative_sheaf_of_Cats}, where the subsheaf of groupoids used for descent data contains all objects but has only fibrewise unitary isomorphisms as morphisms.
Hence, we can still construct a 2-category of bundle gerbes based on these bundles in the same way as we constructed $\BGrb^\nabla(M)$.
Let us work in the setting without connections for now.
In particular, we would still have an equivalence of categories $\BGrb_\FP(M)(\CG_0, \CG_1) \cong \BGrb(M)(\CG_0, \CG_1)$.
However, while this enlargement of the category of bundle gerbes introduces morphisms between every two bundle gerbes, the theory of the resulting 2-category is not particularly rich:

\begin{proposition}
\label{st:uniqueness_of_Hilbert-bundle_1-morphisms}
For any pair $\CG_0, \CG_1 \in \BGrb(M)$ there exists an infinite-rank morphism $\CG_0 \to \CG_1$, and it is unique up to 2-isomorphism.%
\footnote{This statement can be found in~\cite[Proposition 7.1]{BCMMS}.
However, it has not been proven in that article; here we provide a full proof.}
\end{proposition}

\begin{proof}
Any bundle gerbe on $M$ is (stably) isomorphic to a lifting bundle gerbe of a $\PP\sfU(\CH)$ principal bundle on $M$ (see, for instance, \cite{BCMMS}), where $\PP\sfU(\CH)$ denotes the projective unitary group of an infinite-dimensional, separable Hilbert space $\CH$~\cite{BCMMS}.
Therefore, it suffices to investigate this type of bundle gerbe.
Let $\pi \colon P \to M$ be a $\PP\sfU(\CH)$-bundle and denote by $\CG_P$ the associated lifting bundle gerbe.
There is a morphism $\CI \to \CG_P$ given by $(E_P = P {\times} \CH, \alpha, P, \pi)$, where
\begin{equation}
	\alpha_{|(R_{[u]}\, p,\, p)} \colon E_{P|R_{[u]} p} \to E_{P|p} \otimes L_{|[u]}\,, \quad
	\psi \mapsto u\, \psi \otimes [u, 1]
\end{equation}
for $[u] \in \PP\sfU(\CH)$, $\psi \in \CH$, and $L \to \PP\sfU(\CH)$ the line bundle associated to the $\sfU(1)$-bundle $\sfU(\CH) \to \PP\sfU(\CH)$.
Similarly, one can construct a morphism $\CG_P \to \CI$.
Therefore, every bundle gerbe admits morphisms $\CI \to \CG \to \CI$ whose underlying vector bundles are infinite-rank Hilbert bundles.
We can take now compositions of the morphisms $\CI \to \CG_1$ and $\CG_0 \to \CI$ in order to obtain an infinite-rank morphism $\CG_0 \to \CG_1$.
This proves the existence part of the statement.

Now assume that there exist two infinite-rank morphisms $(E,\alpha)$ and $(F, \beta)$ from $\CG_0$ to $\CG_1$, both defined over the minimal surjective submersion $Z = Y_{01} = Y_0 {\times}_M Y_1$, where $Y_i \to M$ is the surjective submersions of $\CG_i$.
We can assume this without restriction because every 1-morphism is 2-isomorphic to a 1-morphism defined over this submersion by the equivalence $\BGrb_\FP(M)(\CG_0, \CG_1) \cong \BGrb(M)(\CG_0, \CG_1)$.
We employ a construction very similar to the definition of the bifunctor $[-,-]$ in Section~\ref{sect:Pairings_and_inner_hom_of_morphisms_in_BGrb}.
Consider the bundle $\sfU(E,F) \to Y_{01}$, whose fibre over $(y_0,y_1) \in Y_{01}$ is the space of unitary isomorphisms from $E_{|(y_0,y_1)}$ to $F_{|(y_0,y_1)}$.
This bundle gives rise to a descent datum via
\begin{equation}
	\delta_{L_1} \circ \beta \circ (-) \circ \alpha^{-1} \circ \delta_{L_0}^{-1} \colon
	\sfU(E,F)_{|(y'_0, y'_1)} \to \sfU(E,F)_{|(y_0, y_1)}\,,
\end{equation}
where $L_i$ denotes the hermitean line bundle of the bundle gerbe $\CG_i$ and $(y_i, y'_i) \in Y_i^{[2]}$ for $i = 0,1$.
Analogously to Corollary~\ref{st:sections_of_reduced_pairing_agree_with_2-homs}, there is a bijection between unitary 2-isomorphisms $(E,\alpha) \to (F, \beta)$ and sections of the descent bundle $\sfR(\sfU(E,F))$ of the descent data found above.
Thus, if we can show that there exists a section of the descent bundle $\sfR(\sfU(E,F))$, that implies that $(E,\alpha) \cong (F, \beta)$.

The bundle $\sfR(\sfU(E,F))$ is a fibre bundle with fibre $\sfU(\CH)$, but not a principal $\sfU(\CH)$-bundle.
Nevertheless, by Kuiper's Theorem, $\sfU(\CH)$ is contractible.
The existence of a global continuous section can be seen by a very compact argument as follows:
a fibre bundle is a Serre fibration, i.e. a fibration in the model category of topological spaces (see e.g.~\cite{Hovey--Model_categories}).
The long exact sequence in homotopy then shows that $\sfR(\sfU(E,F)) \to M$ is a trivial fibration.
Any manifold can be endowed with the structure of a CW-complex (it admits a triangulation) and is thus (homeomorphic to) a cofibrant object in the category of topological spaces.
As the identity $1_M$ is a cofibration, this map admits lifts to the total space of $\sfR(\sfU(E,F))$, which is the same as a section of $\sfR(\sfU(E,F))$.
More generally, one can use arguments from e.g.~\cite{Dold--Partitions_of_unity_and_fibrations,Steenrod--Fibre_bundles} to show that smooth fibre bundles with contractible fibres admit smooth global sections.
\end{proof}

Consequently, allowing for infinite-rank morphisms yields exactly one additional 2-isomorphism class of 1-morphisms between every two bundle gerbes.
The reason is, as expected, Kuiper's Theorem, which implies this uniqueness very directly for the case of $\CG_0 = \CG_1 = \CI$.
In this case the statement is precisely the uniqueness up to isomorphism of infinite-rank Hilbert bundles on any manifold $M$.

Any interesting enlargement of $\BGrb^\nabla(M)$ would, consequently, result from the connections on infinite-rank 1-morphisms only.
However, infinite-rank sections, even with connections, are impracticable from the point of view of 2-Hilbert spaces:
as pointed out in Example~\ref{eg:Hilb_sep_is_not_a_2Hspace}, the category $\Hilb_\sep$ of possibly infinite-dimensional Hilbert spaces does not fit into our framework of 2-Hilbert spaces from Definition~\ref{def:2Hspace}, the contradiction stemming from the fact that taking morphism spaces is not $\Hilb_\sep$-valued and, hence, does not yield a 2-Hilbert space inner product.
We could weaken our definition of 2-Hilbert spaces to allow $\scH(-,-)$ and $\<-,-\>_\scH$ to possibly be different functors.
The natural choice for an inner product on $\Hilb_\sep$ would be the space of Hilbert-Schmidt operators between two Hilbert spaces.
In order to construct the pairing on $\scH_0(\CG, \nabla^\CG)$ in that setting, we would have to use bundles of Hilbert-Schmidt operators as well in order for the pairing $[-,-]$ to produce bundles of Hilbert spaces rather than Banach bundles.
However, if there is any non-trivial parallel Hilbert-Schmidt 2-morphism between two infinite-rank 1-morphisms $\CI_0 \to (\CG, \nabla^\CG)$, this defines a finite-rank 1-morphism for each of its non-zero eigenvalues.
This is contradicted, once again, by Proposition~\ref{st:no-go_for_sections_of_BGrbs} if $\CG$ has non-torsion Dixmier-Douady class.
Thus, the inner product in $\scH_0(\CG, \nabla^\CG)$ would be degenerate whenever $\DD(\CG)$ is non-torsion.%
\footnote{Note that this argument still applies to the situation of infinite-rank 1-morphisms with reduced structure groups like those used to define twisted K-theory~\cite{BCMMS}.}
To summarise, allowing for infinite-rank 1-morphisms does not resolve the obstructions in constructing a sensible 2-Hilbert space of sections for a non-torsion bundle gerbe.

\subsection{Continuous families of Hilbert spaces}

The second way around the determinant argument would be to allow for varying rank of 1-morphisms.
More succinctly, one could weaken the notion of bundles to continuous, or even measurable, families of Hilbert spaces as made precise, for instance, in~\cite[Chapter 7]{Folland--Harmonic_analysis}.
This would presumably also resemble the use of square integrable sections of the prequantum line bundle in the prequantum Hilbert space, rather than smooth sections, more closely.
In geometric constructions, bundles are far more familiar, but, to the end of constructing a 2-Hilbert space of sections, using such families of possibly infinite-dimensional Hilbert spaces might be the natural next step to take from a conceptual point of view.
The global section functor in the construction of the $\Hilb$-valued inner product bifunctor in Section~\ref{sect:2-Hspace_of_a_BGrb} should then be given by the direct integral, as already indicated in that section.
To work out the necessary details for for this framework would, however, go beyond the scope of this thesis.

\chapter{Transgression and dimensional reduction}
\label{ch:Transgression_and_reduction}

\section{Diffeological spaces and bundles}
\label{sect:DfgSp_and_diffeological_bundles}

For the reader's convenience, we briefly recall the definition of a diffeolgical space.
This notion goes back to Souriau~\cite{Souriau:Groupes_Differentiels}; a standard reference is~\cite{Iglesias-Zemmour--Diffeology}, while a concise introduction is given in~\cite{Baez-Hoffnung:Convenient_Spaces}.
Recall that we denote by $\Set$ the category of sets and set $\NN_0 = \NN \cup \{0\}$.
The idea behind diffeological spaces is that a geometry can be probed by the collection of smooth maps into it rather than considering local charts as the fundamental devices that detect geometry.

\begin{definition}[Diffeological space]
\label{def:Dfg_Space}
A \emph{diffeological space} is a set $X$ together with a set $\Plot(X)$ of \emph{plots of $X$}, i.e. maps $\varphi \in \Set(U,X)$, where $U \subset \FR^n$ is an open subset for some $n \in \NN_0$, with the following properties:
\begin{myenumerate}
	\item For a plot $\varphi \in \Set(U,X)$ of $X$ and a smooth map $f \in \Mfd(V,U) \subset \Set(V,U)$, $\varphi \circ f$ is a plot of $X$ as well.%
	\footnote{Note that $U \subset \FR^n$, $V \subset \FR^m$ for $n \neq m$ is allowed here.}
	In other words, $\Plot(X)$ is closed under precomposition by smooth maps.
	
	\item For any open covering $(\iota_a \colon U_a \hookrightarrow U)_{a \in \Lambda}$ of $U$ indexed by $\Lambda \in \Set$ and a map $\varphi \in \Set(U,X)$ such that $\varphi \circ \iota_a \in \Plot(X)$ for all $a \in \Lambda$, $\varphi$ is a plot of $X$ as well.
	
	\item The set $\Plot(X)$ contains all constant maps: $\Set(\FR^0,X) \subset \Plot(X)$.
\end{myenumerate}
The choice of a set $\Plot(X)$ of plots of $X$ is also called a \emph{diffeology on $X$}.
\\
A \emph{morphism of diffeological spaces}, or \emph{diffeological map}, $(X,\Plot(X)) \to (Y,\Plot(Y))$ is a map $f \in \Set(X,Y)$ such that $f \circ \varphi \in \Plot(Y)$ for every plot $\varphi \in \Plot(X)$.
This yields a category of diffeological spaces which we will denote $\DfgSp$.
\end{definition}

Diffeology is a significantly weakened notion of smoothness, but it still allows one to carry over many concepts known from differential geometry such as differential forms, bundles, or parallel transport.

\begin{example}
\label{eg:diffeological_spaces_and_constructions}
We list several basic diffeological structures:
\begin{myenumerate}
	\item On every set $X$ there exists a discrete diffeology whose plots are all constant maps $U \to X$ for any $U \subset \FR^n$, $n \in \NN_0$.
	
	\item Every manifold $M \in \Mfd$ is a diffeological space by setting $\Plot(M)$ to be the set of all smooth maps from open subsets $U \subset \FR^n$ to $M$ for any $n \in \NN_0$.
	
	\item Given two manifolds $M, N \in \Mfd$, the mapping space $\Mfd(N,M)$ can be endowed with a diffeological structure:
	a map $\varphi \colon U \to \Mfd(N,M)$ is a plot of $\Mfd(N,M)$ if the composition
	\begin{equation}
	\begin{tikzcd}
		U {\times} N \ar[r, "\varphi {\times} 1_N"] & \Mfd(N,M) {\times} N \ar[r, "\ev"] & M
	\end{tikzcd}
	\end{equation}
	is smooth, where $\ev \colon \Mfd(N,M) {\times} N \to M$ denotes the evaluation map.
	The most important example for us will be the loop space of a manifold $M \in \Mfd$, of which we define two different versions:
	\begin{equation}
	\begin{aligned}
		LM &\coloneqq \Mfd(S^1,M)\,,
		\\
		L_\rmsi M &\coloneqq \big\{ \gamma \in \Mfd(S^1, M)\, \big|\, \gamma\ \text{has a sitting instant at}\ 1 \in S^1 \big\}\,.
	\end{aligned}
	\end{equation}
	Here the subscript $\rmsi$ is to indicate that $L_\rmsi M$ only contains loops with sitting instant at $1 \in S^1$.
	In general we say that a smooth map $f \in \Mfd(N,M)$ \emph{has a sitting instant at $y \in N$} if there exists an open neighbourhood $V \subset N$ around $y$ such that $f_{|V}$ is constant.%
	\footnote{For our purposes here it is generally sufficient to work with generic smooth loops in $M$.
	However, in order to relate to the results in~\cite{Waldorf--Transgression_I,Waldorf--Transgression_II}, especially to enable us to apply regression techniques if desired, we also consider loops with sitting instants.}
	Note that while we might have described $LM$ as a Frech\'et manifold, this is not possible for $L_\rmsi M$.
	
	\item Given a diffeological space $X \in \DfgSp$ and a set $Y \in \Set$, any map $p \in \Set(X,Y)$ induces a diffeology on $Y$:
	a map $\varphi \in \Set(U,Y)$ is a plot if for every point $y \in U$ there exists an open neighbourhood $V_y \subset U$ such that either $\varphi_{|V_y}$ is constant, or there exists a plot $\psi \in \Set(V_y, X)$ of $X$ such that $p \circ \psi = \varphi_{|V_y}$.
	This diffeology on $Y$ is called the \emph{pushforward diffeology}~\cite{Iglesias-Zemmour--Diffeology}.
	Here our main examples are
	\begin{equation}
	\begin{alignedat}{2}
			\CL_+ M &\coloneqq LM/\Diff_0(S^1)\,, \qquad \CL M &&\coloneqq LM/\Diff(S^1)\,,
			\\
			\CL_{\rmsi +} M &\coloneqq L_\rmsi M/\Diff_0(S^1)\,, \quad \CL_\rmsi M &&\coloneqq L_\rmsi M/\Diff(S^1)\,,
	\end{alignedat}
	\end{equation}
	where $\Diff_0(S^1)$ is the space of orientation-preserving diffeomorphisms from $S^1$ to itself, and $\Diff(S^1)$ is the space of all diffeomorphisms $S^1 \to S^1$.
	Both act on loop spaces by precomposition.
	The diffeology on each of these quotients is the pushforward diffeology of the mapping space diffeology along the quotient map.

	\item Given a family of diffeological spaces $(X_i, \Plot(X_i))_{i \in \Lambda}$, the product $\prod_{i \in \Lambda} X_i$ is naturally endowed with a diffeology whose plots are those maps $\varphi \colon U \to \prod_{i \in \Lambda} X_i$ such that $\pr_{X_i} \circ \varphi \in \Plot(X_i)$ for all $i \in \Lambda$.
	This diffeology is called the \emph{product diffeology on $\prod_{i \in \Lambda} X_i$}.
	In particular, if the indexing set $\Lambda$ is finite, we have~\cite[Article 1.55]{Iglesias-Zemmour--Diffeology}
	\begin{equation}
		\Plot \Big( \prod_{i \in \Lambda} X_i \Big) = \prod_{i \in \Lambda} \Plot(X_i)\,,
	\end{equation}
	which makes the product diffeology easy to describe.
	\qen
\end{myenumerate}
\end{example}

\begin{definition}[Diffeological fibre bundle~{\cite[Article 8.9]{Iglesias-Zemmour--Diffeology}}]
A morphism $\pi \in \DfgSp(E,B)$ of diffeological spaces is called a \emph{(diffeological) fibre bundle with typical fibre $F$} if there exists a diffeological space $F \in \DfgSp$ such that for every $\varphi \in \Plot(B)$ the pullback $\varphi^* \pi \in \DfgSp(\varphi^*E,U)$ is locally trivial with fibre $F$.
That is, around every $x \in U$ there exists an open neighbourhood $V \subset U$ together with an isomorphism of diffeological spaces $(\varphi^*E)_{|V} \arisom V {\times} F$.
\\
In that case, $E$ and $B$ are referred to as the \emph{total space} and the \emph{base space} of the fibre bundle, respectively.
\\
A \emph{diffeological vector bundle} is a diffeological fibre bundle with a vector space structure on each of its fibres such that all local trivialisations restrict to isomorphisms of vector spaces $(\varphi^*E)_{|x} \arisom \{x\} {\times} F$, where the typical fibre $F$ is a vector space.
\end{definition}

\section{The transgression functor}
\label{sect:transgression_functor}

We begin by recalling the construction of the transgression line bundle of a bundle gerbe from~\cite{Waldorf--Transgression_II}.
However, our presentation differs slightly from that reference in that we focus on the hermitean line bundles which are associated to the principal $\sfU(1)$-bundles constructed there.
This is necessary, as it is those hermitean line bundles which will be related to the morphisms of bundle gerbes and the additional categorical structures constructed thereon in Chapter~\ref{ch:bundle_gerbes}.
We then proceed to construct a transgression functor defined on the 2-category $\BGrb^\nabla(M)$, thus extending previously known constructions of transgression.
These were defined on $\BGrb^\nabla_{\rmflat \sim}(M)$ in the case of~\cite{Waldorf--Transgression_II} and on certain subsets of sections of bundle gerbes in the framework of~\cite{CJM--Holonomy_on_D-branes}.
Combining those two approaches allows us to treat transgression on the entire 2-category $\BGrb^\nabla(M)$.

Let $(\CG, \nabla^\CG)$ be a bundle gerbe on $N$ for $N \in \Mfd$.
In order to ease up notation we define the \emph{groupoid of unitary frames of $(\CG,\nabla^\CG)$}:
\begin{equation}
	\frUF(N)(\CG,\nabla^\CG) \coloneqq \bigsqcup_{\rho \in \Omega^2(N,\iu\, \FR)}\, \BGrb^\nabla_{\rmflat \sim}(N) \big( \CI_\rho, (\CG, \nabla^\CG) \big)\,.
\end{equation}
This category has as objects the 1-isomorphisms $\CI_\rho \to (\CG, \nabla^\CG)$ for all $\rho \in \Omega^2(N, \iu\, \FR)$ which satisfy the trace condition~\eqref{eq:trace_condition}.
That is, its objects are 1-isomorphisms $(S, \nabla^S, \beta, X, \xi) \in \BGrb^\nabla(M)(\CI_\rho, (\CG, \nabla^\CG))$ for some $\rho \in \Omega^2(N, \iu\, \FR)$ such that $\curv(\nabla^S) = \xi_Y^*B - \xi_M^*\rho$.
Morphisms in $\frUF(N)(\CG, \nabla^\CG)$ are the unitary parallel 2-isomorphisms between these 1-isomorphisms.
Note that for a fixed $\rho$, the category $\BGrb^\nabla_{\rmflat \sim}(N) \big( \CI_\rho, (\CG, \nabla^\CG) \big)$ may be empty.

Any pair of unitary frames of $(\CG, \nabla^\CG)$, say $(S, \beta)$ and $(S', \beta')$, gives rise to a 1-isomorphism $(S', \beta')^{-1} \circ (S, \beta) \in \BGrb^\nabla_{\rmflat \sim}(N)(\CI_\rho, \CI_{\rho'})$, and, after reduction, to a hermitean line bundle with connection
\begin{equation}
\sfR((S', \beta') \circ (S, \beta)^{-1}) \in \HLBdl^\nabla_{\rho' - \rho,\rmuni}(N)\,.
\end{equation}
For $\omega \in \Omega^2_\cl(N, \iu\, \FR)$, the category $\HLBdl^\nabla_{\omega,\rmuni}(N)$ is the groupoid of hermitean line bundles with connection whose curvature equals $\omega$, together with parallel, unitary isomorphisms.
Conversely, $\HLBdl^\nabla_{\rmuni}(N)$ acts on the groupoid $\frUF(N)(\CG, \nabla^\CG)$ via the tensor product in the 2-category $\BGrb^\nabla(N)$, making $\frUF(N)(\CG, \nabla^\CG)$ into a module category over $\HLBdl^\nabla_\rmuni(N)$.
A module category $\scD$ over a monoidal category $(\scC, \otimes, \One_\scC)$ is called a \emph{torsor category} over $(\scC, \otimes, \One_\scC)$ if the composition
\begin{equation}
\begin{tikzcd}[column sep=1.5cm]
	\scD {\times} \scC \ar[r, "\diag {\times} 1_\scC"] & \scD {\times} \scD {\times} \scC \ar[r, "1_\scD {\times} \otimes"] & \scD {\times} \scD,
\end{tikzcd}
\end{equation}
with $\diag \colon \scD \to \scD {\times} \scD$ denoting the diagonal functor, is an equivalence of categories.

\begin{proposition}[{\cite[Theorem 2.5.4]{Waldorf--Thesis}}]
For any pair $(\CG_i, \nabla^{\CG_i}) \in \BGrb^\nabla(N)$ for $i=0,1$, the groupoid $\BGrb^\nabla_{\rmflat \sim}(N)((\CG_0, \nabla^{\CG_0}), (\CG_1, \nabla^{\CG_1}))$ is a torsor category over the symmetric monoidal groupoid $\HLBdl^\nabla_{0,\rmuni}(N)$.
\end{proposition}

\begin{corollary}
The set $\pi_0 \big(\BGrb^\nabla_{\rmflat \sim}(N)((\CG_0, \nabla^{\CG_1}), (\CG_1, \nabla^{\CG_1})) \big)$ of flat 2-isomorphism classes of flat 1-isomorphisms between any pair of bundle gerbes with connection on $N$ is a torsor over the abelian group of isomorphism classes of flat hermitean line bundles $\pi_0(\HLBdl^\nabla_{0,\rmuni}(N))$ on $N$.
\end{corollary}

In general we have an isomorphism of abelian groups
\begin{equation}
	\pi_0 \big( \HLBdl^\nabla_{0, \rmuni}(N),\, \otimes \big) \cong
	\Ab \big( \pi_1(N), \sfU(1) \big)\,.
\end{equation}
Consider for a moment the situation $N = S^1$.
In this case there is an isomorphism of abelian groups
\begin{equation}
	\hol \colon \pi_0 \big( \HLBdl^\nabla_{0,\rmuni}(S^1) \big) \arisom \sfU(1)
\end{equation}
given by mapping a flat line bundle to its holonomy around $S^1$.
This endows the abelian group $\pi_0(\HLBdl^\nabla_{0,\rmuni}(S^1))$ with the structure of a Lie group.
Further, there is an equality
\begin{equation}
	\frUF(S^1)(\CG, \nabla^\CG) = \BGrb^\nabla(S^1) \big( \CI_0, (\CG, \nabla^\CG) \big)
\end{equation}
because $\Omega^2(S^1, \iu\, \FR) = \{0\}$ for dimensional reasons.
Thus, $\frUF(S^1)(\CG, \nabla^\CG)$ is a torsor category over $\HLBdl^\nabla_\rmuni(S^1)$.

For a loop $\gamma \in LM = \Mfd(S^1,M)$, define
\begin{equation}
	\CT(\CG,\nabla^\CG)_{|\gamma} \coloneqq \big\{ \big[ [S,\beta], z \big]\, \big|\, [S,\beta] \in \pi_0 \big( \frUF(S^1) (\gamma^*(\CG, \nabla^\CG)) \big),\, z \in \FC \big\}\,.
\end{equation}
Here we say that two pairs $([S, \beta], z)$ and $([S', \beta'], z')$ for $[S,\beta], [S',\beta'] \in \frUF(S^1)(\gamma^*(\CG, \nabla^\CG))$ and $z,z' \in \FC$ are equivalent if
\begin{equation}
	z'
	= \hol \big( \sfR((S', \beta')^{-1} \circ (S, \beta)) \big)\, z\,.
\end{equation}
In that case, $\CT(\CG, \nabla^\CG)_{|\gamma}$ resembles the fibre of the associated line bundle of a $\sfU(1)$-principal bundle $\CT^{\sfU(1)}(\CG, \nabla^\CG)$ on $LM$ with fibre over $\gamma$ given by $\pi_0(\frUF(S^1)(\gamma^*(\CG, \nabla^\CG)) \cong \sfU(1)$:
\addtocounter{equation}{1}
\begin{align*}
	\big[ [S, \beta], z \big]
	&= \big[ \big[ (S', \beta') \circ (S', \beta')^{-1} \circ (S, \beta) \big], z \big]
	\\*\theeq
	&= \big[ R_{[(S', \beta')^{-1} \circ (S, \beta)]}\, [S', \beta'], z \big]
	\\
	&= \big[ [S', \beta'],\, \hol \big( \sfR((S', \beta')^{-1} \circ (S, \beta)) \big)\, z \big]
	\\*
	&= \big[ [S', \beta'], z' \big]\,.
\end{align*}
The definitions work verbatim over $L_\rmsi M$ instead of $LM$.
As we endowed $LM$ with the structure of a diffeological space, the line bundle over $LM$ alluded to above cannot be a smooth line bundle, but, if anything, has to be a diffeological line bundle.
The total space of this line bundle is
\begin{equation}
	\CT(\CG, \nabla^\CG) \coloneqq \bigsqcup_{\gamma \in LM}\, \CT(\CG, \nabla^\CG)_{|\gamma} \quad \in \Set\,.
\end{equation}
It comes with a natural projection map $p \colon \CT(\CG, \nabla^\CG) \to LM$ of sets.
Plots of $\CT(\CG, \nabla^\CG)$ should be given locally by unitary frames of $(\CG, \nabla^\CG)$, as these are our only way of probing and controlling the smooth geometry of $(\CG,\nabla^\CG)$.
To that end, define a set $\Plot(\CT(\CG, \nabla^\CG))$ whose elements are maps $\varphi \in \Set(U, \CT(\CG, \nabla^\CG))$ such that around every $y \in U$ there exists an open neighbourhood $\iota_{V_y} \colon V_y \hookrightarrow U$, a smooth function $f_y \in C^\infty(V_y, \FC)$, and a unitary frame $(S_y, \beta_y)$ of the pullback of $(\CG, \nabla^\CG)$ along the composition
\begin{equation}
\begin{tikzcd}[column sep=1.25cm]
	V_y {\times} S^1 \ar[r, "\varphi_{|V_y} {\times} 1_{S^1}"] & \CT(\CG, \nabla^\CG) \times S^1 \ar[r, "p {\times} 1_{S^1}"] & LM {\times} S^1 \ar[r, "\ev"] & M
\end{tikzcd}
\end{equation}
such that
\begin{equation}
	\varphi(y') = \big[ [\iota_{y'}^*(S_y, \beta_y)],\, f(y') \big] \quad \forall\, y' \in V_y\,,
\end{equation}
where $\iota_{y'} \colon S^1 \hookrightarrow V_y {\times} S^1$, $\tau \mapsto (y', \tau)$ is the inclusion of the $S^1$-factor at $y' \in V_y$.
Such a unitary frame exists at least whenever $V_y$ is contractible.
This follows from the classification Theorem~\ref{st:Classification_of_BGrbs_by_Deligne_coho} together with Proposition~\ref{st:DD_triviality_implies_trivialisability}.

\begin{theorem}
\label{st:transgression_line_bundle}
The pair $\big( \CT(\CG, \nabla^\CG), \Plot(\CT(\CG, \nabla^\CG)) \big)$ defines a diffeological hermitean line bundle on $LM$.
\end{theorem}

\begin{proof}
Let $\CT^{\sfU(1)}(\CG, \nabla^\CG)$ denote the $\sfU(1)$-bundle with fibre
\begin{equation}
	\big( \CT^{\sfU(1)}(\CG, \nabla^\CG) \big)_{|\gamma}
	= \pi_0 \big( \frUF(S^1) \big( \gamma^*(\CG, \nabla^\CG) \big) \big)
\end{equation}
constructed from $(\CG, \nabla^\CG)$.
In~\cite{Waldorf--Transgression_II} this bundle has been shown to carry a diffeological structure.%
\footnote{We should point out that in~\cite{Waldorf--Transgression_II} isomorphisms $\gamma^*(\CG, \nabla^\CG) \to \CI_\rho$ have been used, i.e. in the opposite direction than we are using.
However, the results carry over directly; for instance, one can simply consider the dual bundle gerbe and apply the transpose functor.}
The set of plots which we have constructed above is precisely the pushforward diffeology of the product diffeology on $\CT^{\sfU(1)}(\CG, \nabla^\CG) \times \FC$ along the quotient map to the space of $\sfU(1)$-orbits, i.e. to $(\CT^{\sfU(1)}(\CG, \nabla^\CG) \times \FC)/\sfU(1) = \CT(\CG, \nabla^\CG)$ (cf. Example~\ref{eg:diffeological_spaces_and_constructions}).
\end{proof}

\begin{definition}
Let $(\CG, \nabla^\CG) \in \BGrb^\nabla(M)$.
The line bundle $\CT(\CG, \nabla^\CG) \to LM$ is called the \emph{transgression line bundle of $(\CG, \nabla^\CG)$}.
\end{definition}

We now go on to investigate the structure of this line bundle and its dependence on the bundle gerbe with connection $(\CG, \nabla^\CG)$.
To begin with, let $i = 0, 1$ and consider a 1-morphism $(E,\alpha) \in \BGrb^\nabla(M)((\CG_0, \nabla^{\CG_0}), (\CG_1, \nabla^{\CG_1}))$ between two bundle gerbes $(\CG_i, \nabla^{\CG_i}) \in \BGrb^\nabla(M)$.
Let $(S_i, \beta_i)$ be unitary frames of $(\CG_i, \nabla^{\CG_i})$ over $\gamma$, i.e. $(S_i, \beta_i) \in \frUF(S^1)(\gamma^*(\CG_i, \nabla^{\CG_i}))$.
Define
\begin{equation}
\label{eq:transgression_of_morphisms}
\begin{aligned}
	&\CT(E,\alpha)_{|\gamma} \colon \CT(\CG_0, \nabla^{\CG_0})_{|\gamma} \to \CT(\CG_1, \nabla^{\CG_1})_{|\gamma}\,,
	\\
	&\big[ [S_0, \beta_0], z \big] \mapsto \big[ [S_1, \beta_1], \big( \tr \circ \hol \circ \sfR \big( (S_1, \beta_1)^{-1} \circ \gamma^*(E,\alpha) \circ (S_0, \beta_0) \big) \big)\, z \big]\,.
\end{aligned}
\end{equation}
This is well-defined, for if $(S'_i, \beta'_i) \in \frUF(S^1)(\gamma^*(\CG_i, \nabla^{\CG_i}))$, we have
\addtocounter{equation}{1}
\begin{align*}
	&\CT(E,\alpha)_{|\gamma} \big( \big[ [S_0, \beta_0], z \big] \big)
	\\*[0.2cm]
	&= \CT(E,\alpha)_{|\gamma} \Big( \big[ [S_0', \beta_0'],\, \big( \hol \circ \sfR ((S'_0, \beta'_0)^{-1} \circ (S_0, \beta_0)) \big)\, z \big] \Big)
	\\[0.2cm]
	&= \big[ [S_1', \beta_1'],\, \big( \tr \circ \hol \circ \sfR \big( (S'_1, \beta'_1)^{-1} \circ \gamma^*(E,\alpha) \circ (S'_0, \beta'_0) \big) \big)\,
	\\*
	&\qquad \big( \hol \circ \sfR ((S'_0, \beta'_0)^{-1} \circ (S_0, \beta_0)) \big)\, z \big] \theeq
	\\[0.2cm]
	&=\big[ [S_1', \beta_1'],\, \big( \tr \circ \hol \circ \sfR \big( (S'_1, \beta'_1)^{-1} \circ (S_1, \beta_1) \circ (S_1, \beta_1)^{-1} \circ \gamma^*(E,\alpha) \circ (S_0, \beta_0) \big) \big)\, z \big]
	\\*[0.2cm]
	&= \CT(E,\alpha) \big( \big[ [S_1, \beta_1], z \big] \big)\,.
\end{align*}
Here we have used that in the subcategory $\HVBdl^\nabla(M) \subset \BGrb^\nabla(M)(\CI_0, \CI_0)$ the composition induced from the 2-category coincides (here even strictly, not just up to isomorphism) with the tensor product of vector bundles, and that the trace of the holonomy of a tensor product of vector bundles is the product of the traces of the individual holonomies.
Similarly, one can see that for morphisms $(E,\alpha) \in \BGrb^\nabla(M)((\CG_0, \nabla^{\CG_0}), (\CG_1, \nabla^{\CG_1}))$ and $(E',\alpha') \in \BGrb^\nabla(M)((\CG_1, \nabla^{\CG_1}), (\CG_2, \nabla^{\CG_2}))$ we have
\begin{equation}
	\CT \big( (E', \alpha') \circ (E,\alpha) \big) = \CT(E',\alpha') \circ \CT(E,\alpha)\,.
\end{equation}
Moreover, setting $(E,\alpha) = 1_{(\CG, \nabla^\CG)}$ we readily see from~\eqref{eq:transgression_of_morphisms} that $\CT(1_{(\CG, \nabla^\CG)}) = 1_{\CT(\CG, \nabla^\CG)}$.
Since only the trace of the holonomy of $(S_1, \beta_1)^{-1} \circ \gamma^*(E,\alpha) \circ (S_0, \beta_0)$ enters in~\eqref{eq:transgression_of_morphisms}, we see that if there exists a parallel 2-isomorphism $(E,\alpha) \to (F, \beta)$, we obtain a parallel isomorphism of vector bundles with connections from $\sfR((S_1, \beta_1)^{-1} \circ \gamma^*(E,\alpha) \circ (S_0, \beta_0))$ to $\sfR((S_1, \beta_1)^{-1} \circ \gamma^*(F,\beta) \circ (S_0, \beta_0))$, implying $\CT(E,\alpha) = \CT(F, \beta)$.
Hence, transgression only sees 2-isomorphism classes of 1-morphisms of bundle gerbes with connection.

Now consider a diffeological path in $LM$ with sitting instants.
This is a map $f \colon [0,1] \to LM$ such that there exists an open neighbourhood $V \subset [0,1]$ of $\{0,1\}$ with $f_{|V}$ constant and such that the composition
\begin{equation}
\label{eq:mapping_space_comp}
\begin{tikzcd}[column sep=1.25cm]
	\widetilde{f} \colon [0,1] {\times} S^1 \ar[r, "f {\times} 1_{S^1}"] & LM {\times} S^1 \ar[r, "\ev"] & M
\end{tikzcd}
\end{equation}
is smooth.
Denote by $\iota_\sigma \colon S^1 \hookrightarrow [0,1] {\times} S^1$ the inclusion of the $S^1$ factor at parameter value $\sigma \in [0,1]$.
Note that because $[0,1] {\times} S^1$ is 2-dimensional, and because of Proposition~\ref{st:DD_triviality_implies_trivialisability}, there exists a unitary frame $(S, \beta) \colon \CI_\rho \to \widetilde{f}^*(\CG, \nabla^\CG)$ for some $\rho \in \Omega^2([0,1] {\times} S^1,\, \iu\, \FR)$.
Define
\begin{equation}
\label{eq:PT_on_transgression_line_bundle}
	P^{\CT(\CG, \nabla^\CG)}_f \big[ [\iota_0^*(S, \beta)], z \big]
	\coloneqq \big[ [\iota_1^*(S, \beta)],\, \exp \big( \textint_{[0,1] {\times} S^1}\, \rho \big)\, z \big]\,.
\end{equation}
If $(S', \beta') \colon \CI_{\rho'} \to \widetilde{f}^*(\CG, \nabla^\CG)$ is a different unitary frame, we have
\addtocounter{equation}{1}
\begin{align*}
	&P^{\CT(\CG, \nabla^\CG)}_f \big[ [\iota_0^*(S', \beta')],\, \big( \hol \circ \iota_0^* \sfR ((S', \beta')^{-1} \circ (S, \beta)) \big)\, z \big]
	\\*[0.2cm]
	&= \big[ [\iota_1^*(S', \beta')],\, \exp \big( \textint_{[0,1] {\times} S^1}\, \rho' \big)\, \big( \hol \circ \iota_0^* \sfR ((S', \beta')^{-1} \circ (S, \beta)) \big)\, z \big]
	\\[0.2cm]
	&= \big[ [\iota_1^*(S', \beta')],\, \exp \big( \textint_{[0,1] {\times} S^1}\, \rho' \big)\, \exp \big( - \textint_{[0,1] {\times} S^1}\, \curv \big( \sfR ((S', \beta')^{-1} \circ (S, \beta)) \big) \big)
	\\
	&\qquad \big( \hol \circ \iota_1^* \sfR ((S', \beta')^{-1} \circ (S, \beta)) \big)\, z \big] \theeq
	\\[0.2cm]
	&= \big[ [\iota_1^*(S', \beta')],\, \exp \big( \textint_{[0,1] {\times} S^1}\, \rho \big)\, \big( \hol \circ \iota_1^* \sfR ((S', \beta')^{-1} \circ (S, \beta)) \big)\, z \big]
	\\*[0.2cm]
	&= P^{\CT(\CG, \nabla^\CG)}_f \big[ [\iota_0^*(S, \beta)], z \big]\,.
\end{align*}

Thus, the linear unitary isomorphism
\begin{equation}
	P^{\CT(\CG, \nabla^\CG)}_f \colon \CT(\CG, \nabla^\CG)_{|f(0)} \to \CT(\CG, \nabla^\CG)_{|f(1)}
\end{equation}
is well-defined.
If $g \colon [0,1] \to LM$ is a second path in $LM$ such that $g(0) = f(1)$, we can form the concatenation $g*f \colon [0,1] \to LM$, where
\begin{equation}
	g*f(\tau) =
	\begin{cases}
		f(2\tau)\,, & \tau \in [0,\frac{1}{2}]\,,
		\\
		g(2\tau - 1)\,, & \tau \in [\frac{1}{2}, 1]\,.
	\end{cases}
\end{equation}
Note that if $g$ and $f$ have sitting instants at $\tau = 0$ and $\tau = 1$, this is again a smooth path with sitting instants.
Consider a unitary frame $(S, \beta) \colon \CI_\rho \to (\widetilde{g*f})^*(\CG, \nabla^\CG)$ with $\widetilde{g*f}$ defined as in~\eqref{eq:mapping_space_comp}.
Set $j_0 \colon [0,1] {\times} S^1 \to [0,1] {\times} S^1$, $(\sigma, \tau) \mapsto (\frac{1}{2}\, \sigma, \tau)$ and $j_1 \colon [0,1] {\times} S^1 \to [0,1] {\times} S^1$, $(\sigma, \tau) \mapsto (\frac{1}{2}(\sigma + 1), \tau)$.
Note that $\iota_i = j_i \circ \iota_i$ for $i = 0,1$, and $j_0 \circ \iota_1 = j_1 \circ \iota_0$ as maps $S^1 \to [0,1] {\times} S^1$.
Moreover, observe that $j_0^*(S, \beta) \colon \CI_{j_0^*\rho} \to \widetilde{f}^*(\CG, \nabla^\CG)$ as well as $j_1^*(S, \beta) \colon \CI_{j_1^*\rho} \to \widetilde{g}^*(\CG, \nabla^\CG)$ are unitary frames.
We can compute
\addtocounter{equation}{1}
\begin{align*}
	&P^{\CT(\CG, \nabla^\CG)}_{g*f} \big[ [\iota_0^*(S, \beta)], z \big]
	\\*[0.2cm]
	&= \big[ [\iota_1^*(S, \beta)],\, \exp \big( \textint_{[0,1] {\times} S^1}\, \rho \big)\, z \big]
	\\[0.2cm]
	&= \big[ [\iota_1^*(S, \beta)],\, \exp \big( \textint_{[\frac{1}{2},1] {\times} S^1}\, \rho \big)\, \exp \big( \textint_{[0,\frac{1}{2}] {\times} S^1}\, \rho \big)\, z \big] \theeq
	\\[0.2cm]
	&= \big[ [\iota_1^*j_1^*(S, \beta)],\, \exp \big( \textint_{[0,1] {\times} S^1}\, j_1^*\rho \big)\, \exp \big( \textint_{[0,1] {\times} S^1}\, j_0^*\rho \big)\, z \big]
	\\[0.2cm]
	&= P^{\CT(\CG, \nabla^\CG)}_g \big[ [\iota_0^*j_1^*(S, \beta)],\, \exp \big( \textint_{[0,1] {\times} S^1}\, j_0^*\rho \big)\, z \big]
	\\[0.2cm]
	&= P^{\CT(\CG, \nabla^\CG)}_g \big[ [\iota_1^*j_0^*(S, \beta)],\, \exp \big( \textint_{[0,1] {\times} S^1}\, j_0^*\rho \big)\, z \big]
	\\*[0.2cm]
	&= P^{\CT(\CG, \nabla^\CG)}_g \circ P^{\CT(\CG, \nabla^\CG)}_f \big[ [\iota_0^*(S, \beta)],\, z \big]\,.
\end{align*}
Moreover, for $f$ a constant path we can use a constant unitary frame to see that in this case $P^{\CT(\CG, \nabla^\CG)}_f$ is the identity.
Thus, $P^{\CT(\CG, \nabla^\CG)}$ induces a diffeological parallel transport on the line bundle $\CT(\CG, \nabla^\CG)$, making it into a diffeological hermitean line bundle with connection on $LM$ (cf.~\cite{Waldorf--Transgression_II}).
For a 2-category $\scC$ we denote by $h_1 \scC$ the category obtained by identifying 2-isomorphic 1-morphisms in $\scC$.
By slight abuse of notation, we understand $\HLBdl^\nabla(LM)$ to be the category of diffeological line bundles with connection (in the sense of a diffeological parallel transport) on $LM$.
We have proven the following statement:

\begin{theorem}
There exists a symmetric monoidal functor
\begin{equation}
	\CT \colon h_1 \big( \BGrb^\nabla_\rmpar(M), \otimes\, \big) \to \big( \HLBdl^\nabla(LM), \otimes\, \big)\,,
\end{equation}
acting on objects as defined in Theorem~\ref{st:transgression_line_bundle} and on morphisms as in~\eqref{eq:transgression_of_morphisms}, with parallel transport on $\CT(\CG, \nabla^\CG)$ defined in~\eqref{eq:PT_on_transgression_line_bundle}.
\end{theorem}

\begin{definition}[Transgression functor]
\label{def:transgression_functor}
The functor $\CT$ is called the \emph{transgression functor}.
The line bundle with connection $\CT(\CG, \nabla^\CG)$ is called the \emph{transgression line bundle of $(\CG, \nabla^\CG)$}.
\end{definition}

\begin{remark}
While we are employing the diffeological and 2-categorical techniques developed in~\cite{Waldorf--Thesis} and~\cite{Waldorf--Transgression_II,Waldorf--Transgression_I}, the formula for the transgression of morphisms is inspired also by the work~\cite{CJM--Holonomy_on_D-branes}.
Our definition of $\CT$ combines those two approaches in order to obtain a transgression functor defined on the entire category $h_1(\BGrb^\nabla_\rmpar(M))$.
\qen
\end{remark}

From the construction of the parallel transport in~\eqref{eq:PT_on_transgression_line_bundle} we see that the curvature of $\CT(\CG, \nabla^\CG)$ is represented as follows.
Let $f \in \DfgSp(D^2, LM)$ be a diffeological map from the two-dimensional disc to the loop space of $M$.
By Proposition~\ref{st:DD_triviality_implies_trivialisability} there exists a unitary frame $(S, \beta) \colon \CI_\rho \to \widetilde{f}^*(\CG, \nabla^\CG)$ for some $\rho \in \Omega^1(D^2 {\times} S^1, \iu\, \FR)$.
We can compute the holonomy of $\CT(\CG, \nabla^\CG)$ around $\partial f \coloneqq f_{|\partial D^2}$:
From~\eqref{eq:PT_on_transgression_line_bundle} we obtain
\addtocounter{equation}{1}
\begin{align*}
\label{eq:field_strength_og_CTCG_from_hol}
	\hol \big( (\partial f)^* \CT(\CG, \nabla^\CG) \big)
	&= \exp \Big( \int_{\partial D^2 {\times} S^1} \rho \Big)
	\\*
	&= \exp \Big( \int_{D^2 {\times} S^1} \dd \rho \Big) \theeq
	\\
	&= \exp \Big( \int_{D^2 {\times} S^1} \widetilde{f}^* \curv(\nabla^\CG) \Big)\,.
\end{align*}

\begin{definition}[Transgression of differential forms]
\label{def:transgression_of_differential_forms}
For manifolds $M, N$ with $\dim(N) = n$ and $\nu \in \Omega^k(M,\iu\, \FR)$, define the \emph{transgression of $\nu$ to $\Mfd(N,M)$} to be the diffeological $(k-n)$-form (cf.~\cite{Waldorf--Transgression_I}) given pointwise by
\begin{equation}
	(\CT_N \nu)_{|f} (X_0, \ldots, X_{k-n-1})
	= \int_{N} f^* \big( \iota_{X_0 \wedge \ldots \wedge X_{n-k-1}}\, \nu \big)
\end{equation}
for tangent vectors $X_0, \ldots, X_{k-n-1} \in T_f \Mfd(N,M)$, or globaly by
\begin{equation}
	\CT_N \nu = \int_N \ev^* \nu\,,
\end{equation}
where $\ev \colon \Mfd(N,M) {\times} N \to M$ is the evaluation map.
\end{definition}

Equation~\eqref{eq:field_strength_og_CTCG_from_hol} and Definition~\ref{def:transgression_functor} then imply the following proposition:

\begin{proposition}
The field strength of $\CT(\CG, \nabla^\CG)$ is
\begin{equation}
\label{eq:curvature_of_transgerssion_line_bundle}
	\curv \big( \CT(\CG, \nabla^\CG) \big) = \CT_{S^1} \big(\curv(\nabla^\CG) \big)\,.
\end{equation}
\end{proposition}

\begin{remark}
One can reconstruct the bundle gerbe up to 1-isomorphism from the transgression line bundle.
This construction, called \emph{regression}, has been worked out in~\cite{Waldorf--Transgression_II}.
The underlying surjective submersion of the bundle gerbe obtained from this procedure is the based path fibration $\CP M \to M$ (see Diagram~\eqref{eq:tautological_BGrb_diagramm}).%
\footnote{To be precise, we should be referring to the path fibration as a subduction, the diffeological generalisation of surjective submersion.}
A choice for a curving of the bundle gerbe is $\CT_{[0,1]} \curv(\nabla^\CG)$.
The curving of the regression bundle gerbe then agrees with the curving of the original bundle gerbe.

Our sign of the curvature of the transgression line bundle $\CT(\CG, \nabla^\CG)$ is the opposite of the one in~\cite{Waldorf--Transgression_II}.
We have taken different conventions in three places here as compared to those in that reference.
First, we use a different sign convention in the trace condition~\eqref{eq:trace_condition}.
Second, we use unitary frames $\CI_0 \to (\CG, \nabla^\CG)$ here, rather than trivialisations $(\CG, \nabla^\CG) \to \CI_0$.
Trivialisations are related to the dual bundle gerbe, yielding another sign.
Finally, we say that a pair $(\gamma_0, \gamma_1)$ of based paths in $M$ with common end point gives rise to the loop $\overline{\gamma_0} * \gamma_1$, which is the convention motivated by topological field theory, but the opposite of the convention in~\cite{Waldorf--Transgression_II}.
\qen
\end{remark}

We conclude this section by proving a claim made in Remark~\ref{rmk:trace_condition_on_1-morphisms}.

\begin{proposition}
Let $(E,\alpha) \colon (\CG_0, \nabla^{\CG_0}) \to (\CG_1, \nabla^{\CG_1})$ be a 1-morphism which satisfies the \emph{fake curvature condition}
\begin{equation}
	\curv(\nabla^E) - \big( \zeta_{Y_1}^*B_1 - \zeta_{Y_0}^*B_0 \big) \otimes 1_E = 0\,.
\end{equation}
Then $\CT(E,\alpha)$ is parallel.
\end{proposition}

\begin{proof}
Let $f \colon [0,1] \to LM$ be any smooth path, and let $(S_i, \beta_i) \colon \CI_{\rho_i} \to \widetilde{f}^*(\CG_i, \nabla^{\CG_i})$ be unitary frames for $i = 0,1$.
We have
\addtocounter{equation}{1}
\begin{align*}
	&P^{\CT(\CG, \nabla^\CG)}_f \circ \CT(E,\alpha)\, \big[ [\iota_0^*(S_0, \beta_0)], z \big] \theeq
	\\*
	&= \big[ [\iota_1^*(S_1, \beta_1)],\, \exp \big( \textint_{[0,1] {\times} S^1}\, \rho_1 \big)\, \big( \tr \circ \hol \circ \sfR\, \iota_0^*\big( (S_1, \beta_1)^{-1} \circ \widetilde{f}^*(E,\alpha) \circ (S_0, \beta_0) \big) \big)\, z \big]\,.
\end{align*}
Thus, we need to consider the trace of the holonomy of
\begin{equation}
	(E', \nabla^{E'}) \coloneqq \sfR \big( (S_1, \beta_1)^{-1} \circ \widetilde{f}^*(E,\alpha) \circ (S_0, \beta_0) \big) \quad \in \HVBdl^\nabla([0,1] {\times} S^1)
\end{equation}
around the bounding circles of the cylinder $C \coloneqq [0,1] {\times} S^1$.
For $t \in [0,1]$, let $\gamma'_t$ be the path $\sigma \mapsto (t\, \sigma, 1)$ in $C$, where we parameterise $S^1$ as the set of unit-length complex numbers.
Set $\gamma_t(\sigma) = (t, e^{2\pi\, \iu\, \sigma})$ for $t \in [0,1]$.
We obtain a family of piecewise smooth loops based at $(1,0) \in C$ given by
\begin{equation}
	\hat{\gamma} \coloneqq \overline{\gamma'_{(-)}} * \gamma_{(-)} * \gamma'_{(-)} \colon C \to C\,, \quad
	(t, \sigma) \mapsto \hat{\gamma}_t(\sigma) = (\overline{\gamma'_t} * \gamma_t * \gamma'_t) (\sigma)\,,
\end{equation}
where $\overline{\gamma'_t}$ denotes the reversed path, i.e. $\overline{\gamma'_t}(\sigma) = \overline{\gamma'_t}(1-\sigma)$.
We have~\cite{Taubes--Differential_geometry}
\begin{equation}
\label{eq:derivative_of_holonomy}
\begin{aligned}
	&\frac{\dd}{\dd t}_{|t_0}\, \hol \big( \hat{\gamma}_t^*(E',\nabla^{E'}) \big)
	\\
	&= \hol \big( \hat{\gamma}_{t_0}^*(E',\nabla^{E'}) \big)\,
	\Big( - \int_{S^1} P^{(E', \nabla^{E'})\, -1}_{\hat{\gamma}} \circ \iota_{t_0}^* \big( \iota_{\partial_t} \hat{\gamma}_{(-)}^* \curv(\nabla^{E'}) \big) \circ P^{(E', \nabla^{E'})}_{\hat{\gamma}} \Big)\,,
\end{aligned}
\end{equation}
where $\iota_{t_0} \colon S^1 \to C$ denotes the inclusion of $S^1$ at $t_0 \in [0,1]$.
The fake curvature condition amounts to $\curv(\nabla^{E'}) = \rho_1 - \rho_0$ so that equation~\eqref{eq:derivative_of_holonomy} becomes
\begin{equation}
	\frac{\dd}{\dd t}_{|t_0}\, \hol \big( \hat{\gamma}_t^*(E',\nabla^{E'}) \big)
	= \hol \big( \hat{\gamma}_{t_0}^*(E',\nabla^{E'}) \big)\, \Big( \int_{S^1} \iota_{\partial_t} (\rho_1 - \rho_0) \Big)\,.
\end{equation}
Observe that, by construction,
\begin{equation}
	\hol \big( \hat{\gamma}_t^*(E',\nabla^{E'}) \big)
	= \big( P^{(E', \nabla^{E'})}_{\gamma'_t} \big)^{-1} \circ \hol \big( \gamma_t^*(E',\nabla^{E'}) \big) \circ P^{(E', \nabla^{E'})}_{\gamma'_t}\,,
\end{equation}
such that under the trace we obtain
\addtocounter{equation}{1}
\begin{align*}
	\tr \Big( \hol \big( \gamma_1^*(E',\nabla^{E'}) \big) \Big)
	&= \tr \Big( \hol \big( \hat{\gamma}_1^*(E',\nabla^{E'}) \big) \Big)
	\\
	&= \tr \Big( \hol \big( \hat{\gamma}_0^*(E',\nabla^{E'}) \big) \Big)\,
	\exp \Big( \int_C (\rho_1 - \rho_0) \Big) \theeq
	\\
	&=  \tr \Big( \hol \big( \gamma_0^*(E',\nabla^{E'}) \big) \Big)\,
		\exp \Big( \int_C (\rho_1 - \rho_0) \Big)\,.
\end{align*}
Consequently,
\addtocounter{equation}{1}
\begin{align*}
	&P^{\CT(\CG, \nabla^\CG)}_f \circ \CT(E,\alpha)\, \big[ [\iota_0^*(S_0, \beta_0)], z \big]
	\\*[0.2cm]
	&= \big[ [\iota_1^*(S_1, \beta_1)],\, \exp \big( \textint_{[0,1] {\times} S^1}\, \rho_1 \big)\, \big( \tr \circ \hol \circ \sfR\, \iota_0^*\big( (S_1, \beta_1)^{-1} \circ \widetilde{f}^*(E,\alpha) \circ (S_0, \beta_0) \big) \big)\, z \big]
	\\[0.2cm]
	&=\big[ [\iota_1^*(S_1, \beta_1)],\, \exp \big( \textint_{[0,1] {\times} S^1}\, \rho_1 \big)\, \exp \big( \textint_{[0,1] {\times} S^1}\, (\rho_0 - \rho_1) \big) \theeq
	\\
	&\qquad \big( \tr \circ \hol \circ \sfR\, \iota_1^*\big( (S_1, \beta_1)^{-1} \circ \widetilde{f}^*(E,\alpha) \circ (S_0, \beta_0) \big) \big)\, z \big]
	\\*[0.2cm]
	&= \CT(E,\alpha) \circ P^{\CT(\CG, \nabla^\CG)}_f\, \big[ [\iota_0^*(S_0, \beta_0)], z \big]\,.
\end{align*}
That is, $\CT(E,\alpha)$ intertwines the parallel transports on $\CT(\CG_0, \nabla^{\CG_0})$ and $\CT(\CG_1, \nabla^{\CG_1})$.
\end{proof}

\section{Transgression of categorical structures}
\label{sect:transgression_of_additional_structures}

We proceed to investigate the compatibility of the transgression functor $\CT$ from Definition~\ref{def:transgression_functor} with the additional categorical structures found in Chapter~\ref{ch:bundle_gerbes}.

\subsection{Tensor products}

Consider $(\CG_i, \nabla^{\CG_i}) \in \BGrb^\nabla(M)$ for $i = 0,1,2,3$.
Let $\gamma \in LM$ and choose unitary frames $(S_i, \beta_i) \colon \CI_0 \to \gamma^*(\CG_i, \nabla^{\CG_i})$ over $S^1$.
There is a morphism of diffeological line bundles over $LM$ which reads as
\begin{equation}
\begin{aligned}
	(\phi_{\CG_0, \CG_1})_{|\gamma} \colon \CT(\CG_0, \nabla^{\CG_0})_{|\gamma} \otimes \CT(\CG_1, \nabla^{\CG_1})_{|\gamma} &\to \CT \big( (\CG_0, \nabla^{\CG_0}) \otimes (\CG_1, \nabla^{\CG_1}) \big)_{|\gamma}\,,
	\\
	\big[ [S_0, \beta_0], z \big] \otimes \big[ [S_1, \beta_1], z' \big] &\mapsto \big[ [(S_0, \beta_0) \otimes (S_1, \beta_1)], z\, z' \big]\,.
\end{aligned}
\end{equation}
One can see that this is a unitary isomorphism of line bundles.
It is defined with respect to the unitary frames that are used to define the plots on these line bundles, such that it is, in fact, diffeological.
Well-definedness follows from the fact that the holonomy of a tensor product of line bundles is the product of the individual holonomies.
The same fact implies that $\phi_{\CG_0, \CG_1}$ is parallel, as can be seen from~\eqref{eq:PT_on_transgression_line_bundle}.
It follows, moreover, from the associativity of the tensor product in $\BGrb^\nabla(M)$ that
\begin{equation}
	\phi_{\CG_0, \CG_1 \otimes \CG_2} \circ \big( 1_{\CT(\CG_0, \nabla^{\CG_0})} \otimes \phi_{\CG_1, \CG_2} \big)
	= \phi_{\CG_0 \otimes \CG_1, \CG_2} \circ \big( \phi_{\CG_0, \CG_1} \otimes 1_{\CT(\CG_2, \nabla^{\CG_2})} \big)\,.
\end{equation}
By using the compatibility of $\tr \circ \hol$ with tensor products on can check that $\phi$ defines a natural isomorphism
\begin{equation}
	\phi_{-,-} \colon ( - \otimes - ) \circ (\CT {\times} \CT) \to \CT \circ ( - \otimes - )\,,
\end{equation}
where the tensor product on the source side is in $\HLBdl^\nabla(LM)$, and the one on the target side is in $h_1(\BGrb^\nabla(M))$.

\subsection{Direct sum}

Consider a pair of 1-morphisms $(E,\alpha), (E', \alpha') \in \BGrb^\nabla(M)((\CG_0, \nabla^{\CG_0}), (\CG_1, \nabla^{\CG_1}))$.
In Section~\ref{sect:Additive_structures_on_morphisms_in_BGrb} we defined a functorial direct sum of 1-morphisms between a given pair of bundle gerbes and showed in Theorem~\ref{st:direct_sum_structure_on_morphisms_of_BGrbs} that it is compatible with composition and the tensor product in $\BGrb^\nabla(M)$.
Consequently, for $(S_i, \beta_i) \colon \CI_0 \to \gamma^*(\CG_i, \nabla^{\CG_i})$, for $i = 0,1$, unitary frames of $(\CG_i, \nabla^{\CG_i})$ over $\gamma \in LM$, we obtain
\addtocounter{equation}{1}
\begin{align*}
	&\CT \big( (E,\alpha) \oplus (E', \alpha') \big)_{|\gamma} \big[ [S_0, \beta_0], z \big]
	\\*[0.2cm]
	&= \big[ [S_1, \beta_1], \big( \tr \circ \hol \circ \sfR \big( (S_1, \beta_1)^{-1} \circ \gamma^*\big( (E,\alpha) \oplus (E', \alpha') \big) \circ (S_0, \beta_0) \big) \big)\, z \big]
	\\[0.2cm]
	&= \big[ [S_1, \beta_1], \big( \tr \circ \hol \big( \sfR \big( (S_1, \beta_1)^{-1} \circ \gamma^* (E,\alpha) \circ (S_0, \beta_0) \big)
	\\*
	&\qquad \oplus \sfR \big( (S_1, \beta_1)^{-1} \circ \gamma^* (E',\alpha') \circ (S_0, \beta_0) \big) \big)\, z \big] \theeq
	\\[0.2cm]
	&= \big[ [S_1, \beta_1], \big( \tr \circ \hol \circ \sfR \big( (S_1, \beta_1)^{-1} \circ \gamma^* (E,\alpha) \circ (S_0, \beta_0) \big) \big)\, z
	\\*
	&\qquad + \big( \tr \circ \hol \circ \sfR \big( (S_1, \beta_1)^{-1} \circ \gamma^* (E',\alpha') \circ (S_0, \beta_0) \big) \big)\, z \big]
	\\*[0.2cm]
	&= \big( \CT (E,\alpha)_{|\gamma} + \CT (E', \alpha')_{|\gamma} \big) \big[ [S_0, \beta_0], z \big]\,.
\end{align*}

We deduce the following proposition:

\begin{proposition}
The transgression functor $\CT \colon h_1 (\BGrb^\nabla(M)) \to \HLBdl^\nabla(LM)$ is additive, i.e. it respects the enrichment in $\Ab$ in its source and target.
\end{proposition}

\subsection{Duals}

Recall the definition of the Riesz dual functor $\Theta$ from Definition~\ref{def:Riesz_dual_functor}.
Its functoriality readily implies

\begin{proposition}
There is a natural isomorphism
\begin{equation}
	\CT(\CG, \nabla^\CG) \to \overline{\CT\big( (\CG, \nabla^\CG)^* \big)}\,, \quad
	\big[ [S, \beta], z \big] \mapsto \big[ [ \Theta (S, \beta)], \overline{z} \big]\,.
\end{equation}
Thus, there is a natural isomorphism
\begin{equation}
	\CT\big( (\CG, \nabla^\CG)^* \big) \cong \big( \CT(\CG,\nabla^\CG) \big)^*\,.
\end{equation}
\end{proposition}

\subsection{Pairings and bundle metrics}

In Section~\ref{sect:higher_geometric_structures} we interpreted the bifunctor $[-,-]$ from Theorem~\ref{st:internal_hom_of_morphisms_in_BGrb--existence_and_naturality} and Theorem~\ref{st:2-var_adjunction_on_BGrb} as a higher version of a hermitean bundle metric on morphisms of bundle gerbes.
Consider bundle gerbes with connections $(\CG_i, \nabla^{\CG_i}) \in \BGrb^\nabla(M)$ for $i = 0,1$, and 1-morphisms $(E,\alpha), (E', \alpha') \in \BGrb^\nabla(M) ((\CG_0, \nabla^{\CG_0}), (\CG_1, \nabla^{\CG_1}))$.
We compute
\begin{equation}
\label{eq:transgression_after_higher_bundle_metric}
	\CT \big( \big[ (E,\alpha), (E', \alpha') \big] \big)_{|\gamma}
	= \tr \circ \hol \circ \gamma^* \sfR \big[ (E,\alpha), (E', \alpha') \big]\,.
\end{equation}
If we transgress the morphisms first and then use the bundle metric induced on the bundle of morphisms from $\CT(\CG_0, \nabla^{\CG_0})$ to $\CT(\CG_1, \nabla^{\CG_1})$, we obtain
\addtocounter{equation}{1}
\begin{align*}
\label{eq:bundle_metric_after_transgression}
	&\big( (E,\alpha), (E', \alpha') \big)
	\\*[0.2cm]
	&\mapsto\ \overline{\big( \tr \circ \hol \circ \sfR \big( (S_1, \beta_1)^{-1} \circ \gamma^*(E,\alpha) \circ (S_0, \beta_0) \big) \big)}
	\\*
	&\qquad \cdot \big( \tr \circ \hol \circ \sfR \big( (S_1, \beta_1)^{-1} \circ \gamma^*(E',\alpha') \circ (S_0, \beta_0) \big) \big)
	\\[0.2cm]
	&= \big( \tr \circ \hol \circ \Theta \circ \sfR \big( (S_1, \beta_1)^{-1} \circ \gamma^*(E,\alpha) \circ (S_0, \beta_0) \big) \big) \theeq
	\\*
	&\qquad \cdot \big( \tr \circ \hol \circ \sfR \big( (S_1, \beta_1)^{-1} \circ \gamma^*(E',\alpha') \circ (S_0, \beta_0) \big) \big)
	\\[0.2cm]
	&= \tr \circ \hol \circ \sfR \big( \big( (S_1, \beta_1)^{-1} \circ \gamma^*(E',\alpha') \circ (S_0, \beta_0) \big)
	\otimes \Theta \big( (S_1, \beta_1)^{-1} \circ \gamma^*(E,\alpha) \circ (S_0, \beta_0) \big) \big)\,.
\end{align*}
As before, $(S_i, \beta_i)$ are unitary frames of $\gamma^*(\CG_i, \nabla^{\CG_i})$ over $S^1$.
In the first identity we have made use of the fact that the dual of a bundle with connection has the inverse transpose holonomy.
If the holonomy is unitary, as is the case for the above hermitean bundles, the trace of the holonomy of the dual bundle is, consequently, the complex conjugate of the original holonomy.

\begin{lemma}
\label{st:pairing_and_isomorphisms_for_metrics}
Consider bundle gerbes $(\CG_i, \nabla^{\CG_i}), (\CG'_i, \nabla^{\CG'_i}) \in \BGrb^\nabla(M)$ for $i = 0,1$.
Let $(E,\alpha), (E', \alpha') \in \BGrb^\nabla(M) ((\CG_0, \nabla^{\CG_0}), (\CG_1, \nabla^{\CG_1}))$, and consider two 1-isomorphisms $(J_i, \nu_i) \in \BGrb^\nabla(M) ((\CG_i, \nabla^{\CG_i}), (\CG'_i, \nabla^{\CG'_i}))$.
The following statements hold true:
\begin{myenumerate}
	\item There exists a parallel, unitary 2-isomorphism
	\begin{equation}
		\big[ (J_1, \nu_1) \circ (E,\alpha) \circ (J_0, \nu_0)^{-1},\, (J_1, \nu_1) \circ (E', \alpha') \circ (J_0, \nu_0)^{-1} \big]
		\cong \big[ (E,\alpha), (E', \alpha') \big]\,.
	\end{equation}
	
	\item There exists a parallel, unitary isomorphism of hermitean vector bundles with connection
	\begin{equation}
	\begin{aligned}
		&\sfR \big( (J_1, \nu_1) \circ (E', \alpha') \circ (J_0, \nu_0)^{-1} \big) \otimes \big( \Theta \circ \sfR\, \big( (J_1, \nu_1) \circ (E,\alpha) \circ (J_0, \nu_0)^{-1} \big) \big)
		\\
		&\cong \sfR\, \big[ (E,\alpha), (E', \alpha') \big]\,.
	\end{aligned}
	\end{equation}
\end{myenumerate}
\end{lemma}

\begin{proof}
Ad (1):
We only consider the postcomposition by $(J_1, \nu_1)$.
The precomposition part is analogous.
There are natural isomorphisms, in $\BGrb^\nabla(M)$ as well as in $\BGrb^\nabla_\rmpar(M)$,
\addtocounter{equation}{1}
\begin{align*}
	&\BGrb^\nabla(M) \Big( \big[ (E,\alpha), (E', \alpha') \big],\, \big[ (J_1, \nu_1) \circ (E,\alpha),\, (J_1, \nu_1) \circ (E', \alpha') \big] \Big)
	\\*
	&\cong \BGrb^\nabla(M) \Big( \big( (J_1, \nu_1) \circ (E,\alpha) \big) \otimes \big[ (E,\alpha), (E', \alpha') \big],\, (J_1, \nu_1) \circ (E', \alpha') \Big)
	\\
	&\cong \BGrb^\nabla(M) \Big( (J_1, \nu_1) \circ \big( (E,\alpha)  \otimes \big[ (E,\alpha), (E', \alpha') \big] \big),\, (J_1, \nu_1) \circ (E', \alpha') \Big)
	\\
	&\cong \BGrb^\nabla(M) \Big( (E,\alpha) \otimes \big[ (E,\alpha), (E', \alpha') \big],\, (E', \alpha') \Big) \theeq
	\\*
	&\cong \BGrb^\nabla(M) \Big( \big[ (E,\alpha), (E', \alpha') \big],\, \big[ (E,\alpha), (E', \alpha') \big] \Big)\,.
\end{align*}
Here, the first isomorphism is according to Theorem~\ref{st:2-var_adjunction_on_BGrb}, the second exists since tensoring by endomorphisms of $\CI_0$ is compatible with composition, the third uses that composition by 1-isomorphisms in a 2-category yields an equivalence of morphism categories, and the last isomorphism is, again, Theorem~\ref{st:2-var_adjunction_on_BGrb}.
All of these isomorphisms preserve unitaries and 2-isomorphisms (see Section~\ref{sect:Pairings_and_inner_hom_of_morphisms_in_BGrb} and Appendix~\ref{app:Proof_of_adjunction_theorem}).
The last set of morphisms contains the identity 2-morphism on $[(E,\alpha), (E',\alpha')]$ and, thus, a parallel, unitary 2-isomorphism as required.

Ad (2):
Recall from ~\eqref{eq:internal_hom_and_tensors} that on $\HVBdl^\nabla(M) \subset \BGrb^\nabla(M)(\CI_0, \CI_0)$, the bifunctor $[-,-]$ is naturally isomorphic to $(-) \otimes \Theta(-)$.
Thus,
\addtocounter{equation}{1}
\begin{align*}
	&\sfR \big( (J_1, \nu_1) \circ (E', \alpha') \circ (J_0, \nu_0)^{-1} \big) \otimes \big( \Theta \circ \sfR\, \big( (J_1, \nu_1) \circ (E,\alpha) \circ (J_0, \nu_0)^{-1} \big) \big)
	\\*
	&\cong \big[ \sfR \big( (J_1, \nu_1) \circ (E', \alpha') \circ (J_0, \nu_0)^{-1} \big),\, \sfR\, \big( (J_1, \nu_1) \circ (E,\alpha) \circ (J_0, \nu_0)^{-1} \big) \big] \theeq
	\\
	&\cong \sfR\, \big[ (J_1, \nu_1) \circ (E', \alpha') \circ (J_0, \nu_0)^{-1},\, (J_1, \nu_1) \circ (E,\alpha) \circ (J_0, \nu_0)^{-1} \big]
	\\*
	&\cong \sfR\, \big[ (E,\alpha), (E', \alpha') \big]\,,
\end{align*}
where in the last step we have used part (1) of the statement.
Note that the compatibility of $\sfR$ and $[-,-]$ can be derived from that of $\sfR$ with the tensor product using the adjointness properties of $[-,-]$ and $\otimes$.
\end{proof}

Now we can combine equations~\eqref{eq:transgression_after_higher_bundle_metric} and~\eqref{eq:bundle_metric_after_transgression} with Lemma~\ref{st:pairing_and_isomorphisms_for_metrics} to infer

\begin{proposition}
The transgression functor sends the higher bundle metric $[-,-]$ of the bundle gerbes on the source side to the bundle metric on the target side:
\begin{equation}
	\CT \circ [-,-] = h_{\CT \CG_1 \otimes\CT\CG_0^*} \circ (\CT {\times} \CT)\,,
\end{equation}
where we have omitted connections for the sake of readability.
\end{proposition}

This statement further justifies our interpretation of the bifunctor $[-,-]$ as a higher version of a bundle metric.

\section{Unparameterised loops, orientation, and Reality}
\label{sect:unparameterised_loops_and_Reality}

So far, we have investigated the transgression line bundle over parameterised loops $\gamma \colon S^1 \to M$.
Different parameterisations of the same loop are related by precomposition by diffeomorphisms of $S^1$.
We denote the group of diffeomorphisms of $S^1$ to itself by $\Diff(S^1) = \Mfd_\sim(S^1, S^1)$ and write $\Diff_0(S^1)$ for its subgroup of orientation-preserving diffeomorphisms and $\Diff_1(S^1)$ for its subset of orientation-reversing diffeomorphisms.
Note that all these sets are canonically endowed with a mapping space diffeology (cf. Example~\ref{eg:diffeological_spaces_and_constructions}).

Let, as before, $\gamma \in LM$, $(\CG, \nabla^\CG) \in \BGrb^\nabla(M)$ and $(S, \beta) \colon \CI_0 \to \gamma^*(\CG, \nabla^\CG)$ be a unitary frame.
Orientation-preserving diffeomorphisms $g \in \Diff_0(S^1)$ act on the line bundle $\CT(\CG, \nabla^\CG)$ via
\begin{equation}
	R^{LM}_g \colon \CT(\CG, \nabla^\CG)_{|\gamma} \to \CT(\CG, \nabla^\CG)_{|\gamma \circ g} \,, \quad
	R^{LM}_g \big[ [S, \beta], z \big] \coloneqq \big[ [g^*(S, \beta)], z \big]\,.
\end{equation}
The group $\Diff(S^1)$ is naturally $\RZ_2$-graded, by defining, for $g \in \Diff(S^1)$, a grading
\begin{equation}
	\sfor(g) =
	\begin{cases}
		0\,, & g \text{ is orientation-preserving}\,,
		\\
		1\,, & g \text{ is orientation-reversing}\,.
	\end{cases}
\end{equation}
For $g \in \Diff(S^1)$, let
\begin{equation}
	R^{LM}_g \colon \CT(\CG, \nabla^\CG)_{|\gamma} \to \CT \big( (\CG, \nabla^\CG) \big)^{1 - 2\, \sfor(g)}_{|\gamma \circ g} \,, \quad
	R^{LM}_g \big[ [S, \beta], z \big] \coloneqq \big[ [g^*\Theta^{\sfor(g)}(S, \beta)], z \big]\,,
\end{equation}
where $\Theta^{\sfor(g)} = \Theta$ for $\sfor(g) = 1$ and otherwise $\Theta^{\sfor(g)}$ is the identity.
In order to see that this is well-defined, consider an orientation-reversing diffeomorphism $g \colon S^1 \to S^1$, and two unitary frames $(S, \beta)$ and $(S',\beta')$ of $\gamma^*(\CG, \nabla^\CG)$.
We have
\addtocounter{equation}{1}
\begin{align*}
	&R^{LM}_g \big[ [S', \beta'],\, \hol \big( \sfR((S', \beta')^{-1} \circ (S, \beta)) \big)\, z \big]
	\\*
	&= \big[ [g^*\Theta(S', \beta')],\, \hol \big( \sfR((S', \beta')^{-1} \circ (S, \beta)) \big)\, z \big] \theeq
	\\
	&= \big[ [g^*\Theta(S, \beta)],\, \hol \big( g^* \big( \Theta \circ \sfR((S, \beta)^{-1} \circ (S', \beta')) \big) \big)\, \hol \big( \sfR((S', \beta')^{-1} \circ (S, \beta)) \big)\, z \big]
	\\
	&= \big[ [g^*\Theta(S, \beta)],\, \hol \big( \sfR((S, \beta)^{-1} \circ (S', \beta')) \big)\, \hol \big( \sfR((S', \beta')^{-1} \circ (S, \beta)) \big)\, z \big]
	\\*
	&= R^{LM}_g \big[ [S, \beta], z \big]\,.
\end{align*}
Here we have used that $\Theta$ acts on ordinary vector bundles by taking the dual vector bundle.
Since $\sfR((S, \beta)^{-1} \circ (S', \beta'))$ is a hermitean line bundle with connection, both taking its dual and pulling back along an orientation-reversing diffeomorphism $S^1 \to S^1$ invert the holonomy.

\begin{proposition}
For any bundle gerbe $(\CG, \nabla^\CG)$ on $M$, the transgression line bundle $\CT(\CG, \nabla^\CG)$ is equivariant with respect to the $\Diff_0(S^1)$-action on $LM$.
Therefore, it descends to a diffeological hermitean line bundle with connection $\CT_0(\CG, \nabla^\CG)$ on the space $\CL_+ M$ of unparameterised oriented loops in $M$.
\end{proposition}

\begin{proof}
We have to check the compatibility of the parallel transport on $\CT(\CG, \nabla^\CG)$ with the $\Diff_0(S^1)$-action.
Let $f \colon [0,1] \to LM$ be a diffeological path in $LM$ and $g \colon [0,1] \to \Diff_0(S^1)$, $\sigma \mapsto g_\sigma$.
Let $(S, \beta) \colon \CI_\rho \to \widetilde{f}^*(\CG, \nabla^\CG)$ be a unitary frame of $(\CG, \nabla^\CG)$ over the path $f$, with $\tilde{f}$ as in~\eqref{eq:mapping_space_comp}.
We compute
\begin{equation}
\begin{aligned}
	R^{LM}_{g_1} \circ P^{\CT(\CG, \nabla^\CG)}_f\, \big[ [\iota_0^*(S, \beta)], z \big]
	&= R^{LM}_{g_1} \big[ [\iota_1^*(S, \beta)],\, \exp\big( \textint_{[0,1] {\times} S^1} \rho \big)\, z \big]
	\\
	&= \big[ [g_1^*\iota_1^*(S, \beta)],\, \exp\big( \textint_{[0,1] {\times} S^1} \rho \big)\, z \big]\,.
\end{aligned}
\end{equation}
Note that $\widetilde{f \circ g}$ is the composition
\begin{equation}
\begin{tikzcd}[column sep=1.5cm]
	[0,1] {\times} S^1 \ar[r, "(f {\times} g) {\times} 1_{S^1}"] & LM {\times} \Diff_0(S^1) {\times} S^1 \ar[r, " (- \circ -) {\times} 1_{S^1}"] & LM {\times} S^1 \ar[r, "\ev"] & M\,.
\end{tikzcd}
\end{equation}
Writing $\hat{g} (\sigma, \tau) = (\sigma, g_\sigma(\tau))$, we obtain a diffeomorphism $\hat{g} \colon C \to C$, where $C \coloneqq [0,1] {\times} S^1$ is the cylinder.
We have $(\widetilde{f \circ g}) (\sigma, \tau) = f(\sigma)(g_\sigma(\tau)) = \widetilde{f} \circ \hat{g} (\sigma, \tau)$.
Consequently, $\hat{g}^*(S, \beta)$ is a unitary frame of $(\widetilde{f \circ g})^* (\CG, \nabla^\CG)$.
Observe, moreover, that
\begin{equation}
	\iota_\sigma \circ g_\sigma (\tau) = (\sigma, g_\sigma(\tau)) = \hat{g}(\sigma, \tau) = \hat{g} \circ \iota_\sigma (\tau)\,.
\end{equation}
Using this, we compute
\addtocounter{equation}{1}
\begin{align*}
	P^{\CT(\CG, \nabla^\CG)}_{R^{LM}_g f} \circ R^{LM}_{g_0}\, \big[ [\iota_0^*(S, \beta)], z \big]
	&= P^{\CT(\CG, \nabla^\CG)}_{R^{LM}_g f} \big[ [g_0^* \iota_0^*(S, \beta)],\, z \big]
	\\
	&= P^{\CT(\CG, \nabla^\CG)}_{R^{LM}_g f} \big[ [\iota_0^* \hat{g}^*(S, \beta)],\, z \big] \theeq
	\\
	&= \big[ [\iota_1^* \hat{g}^*(S, \beta)],\, \exp \big( \textint_{[0,1] {\times} S^1} \hat{g}^*\rho \big)\,  z \big]
	\\
	&= \big[ [g_1^*\iota_1^*(S, \beta)],\, \exp \big( \textint_{[0,1] {\times} S^1} \rho \big)\,  z \big]\,.
\end{align*}
In the last step we have used that integrals of differential forms are invariant under pullback along orientation-preserving diffeomorphisms.
\end{proof}

Via the action on loops in $M$, any choice of orientation-reversing diffeomorphism on $S^1$ induces an involution $\rev \colon \CL_+ M \to \CL_+ M$.
For instance, under the diffeomorphism $S^1 \cong \sfU(1)$, we can consider $\inv \in \Diff_1(S^1)$, $\inv(\lambda) = \lambda^{-1}$ for $\lambda \in \sfU(1)$.
If $[[S, \beta], z] \in \CT(\CG, \nabla^\CG)_{|\gamma}$, then, for $g \in \Diff_0(S^1)$,
\addtocounter{equation}{1}
\begin{align*}
	R^{LM}_\inv \circ R^{LM}_g\, \big[ [S, \beta], z \big]
	&= \big[ [\inv^* \Theta (g^* (S, \beta))], z \big]
	\\*
	&= \big[ [(g \circ \inv)^* \Theta (S, \beta)], z \big] \theeq
	\\
	&= \big[ [(\inv \circ g \circ \inv)^* \inv^* \Theta (S, \beta)], z \big]
	\\*
	&= R^{LM}_{(\inv \circ g \circ \inv)} \circ R^{LM}_\inv\, \big[ [S, \beta], z \big]\,.
\end{align*}
This implies that $R^{LM}_\inv \circ R^{LM}_g\, [ [S, \beta], z ]$ and $R^{LM}_\inv\, [ [S, \beta], z ]$ define the same element in the descended bundle $\CT_0(\CG, \nabla^\CG)$ on $\CL_+ M$.
Consequently, the involution $\rev = R^{LM}_\inv$ has a lift $\widehat{\rev}_\CG$ to the line bundle with connection $\CT_0(\CG, \nabla^\CG)$ to yield a commutative diagram
\begin{equation}
\begin{tikzcd}
	\CT_0(\CG, \nabla^\CG) \ar[r, "\widehat{\rev}_\CG"] \ar[d] & \CT_0(\CG, \nabla^\CG)^* \ar[d]
	\\
	\CL_+ M \ar[r, "\rev"] & \CL_+ M
\end{tikzcd}
\end{equation}
Such a lift to a vector bundle of an involution on the base space is called a~\emph{Real structure}.
Thus, we have proven

\begin{theorem}
\label{st:Real_structure_on_T_0CG}
The line bundle $\CT_0(\CG, \nabla^\CG) \to \CL_+ M$ on the space of unparameterised oriented loops in $M$ has a Real structure $\widehat{\rev}_\CG$ with respect to the $\RZ_2$-action induced by orientation-reversal on $S^1$.
\end{theorem}

\begin{remark}
This Real structure has applications in the definition of the so-called unoriented surface holonomy of a bundle gerbe~\cite{SSW--Unoriented_WZW}, or the square root of the holonomy of a bundle gerbe~\cite{Gawedzki--BGrbs_for_TIs,Gawedzki--FKM_and_sewing_matrix} and is, further, interesting from the point of view of topological field theory.
Here it makes manifest the property of an oriented $(2,1)$-dimensional topological field theory to assign to an unparameterised loop with reversed orientation the dual of the vector, or Hilbert, space that it assigns to the original unparameterised loop.
\qen
\end{remark}

\section{Dimensional reduction}
\label{sect:dimensional_reduction}

The line bundle over loop space associated to a bundle gerbe has several important applications.
One of them, which has not yet been studied in the literature, is dimensional reduction (not to be confused with regression~\cite{Waldorf--Transgression_II}).
In view of higher geometric structures in general, not only as elaborated in Chapter~\ref{ch:2Hspaces_from_bundle_gerbes}, dimensional reduction is a powerful tool to relate higher geometric structures to ones which are already known and familiar.
In this way, one can test not only proposals for higher geometric structures, but also for formalisms of higher geometric quantisation.
Using reduction, structures in higher geometric quantisation can be related to their analogues in ordinary geometric quantisation.
From a string theory point of view, dimensional reduction is relevant to compactifications as well as to obtaining (super) string theory limits of M-theory.
Another, related, application of dimensional reduction might be to T-duality~\cite{Bunke-Nikolaus--T-duality_via_gerbey_geometry_and_reductions}.

\subsection{The reduction functor}

We begin by constructing the reduction functor, before working out an explicit example.
Consider a fibre bundle $K \to M$ with fibre $F \in \Mfd$.
We call $K \to M$ \emph{fibre-orientable} if there exists an open covering $\CU = \bigsqcup_{a \in \Lambda} U_a$ of $M$, for some indexing set $\Lambda$, and local trivialisations $\phi_a \colon K_{|U_a} \to U_a {\times} F$ such that the resulting transition functions $\phi_{ab} = \phi_{b|U_{ab}} \circ \phi_{a|U_{ab}}^{-1}$ are of the form $\phi_{ab} \colon U_{ab} \to \Diff_0(F)$, i.e. they consist of smooth families of orientation-preserving diffeomorphisms on $F$.
A choice of such a trivialisation and an orientation on the typical fibre $F$ induces an orientation on each fibre of the $F$-bundle $K$, and all of these orientations are consistent under the transition functions.
Such a family of consistent orientations of the fibres is called a \emph{fibre orientation}.
Note that there are two possible induced fibre-orientations on a fibre-orientable bundle.
A fibre-orientable bundle with a choice of orientation is called a \emph{fibre-oriented bundle}.

\begin{example}
For $\sfG$ a connected Lie group, every principal $\sfG$-bundle is fibre-orientable:
all local transition functions are of the form $\phi_{ab}(x)(\lambda) = \tilde{\lambda}_{ab}(x)\, \lambda$ for some function $\tilde{\lambda} \colon U_{ab} \to \sfG$.
In any Lie group $\sfG$, multiplication by elements from the connected component of the identity necessarily preserves the orientation, as the map $\sfG \to \Diff(\sfG)$, $g \mapsto \ell_g$ is continuous, where $\ell_g \in \Diff(\sfG)$ denotes the diffeomorphism $\ell_g(h) = g\, h$ for $g,h \in \sfG$.
Moreover, a choice of orientation on $\sfG$ naturally induces an orientation on the fibres via the simply transitive action of $\sfG$ on the fibres.
\qen
\end{example}

Any fibre $K_{|x} \subset K$ of a fibre-oriented $S^1$-bundle $K \to M$ can be viewed as an unparameterised oriented loop in $M$.
In order to obtain a parameterised representative, one can use a unitary frame (which is compatible with the orientation).
Therefore, any oriented $S^1$-bundle $K \to M$ has associated to it a canonical map
\begin{equation}
	\widehat{K} \colon M \to \CL_+ K\,,
\end{equation}
that sends $x \in M$ to the oriented fibre $K_{|x}$.
From our findings in Section~\ref{sect:transgression_functor} and Section~\ref{sect:unparameterised_loops_and_Reality}, we infer the following statement.

\begin{theorem}
\label{st:reduction_functor}
For any fibre-oriented $S^1$-bundle $q \colon K \to M$, there exists a functor
\begin{equation}
	q_*  \colon h_1 \big( \BGrb^\nabla(K), \otimes\, \big) \to \big( \HLBdl^\nabla(M), \otimes\, \big)\,, \quad
	q_* = \widehat{K}^* \circ \CT_0\,.
\end{equation}
It is symmetric monoidal and additive, and maps the higher bundle metric $[-,-]$ to the bundle metric of the target hermitean line bundle.
Moreover, we have
\begin{equation}
	\curv \big( q_* (\CG, \nabla^\CG) \big) = q_*\, \curv (\nabla^\CG)\,,
\end{equation}
where $q_*$ is the pushforward of differential forms, i.e. integration along the fibre.
\end{theorem}

\begin{proof}
The functoriality of $q_*$ is immediate from its definition since it is a composition of two functors.
Its additional properties simply follow from the properties of the transgression functor $\CT$ which we have worked out in Sections~\ref{sect:transgression_functor}, \ref{sect:transgression_of_additional_structures} and~\ref{sect:unparameterised_loops_and_Reality}.
The curvature identity follows from~\eqref{eq:curvature_of_transgerssion_line_bundle} and Definition~\ref{def:transgression_of_differential_forms}.
\end{proof}

\begin{definition}[Dimensional reduction]
\label{def:reduction_functor}
Let $K \to M$ be a fibre-oriented $S^1$-bundle on $M$.
The functor $q_* \colon h_1 ( \BGrb^\nabla(K), \otimes\, ) \to ( \HLBdl^\nabla(M), \otimes\, )$ is called \emph{dimensional reduction}, or \emph{pushforward} along $q$.
\end{definition}

\subsection{Example: Revisiting the decomposable bundle gerbe}

Even though the reduction functor introduced in Theorem~\ref{st:reduction_functor} is constructed using infinite-dimensional spaces and may, at first sight, appear somewhat unwieldy, we demonstrate in this section that it is, in fact, well-behaved and practicable.

The example of a torsion bundle gerbe in Section~\ref{sect:Ex:Cup_product_lens_space_BGrbs} is of a form which allows for a dimensional reduction.
Recall that there we defined a bundle gerbe $(\CG, \nabla^\CG)$ on the manifold $M_p \coloneqq \bbL_p {\times} S^1$, where $\bbL_p = S^3/ \RZ_p$ is a lens space.
The defining diagram for the bundle gerbe was~\eqref{eq:cup_product_BGrb_diagram}.
There are two possible $S^1$-fibrations to consider here, namely the projection $M_p \to \bbL_p$ and $M_p \to S^2 {\times} S^1$ induced by the fibration $\bbL_p \to S^2$.
Both are principal fibre bundles and, hence, oriented once an orientation on $\sfU(1)$ has been fixed.
Here we focus on the projection $q = \pr_{\bbL_p} \colon M_p \to \bbL_p$.
In this section we are going to prove

\begin{proposition}
\label{st:reduction_of_cup_product_BGrb}
For the decomposable bundle gerbe $(\CG, \nabla^\CG) \in \BGrb^\nabla(M_p)$ introduced in Section~\ref{sect:Ex:Cup_product_lens_space_BGrbs}, with $q \colon M_p = \bbL_p {\times} S^1 \to \bbL_p$ the projection onto the lens space, there exists a canonical isomorphism of hermitean line bundles with connections
\begin{equation}
	q_*(\CG, \nabla^\CG) \cong (J, \nabla^J)\,.
\end{equation}
\end{proposition}

In order to construct this isomorphism, let $x \in \bbL_p$, and set $\iota_x \colon S^1 \hookrightarrow M_p$, $\tau \mapsto (x, \tau)$.
We investigate the fibre of $q_* (\CG, \nabla^\CG)$ over $x$.
This pushforward line bundle is associated to the $\sfU(1)$-bundle on $\bbL_p$ whose fibre over $x \in \bbL_p$ is $\pi_0(\frUF(\iota_x^*(\CG, \nabla^\CG)))$.
Thus, we have to consider unitary frames of $\iota_x^*(\CG, \nabla^\CG) \in \BGrb^\nabla(S^1)$.
Restricting diagram~\eqref{eq:cup_product_BGrb_diagram} to $x \in \bbL_p$, we are left with a bundle gerbe over $S^1$ of the form
\begin{equation}
\begin{tikzcd}
	(J, \nabla^J)^{\otimes \RZ}_{|x} \ar[d] & 
	\\
	\FR {\times} \RZ \ar[r, shift left=0.1cm] \ar[r, shift left=-0.1cm] & \FR \ar[d, "\pi_x"]
	\\
	 & S^1
\end{tikzcd}
\end{equation}
Let $(S, \beta)$ be a unitary frame of $\iota_x^*(\CG, \nabla^\CG)$.
We may assume that this is defined over the minimal surjective submersion, i.e. that the underlying hermitean line bundle $(S, \nabla^S)$ with connection is defined over $\FR$:
\begin{equation}
\begin{tikzcd}
	(J, \nabla^J)^{\otimes \RZ}_{|x} \ar[d] & S \ar[d]
	\\
	\FR {\times} \RZ \ar[r, shift left=0.1cm] \ar[r, shift left=-0.1cm] & \FR \ar[d, "\pi_x"]
	\\
	 & S^1
\end{tikzcd}
\end{equation}
The isomorphism $\beta$ acts as
\begin{equation}
	\beta_{|(r,n)} \colon S_{|r+n} \to S_{|r} \otimes J^{\otimes n}_{|x}
\end{equation}
and is compatible with the bundle gerbe multiplication.
For $r,s \in \FR$, let
\begin{equation}
	P^S_{r, s} = \colon S_{|r} \to S_{|r+s}
\end{equation}
denote the parallel transport in $(S, \nabla^S)$ along the path $\gamma \colon [0,1] \to \FR$, $\sigma \mapsto r + s\, \sigma$.
We can then define a unitary isomorphism
\begin{equation}
	h^{(S, \beta)}_{|(r,n)} \coloneqq \beta_{|(r,n)} \circ \big( P^S_{r,n} \otimes 1_{J^{\otimes n}} \big) \colon S_{|r} \to S_{|r} \otimes J^{\otimes n}_{|x}\,.
\end{equation}
Let $V,W \in \Hilb$, with $W$ one-dimensional.
Denote the inner product on any Hilbert space $V$ by $\<-,-\>_V$.
Choose an orthonormal basis $(e_i)_{i = 1, \ldots, n}$ of $V$ and a unit-length vector $w \in W$ with $w^\vee \in W^*$ denoting its dual.
There exist trace-like maps
\begin{equation}
\label{eq:trace-like_morphisms}
\begin{aligned}
	&\tr_{\rml,W} \colon \Hilb \big( W \otimes V, V \big) \to W^*\,, \quad
	\psi \mapsto \sum_{i = 1}^n\, \< e_i, \psi(w \otimes e_i) \>_V\, w^\vee\,,
	\\
	&\tr_{\rmr,W} \colon \Hilb \big( V, V \otimes W \big) \to W\,, \quad
	\psi \mapsto \sum_{i = 1}^n\, \< e_i \otimes w, \psi(e_i) \>_V\, w\,,
\end{aligned}
\end{equation}
Like the ordinary trace, the maps in~\eqref{eq:trace-like_morphisms} are defined independently of the choices of bases.
Hence, given a unitary frame $(S, \beta)$ of $\iota_x^*(\CG, \nabla^\CG)$, we obtain a canonical element
\begin{equation}
\label{eq:map_j_for_cup_prod_pushforward}
	j_x(S, \beta) \coloneqq \tr_{\rmr,J} \Big( h^{(S, \beta)}_{|(r,1)} \Big) \in J_{|x}\,.
\end{equation}
Decomposing the parallel transport, one can check that this is defined independently of the choice of $r \in \FR$.
In fact, because of the unitarity of $\beta$ and the compatibility of $\nabla^S$ with the hermitean metric, 
\begin{equation}
	j_x(S, \beta) \in \UF(J)_{|x}\,,
\end{equation}
where $\UF(J)$ denotes the bundle of orthonormal, or unitary, frames of $J$.

Consider a second unitary frame $(S', \beta') \colon \CI_0 \to \iota_x^*(\CG, \nabla^\CG)$ which is 2-isomorphic to $(S, \beta)$.
We may assume that the 2-isomorphism is of the form $[\FR, 1_\FR, \phi]$.
This induces a commutative diagram
\begin{equation}
\label{eq:phi_beta_P_diagram}
\begin{tikzcd}[column sep=1.5cm, row sep=1cm]
	S_{|r} \ar[d, "\phi_{|r}"] \ar[r, "P^S_{r,n}"] & S_{|r+n} \ar[d, "\phi_{|r+n}"] \ar[r, "\beta_{|(r,n)}"] & S_{|r} \otimes J^{\otimes n}_{|x} \ar[d, "\phi_{|r} \otimes 1"]
	\\
	S'_{|r} \ar[r, "P^{S'}_{r,n}"'] & S'_{|r+n} \ar[r, "\beta'_{|(r,n)}"'] & S'_{|r} \otimes J^{\otimes n}_{|x}
\end{tikzcd}
\end{equation}
Here, the right-hand square commutes because of the compatibility of $\phi$ with $\beta$ and $\beta'$, while the left-hand square commutes because $\phi$ is parallel.
Consequently,
\begin{equation}
	(\phi_{|r} \otimes 1) \circ h^{(S, \beta)}_{|(r,n)} = h^{(S', \beta')}_{|(r,n)} \circ \phi_{|r}\,.
\end{equation}
Choosing normalised elements $k \in J_{|x}$ and $e \in S_{|r}$, \eqref{eq:phi_beta_P_diagram} implies
\addtocounter{equation}{1}
\begin{align*}
	j_x(S, \beta)
	&= \< e \otimes k, h^{(S, \beta)}_{|(r,1)} (e)\>_{S_{|r}}\, k
	\\*[0.2cm]
	&= \< e \otimes k, (\phi_{|r} \otimes 1)^{-1} \circ h^{(S', \beta')}_{|(r,-1)} \circ \phi_{|r} (e)\>_{S_{|r}}\, k \theeq
	\\[0.2cm]
	&= \< \phi_{|r} (e) \otimes k, h^{(S', \beta')}_{|(r,1)} \big( \phi_{|r} (e) \big) \>_{S'_{|r}}\, k
	\\*[0.2cm]
	&= j_x(S', \beta')\,.
\end{align*}
In the last step we have used the unitarity of $\phi$ and the invariance of the morphisms in~\eqref{eq:trace-like_morphisms} under changes of orthonormal bases.
This shows that we have defined a map
\begin{equation}
\label{eq:torsor_map_for_cup_product_reduction}
	j_x \colon \pi_0 \big( \frUF(\iota_x^*(\CG, \nabla^\CG)) \big) \to \UF(J)_{|x}\,.
\end{equation}

We continue by investigating the action of hermitean line bundles with connection on unitary frames.
Let $(T, \nabla^T) \in \HLBdl^\nabla(S^1)$, which is necessarily flat by dimensional reasons.
Its action on $(S, \beta)$ is given by
\begin{equation}
	(S, \beta) \otimes T = (S \otimes \pi_x^*T, \beta \otimes 1_{\pi_x^*T})\,.
\end{equation}
The parallel transport in $\pi_x^*(T, \nabla^T)$ on $\FR$ along any path from $r$ to $r+1$ is just $\hol(T, \nabla^T)$, so that
\begin{equation}
	j_x \big( (S, \beta) \otimes T \big) = j_x(S, \beta)\, \hol(T, \nabla^T)\,.
\end{equation}
Hence, the map $j_x$ from~\eqref{eq:torsor_map_for_cup_product_reduction} intertwines the $\sfU(1)$-actions on $\pi_0(\frUF(\iota_x^*(\CG, \nabla^\CG)))$ and $\UF(J)_{|x}$.
That is, $j_{|x}$ is a morphism of $\sfU(1)$-torsors and, thus, an isomorphism.
This induces an isomorphism between the fibres of $J$ and the pushforward $q_*(\CG, \nabla^\CG)$.
As all operations involved are smooth, we obtain an isomorphism of hermitean line bundles $J \cong q_*(\CG, \nabla^\CG)$.
Note that, so far, this does not include the connections or parallel transports on these line bundles.

In order to address parallel transports, we consider a path $\gamma \colon [0,1] \to \bbL_p$.
It induces a map $\hat{\gamma} = \gamma {\times} 1_{S^1} \colon C = [0,1] {\times} S^1 \to \bbL_p {\times} S^1 = M_p$.
We would like to compare the parallel transport along $\gamma$ in $(J, \nabla^J)$ and in $q_*(\CG, \nabla^\CG)$.
In order to find the parallel transport along $\gamma$ in $q_*(\CG, \nabla^\CG)$, consider a unitary frame $(S, \beta) \colon \CI_\rho \to \hat{\gamma}^*(\CG, \nabla^\CG)$ on the cylinder $C$ with underlying hermitean line bundle with connection $(S, \nabla^S) \to [0,1] {\times} \FR$.
As before, we write $\iota_\sigma \colon S^1 \hookrightarrow C$, $\iota_\sigma(\tau) = (\sigma, \tau)$ for $\sigma \in [0,1]$.
From~\eqref{eq:PT_on_transgression_line_bundle} we have
\begin{equation}
	P^{q_*(\CG, \nabla^\CG)}_\gamma\, \big[ [\iota_0^*(S, \beta)], z \big]
	= \big[ [\iota_1^*(S, \beta)],\, \exp \big( \textint_C\, \rho \big)\, z \big]\,.
\end{equation}
Using the map $j$ from~\eqref{eq:map_j_for_cup_prod_pushforward}, we obtain
\begin{equation}
\begin{aligned}
	&j_{\gamma(1)} \circ P^{q_*(\CG, \nabla^\CG)}_\gamma \circ j_{\gamma(0)}^{-1} \colon J_{|\gamma(0)} \to J_{|\gamma(1)}\,,
	\\[0.2cm]
	& \qquad j_{\gamma(0)} \big[ \big( \iota_0^*(S, \beta) \big), z \big] \longmapsto j_{\gamma(1)} \big[ \big( \iota_1^*(S, \beta) \big), \exp \big( \textint_C\, \rho \big)\, z \big]\,.
\end{aligned}
\end{equation}

On the other hand, let $r \in \FR$, $e_0 \in S_{|(0,r)}$, $e_1 \in S_{|(1,r)}$, $k_0 \in J_{|\gamma(0)}$, and $k_1 \in J_{|\gamma(1)}$ be normalised elements.
We set
\begin{equation}
	h_r \gamma \colon [0,1] \to [0,1] {\times} \FR\,, \quad
	h_r \gamma (\sigma) \coloneqq \big( \sigma,\, r \big)
\end{equation}
and compute
\addtocounter{equation}{1}
\begin{align*}
\label{eq:PT_comparison_for_reduction_of_cuppr_BGrb_I}
	&P^J_\gamma \big( j_{\gamma(0)} (\iota_0^*(S, \beta)) \big)
	\\*[0.2cm]
	&= \< e_0 \otimes k_0,\, \beta_{|(0, r, 1)} \circ P^S_{0, r, 1} (e_0) \>_{S_{|(0,r)}}\ P^J_\gamma(k_0) \theeq
	\\[0.2cm]
	&= \< (P^S_{h_r \gamma})^{-1} (e_1) \otimes  (P^J_\gamma)^{-1} (k_1) ,\, \beta_{|(0, r, 1)} \circ P^S_{0, r, 1} \circ (P^S_{h_r \gamma})^{-1} (e_1) \>_{S_{|(0,r)}}\ k_1
	\\[0.2cm]
	&= \< e_1 ,\, \beta_{|(1, r, 1)} \circ P^S_{h_{r+1} \gamma} \circ P^S_{0, r, 1} \circ (P^S_{h_r \gamma})^{-1} (e_1) \>_{S_{|(1,r)}}\ k_1\,.
\end{align*}
Here we have used the independence of the morphisms~\eqref{eq:trace-like_morphisms} of the choice of orthonormal bases and the fact that $\beta$ is parallel.
Consider the following non-commutative digram:
\begin{equation}
\begin{tikzcd}[column sep=1.5cm, row sep=1.25cm]
	S_{|(1, r+1)} \ar[r, "P^S_{h_{r+1} \gamma}"]  & S_{|(0, r+1)}
	\\
	S_{|(0, r)} \ar[u, "P^S_{0, r, 1}"] & S_{|(1, r)} \ar[l, "P^{S\, -1}_{h_r \gamma}"] \ar[u, "P^S_{1, r, 1}"']
\end{tikzcd}
\end{equation}
Its failure to commute is the holonomy of $(S, \nabla^S)$ around the loop (or rectangle) with horizontal edges $h_{r+1} \gamma$ and $h_r \gamma$, and vertical edges given by the paths $\sigma \mapsto (0, r + \sigma)$ and $\sigma \mapsto (1, r + \sigma)$.
Define a map $f_r \colon [0,1]^2 \to [0,1] {\times} \FR$, $(\sigma_0,\sigma_1) \mapsto (\sigma_0, r + \sigma_1)$.
Observe that $\pi \circ f_r$ descends to define a map $\pi \circ f_r \colon [0,1]^2 \to C$, $(\sigma_0, \sigma_1) \mapsto (\sigma_0, e^{2\pi\, \iu\, (\sigma_1 + r)})$.
Then,
\addtocounter{equation}{1}
\begin{align*}
\label{eq:PT_commutation_for_cup_product_reduction}
	&P^{S\, -1}_{1, r, 1} \circ P^S_{h_{r+1} \gamma} \circ P^S_{0, r, 1} \circ P^{S\, -1}_{h_r \gamma}
	\\[0.2cm]
	&= \hol \big( (\partial f_r)^*(S, \nabla^S) \big)^{-1}
	\\[0.2cm]
	&= \exp \Big( \int_{[0,1]^2} - f_r^* \curv(\nabla^S) \Big) \theeq
	\\[0.2cm]
	&= \exp \Big( \int_{[0,1]^2} - \big( \hat{\gamma}^*B - \pi^*\rho \big) \Big)
	\\*[0.2cm]
	&= \exp \Big( \int_{C} \rho \Big)\,,
\end{align*}
where we have used that $B_{|(x,r)} = r\, \curv(\nabla^J)$ (cf. equation~\eqref{eq:cup_product_BGrb--curving}) with $x \in \bbL_p$ and $r \in \FR$, and that $\curv(\nabla^J)$ is constant along the $\FR$-direction.
Plugging this into~\eqref{eq:PT_comparison_for_reduction_of_cuppr_BGrb_I}, we obtain
\addtocounter{equation}{1}
\begin{align*}
	&P^J_\gamma \big( j_{\gamma(0)} (\iota_0^*(S, \beta)) \big)
	\\*[0.2cm]
	&= \< e_1 ,\, \beta_{|(1, r, 1)} \circ P^S_{h_{r+1} \gamma} \circ P^S_{0, r, 1} \circ P^{S\, -1}_{h_r \gamma} (e_1) \>_{S_{|(1,r)}}\ k_1 \theeq
	\\[0.2cm]
	&= \exp \Big( \int_{C} \rho \Big) \ \< e_1 ,\, \beta_{|(1, r, 1)} \circ P^S_{1, r, 1} (e_1) \>_{S_{|(1,r)}}\ k_1
	\\*[0.2cm]
	&= \exp \Big( \int_{C} \rho \Big)\ j_{\gamma(1)} \big( \iota_1^*(S, \beta) \big)\,.
\end{align*}

Consequently, we deduce that
\begin{equation}
	P^J_\gamma \circ j_{\gamma(0)} = j_{\gamma(1)} \circ P^{q_*(\CG, \nabla^\CG)}_\gamma\,,
\end{equation}
thus completing the proof of Proposition~\ref{st:reduction_of_cup_product_BGrb}.

In particular, sections of $(\CG, \nabla^\CG)$ canonically reduce to sections of $J$ over $\bbL_p$.
From the results in Section~\ref{sect:transgression_of_additional_structures} we know, moreover, that the algebraic structures on bundle gerbes and their morphisms which we introduced in Chapter~\ref{ch:bundle_gerbes} are mapped to their analogues on the category $\HLBdl^\nabla(M)$ precisely as we would expect from the analogies laid out in Section~\ref{sect:higher_geometric_structures}.

For $(E,\alpha) \in \Gamma(M_p, (\CG, \nabla^\CG))$ we can set
\begin{equation}
	\tr_{\rmr, J} \big( \alpha_{|(x,r,1)} \circ P^E_{(x,r, 1)} \big) \in J_{|x}\,.
\end{equation}
This reduction of sections has been worked out in~\cite{Bunk-Szabo--Fluxes_brbs_2Hspaces}.
It produces a twisted Wilson loop map
\begin{equation}
	\rmW \colon \pi_0 \big( \Gamma \big( M_p, (\CG, \nabla^\CG) \big)_\sim \big) \to \Gamma(\bbL_p, J)\,.
\end{equation}
Observe that for higher functions $(F, \nabla^F) \in \HVBdl^\nabla(M_p) \cong \BGrb^\nabla(M_p)(\CI_0, \CI_0)$ there is a natural reduction
\begin{equation}
\begin{aligned}
	&\rmW= \tr \circ \hol \circ \iota_{(-)}^* \colon \HVBdl^\nabla(M_p) \to C^\infty(\bbL_p)\,,
	\\
	&\big( \tr \circ \hol \circ \iota_{(-)}^* (F, \nabla^F) \big) (x) = \tr \circ \hol \big( \iota_x^*(F, \nabla^F) \big) \,.
\end{aligned}
\end{equation}
Under this reduction, the $\HVBdl^\nabla(M_p)$-module structure on $\Gamma(M_p, (\CG, \nabla^\CG))$ is mapped to the $C^\infty(\bbL_p, \FC)$-module structure on $\Gamma(\bbL_p, J)$.
That is,
\begin{equation}
	\rmW \big( F \otimes (E,\alpha) \big) = \rmW(F, \nabla^F) \cdot \rmW(E, \alpha)\,.
\end{equation}

\chapter{Conclusions}
\label{ch:conclusions}

In this thesis, we have presented a construction of a 2-Hilbert space of sections associated to a torsion bundle gerbe $(\CG, \nabla^\CG)$ on a manifold $M$.
To that end, we first had to enlarge the 2-category of bundle gerbes in Chapter~\ref{ch:bundle_gerbes} so that suitable additive structures would exist.
We examined in detail the morphism categories in the 2-categories $\BGrb^\nabla(M)$ and $\BGrb^\nabla_\rmpar(M)$ in Chapter~\ref{ch:structures_on_morphisms_of_bgrbs}, introducing an additive monoidal structure on these morphism categories and showing that they carry a rig-module category structure over the rig-category of higher functions $\BGrb^\nabla(M)(\CI_0, \CI_0)$, or $\BGrb^\nabla_\rmpar(M)(\CI_0, \CI_0)$, respectively.
Moreover, we introduced the Riesz dual functor $\Theta$ in Theorem~\ref{st:Riesz_dual_is_functorial} and the bifunctor $[-,-]$ in Theorem~\ref{st:internal_hom_of_morphisms_in_BGrb--existence_and_naturality}, and proved that these turn the rig-module actions of higher functions into closed module structures on morphism categories in $\BGrb^\nabla_{(\rmpar)}(M)$.

The structures introduced and results obtained in Chapter~\ref{ch:structures_on_morphisms_of_bgrbs} provide further progress in understanding hermitean bundle gerbes with connection as higher geometric versions of hermitean line bundles with connection.
In contrast to other approaches, we have focused on the higher vector bundle perspective rather than working with higher principal bundles, and we showed that also from this point of view, bundle gerbes, with their morphism categories from~\cite{Waldorf--More_morphisms} extended as in Chapter~\ref{ch:structures_on_morphisms_of_bgrbs}, deserve to be regarded as higher line bundles.

The biggest obstruction to a fully satisfying theory, however, is the no-go statement of Proposition~\ref{st:no-go_for_sections_of_BGrbs}.
It shows that the notion of morphism of bundle gerbes that we used here has certain limitations.
We investigated these limitations in detail in Section~\ref{sect:ways_around_the_torsion_constraint}, observing that they are, unfortunately, hard to overcome, if at all.
It would be desirable to either find a better notion of morphism of bundle gerbes that circumvents the no-go statement, or to obtain a notion of higher vector bundles in another approach to higher bundles (as in, for instance, those in~\cite{FSS--Cech_cocycles_and_diff_char_classes,FSS--Higher_stacky_perspective,NSS--Principal_infty_bundles_I,NSS--Principal_infty_bundles_II}).
Interesting progress in the latter direction has been made in~\cite{Nuiten--MSc_Thesis}.

In Chapter~\ref{ch:2Hspaces_from_bundle_gerbes} we demonstrated that the closed rig-module structure over higher functions that we constructed on sections of bundle gerbes gives rise to the structure of a 2-Hilbert space on the category of sections.
We introduced a very strong notion of 2-Hilbert space, which, we feel, encodes the various structures one would like to see on a 2-Hilbert space in a compact and practicable way.
It seems one should investigate the application of the language two-variable adjunctions, or closed module structures, to 2-Hilbert spaces more closely.
For example, we suspect that Definition~\ref{def:2Hspace} could be weakened, as some of the properties we asked for might already be implied by the use of closed module structures.

Composing the module action of higher functions with the inclusion of higher numbers as constant functions, i.e. the functor $\sfc \colon \Hilb \hookrightarrow \BGrb^\nabla_\rmpar(M)$, we obtained the structure of a 2-Hilbert space on $\Gamma(M, (\CG, \nabla^\CG))$.
We pointed out how the inner product on this 2-Hilbert space can be regarded as a complete analogue of the inner product on the space of sections $\Gamma(M, L)$ of a hermitean line bundle on $M$.

Chapter~\ref{ch:Transgression_and_reduction} contained the merging of the transgression functors from~\cite{Waldorf--Transgression_II} and~\cite{CJM--Holonomy_on_D-branes}, and their application to a geometric construction of a pushforward, or dimensional reduction functor from bundle gerbes to line bundles.
Here it would be beneficial to see an example of a decomposable torsion bundle gerbe (in the sense of Section~\ref{sect:Ex:Cup_product_lens_space_BGrbs}) on the total space of a fibre-orientable $S^1$-bundle $K \to M$ over a symplectic manifold $M$.
We almost achieved this in Sections~\ref{sect:Ex:Cup_product_lens_space_BGrbs} and~\ref{sect:dimensional_reduction}, but our base $M_p$ is three-dimensional and, thus, not symplectic.
If the base was symplectic, however, one could not only relate the mere geometric structures of a bundle gerbe on $K$ and a line bundle on $M$, but also their interpretations as (higher) prequantum line bundles.
This could, hopefully, lead to more intuition about how to conclusively set up higher geometric quantisation.
Simple examples for geometries of the desired type can be obtained, for instance, by starting from a symplectic manifold of the form $(M, \dd \eta)$ with a torsion prequantum line bundle $(L, \nabla^L)$ and then constructing the cup product bundle gerbe with connection over $M {\times} S^1$ in the fashion of Section~\ref{sect:Ex:Cup_product_lens_space_BGrbs}.
Furthermore, the pushforward functor of bundle gerbes might have interesting applications to T-duality (cf.~\cite{Bunke-Nikolaus--T-duality_via_gerbey_geometry_and_reductions}).

Finally, there are several questions unanswered about higher geometric quantisation.
For example, it is unclear how observables, i.e. higher functions, should be represented on the 2-Hilbert space of sections of $(\CG, \nabla^\CG)$.
Even more fundamentally, it is unclear what the higher version of the Poisson bracket on ordinary functions should be.
On the subcategory of higher functions which are of the form of a trivial hermitean line bundle with connection on $M$, a good candidate for such a structure has been discovered and investigated in~\cite{Rogers--L_infty_algebras}.
From our considerations, however, one has to conclude that this Lie 2-algebra structure should somehow be defined on all of $\BGrb^\nabla_{(\rmpar)}(M)$, or at least on the equivalent subcategory $\HVBdl^\nabla_{(\rmpar)}(M)(\CI_0, \CI_0)$.
Further, it is completely unclear, unfortunately, how to define a higher covariant derivative of a section of a bundle gerbe, or even ``along what'' this should be taken.
Candidates are elements of $TM$, its second exterior power $\Lambda^2 TM$, or $TM \oplus T^*M$, motivated by the Courant algebroids considered, for example, in~\cite{Rogers--2PG_Courant_algbds_and_prequan}.
Given such an operation, one could try to write down a higher Kostant-Souriau prequantisation formula.
The next, and probably biggest problem in higher geometric quantisation, which we left completely untouched here, but which has been considered in~\cite{Rogers--L_infty_algebras}, is higher polarisation.
Once more, it would be good to have a notion of higher covariant derivative at hand in order to obtain intuition about what the higher analogue of a polarisation should be.

While we have been able to categorify the construction of the prequantum Hilbert space in higher geometric quantisation, at this point there is too little background for us to be able to carry out the full higher geometric quantisation programme.
However, any progress in this direction will certainly contain leaps forward in our understanding of higher geometric structures and the structure of spacetime as modelled by string theory.


\begin{appendix}

\chapter{Special morphisms of bundle gerbes}
\label{ch:App:special_morphisms}

\section{Remarks on monoidal structures and strictness}
\label{app:monoidal_structures_and_strictness}

Several conventions adopted from~\cite{Waldorf--More_morphisms,Waldorf--Thesis} are employed in the main text to simplify our treatment of monoidal structures.
For the reader's convenience, we list here the simplifications made, so that one may keep in mind what we have chosen not to display explicitly.
These conventions greatly simplify expressions throughout the text, and especially the structural isomorphisms in the full 2-category of bundle gerbes.
We choose to work as if certain symmetric monoidal structures, outlined below, were strictly associative and unital.
By Mac Lane's coherence theorem~\cite[Section~VII]{ML--Categories_for_the_working_mathematician}, the unitors and associators can be reinserted, if desired, in a unique way.
Note though, that apart from having to add the adjective \emph{strict} at certain points in the main text (we point out below where this happens), there is no mathematical content in those isomorphisms relevant to our arguments and computations.

The most fundamental simplification we make is to work as if that the symmetric monoidal structure on $\HVBdl^\nabla$ is strictly associative and unital.
That is, we never display the isomorphisms $E \otimes I_0 \cong I_0 \cong I_0 \otimes E$, since for all our intends and purposes, these isomorphisms are irrelevant and just unnecessarily clutter explicit expressions.
Moreover, we take the same convention for the additive monoidal structure $\oplus$ on $\HVBdl^\nabla(M)$, i.e. treat it as strictly associative an unital, remembering Mac Lane's coherence theorem.

We treat analogously the category $\Mfd \ddownarrow M$ of surjective submersions over $M$.
Its objects are pairs $(Y, \pi)$ which constitute a surjective submersion $\pi \colon Y \to M$, while morphisms $(\Mfd \ddownarrow M)((Y_0, \pi_0),(Y_1,\pi_1))$ are smooth maps $f \in \Mfd(Y_0, Y_1)$ such that $\pi_1 \circ f = \pi_0$.
Like $\HVBdl^\nabla(M)$, the category $\Mfd \ddownarrow M$ is (cartesian) symmetric monoidal under the operation $(Y_0,\pi_0) \otimes (Y_1,\pi_1) \coloneqq (Y_0 {\times}_M Y_1,\, \pi_0 \circ \pr_{Y_0} = \pi_1 \circ \pr_{Y_1})$.
As in the aforementioned cases, we work as if this symmetric monoidal structure was strictly associative and unital.

Finally, we assume the convention that pullbacks of vector bundles compose strictly, i.e. that for $f \in \Mfd(M_0, M_1)$ , $g \in \Mfd(M_1, M_2)$, and $(E, \nabla^E) \in \HVBdl^\nabla(M_0)$ we have $g^* f^* (E, \nabla^E) = (g \circ f)^* (E,\nabla^E)$, rather than spelling out the isomorphism involved here (cf. Definition~\ref{def:presheaf_of_Cats}).

As a consequence of these conventions, we are forced to view the tensor product of bundle gerbes as strictly associative and unital.
Strictly speaking, however, the unitors of this tensor product on objects involve all three of the aforementioned unitors of symmetric monoidal structures.
An analogous statement holds true for the composition of 1-morphisms of bundle gerbes with connections on $M$, whose construction is very similar to the tensor product.
In order to stay consistent with our simplifying conventions, we further have to view the composition of 1-morphisms in the 2-category of bundle gerbes with connection on $M$ as strictly associative.
It is not strictly unital, though, for the non-trivial line bundles which are part of the respective bundle gerbe appear in the identity 1-morphisms.

\section{Morphisms of bundle gerbes and descent}
\label{app:special_morphisms_and_descent}

\begin{lemma}
\label{st:t_mu-lemma}
Let $(\CG,\nabla^\CG) = (L, \nabla^L, \mu, B, Y, \pi)$ be a bundle gerbe with connection on $M$.
There is a canonical isomorphism of line bundles with connection
\begin{equation}
	t_\mu \colon s_0^*L \to I_0\,,
\end{equation}
with $I_0$ denoting the trivial hermitean line bundle with connection on $Y$, given by the composition\footnote{To ease up notation, we do not display the pullback of $L$ along the inclusion $Y {\times}_Y Y \hookrightarrow Y^{[2]}$.}
\begin{equation}
	s_0^*L
	\xrightarrow{~1 \otimes \delta_{s_0^*L}^{-1}~} s_0^*L \otimes s_0^*L \otimes s_0^* L^*
	\xrightarrow{~s_0^*s_0^*\mu \otimes \One~} s_0^*L \otimes s_0^*L^*
	\xrightarrow{~\delta_{s_0^*L}^{-1}~} I_0\,.
\end{equation}
It satisfies
\begin{equation}
\label{eq:t_mu--properties}
\begin{aligned}
	d_1^*t_\mu \otimes 1_L &= s_0^* \mu \colon d_1^* s_0^* L \otimes L \to L \quad \text{over } Y^{[2]}\,,
	\\
	1_L \otimes d_0^*t_\mu &= s_1^* \mu \colon L \otimes d_0^* s_0^*L \to L \quad \text{over } Y^{[2]}\,,
	\\
	t_\mu^{-{\sft}} &= t_{\mu^{-{\sft}}} \quad \text{over } Y\,.
\end{aligned}
\end{equation}
\end{lemma}

\begin{proof}
Well-definedness of $t_\mu$ as well as the first two identities in~\eqref{eq:t_mu--properties} are proven in \cite[Lemma 2.1.3]{Waldorf--Thesis}.
The last identity follows from
\begin{equation}
	\big( \delta_L^{-\sft} (z \cdot 1_{L^*}) \big) \big( \delta_L (z' \cdot 1_L) \big)
	= z\, z'
	= \big( \delta_{L^*} (z \cdot 1_{L^*}) \big) \big( \delta_L (z' \cdot 1_L) \big)\,,
\end{equation}
showing that $\delta_L^{-\sft} = \delta_{L^*}$ as required.
\end{proof}

For several constructions regarding morphisms of bundle gerbes we will need descent expressions for morphisms.
Given a surjective submersion $\pi \colon Y \to M$, there is a functor $\Asc_\pi \colon \HVBdl^\nabla(M) \to \Desc(\HVBdl^\nabla,\pi)$ to the descent category of the relative sheaf of categories $\HVBdl^\nabla$ as described in Section~\ref{sect:Coverings_and_sheaves_of_cats}.
It acts as $(E,\nabla^E) \mapsto (\pi^*(E,\nabla^E), \alpha_E)$, where $\alpha_E \colon d_0^*\pi^*E \to d_1^*\pi^*E$ is the canonical isomorphism over $Y^{[2]}$ which identifies fibres of the pullback over points in $Y$ with the same image in $M$.%
\footnote{Recall from Section~\ref{sect:Coverings_and_sheaves_of_cats} that, strictly speaking, we should write $\Desc(\HVBdl^\nabla, \HVBdl^\nabla_{\rmuni \sim},\pi)$ here, since $\HVBdl^\nabla$ is a relative sheaf of categories with respect to the sheaf of groupoids $\HVBdl^\nabla_{\rmuni \sim}$.}
The functor $\Asc_\pi$ is a weak inverse to the descent functor $\Desc_\pi \colon \Desc(\HVBdl^\nabla,\pi) \to \HVBdl^\nabla(M)$, establishing the fact that $\HVBdl^\nabla$ is a relative sheaf of categories on the Grothendieck site $(M, \tau_\ssub)$.
In the following, we approach the statement that for any pair of bundle gerbes, the category $\BGrb^\nabla(M)((\CG_0,\nabla^{\CG_0}), (\CG_1,\nabla^{\CG_1}))$ is a sheaf of categories over $Y_{01} = Y_0 {\times}_M Y_1$.

\begin{lemma}
\label{st:dd-def_and_properties}
Let $(E, \nabla^E, \alpha, Z, \zeta) \in \BGrb^\nabla(M)((\CG_0,\nabla^{\CG_0}),(\CG,\nabla^{\CG_1}))$.
There is a unitary parallel isomorphism of hermitean vector bundles with connection $\dd_{(E,\alpha)} \colon p_1^*E \to p_0^*E$ over $Z {\times}_{Y_{01}} Z$ defined by 
\begin{equation}
\label{eq:def_dd-morphism_for_1morph_descent}
	\dd_{(E,\alpha)} = \big( \zeta_{Y_1}^*t_{\mu_1} \otimes 1_E \big) \circ \big(\alpha_{\vert Z{\times_{Y_{01}}}Z} \big) \circ \big( 1_E \otimes \zeta_{Y_0}^*t_{\mu_0} \big)^{-1} 
\end{equation} 
with the following properties:	
\begin{myenumerate}
	\item Over $(z_0,z_1,z_2)\in Z {\times_{Y_{01}}} Z {\times_{Y_{01}}} Z$, the isomorphism $\dd_{(E,\alpha)}$ satisfies a cocycle relation
	\begin{equation}
	\label{eq:dd-cocycle_condition}
		\dd_{(E,\alpha)|(z_0,z_1)} \circ \dd_{(E,\alpha)|(z_1,z_2)}
		=\dd_{(E,\alpha)|(z_0,z_2)}\,.
	\end{equation}

	\item Over $(Z {\times}_M Z) {\times}_{Y_{01}^{[2]}} (Z {\times}_M Z)$, there is a commutative diagram
	\begin{equation}
	\label{eq:comp_alpha_dd}
	\xymatrixcolsep{2cm}
	\xymatrixrowsep{1.5cm}
	\xymatrix{
		\pr_{23}^* \zeta_{Y_0}^{[2]*}L_0 \otimes \pr_{13}^*p_1^*E \ar@{->}[r]^-{\pr_{23}^*\alpha} \ar@{->}[d]_-{1 \otimes \pr_{13}^*\dd_{(E,\alpha)}} & \pr_{02}^*p_1^*E \otimes \pr_{23}^* \zeta_{Y_1}^{[2]*}L_1 \ar@{->}[d]^-{\pr_{02}^*\dd_{(E,\alpha)} \otimes 1}
		\\
		\pr_{01}^* \zeta_{Y_0}^{[2]*}L_0 \otimes \pr_{13}^*p_0^*E \ar@{->}[r]_-{\pr_{01}^*\alpha} & \pr_{02}^*p_0^*E \otimes \pr_{01}^* \zeta_{Y_1}^{[2]*}L_1
	}
	\end{equation}
	In other words, $\alpha$ is a descent morphism
	\begin{equation}
	\begin{aligned}
		\alpha \in \Desc&\big(\HVBdl^\nabla, \zeta {\times}_M \zeta\big)
		\\
		&\Big( \big( \zeta_{Y_0}^{[2]*} L_0 \otimes \pr_1^*E,\, 1 \otimes \pr_{13}^*\dd_{(E,\alpha)} \big),\, \big( \pr_0^*E \otimes \zeta_{Y_1}^{[2]*} L_1,\, \pr_{02}^* \dd_{(E,\alpha)} \big) \Big)\,.
	\end{aligned}
	\end{equation}

	\item The isomorphism $\dd_{(E,\alpha)}$ satisfies the identity
	\begin{equation}
		\dd_{(E,\alpha)}^{-{\sft}} = \dd_{(E^*,\alpha^{-{\sft}})}\,.
	\end{equation}
\end{myenumerate}
\end{lemma}

\begin{proof}
Statements (1) and (2) can be found as \cite[Lemma 2.1.4]{Waldorf--Thesis}.
	Item (3) follows from Lemma~\ref{st:t_mu-lemma}.
\end{proof}

\begin{lemma}
	\label{st:dd_and_2-morphisms}
	For every 2-morphism $[W,\omega,\psi] \colon (E, \nabla^E, \alpha, Z, \zeta) \to (E', \nabla^{E'}, \alpha', Z', \zeta')$, there is a commuting diagram over $W {\times}_{Y_{01}} W$:
	\begin{equation}
	\xymatrixrowsep{1.5cm}
	\xymatrixcolsep{3.5cm}
	\xymatrix{
		\pr_1^* \omega_Z^*E \ar@{->}[r]^-{(\omega_Z {\times}_{Y_{01}} \omega_Z)^*\dd_{(E,\alpha)}} \ar@{->}[d]_-{\pr_1^*\psi}	& \pr_0^*\omega_Z^* E \ar@{->}[d]^-{\pr_0^*\psi}
		\\
		\pr_1^* \omega_{Z'}^* E' \ar@{->}[r]_-{(\omega_{Z'} {\times}_{Y_{01}} \omega_{Z'})^* \dd_{(E',\alpha')}}
		& \pr_0^* \omega_{Z'}^*E'
	}
	\end{equation}
	That is, under the isomorphism $Y_{01} \cong Y_{01} {\times}_{Y_{01}} Y_{01}$,
	\begin{equation}
\begin{aligned}
		\psi \in\Desc&\big( \HVBdl^\nabla, (\zeta' {\times}_{Y_{01}} \zeta) \circ \omega \big)
		\\
		&\big( (\omega_Z^*E, (\omega_Z {\times}_{Y_{01}} \omega_Z)^*\dd_{(E,\alpha)}),\, (\omega_{Z'}^*E', (\omega_{Z'}  {\times}_{Y_{01}} \omega_{Z'})^*\dd_{(E',\alpha')}) \big)\,.
\end{aligned}
	\end{equation}
\end{lemma}
\begin{proof}
Let $(w_0,w_1) \in W {\times}_{Y_{01}} W$.
Without explicitly displaying the pullbacks, we have
\begin{equation}
\begin{aligned}
	&\psi_{|w_0} \circ \dd_{(E,\alpha)|(w_0,w_1)}
	\\*
	&= \psi_{|w_0} \circ \big( \pr_{Y_1}^*t_{\mu_1} \otimes 1_E \big)_{|w_1} \circ \big(\alpha_{\vert Z{\times_{Y_{01}}}Z}\big)_{|(w_0,w_1)} \circ \big( 1_E \otimes \pr_{Y_0}^*t_{\mu_0} \big)^{-1}_{|w_0}
	\\
	&= \big( \pr_{Y_1}^*t_{\mu_1} \otimes 1_{E'} \big)_{|w_1} \circ (1 \otimes\psi_{|w_0}) \circ \big(\alpha_{\vert Z{\times_{Y_{01}}}Z}\big)_{|(w_0,w_1)} \circ \big( 1_E \otimes \pr_{Y_0}^*t_{\mu_0} \big)^{-1}_{|w_0}
	\\
	&= \big( \pr_{Y_1}^*t_{\mu_1} \otimes 1_{E'} \big)_{|w_1} \circ \big(\alpha'_{\vert Z' {\times_{Y_{01}}} Z'}\big)_{|(w_0,w_1)} \circ (\psi_{|w_1}\otimes 1) \circ \big( 1_E \otimes \pr_{Y_0}^*t_{\mu_0} \big)^{-1}_{|w_0}
	\\*
	&= \dd_{(E',\alpha')|(w_0,w_1)}\circ \psi_{|w_1}\,,
\end{aligned}
\end{equation}
where we used that the 2-morphisms intertwine $\alpha$ and $\alpha'$, i.e. equation~\eqref{eq:2-morphism_compatibility_with_alphas}.
\end{proof}

The following two statements are refinements of~\cite[Theorem 2.4.1]{Waldorf--Thesis}.

\begin{proposition}
\label{st:2-morphisms_have_simple_representatives}
Every 2-morphism $[W,\omega,\psi] \colon (E, \nabla^E, \alpha, Z, \zeta) \to (E', \nabla^{E'}, \alpha', Z', \zeta')$ has a unique representative of the form $(Z {\times}_{Y_{01}} Z', 1, \phi)$.
\end{proposition}

\begin{proof}
Over the diagonal $Z {\times}_Z Z \subset Z^{[2]}$, we have $\dd_{(E,\alpha)|Z {\times}_Z Z} = 1_{E|Z {\times}_Z Z}$.
This follows from the cocycle relation $\dd_{(E,\alpha)|(z,z)} \circ \dd_{(E,\alpha)|(z,z)} = \dd_{(E,\alpha)|(z,z)}$, composed by $\dd_{(E,\alpha)|(z,z)}^{-1}$.
We set $\hat{Z} = Z {\times}_{Y_{01}} Z'$, so that we obtain restrictions $\hat{\omega}_Z \coloneqq \omega^{[2]}_{Z|W {\times}_{\hat{Z}} W} \colon W {\times}_{\hat{Z}} W \to Z {\times}_Z Z \cong Z$ and $\hat{\omega}_{Z'} \coloneqq \omega^{[2]}_{Z'|W {\times}_{\hat{Z}} W} \colon W {\times}_{\hat{Z}} W \to Z' {\times}_{Z'} Z' \cong Z'$.
By the above observation, we have that
\begin{equation}
\begin{aligned}
	&\big( \omega_Z^*(E, \nabla^E),\, (\omega_Z \times_Z \omega_Z)^* \dd_{(E,\alpha)} \big)
	= \Asc_{\omega_Z}( E, \nabla^E)
	= \Asc_{\omega}(\pr_Z^*(E,\nabla^E))\,,
	\\
	&\big( \omega_{Z'}^*(E', \nabla^{E'}),\, (\omega_{Z'} \times_{Z'} \omega_{Z'})^* \dd_{(E',\alpha')} \big)
	= \Asc_{\omega_{Z'}}( E', \nabla^{E'})
	= \Asc_\omega( \pr_{Z'}^*(E', \nabla^{E'}))\,.
\end{aligned}
\end{equation}
Thus, Lemma~\ref{st:dd_and_2-morphisms} implies that $\psi$ is a descent morphism
\begin{equation}
	\psi \in \Desc \big( \HVBdl^\nabla, \omega \big) \big( \Asc_{\omega}(\pr_Z^*(E,\nabla^E)), \Asc_\omega( \pr_{Z'}^*(E', \nabla^{E'})) \big)\,.
\end{equation}
As $\Asc_\omega$ is an equivalence of categories, there exists a unique morphism of vector bundles $\phi \in \HVBdl^\nabla(Z {\times}_M Z') (\pr_Z^*(E,\nabla^E), \pr_{Z'}^*(E', \nabla^{E'}))$ such that $\omega^*\phi = \Asc_\omega \phi = \psi$.
This morphism is, in fact, a 2-morphism of bundle gerbes:
substituting $\psi = \omega^*\phi$ into the compatibility condition~\eqref{eq:2-morphism_compatibility_with_alphas} of $\psi$ with $\alpha$ and $\alpha'$, this relation becomes the pullback along $\omega^{[2]}$ of the compatibility condition of $\phi$.
Now, since $\omega$ is surjective, the relation for $\phi$ follows.
Moreover, $\phi$ represents the same 2-morphism as $(W,\omega,\psi)$.
To see this, choose $(X,\chi)$ with $X = W$ and surjective submersions to $Z {\times}_{Y_{01}} Z'$ and $W$ to be $\omega$ and $1_W$, respectively.
Finally, let $(W',\omega',\psi')$ be another representative of $[W,\omega,\psi]$.
By the above construction, it gives rise to a morphism of hermitean vector bundles with connection $\phi'$ over $Z {\times}_{Y_{01}} Z'$, intertwining $\alpha$ and $\alpha'$, and such that $\omega^{\prime*}\phi' = \psi'$ over $W'$.
Let $(X,\chi_W, \chi_{W'})$ establish the equivalence $(W,\omega,\psi) \sim (W',\omega',\psi')$.
Then, $\chi_{W'}^* \omega^{\prime*}\phi' = \chi_{W'}^* \psi' = \chi_W^*\psi = \chi_W^* \omega^* \phi$, whence, since all maps $\chi_W$, $\chi_{W'}$, $\omega$ and $\omega'$ are surjective (or since $\Asc_\chi$ is an equivalence), we have $\phi' = \phi$.
\end{proof}

\begin{theorem}
\label{st:sheaf of categories_of_BGrb_morphisms}
For any pair $(\CG_i, \nabla^{\CG_i}) = (L_i, \nabla^{L_i}, \mu_i, B_i, Y_i, \pi_i)$, with $i = 0,1$, of bundle gerbes on $M$, there exists a relative sheaf of categories $\BGrb^\nabla_\sh(M)((\CG_0, \nabla^{\CG_0}),(\CG_1, \nabla^{\CG_1}))$ on $Y_{01}$, which assigns to a surjective submersion $\zeta \colon Z \to Y_{01}$ the bundle gerbe morphisms from $(\CG_0, \nabla^{\CG_0})$ to $(\CG_1, \nabla^{\CG_1})$ with surjective submersion $\zeta$, as well as the 2-morphisms between these.
\end{theorem}

\begin{proof}
From Lemma~\ref{st:dd-def_and_properties} (1) we see that any bundle gerbe morphism $(E, \nabla^E, \alpha, Z, \zeta)$ gives rise to a descent vector bundle $\Desc_\zeta(E,\nabla^E) \in \HVBdl^\nabla(Y_{01})$, stemming from the descent data defined using $\dd_{(E,\alpha)}$.
Statement (2) of the same lemma implies that there is a morphism $\Desc_\zeta(\alpha) \colon \pr_{Y_0}^{[2]*}L_0 \otimes d_1^*\Desc_\zeta(E,\nabla^E) \to d_0^*\Desc_\zeta(E,\nabla^E) \otimes \pr_{Y_0}^{[2]*} L_1$ over $Y_{01}$.
The compatibility of $\Desc_\zeta$ with the $\mu_i$ follows from the functoriality of $\Desc_\zeta$, which transforms the required identity into $\Desc_\zeta$ applied to the compatibility relation for $\alpha$ itself, together with the fact that $\Desc_\zeta$ is fully faithful.
Proposition~\ref{st:2-morphisms_have_simple_representatives} first implies that any 2-morphism $(E,\alpha) \to (E', \alpha')$ of morphisms defined over $\omega \colon W \to \hat{Z} = Z {\times}_{Y_{01}} Z'$ has a representative $\phi$ over $\hat{Z}$, and Lemma~\ref{st:dd_and_2-morphisms} together with the full faithfulness of $\Desc_\zeta$ then shows that $[\hat{Z}, 1, \Desc_\zeta \phi]$ still provides a 2-morphism of bundle gerbes.
The functor that sends $(F, \nabla^{F}, \beta, Y_{01}, 1_{Y_{01}})$ to $(\zeta^*(F,\nabla^{F}), \zeta^{[2]*}\beta, \zeta, Z)$ provides a weak inverse for the descent functor described above.
\end{proof}

We denote by $\BGrb_\FP^\nabla(M)((\CG_0,\nabla^{\CG_0}), (\CG_1,\nabla^{\CG_1})) \subset \BGrb^\nabla(M)((\CG_0,\nabla^{\CG_0}), (\CG_1,\nabla^{\CG_1}))$ the full subcategory which has 1-morphisms defined only over $Y_{01}$.
In~\cite{Waldorf--Thesis}, this is also required to have only those 2-morphisms which are defined over $Y_{01}^{[2]}$, but we know already from Proposition~\ref{st:2-morphisms_have_simple_representatives} that this is the case for all 2-morphisms.

\begin{corollary}[{\cite[Theorem 2.4.1]{Waldorf--Thesis}}]
\label{st:BGrb_FP_hookrightarrow_BGrb_is_equivalence}
The inclusion
\begin{equation}
	\sfS: \BGrb_\FP^\nabla(M) \big( (\CG_0,\nabla^{\CG_0}), (\CG_1,\nabla^{\CG_1}) \big) \overset{\cong}{\hookrightarrow} \BGrb^\nabla(M) \big( (\CG_0,\nabla^{\CG_0}), (\CG_1,\nabla^{\CG_1}) \big)
\end{equation}
is an equivalence of categories.
It has a canonical inverse functor $\sfR$ induced by the descent functor of $\HVBdl^\nabla$ as in the proof of Theorem~\ref{st:sheaf of categories_of_BGrb_morphisms}.
\end{corollary}

\section{Mutual surjective submersions}
\label{app:BGrbs_with_mutual_surjective_submersions}

Consider two bundle gerbes $(\CG_i, \nabla^{\CG_i}) = (L_i, \nabla^{L_i}, \mu_i, B_i, Y_i, \pi_i) \in \BGrb^\nabla(M)$, for $i = 0,1$, in the special case where $(Y_0, \pi_0) = (Y_1, \pi_1) = (Y,\pi)$, i.e. where the source and target bundle gerbes are defined over the same surjective submersion.
For a 1-morphism of the form $(E, \nabla^E, \alpha, Y^{[2]}, 1_{Y^{[2]}}) \in \BGrb^\nabla(M)((\CG_0, \nabla^{\CG_0}), (\CG_1, \nabla^{\CG_1}))$ we set
\begin{equation}
	\sfF (E, \nabla^E, \alpha, Y^{[2]}, 1_{Y^{[2]}})
	= \big( s_0^*(E, \nabla^E),\, s_0^{[2]*} \alpha \big)\,,
\end{equation}
where $s_0 \colon Y \to Y^{[2]}$ is the diagonal map.
This comes with a unitary, parallel isomorphism of hermitean vector bundles with connection
\begin{equation}
	s_0^{[2]*} \alpha \colon L_0 \otimes d_0^*(s_0^*E) \to d_1^*(s_0^*E) \otimes L_1
\end{equation}
over $Y^{[2]}$, which satisfies the relation
\begin{equation}
\label{eq:reduced_1-morphism_condition}
	(1 \otimes \mu_1) \circ \big( d_2^*(s_0^{[2]*} \alpha) \otimes 1 \big) \circ \big( 1 \otimes d_0^* (s_0^{[2]*} \alpha) \big)
	= d_1^*(s_0^{[2]*} \alpha) \circ (\mu_0 \otimes 1)
\end{equation}
over $Y^{[3]}$.
This is the pullback along $s_0^{[3]} \colon Y^{[3]} \to (Y^{[2]})^{[3]} \cong Y^{[6]}$ of the compatibility relation~\eqref{eq:1-morphisms_compatibility_with_BGrb_multiplications}.
Similarly, for $(E', \nabla^{E'}, \alpha', Y^{[2]}, 1_{Y^{[2]}}) \in \BGrb^\nabla(M)((\CG_0, \nabla^{\CG_0}), (\CG_1, \nabla^{\CG_1}))$ and a 2-morphism $[Y^{[2]}, 1_{Y^{[2]}}, \psi] \colon (E, \alpha) \to (E', \alpha')$, we set
\begin{equation}
	\sfF[Y^{[2]}, 1_{Y^{[2]}}, \psi] = s_0^*\psi \colon s_0^*(E,\nabla^E) \to s_0^*(E', \nabla^{E'})\,.
\end{equation}
This satisfies
\begin{equation}
\label{eq:reduced_2-morphism_condition}
	(d_1^*\psi \otimes 1) \circ (s_0^{[2]*} \alpha)
	= (s_0^{[2]*} \alpha') \circ (1 \otimes d_0^*\psi)
\end{equation}
over $Y^{[2]}$, which is the pullback of~\eqref{eq:2-morphism_compatibility_with_alphas} along $s_0^{[2]} \colon Y^{[2]} \to Y^{[4]}$.

\begin{definition}[Twisted vector bundles]
Let $(\CG_i, \nabla^{\CG_i}) = (L_i, \nabla^{L_i}, \mu_i, B_i, Y_i, \pi_i) \in \BGrb^\nabla(M)$, for $i = 0,1$, with $(Y_0, \pi_0) = (Y_1, \pi_1) = (Y,\pi)$.
We define a category with objects being triples $(F, \nabla^F, \beta)$, where $(F,\nabla^F) \in \HVBdl^\nabla(Y)$ and $\beta$ is a parallel unitary isomorphism $L_0 \otimes d_0^*F \to d_1^*F \otimes L_1$, satisfying~\eqref{eq:reduced_1-morphism_condition}, and morphisms $(F, \nabla^F, \beta) \to (F', \nabla^{F'}, \beta')$ given by morphisms of hermitean vector bundles $\phi \colon F \to F'$ satisfying~\eqref{eq:reduced_2-morphism_condition}.
This category is denoted $\HVBdl^\nabla((\CG_0, \nabla^{\CG_0}), (\CG_1, \nabla^{\CG_1}))$ and we refer to it as the category of \emph{$((\CG_0, \nabla^{\CG_0}){-}(\CG_1, \nabla^{\CG_1}))$-twisted hermitean vector bundles with connections}.
\end{definition}

The nomenclature is in accordance with previous literature on twisted K-theory (see, for instance, \cite{Karoubi--Twisted_bundles_and_twisted_K-theory}).
We have constructed a functor
\begin{equation}
	\sfF \colon \BGrb^\nabla_\FP(M)((\CG_0, \nabla^{\CG_0}), (\CG_1, \nabla^{\CG_1})) \to \HVBdl^\nabla((\CG_0, \nabla^{\CG_0}), (\CG_1, \nabla^{\CG_1}))\,,
\end{equation}
and we obtain a functor $\sfG$ in the opposite direction by setting
\begin{equation}
	\sfG (F, \nabla^F, \beta) = \big( (L_0, \nabla^{L_0}) \otimes d_0^*(F, \nabla^F), \sfG\beta, Y^{[2]}, 1_{Y^{[2]}} \big)\,,
	\quad
	\sfG \phi = 1 \otimes d_0^*\phi\,,
\end{equation}
where $\sfG \beta$ is the composition
\begin{equation}
	\begin{tikzcd}[column sep=2cm, row sep=1.5cm]
		L_{0|(y_0, y_2)} \otimes (L_{0|(y_2, y_3)} \otimes F_{|y_3}) \ar[r, "\mu_{0|(y_0,y_2,y_3)} \otimes 1"] \ar[d, dashed, "{\sfG \beta}"']& L_{0|(y_0,y_3)} \otimes F_{|y_3} \ar[d, "\mu_{0|(y_0,y_1,y_3)}^{-1} \otimes 1"]
		\\
		(L_{0|(y_0,y_1)} \otimes F_{|y_1}) \otimes L_{1|(y_1,y_3)} & L_{0|(y_0, y_1)} \otimes L_{0|(y_1, y_3)} \otimes F_{|y_3} \ar[l, "1 \otimes \beta_{|(y_1,y_3)}"]
	\end{tikzcd}
\end{equation}
for $y_0,\ldots,y_3 \in Y^{[4]}$.
The isomorphism $t_{\mu_0}$ induces a natural isomorphism $\sfF \circ \sfG \to 1$, while a natural isomorphism $\eta \colon \sfG \circ \sfF \to 1$ is provided by $\eta_{(E,\alpha)} = (1 \otimes t_{\mu_1}^{-1}) \circ \alpha$ (with pullbacks from $Y {\times}_M (Y {\times}_Y Y {\times}_Y Y) \subset Y^{[4]}$ omitted).
Thus, we conclude the following simpler description of the morphism categories of bundle gerbes over a mutual surjective submersion in terms of twisted vector bundles:

\begin{proposition}
\label{st:morphism_categories_and_twisted_HVBdls_for_same_sur_sub}
For two bundle gerbes with connections $(\CG_i, \nabla^{\CG_i})$ on $M$ which are defined with respect to the same surjective submersions, there are equivalences of categories
\begin{equation}
	\BGrb^\nabla(M)(\CG_0, \CG_1)
	\cong \BGrb^\nabla_\FP(M)(\CG_0, \CG_1)
	\cong \HVBdl^\nabla(\CG_0, \CG_1)
\end{equation}
(where we have refrained from displaying the connections for spatial reasons).
In particular,
\begin{equation}
	\BGrb^\nabla(M)(\CI_0, \CI_0) \cong \HVBdl^\nabla(M)\,.
\end{equation}
\end{proposition}

Note that the equivalences in Proposition~\ref{st:morphism_categories_and_twisted_HVBdls_for_same_sur_sub} restrict to $\BGrb^\nabla_{\rmflat}(M)$, $\BGrb^\nabla_{\rmflat \sim}(M)$, or $\BGrb(M)$, with the respective modifications of the target category of twisted vector bundles.

Consider a family $(\CG_i, \nabla^{\CG_i})$, $i = 0, \ldots, m$, of bundle gerbes over $M$, all defined with respect to a mutual surjective submersion $\pi \colon Y \to M$.
Let $\Lambda \in \Set$, and write $Y^{[\Lambda]}$ for the $\Lambda$-indexed fibre product of $Y$ over $M$.
We assume there are unitary, parallel isomorphisms $\gamma_{ij} \colon L_j \to L_i$ in $\HLBdl^\nabla(Y^{[2]})$ which satisfy a cocycle condition and intertwine the bundle gerbe multiplications, i.e. we have
\begin{equation}
\label{eq:compatibility relations for naive morphisms}
	\gamma_{ij} \circ \gamma_{jk} = \gamma_{ik}\,, \qquad 
	\mu_i \circ (p_{01}^*\gamma_{ij} \otimes p_{12}^*\gamma_{ij}) = p_{02}^*\gamma_{ij} \circ \mu_j \quad \forall\, i,j,k \in \{0, \ldots, m\}\,.
\end{equation}
We define an auxiliary 2-category $\scC_\Lambda$ as follows.
First, we set $\obj(\scC_\Lambda) = \Lambda$.
For $a,b \in \Lambda$ and $n \in \NN_0$, a 1-morphism in $\scC_\Lambda(a,b)$ is a chain
\begin{equation}
	p_{a,c_0}^*L_{i_0} \otimes p_{c_0, c_1}^*L_{i_1} \otimes \cdots \otimes p_{c_{n-1}, c_n}^*L_{i_n} \otimes p_{c_n, b}^*L_{i_{n+1}} \in \HLBdl^\nabla(Y^{[\Lambda]})\,,
\end{equation}
for $i_0, \ldots, i_{n+1} \in \{0, \ldots, m \}$ and $c_0, \ldots, c_{n} \in \Lambda$, or a one-element chain $p_{ab}^*L_i$.%
\footnote{If we were adhering to the weak associativity of the tensor product here, we should allow for different bracketings in the 1-morphisms, and associators in the 2-morphisms.}
The identity 1-morphism on $a$ is $p_{aa}^*L_0$.
As 2-morphisms between these chains we allow any morphisms of hermitean line bundles over $Y^{[\Lambda]}$ constructed from tensor products and compositions of pullbacks of the $\gamma_{ij}$ as well as the $\mu_i$ and identity morphisms of hermitean line bundles.
Composition of 1-morphisms and horizontal composition of 2-morphisms is given by concatenation of chains, vertical composition is induced from that in $\HLBdl^\nabla(Y^{[\Lambda]})$.
In particular, $\scC_\Lambda$ is a 2-groupoid, i.e. has only invertible 1-morphisms and 2-morphisms.
The crucial observation here is that because of the associativity of the bundle gerbe multiplications and their compatibility with the $\gamma_{ij}$, for every 1-morphism $a \to b$ there is a unique 2-isomorphism to $p_{ab}^*L_0$.
Consequently, any two 1-morphisms in $\scC_\Lambda$ are 2-isomorphic by a unique 2-isomorphism, which is any composition of the $\gamma_{ij}$ and $\mu_i$ with the correct domain and codomain 1-morphisms.

Define a strict auxiliary 2-category as follows.
Let $\scC^\rmd_\Lambda$ denote the discrete 2-category over the free category on $\Lambda$.
That is, $\scC^\rmd_\Lambda$ has objects $a \in \Lambda$.
A 1-morphism $a \to b$ in $\scC^\rmd_\Lambda$ is a finite word $(a, a'_0, \ldots, a'_k, b) = (a, \vec{a}', b)$ for $k \in \NN$, or the word $(a,b)$.
The identity 1-morphism on $a$ is the word $(a,a)$, and composition is given by $(b, \vec{b}', c) \circ (a, \vec{a}', b) = (a, \vec{a}', \vec{b}', c)$.
All 2-morphisms are identities.
We then have a 2-functor
\begin{equation}
	\widehat{L} \colon \scC^\rmd_\Lambda \to \scC_\Lambda\,, \quad
	a \mapsto a\,, \quad
	(a, a'_0, \ldots, a'_k, b) \mapsto p_{a, a'_0}^*L_0 \otimes \ldots \otimes p_{a'_k, b}^*L_0\,,
\end{equation}
and sending identity 2-morphisms to identities 2-morphisms.
This is not a strong 2-functor: there are natural 2-isomorphisms
\begin{equation}
	\widehat{L} \big( (b, \vec{b}', c) \circ (a, \vec{a}', b) \big) \to \widehat{L}(b, \vec{b}', c) \circ \widehat{L}(a, \vec{a}', b)
\end{equation}
given by the unique 2-morphisms between these 1-morphisms.
Because of the coherence conditions on the $\gamma_{ij}$ and the $\mu_i$, $\widehat{L}$ is fully faithful and essentially surjective in the 2-categorical sense, and, thus, a 2-equivalence~\cite{Leinster--Basic_bicategories}.
This is very similar to the statement of Mac Lane's famous coherence theorem~\cite[Section VII.2, Theorem 2]{ML--Categories_for_the_working_mathematician}.

With the above constructions at hand, we turn to \emph{naive isomorphisms}.
These consist of isomorphisms of the line bundles over $Y^{[2]}$ underlying the bundle gerbes that are compatible with the bundle gerbe multiplications in the sense of~\eqref{eq:compatibility relations for naive morphisms}.

\begin{proposition}
\label{st:From_naive_isomorphisms_to_1-isomorphisms--same_sursub}
Let $(\CG_i, \nabla^{\CG_i}) \in \BGrb^\nabla(M)$, $i = 0,1,2,3$ be defined over a mutual surjective submersion $(Y_i, \pi_i) = (Y, \pi)$.
\begin{myenumerate}
	\item If $\gamma \colon L_0 \to L_1$ is an isomorphism in $\HLBdl^\nabla_\rmuni(Y^{[2]})$ which satisfies $\mu_1 \circ (p_{01}^*\gamma \otimes p_{12}^*\gamma) = p_{02}^*\gamma \circ \mu_0$ over $Y^{[3]}$, there exists a 1-isomorphism canonically associated to $\gamma$, given by
	\begin{equation}
		\sfJ(\gamma) \coloneqq \big( L_0, \nabla^{L_0}, (1 \otimes p_{13}^*\gamma) \circ p_{013}^*\mu_0^{-1} \circ p_{023}^*\mu_0, Y^{[2]}, \pi_M \big)\,.
	\end{equation}
	Note that $\sfJ(1_{L_i}) = 1_{\CG_i}$.
	
	\item If $\epsilon \colon L_1 \to L_2$ is a second such isomorphism, there exists a canonical 2-isomorphism
		\begin{equation}
			j_{\epsilon, \gamma} \colon \sfJ(\epsilon) \circ \sfJ(\gamma) \to \sfJ(\epsilon \circ \gamma)\,, \quad 
			j_{\epsilon, \gamma} = \big( Y^{[3]}, 1_{Y^{[3]}}, \mu_0 \circ (p_{12}^*\gamma^{-1} \otimes 1) \big)\,.
		\end{equation}
	For $\rho \colon L_2 \to L_3$ another isomorphism as above, this satisfies
	\begin{equation}
		j_{\rho, \epsilon \circ \gamma} \circ_2 (1_{J(\rho)} \circ_1 j_{\epsilon, \gamma}) = j_{\rho \circ \epsilon, \gamma} \circ_2 (j_{\rho, \epsilon} \circ_1 1_{J(\gamma)})\,.
	\end{equation}
	
	\item Moreover, with $\lambda$ and $\rho$ denoting the left and right unitors in $\BGrb^\nabla(M)$, respectively (see Example~\ref{eg:2-morphisms--new_from_old_identity_unitors}), we have
	\begin{equation}
			j_{\gamma, 1_{L_0}} = \lambda_{J(\gamma)}\,, \qquad
			j_{1_{L_1}, \gamma} = \rho_{J(\gamma)}\,.
	\end{equation}
\end{myenumerate}
\end{proposition}

\begin{proof}
The proof is a consequence of the discussion preceding the proposition.
All consistencies which have to be checked amount to comparing morphisms constructed as compositions of tensor products of $\gamma_{ij}$ and $\mu_i$.
That is, we are only comparing morphisms in categories $\scC_\Lambda$, where $\Lambda$ has to be taken as an appropriate finite subset of $\NN$.
Regarding (2) and (3) it is important to note that the 1-morphisms involved are still of the form $(E,\alpha)$ with $E$ a tensor product of pullbacks of line bundles $L_i$ and $\alpha$ constructed as a composition of tensor products of $\mu_i$ and their inverses.
Hence, they still lie within the scope of the above abstract argument and are, hence, unique by the properties of $\scC_\Lambda$.
\end{proof}

If we denote by $\BGrb^\nabla(Y \overset{\pi}{\to} M)$ the symmetric monoidal groupoid of bundle gerbes on $M$ which have surjective submersion $(Y,\pi)$ and morphisms given by naive isomorphisms, we have thus proven the following theorem.

\begin{theorem}
The assignments $\sfJ$ and $j$ combine to form a 2-functor
\begin{equation}
	\BGrb^\nabla(Y \overset{\pi}{\to} M) \hookrightarrow \BGrb^\nabla(M)
\end{equation}
which includes the category of bundle gerbes defined over $(Y, \pi)$ with naive morphisms (after adding identity 2-morphisms) into the 2-category $\BGrb^\nabla(M)$.
\end{theorem}

\chapter{Proofs}
\label{ch:App:proofs}

\section{Proof of Theorem~\ref{st:direct_sum_structure_on_morphisms_of_BGrbs}}
\label{app:proof_of_direct_sum_theorem}

In this appendix we prove Theorem~\ref{st:direct_sum_structure_on_morphisms_of_BGrbs}.
The first  step is to check that~\eqref{eq:direct_sum_on_1-morphisms} defines a 1-morphisms in $\BGrb^\nabla(M)$.
The only non-trivial thing to do here is to check that $\alpha^\oplus \coloneqq \sfd_r^{-1} \circ ( \pr_{Z'}^{[2]*} \alpha' \oplus \pr_Z^{[2]*} \alpha) \circ \sfd_l$ satisfies the compatibility condition with the bundle gerbe multiplications from Definition~\ref{def:1-morphisms_of_BGrbs}.
That is, we need to show that
\begin{equation}
	\big( 1 \otimes \hat{\zeta}_{Y_1}^{[3]*} \mu_1 \big) \circ (d_2^*\alpha^\oplus \otimes 1) \circ (1 \otimes d_0^* \alpha^\oplus)
	= d_1^*\alpha^\oplus \circ \big( \hat{\zeta}_{Y_0}^{[3]*} \mu_0 \otimes 1 \big)\,.
\end{equation}
An explicit manipulation yields\\
\addtocounter{equation}{1}
\begin{align*}
\label{eq:direct_sum_of_1-morphisms_is_1-morphism}
	&\big( 1 \otimes \hat{\zeta}_{Y_1}^{[3]*} \mu_1 \big) \circ (d_2^*\alpha^\oplus \otimes 1) \circ (1 \otimes d_0^* \alpha^\oplus)
	\\*[0.2cm]
	&= 	\big( 1 \otimes \hat{\zeta}_{Y_1}^{[3]*} \mu_1 \big) \circ \Big( d_2^* \big( \sfd_r^{-1} \circ ( \pr_{Z'}^{[2]*} \alpha' \oplus \pr_Z^{[2]*} \alpha) \circ \sfd_l \big) \otimes 1 \Big)
	\\*
	&\qquad \circ \Big( 1 \otimes d_0^*\big( \sfd_r^{-1} \circ ( \pr_{Z'}^{[2]*} \alpha' \oplus \pr_Z^{[2]*} \alpha) \circ \sfd_l \big) \Big)
	\\[0.2cm]
	&= \big( 1 \otimes \hat{\zeta}_{Y_1}^{[3]*} \mu_1 \big) \circ (\sfd_r^{-1} \otimes 1) \circ \sfd_r^{-1} \circ \Big( \big( d_2^*\pr_{Z'}^{[2]*} \alpha' \otimes 1 \big) \oplus \big( d_2^*\pr_Z^{[2]*} \alpha \otimes 1 \big) \Big)
	\\*
	&\qquad \circ \Big( \big( 1 \otimes d_0^*\pr_{Z'}^{[2]*} \alpha' \big) \oplus \big( 1 \otimes d_0^*\pr_Z^{[2]*} \alpha \big) \Big) \circ \sfd_l \circ (1 \otimes \sfd_l)
	\\[0.2cm]
	&= \sfd_r^{-1} \circ \Big( \pr_{Z'}^{[3]*} \Big( \big( 1 \otimes \zeta_{Y_1}^{\prime[3]*} \mu_1 \big) \circ \big( d_2^*\alpha' \otimes 1 \big) \circ \big( 1 \otimes d_0^*\alpha' \big) \Big) \theeq
	\\*
	&\qquad \oplus \pr_Z^{[3]*} \Big( \big( 1 \otimes \zeta_{Y_1}^{[3]*} \mu_1 \big) \circ \big( d_2^* \alpha \otimes 1 \big) \circ \big( 1 \otimes d_0^* \alpha \big) \Big) \Big) \circ \sfd_l \circ (1 \otimes \sfd_l)
	\\[0.2cm]
	&= \sfd_r^{-1} \circ \Big( \pr_{Z'}^{[3]*} \Big( d_1^* \alpha' \circ \big( \zeta_{Y_0}^{\prime[3]*} \mu_0 \otimes 1 \big) \Big) \oplus \pr_Z^{[3]*} \Big( d_1^*\alpha \circ \big( \zeta_{Y_0}^{[3]*} \mu_0 \otimes 1 \big) \Big) \Big) \circ \sfd_l \circ (1 \otimes \sfd_l)
	\\[0.2cm]
	&= \sfd_r^{-1} \circ \big( d_1^*\pr_{Z'}^{[2]*} \alpha' \oplus d_1^*\pr_Z^{[2]*} \alpha \big) \circ \sfd_l \circ \big( \hat{\zeta}_{Y_0}^{\prime[3]*} \mu_0 \otimes 1 \big)
	\\*
	&= d_1^*\alpha^\oplus \circ \big( \hat{\zeta}_{Y_0}^{[3]*} \mu_0 \otimes 1 \big)\,,
\end{align*}
as required.

Next, we have to check that the sum of two 2-morphisms
\begin{equation}
\begin{aligned}
	&[Z {\times}_{Y_{01}} X,\, 1,\, \psi] \colon (E, \nabla^E, \alpha, Z, \zeta) \to (F, \nabla^F, \beta, X, \chi)\,,
	\\
	&[Z' {\times}_{Y_{01}} X',\, 1,\, \psi'] \colon (E', \nabla^{E'}, \alpha', Z', \zeta') \to (F', \nabla^{F'}, \beta', X', \chi')
\end{aligned}
\end{equation}
as defined in~\eqref{eq:direct_sum_on_2-morphisms} is in fact a 2-morphism in $\BGrb^\nabla(M)$.
Note that equivalences of representatives of the original morphisms induce equivalences of the direct sum from~\eqref{eq:direct_sum_on_2-morphisms} via their fibre product over $Y_{01}$, so that it is sufficient to use representatives of 2-morphisms over the minimal surjective submersion (where we appeal to Proposition~\ref{st:2-morphisms_have_simple_representatives}).
The direct sum then reads
\addtocounter{equation}{1}
\begin{align*}
	&[Z' {\times}_{Y_{01}} X',\, 1,\, \psi'] \oplus [Z {\times}_{Y_{01}} X,\, 1,\, \psi] \theeq
	\\
	&= \big[ Z' {\times}_{Y_{01}} X' {\times}_{Y_{01}} Z {\times}_{Y_{01}} X,\, 1,\, \pr_{(Z' {\times}_{Y_{01}} X')}^* \psi' \oplus \pr_{(Z {\times}_{Y_{01}} X)}^* \psi \big]\,,
\end{align*}
and the compatibility with the structural morphisms of the source and target 1-morphisms in $\BGrb^\nabla(M)$ is checked analogously to the computation in~\eqref{eq:direct_sum_of_1-morphisms_is_1-morphism}, by distributing the structural morphisms onto the pullbacks of $\psi$ and $\psi'$, then using the compatibilities of the original 2-morphisms, and finally distributing back to obtain the desired expression.
We have thus shown that the direct sums of 1-morphisms and 2-morphisms in $\BGrb^\nabla(M)$ exist as claimed in Theorem~\ref{st:direct_sum_structure_on_morphisms_of_BGrbs}.

The second step is to show that $\oplus$ defines a functor with respect to vertical composition of 2-morphisms.
Consider another pair of 2-morphisms
\begin{equation}
\begin{aligned}
	&[X {\times}_{Y_{01}} U,\, 1,\, \phi] \colon (F, \nabla^F, \beta, X, \chi) \to (G, \nabla^G, \gamma, U, \xi)\,,
	\\
	&[X' {\times}_{Y_{01}} U',\, 1,\, \phi'] \colon (F', \nabla^{F'}, \beta', X', \chi') \to (G', \nabla^{G'}, \gamma', U', \xi')
\end{aligned}
\end{equation}
between the same bundle gerbes.
We have
\begin{equation}
\begin{aligned}
	&[X {\times}_{Y_{01}} U,\, 1,\, \phi] \circ_2 [Z {\times}_{Y_{01}} X,\, 1,\, \psi]
	\\
	&= \big[ Z {\times}_{Y_{01}} X {\times}_{Y_{01}} U,\, 1,\, \pr_{(X {\times}_{Y_{01}} U)}^*\phi \circ \pr_{(Z {\times}_{Y_{01}} X)}^*\psi \big]\,,
\end{aligned}
\end{equation}
so that
\addtocounter{equation}{1}
\begin{align*}
	&\Big( [X' {\times}_{Y_{01}} U',\, 1,\, \phi'] \circ_2 [Z' {\times}_{Y_{01}} X',\, 1,\, \psi] \Big)
	\oplus \Big( [X {\times}_{Y_{01}} U,\, 1,\, \phi] \circ_2 [Z {\times}_{Y_{01}} X,\, 1,\, \psi] \Big)
	\\
	&= \big[ Z' {\times}_{Y_{01}} X' {\times}_{Y_{01}} U' {\times}_{Y_{01}} Z {\times}_{Y_{01}} X {\times}_{Y_{01}} U,\, 1, \theeq
	\\*
	&\qquad \pr_{(X' {\times}_{Y_{01}} U')}^*\phi' \circ \pr_{(Z' {\times}_{Y_{01}} X')}^*\psi' \oplus \pr_{(X {\times}_{Y_{01}} U)}^*\phi \circ \pr_{(Z {\times}_{Y_{01}} X)}^*\psi \big]
	\\
	&= \Big( [X' {\times}_{Y_{01}} U',\, 1,\, \phi'] \oplus [X {\times}_{Y_{01}} U,\, 1,\, \phi] \Big)
	\circ_2 \Big( [Z' {\times}_{Y_{01}} X',\, 1,\, \psi] \oplus [Z {\times}_{Y_{01}} X,\, 1,\, \psi] \Big)\,.
\end{align*}
Note that the composition of the direct sums takes place over $X' {\times}_{Y_{01}} X$ and we have to use an equivalence of representatives of 2-morphisms which reorders the factors in the six-fold fibre product over $Y_{01}$.
This hows that the direct sum is compatible with vertical composition.

Regarding identity 2-morphisms, note that the morphisms $\dd_{(E,\alpha)}$ defined in Lemma~\ref{st:dd-def_and_properties} satisfy
\begin{equation}
	\dd_{(E', \alpha') \oplus (E,\alpha)} = \pr_{(Z' {\times}_{Y_{01}} Z')}^* \dd_{(E',\alpha')} \oplus \pr_{(Z {\times}_{Y_{01}} Z)}^* \dd_{(E,\alpha)}\,.
\end{equation}
This readily implies that $1_{(E',\alpha')} \oplus 1_{(E,\alpha)} = 1_{(E',\alpha') \oplus (E, \alpha)}$, and, consequently, that $\oplus$ is a functor.

The functor $\oplus$ is built from fibre products and the direct sum of vector bundles, whence it is associative and unital with respect to the unit element given by the zero 1-morphism $0 = (0, 1_0, Y_{01}, 1_{Y_{01}} ) \in \BGrb^\nabla(M)((\CG_0, \nabla^{\CG_0}), (\CG_1, \nabla^{\CG_1}))$.
Here, $1_0$ is the identity morphism on the rank-zero hermitean vector bundle with connection.
The symmetry of $\oplus$ is implemented by
\begin{equation}
\begin{aligned}
	\sigma_{(E',\alpha'),(E,\alpha)} &= \big[ Z' {\times}_{Y_{01}} Z {\times}_{Y_{01}} Z {\times}_{Y_{01}} Z',\, 1,
	\\
	&\qquad \big( \pr_{(Z {\times}_{Y_{01}} Z)}^* \dd_{(E,\alpha)} \oplus \pr_{(Z' {\times}_{Y_{01}} Z')}^* \dd_{(E',\alpha')} \big) \circ \sw^\oplus \big]\,,
\end{aligned}
\end{equation}
with $\sw^\oplus$ denoting the swap of the summands in a direct sum in $\HVBdl^\nabla$.
This makes $\oplus$ into a symmetric monoidal structure on $\BGrb^\nabla(M)((\CG_0, \nabla^{\CG_0}), (\CG_1, \nabla^{\CG_1}))$, thus completing the proof of (1).

We now proceed to part (2).
The additivity of $\oplus$ with respect to the sum from Proposition~\ref{st:additive_structure_on_2-morphisms_of_BGrbs} is straightforward.
Regarding distributivity, let $(E,\alpha), (E', \alpha) \colon (\CG_0, \nabla^{\CG_0}) \to (\CG_1, \nabla^{\CG_1})$ and $(F, \nabla^F, \beta, X, \xi) \colon (\CG_2, \nabla^{\CG_2}) \to (\CG_3, \nabla^{\CG_3})$ be 1-morphisms in $\BGrb^\nabla(M)$.
First, note that $(F, \beta) \otimes \big( (E',\alpha') \oplus (E, \alpha) \big)$ is defined over $X {\times}_M (Z' {\times}_{Y_{01}} Z)$, while $\big( (F, \beta) \otimes (E', \alpha') \big) \oplus \big( (F, \beta) \otimes (E, \alpha) \big)$ is defined over $(X {\times}_M Z') {\times}_{Y_{0123}} (X {\times}_M Z)$.
We can, thus, define a 2-morphism $\delta_{l, (F,E',E)}$ over
\begin{equation}
	W = \big( X {\times}_M (Z' {\times}_{Y_{01}} Z) \big) {\times}_{Y_{0123}} \big( (X {\times}_M Z') {\times}_{Y_{0123}} (X {\times}_M Z) \big)\,,
\end{equation}
using the identity $1_W$ as a surjective submersion and using the unitary, parallel isomorphism of hermitean vector bundles with connections given by (omitting the pullbacks)
\begin{equation}
	\widehat{\delta}_{l,(F,E',E)} = \Big( \big( \dd_{(F, \beta)} \otimes \dd_{(E',\alpha')} \big) \oplus \big( \dd_{(F, \beta)} \otimes \dd_{(E,\alpha)} \big) \Big) \circ \sfd_l\,.
\end{equation}
Here, $\sfd_l$ distributes the pullback of $F$ over the sum of the pullbacks of $E'$ and $E$ from the left.
Using the commutation relations of morphisms of vector bundles with the distributivity isomorphisms $\sfd_l$, together with part (2) of Proposition~\ref{st:dd-def_and_properties} then yields the necessary compatibility of $\widehat{\delta}_{l,(F,E',E)} $ with $\beta$, $\alpha'$ and $\alpha$, thus proving that $\delta_{l, (F,E',E)}= [W, 1, \widehat{\delta}_{l,(F,E',E)}]$ is a 2-isomorphism as desired.
The naturality of $\delta_l = \delta_{l, (-,-,-)}$ follows by analogous arguments from Lemma~\ref{st:dd_and_2-morphisms}.
This completes the proof of the second claim in (2).
The third claim, i.e. the distributivity from the right, is analogous.
Finally, one can check that $\delta_{l, (F,-,-)}$ turns $(F,\beta) \otimes (-)$ into a monoidal functor.

To see (3), let us not consider $(E,\alpha), (E', \alpha) \colon (\CG_0, \nabla^{\CG_0}) \to (\CG_1, \nabla^{\CG_1})$ as before, and $(F, \nabla^F, \beta, X, \xi) \colon (\CG_1, \nabla^{\CG_1}) \to (\CG_2, \nabla^{\CG_2})$.
Now, $(F, \beta) \circ \big( (E',\alpha') \oplus (E,\alpha) \big)$ is defined over $(Z' {\times}_{Y_{01}} Z) {\times}_{Y_1} X$, while the morphism $\big( (F, \beta) \circ (E',\alpha') \big) \oplus \big( (F, \beta) \circ (E,\alpha) \big)$ is defined over $(Z' {\times}_{Y_1} X) {\times}_{Y_{012}} (Z {\times}_{Y_1} X)$.
We thus set
\begin{equation}
	W' = \big( (Z' {\times}_{Y_1} X) {\times}_{Y_{012}} (Z {\times}_{Y_1} X) \big) {\times}_{Y_{012}} \big( (Z' {\times}_{Y_{01}} Z) {\times}_{Y_1} X \big)\,.
\end{equation}
Over this space, we use a very similar expression as for the distributivity of the tensor product:
\begin{equation}
	\widehat{\sfc}_{l,(F,E',E)} = \big( \dd_{(F, \beta) \circ (E',\alpha')} \oplus \dd_{(F, \beta) \circ (E,\alpha)} \big) \circ \sfd_r\,,
\end{equation}
where $\sfd_r$ distributes the pullback of $F$ over the sum of the pullbacks of $E'$ and $E$.
This is, indeed, compatible with the structural isomorphisms:
we have, omitting pullbacks,
\addtocounter{equation}{1}
\begin{align*}
	&(\widehat{\sfc}_{l,(F,E',E)} \otimes 1_{L_2}) \circ \big( (1 \otimes \beta) \circ (\alpha^\oplus \otimes 1) \big)
	\\*
	&= \Big( \big( ( \dd_{(F, \beta) \circ (E',\alpha')} \oplus \dd_{(F, \beta) \circ (E,\alpha)}) \circ \sfd_{r,F} \big) \otimes 1_{L_2} \Big)
	\\
	&\qquad \circ \Big( (1_{E' \oplus E} \otimes \beta) \circ \big( (\sfd_{r,L_1})^{-1} \circ (\alpha' \oplus \alpha) \circ \sfd_{l,L_0} \big) \otimes 1_F \big) \Big)
	\\
	&= \Big( \big( \dd_{(F, \beta) \circ (E',\alpha')} \oplus \dd_{(F, \beta) \circ (E,\alpha)} \big) \circ \sfd_{r, F} \otimes 1_{L_2} \Big) \theeq
	\\
	&\qquad \circ \big( (1_{E' \oplus E} \otimes \beta) \circ (\sfd_{r,L_1 \otimes F})^{-1} \circ (\alpha' \otimes 1_F \oplus \alpha \otimes 1_F) \circ \sfd_{l,L_0} \big) \circ (1_{L_0} \otimes \sfd_{r,F})
	\\
	&= (\sfd_{r, L_2})^{-1} \circ \big( (\dd_{(F, \beta) \circ (E',\alpha')} \otimes 1_{L_2}) \oplus (\dd_{(F, \beta) \circ (E,\alpha)} \otimes 1_{L_2}) \big) \circ \sfd_{r, L_2}
	\\
	&\qquad \circ (\sfd_{r,L_2})^{-1} \circ \Big( \big( (1_{E'} \otimes \beta) \circ (\alpha' \otimes 1_F) \big) \oplus \big( (1_E \otimes \beta) \circ (\alpha \otimes 1_F)\big) \Big) \circ \sfd_{l,L_0} \circ (1_{L_0} \otimes \sfd_{r,F})
	\\
	&= (\sfd_{r, L_2})^{-1} \circ \Big( \big( (1_{E'} \otimes \beta) \circ (\alpha' \otimes 1_F) \big) \oplus \big( (1_E \otimes \beta) \circ (\alpha \otimes 1_F)\big) \Big)
	\\
	&\qquad \circ \big( (1_{L_0} \otimes \dd_{(F, \beta) \circ (E',\alpha')}) \oplus ( 1_{L_0} \otimes \dd_{(F, \beta) \circ (E,\alpha)}) \big) \circ \sfd_{l,L_0} \circ (1_{L_0} \otimes \sfd_{r,F})
	\\
	&= (\sfd_{r, L_2})^{-1} \circ \Big( \big( (1_{E'} \otimes \beta) \circ (\alpha' \otimes 1_F) \big) \oplus \big( (1_E \otimes \beta) \circ (\alpha \otimes 1_F)\big) \Big) \circ \sfd_{l,L_0}
	\\
	&\qquad \circ \Big( 1_{L_0} \otimes \big( \dd_{(F, \beta) \circ (E',\alpha')} \oplus \dd_{(F, \beta) \circ (E,\alpha)} \big) \circ \sfd_{r,F} \Big)
	\\
	&= (\sfd_{r, L_2})^{-1} \circ \Big( \big( (1_{E'} \otimes \beta) \circ (\alpha' \otimes 1_F) \big) \oplus \big( (1_E \otimes \beta) \circ (\alpha \otimes 1_F)\big) \Big) \circ \sfd_{l,L_0}
	\\
	&\qquad \circ (1_{L_0} \otimes \widehat{\sfc}_{l,(F,E',E)})\,,
\end{align*}
which is the desired compatibility condition (compare~\eqref{eq:2-morphism_compatibility_with_alphas}).
The crucial ingredient in this computation is the commutative diagram in part (2) of Lemma~\ref{st:dd-def_and_properties}, which makes the fourth equality possible.
Naturality again follows from Lemma~\ref{st:dd_and_2-morphisms}.
This completes the proof of the first half of (3); the second half is analogous.
We have, thus, completed the proof of Theorem~\ref{st:direct_sum_structure_on_morphisms_of_BGrbs}.\qed

\section{Proof of Theorem~\ref{st:internal_hom_of_morphisms_in_BGrb--existence_and_naturality}}
\label{app:Proof_of_adjunction_theorem}

Here we prove Theorem~\ref{st:internal_hom_of_morphisms_in_BGrb--existence_and_naturality}.
We use the conventions and nomenclature introduced there.

The right adjoint acts on a 1-morphism $(G, \nabla^G, \gamma, U, \chi) \colon (\CG_0, \nabla^{\CG_0}) \otimes (\CG_2, \nabla^{\CG_2}) \to (\CG_1, \nabla^{\CG_1}) \otimes (\CG_3, \nabla^{\CG_3})$ as (cf.~\eqref{eq:internal_hom_on_1-morphisms})
\begin{align}
	\big[ (F, \beta), (G,\gamma) \big]
	&= \Big( \pr_U^* (G, \nabla^G) \otimes \pr_X^* (F, \nabla^F)^*,
	\\*
	&\qquad (1 \otimes \pr_{Y_3}^{[2]*}\delta_{L_3}) \circ ( \pr_U^{[2]*} \gamma \otimes \pr_X^{[2]*} \beta^{-\sft}) \circ (1 \otimes \pr_{Y_1}^{[2]*}\delta_{L_1}^{-1}),\,
	\notag\\*
	&\qquad X {\times}_{Y_{13}} U,\, \chi_{Y_{02}} \circ \pr_U \Big)\,, \notag
\end{align}
The compatibility with $\mu_0$ and $\mu_2$ follows from the identity $\delta_{L_i} \circ (\mu_i^{-\sft} \otimes \mu_i) \circ (\delta_{L_i} \otimes \delta_{L_i}) = 1_{I_0}$.
Hence, $[(F, \beta), (G,\gamma)]$ is a 1-morphism as claimed.

Defining the action on 2-morphisms is complicated by the fact that $(E,\alpha)$, $(F, \beta)$ and $(G, \gamma)$ are all defined with respect to different surjective submersions.
In order to circumvent these complications, we employ the inclusion
\begin{equation}
	\sfS \colon \BGrb^\nabla_\FP(M) \big( (\CG_i, \nabla^{\CG_i}), (\CG_j, \nabla^{\CG_j}) \big) \overset{\cong}{\hookrightarrow}
	\BGrb^\nabla(M) \big( (\CG_i, \nabla^{\CG_i}), (\CG_j, \nabla^{\CG_j}) \big)
\end{equation}
from Theorem~\ref{st:BGrb_FP_hookrightarrow_BGrb_is_equivalence}, with its canonical inverse given by the functor $\sfR$.
Let $\eta \colon 1 \to \sfR \circ \sfS$ and $\epsilon \colon \sfS \circ \sfR \to 1$ be the unit and counit of the equivalence $(\sfS, \sfR)$, respectively.
Consider a second 1-morphism $(G', \nabla^{G'}, \gamma', U', \chi') \colon (\CG_0, \nabla^{\CG_0}) \otimes (\CG_2, \nabla^{\CG_2}) \to (\CG_1, \nabla^{\CG_1}) \otimes (\CG_3, \nabla^{\CG_3})$ and a 2-morphism $\widehat{\psi} = [U {\times}_{Y_{0123}} U',\, 1,\, \psi] \colon (G, \gamma) \to (G', \gamma')$.
This gives rise to a 2-morphism
\begin{equation}
	\sfR \widehat{\psi} = [Y_{0123}, 1, \sfR\psi] \colon \sfR(G, \gamma) \to \sfR(G', \gamma')
\end{equation}
living in $\BGrb^\nabla_\FP(M)(\sfR(G, \gamma), \sfR(G', \gamma'))$.
Observe that
\begin{equation}
	\big[ Y_{0123}, \pr_{Y_{02}}, \sfR \psi \otimes 1_{{\pr_{Y_{13}}^*}F^*} \big] \colon \big[ \sfR(F, \beta), \sfR(G, \gamma) \big] \to \big[ \sfR(F, \beta), \sfR(G', \gamma') \big]
\end{equation}
is a 2-morphism in $\BGrb^\nabla(M)$.
This can either be checked explicitly, or it can be observed from $1_{{\pr_{Y_{13}}^*}F^*} = \pr_{Y_{13}}^* \dd_{\sfR \Theta(F, \beta)}$ and the compatibility of $\dd_{\sfR \Theta(F, \beta)}$ with $\beta^{-\sft}$ (cf. Lemma~\ref{st:dd-def_and_properties}).
Therefore, we know that this 2-morphism induces a 2-morphism
\begin{equation}
	[ Y_{02}, 1, \Desc(\sfR \psi \otimes 1_{F^*})] \colon \sfR \big[ \sfR(F, \beta), \sfR(G, \gamma) \big] \to \sfR \big[ \sfR(F, \beta), \sfR(G', \gamma') \big]\,.
\end{equation}
Here and in the remainder of this proof we will suppress indication of the surjective submersion whose descent or reduction functor is used, as it is evident from context.
Observe that, since descent of vector bundles is compatible with duals and tensor products, and since $\sfR$ is an equivalence, there are canonical 2-isomorphisms
\begin{equation}
\begin{aligned}
	&\sfR \big[ \sfR(F, \beta), \sfR(G, \gamma) \big] \cong \sfR[(F, \beta), (G, \gamma)]\,,
	\\
	&\sfR \big[ \sfR(F, \beta), \sfR(G', \gamma') \big] \cong \sfR[(F, \beta), (G', \gamma')]\,.
\end{aligned}
\end{equation}
Adjoining these to $[ Y_{02}, 1, \Desc(\sfR \psi \otimes 1_{F^*})]$, we obtain a 2-isomorphism
\begin{equation}
	\sfF_{[F,-]} \widehat{\psi} \colon \sfR[(F, \beta), (G, \gamma)] \to \sfR[(F, \beta), (G', \gamma')]\,.
\end{equation}
Moreover, the assignment
\begin{equation}
\begin{aligned}
	&\sfF_{[F,-]} \colon \BGrb^\nabla(M)( \CG_0 \otimes \CG_1,\, \CG_2 \otimes \CG_3)
	\to \BGrb^\nabla_\FP(M)( \CG_0, \CG_2),\,
	\\
	&\Big( (G, \gamma) \overset{\widehat{\psi}}{\longrightarrow} (G', \gamma') \Big)
	\longmapsto
	\Big( \sfR[(F, \beta), (G, \gamma)] \overset{\sfF_{[F,-]} \widehat{\psi}}{\longrightarrow} \sfR[(F, \beta), (G', \gamma')]  \Big)
\end{aligned}
\end{equation}
is functorial since $\sfR$ and $\Desc$ are functorial.
We define $[(F, \beta), \widehat{\psi}]$ to be the dashed arrow in
\begin{equation}
\begin{tikzcd}[column sep=2cm, row sep=1.25cm]
	{[(F, \beta), (G, \gamma)]} \ar[r, dashed, "{[(F, \beta), \widehat{\psi}]}"] \ar[d, "\epsilon^{-1}_{[F,G]}"'] & {[(F, \beta), (G', \gamma')]}
	\\
	{\sfS\sfR [(F, \beta), (G, \gamma)]} \ar[r, "\sfF_{[F,-]} \widehat{\psi}"'] & {\sfS \sfR [(F, \beta), (G', \gamma')]} \ar[u, "\epsilon_{[F,G']}"']
\end{tikzcd}
\end{equation}

To see that $(-) \otimes (F, \beta) \dashv [(F, \beta), - ]$, we define a bijection
\begin{equation}
	\tau^F_{F,G} \colon \BGrb^\nabla(M) \big( (E, \alpha) \otimes (F, \beta) ,\, (G, \gamma) \big) \longrightarrow
	\BGrb^\nabla(M) \big( (E, \alpha),\, [(F, \beta), (G, \gamma)] \big)
\end{equation}
as follows.
Let $\widehat{\phi} = [(Z {\times}_M X) {\times}_{Y_{0123}} U, 1, \phi] \colon (E, \alpha) \otimes (F, \beta) \to (G, \gamma)$ be a 2-morphism.
It gives rise to a 2-morphism
\begin{equation}
	\sfR \widehat{\phi} = [Y_{0123}, 1, \sfR\phi] \colon \sfR(E, \alpha) \otimes \sfR(F, \beta) \to \sfR(G, \gamma)\,,
\end{equation}
defined over the identity surjective submersion on $Y_{0123}$.
In particular, $\sfR \phi \colon \pr_{Y_{02}}^*\sfR E \otimes \pr_{Y_{13}}^*\sfR F \to \sfR G$.
Define a morphism of hermitean vector bundles with connection via
\begin{equation}
	\tau(\sfR \phi) \colon \pr_{Y_{02}}^*\sfR E \to \pr_{Y_{13}}^*\sfR F^* \otimes \sfR G\,,
	\quad
	\tau(\sfR \phi) (e)\, (f) \coloneqq \sfR \phi (e \otimes f)\,,
\end{equation}
where $e \otimes f \in \pr_{Y_{02}}^*\sfR E \otimes \pr_{Y_{13}}^*\sfR F$.
This yields a 2-morphism
\begin{equation}
	\big[ Y_{0123}, \pr_{Y_{02}}, \tau(\sfR \phi) \big] \colon \sfR(E,\alpha) \to [ \sfR(F, \beta), \sfR(G, \gamma)]\,.
\end{equation}
Applying $\sfR$ and composing by the 2-isomorphism $[ \sfR(F, \beta), \sfR(G, \gamma)] \cong \sfR [(F, \beta), (G, \gamma)]$, we obtain a 2-morphism
$\check{\tau}^F_{E,G} (\widehat{\phi}) \colon \sfR(E,\alpha) \to \sfR [(F, \beta), (G, \gamma)]$, and we can define $\tau^F_{E,G}$ as the dashed arrow in
\begin{equation}
\label{eq:def of tauFEG}
\begin{tikzcd}[column sep=1.75cm, row sep=1.25cm]
	(E,\alpha) \ar[r, dashed, "{\tau^F_{E,G} \widehat{\phi}}"] \ar[d, "{\epsilon^{-1}_E}"'] & {[(F, \beta),(G, \gamma)]}
	\\
	\sfS \sfR(E, \alpha) \ar[r, "{\check{\tau}^F_{E,G}(\widehat{\phi})}"'] & \sfS \sfR {[(F, \beta), (G, \gamma)]} \ar[u, "\epsilon_{[F,G]}"']
\end{tikzcd}
\end{equation}
This defines a bijection since the functors $\sfS$ and $\sfR$ are equivalences of categories, and the operation $\sfR \phi \mapsto \tau(\sfR\phi)$ is a bijection.

In fact, $\tau$ establishes the internal hom adjunction in $\HVBdl^\nabla$, and is, therefore, natural with respect to composition.
Explicitly, for $\psi_j$, $j = 0, \ldots, 3$ morphisms of hermitean vector bundles with connections such that $\psi_3 \circ \psi_2 \circ (\psi_1 \otimes \psi_0)$ makes sense, we have
\begin{equation}
	\tau \big( \psi_3 \circ \psi_2 \circ (\psi_1 \otimes \psi_0) \big) =
	(\psi_3 \otimes \psi_0^\sft) \circ \tau(\psi_2) \circ \psi_1\,.
\end{equation}
Moreover, note that, by the functoriality of $\sfR$,
\begin{equation}
	\sfR (1_{(F, \beta)})
	= \sfR [X {\times}_{Y_{13}} X, 1, \dd_{(F, \beta)}]
	= 1_{\sfR(F, \beta)}
	= [ Y_{13}, 1, 1_F ]\,.
\end{equation}
Let $(E', \nabla^{E'}, \alpha', Z', \zeta') \colon (\CG_0, \nabla^{\CG_0}) \to (\CG_2, \nabla^{\CG_2})$ and let $[Z {\times}_{Y_{02}} Z, 1, \nu] \colon (E',\alpha') \to (E', \alpha')$ be a 2-morphism.
We obtain a commutative diagram over $Y_{0123}$, which reads as
\begin{equation}
\label{eq:[-,-]_naturality_diagram_for_tau}
\begin{tikzcd}[column sep=5cm, row sep=1.5cm]
	\pr_{Y_{02}}^*\sfR E' \ar[r, "\tau \big( \sfR \psi \circ \sfR \phi \circ (\sfR \nu \otimes 1_{\sfR F}) \big)"] \ar[d, "\sfR \nu"'] & \pr_{Y_{13}}^* \sfR F^* \otimes \sfR G'
	\\
	\pr_{Y_{02}}^*\sfR E \ar[r, "\tau(\sfR \phi)"'] & \pr_{Y_{13}}^* \sfR F^* \otimes \sfR G' \ar[u, "1_{\sfR F^*} \otimes \sfR \psi"']
\end{tikzcd}
\end{equation}
Observing that $[ Y_{0123}, \pr_{Y_{02}}, \sfR \psi \circ \sfR \phi \circ (\sfR \nu \otimes 1_{\sfR F})]$ is a descent to $Y_{0123}$ of $\widehat{\psi} \circ \widehat{\phi} \circ (\widehat{\nu} \otimes 1_{(F, \beta)})$ and applying descent along the projection $\pr_{Y_{02}} \colon Y_{0123} \to Y_{02}$ to~\eqref{eq:[-,-]_naturality_diagram_for_tau}, we obtain the diagram
\begin{equation}
\label{eq:naturality of check tau}
\begin{tikzcd}[column sep=5cm, row sep=1.5cm]
	\sfS \sfR (E', \alpha') \ar[r, "\check{\tau}^F_{E',G'} \big( \widehat{\psi} \circ \widehat{\phi} \circ (\widehat{\nu} \otimes 1_{(F, \beta)}) \big)"] \ar[d, "\sfR \widehat{\nu}"'] & \sfR {[(F, \beta), (G', \gamma')]}
	\\
	\sfS \sfR (E, \alpha) \ar[r, "\check{\tau}^F_{E,G} \widehat{\phi}"'] & \sfR {[(F, \beta), (G, \gamma)]} \ar[u, "{[(F, \beta), \widehat{\psi}]}"']
\end{tikzcd}
\end{equation}
This implies that the following diagram commutes:
\begin{equation}
\begin{scriptsize}
\begin{tikzcd}[column sep=1.25cm, row sep=1cm]
	(E', \alpha') \ar[rrrr, "\tau^F_{E',G'} \big( \widehat{\psi} \circ \widehat{\phi} \circ (\widehat{\nu} \otimes 1_F) \big)"] \ar[dddd, "\epsilon_{E'}^{-1}"'] \ar[dr, "\widehat{\nu}"'] & & & & {[(F, \beta), (G', \gamma')]}
	\\
	& (E,\alpha) \ar[rr, "\tau^F_{E,G} (\widehat{\phi})"] \ar[dd, "\epsilon_E^{-1}"'] & & {[(F, \beta), (G, \gamma)]}  \ar[ur, "{[(F, \beta), \widehat{\psi}]}"'] &
	\\
	& & & &
	\\
	& \sfS \sfR (E,\alpha) \ar[rr, "\check{\tau}^F_{E,G} (\widehat{\phi})"'] & & \sfS \sfR {[(F, \beta), (G, \gamma)]} \ar[uu, "\epsilon_{[F,G]}"'] \ar[dr, "{\sfS \sfR [(F, \beta), \widehat{\psi}]}"] &
	\\
	\sfS \sfR (E', \alpha') \ar[ur, "\sfS \sfR \widehat{\nu}"] \ar[rrrr, "\check{\tau}^F_{E',G'} \big( \widehat{\psi} \circ \widehat{\phi} \circ (\widehat{\nu} \otimes 1_F) \big)"'] & & & & \sfS \sfR {[(F, \beta), (G', \gamma')]} \ar[uuuu, "\epsilon_{[F,G']}"']
\end{tikzcd}
\end{scriptsize}
\end{equation}
Here, the left and right quadrangles commute by the naturality of $\epsilon$ and~\eqref{eq:def of tauFEG}, the outer and centre squares commutes by the definition~\eqref{eq:def of tauFEG} of $\tau^F_{E,G} (\widehat{\phi})$, while the commutativity of the bottom quadrangle is diagram~\eqref{eq:naturality of check tau}.
This completes the proof of the naturality of $\tau^F_{E,G}$.
\qed

\end{appendix}

\phantomsection
\addcontentsline{toc}{chapter}{References}
\bibliographystyle{amsplainurl}
\bibliography{Thesis_Bib}

\end{document}